\documentclass[11pt, twoside]{article}
\pdfoutput=1
\usepackage[numbers, sort&compress, square]{natbib}
\usepackage[english]{babel}
\usepackage[utf8]{inputenc}
\usepackage[font=small,labelfont=bf]{caption}
\usepackage{latexsym}
\usepackage{marvosym}
\usepackage{color}
\usepackage{amsfonts}
\usepackage{amsmath}
\usepackage{amsthm}
\usepackage{amssymb}
\usepackage{fancyhdr}
\usepackage{esdiff}
\usepackage{empheq}
\usepackage{graphicx}
\usepackage{mathrsfs}
\usepackage{wrapfig}
\usepackage[all]{xy}
\usepackage{comment}
\usepackage{enumerate}
\usepackage[babel]{csquotes}
\usepackage{geometry}
\usepackage{subcaption}
\usepackage{multirow}
\usepackage{tabularx}
\usepackage{bm}

\usepackage{./mystyle2}

\usepackage{xcolor}
\usepackage{epstopdf}
\usepackage
{hyperref}
\usepackage{tikz}
\usepackage{tcolorbox}
\usetikzlibrary{decorations.markings}
\usetikzlibrary{arrows,matrix}
\usetikzlibrary{calc}
\usepackage{marvosym}
\usetikzlibrary{patterns}

\fancypagestyle{subsection}{
\fancyhf{}

\fancyhead[
RO]
{\nouppercase{\small\leftmark}}
\fancyhead[
LE]
{\nouppercase{\small\rightmark}}
\fancyfoot[RO,LE]{\thepage}}

\fancypagestyle{section}{
\fancyhf{}

\fancyhead[LO,RE]{}
\fancyhead[RO,LE]{\nouppercase{\small\leftmark}}
\fancyfoot[RO,LE]{\thepage}}

\fancypagestyle{plain}{
\fancyhf{}

\fancyhead[LO,RE]{}
\fancyhead[RO,LE]{}
\fancyfoot[RO,LE]{\thepage}}

\def\beq{\begin{equation}\begin{aligned}}
\def\eeq{\end{aligned}\end{equation}}

\newcommand*{\warning}{{\fontencoding{U}\fontfamily{futs}\selectfont\char 66\relax}}

\definecolor{blue}{rgb}{0,0,0.8}
\definecolor{blue1}{rgb}{0,0,1.5}
\definecolor{red}{rgb}{0.85,0,0}
\definecolor{red1}{rgb}{1,0,0}
\definecolor{green}{rgb}{0,0.7,0}

\theoremstyle{definition}
\newtheorem{theoremc}{Theorem}

\newtheorem{definition}{Definition}[section]

\theoremstyle{definition}
\newtheorem{theorem}[definition]{Theorem}

\newtheorem{lemma}[definition]{Lemma}

\newtheorem{proposition}[definition]{Proposition}

\newtheorem{remark}[definition]{Remark}

\newtheorem{assumption}[definition]{Assumption}

\begin{document}
\title{Two-Dimensional Yang-Mills Theory on Surfaces With Corners in Batalin-Vilkovisky Formalism}

\author[a]{R. Iraso,}
\author[b,c]{P. Mnev}
\affiliation[a]{International School of Advanced Studies (SISSA),\\via Bonomea 265, I-34136 Trieste, Italy}
\affiliation[b]{University of Notre Dame,\\Notre Dame, IN 46556, USA}
\affiliation[c]{St. Petersburg Department of V. A. Steklov Institute of Mathematics \\of the Russian Academy of Sciences,\\ 
Fontanka 27, St. Petersburg 191023, Russia}
\emailAdd{riccardo.iraso@gmail.com}
\emailAdd{pmnev@nd.edu}

\keywords{}

\date{\today}

\abstract{
 In this paper we recover the non-perturbative partition function of 2D~Yang-Mills theory from the perturbative path integral. 
 To achieve this goal, we study the perturbative path integral quantization for 2D~Yang-Mills theory on surfaces with boundaries and corners in the Batalin-Vilkovisky formalism (or, more precisely, in its adaptation to the setting with boundaries, compatible with gluing and  cutting -- the BV-BFV formalism).
 We prove that cutting a surface (e.g. a closed one) into simple enough pieces -- building blocks -- and choosing a convenient gauge-fixing on the pieces, and assembling back the partition function on the surface, one recovers the known non-perturbative answers for 2D~Yang-Mills theory.\\\\\\\\\\\\\\\\\\\\
\newpage
}

\maketitle

\pagestyle{section}\thispagestyle{plain}\setcounter{page}{1}
\section{Introduction}

In this paper we study the perturbative path integral quantization of 2D~Yang-Mills theory, defined classically by the first-order action functional
\begin{equation}\label{S_cl}
 S_\mathrm{YM}^{\mathrm{cl}}=\int_\Sigma \langle B, \mathrm{d}A+A\wedge A\rangle +\frac12 \mu\, (B,B) ~,
\end{equation}
with fields~$A$, a $1$-form on the surface~$\Sigma$ with coefficients in a semi-simple Lie algebra $\mathfrak{g}=\mathrm{Lie}(G)$\,, for~$G$ a compact simply-connected structure group, and~$B$, a $0$-form valued in $\mathfrak{g}^*$\,, and where $\mu$ is a fixed background $2$-form (the ``area form'') on~$\Sigma$.
Here $\langle,\rangle$ is the pairing between~$\mathfrak{g}$ and~$\mathfrak{g}^*$ and $(,)$ is the inverse Killing form.%
\footnote{
The case $\mu=0$ of the action functional~\eqref{S_cl} defines the so-called non-abelian~$BF$ theory, which is a topological field theory, i.e., is invariant under diffeomorphisms of surfaces.
}
We study 2D~Yang-Mills theory in the Batalin-Vilkovisky formalism on oriented surfaces~$\Sigma$ with boundaries and corners allowed.%
\footnote{\label{footnote: non-or} 
We assume orientability for convenience of the formalism, but in fact one can define the theory on  non-orientable surfaces as well, twisting the field~$B$ by the orientation line bundle and defining~$\mu$ to be a density on~$\Sigma$\,, rather than a $2$-form. 
The integral~(\ref{S_cl}) is then also understood as an integral of a density.
} 
The quantization is constructed in such a way that it is 
compatible with gluing and cutting of surfaces.

Our primary motivating goal is to construct explicit partition functions of 2D Yang-Mills theory on arbitrary surfaces via the perturbative path integral quantization, $Z= \int \mathrm{e}^{\frac{\mathrm{i}}{\hbar} S_{\mathrm{YM}}}$, and to compare them with the known non-perturbative answers~\cite{Migdal:1975zg, witten:2d_quantum_gauge} formulated in terms of the representation-theoretic data of the structure group~$G$\,. 
There are two immediate problems to deal with:

\textbf{Gauge symmetry.} To define the Feynman diagrams giving the perturbative expansion of the path integral, one needs to deal with the gauge symmetry of the action creating the degeneracy of stationary phase points in the path integral. 
To do this, we employ the Batalin-Vilkovisky~(BV) formalism. 
In the BV formalism, the classical fields~$A,B$ are promoted to non-homogeneous differential forms on the surface (whose homogeneous components are the original classical fields, the Faddeev-Popov ghost for the gauge symmetry and the anti-fields for those) and the gauge-fixing consists in the choice of a Lagrangian submanifold in the BV fields.

\textbf{Computability of the perturbative answers.} Generally, Feynman diagrams are given by certain integrals over configuration spaces of points on the surface, with the integrand given by a product of propagators which depend on the details of the gauge-fixing, and typically these integrals are very hard to compute. 
The remedy for this comes from two ideas: 
\begin{enumerate}
\item Firstly, we employ the BV-BFV refinement of the Batalin-Vilkovisky formalism constructed in~\cite{CMR:classical_BV, CMR:pert_quantum_BV} -- a refinement adapted to gauge theories on manifolds with boundaries allowed, compatible with gluing and cutting (thus, it is a version of BV quantization compatible with Atiyah-Segal functorial picture of QFT). 
We give a brief overview of the BV-BFV setup in Section~\ref{background} below.
In this formalism, we can recover the perturbative partition function on a surface from cutting the surface into pieces -- the appropriate ``building blocks of surfaces'', calculating the perturbative partition function on the pieces and then assembling back into the answer for the whole surface via the gluing formula.
\item Secondly, to compute the answers on our building blocks, we employ special gauge-fixings which allow for explicit computation of Feynman diagrams on the building block. 
E.g. we use the axial gauge for cylinders.
This procedure is equivalent to imposing a very special gauge-fixing for the theory on the entire surface; it involves the data of cutting of the surface into building blocks.

For instance, the following Feynman graph for 2D Yang-Mills theory on a sphere 
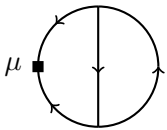
\[
 \begin{tikzpicture}[scale=.4, rotate=180]
  \coordinate (x0) at (0,2);
  \coordinate (x1) at (0,-2);
  \coordinate (x2) at (2,0);
  \coordinate (x3) at (-2,0);
  
  \draw[color=black, thick] (x0) to[out=180, in=90] (x3);
  \draw[decoration={markings, mark=at position 0 with {\arrow{>}}},
        postaction={decorate},color=black, thick] (x3) to[out=-90, in=180] (x1);
  \draw[decoration={markings, mark=at position 0.5 with {\arrow{>}}},
        postaction={decorate},color=black, thick] (x1) to[out=0, in=-90] (x2) node[left=2]{$\mu$};
  \draw[decoration={markings, mark=at position 0.5 with {\arrow{<}}},
        postaction={decorate},color=black, thick] (x2) to[out=90, in=0] (x0);
  \draw[decoration={markings, mark=at position 0.5 with {\arrow{<}}},
        postaction={decorate},color=black, thick] (x0) to (x1);
        
  \draw[color=black, fill=black, thick] ($(x2)-(.15,.15)$) rectangle ($(x2)+(.15,.15)$);          
  
 \end{tikzpicture}
\]
is given by a complicated integral in Lorenz gauge on the sphere (e.g. the one associated to the standard round metric), but in our approach it is explicitly computable, once we split the sphere into four pieces (figure~\ref{feynman_salami}).
\end{enumerate}

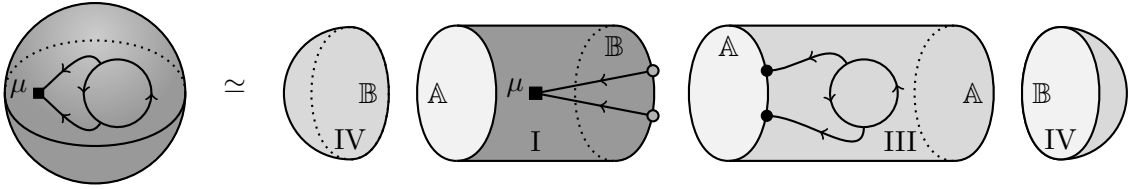
\begin{figure}[h!]
 \centering
  \[
   \begin{tikzpicture}[scale=.3, baseline=(y.base)] 

    \coordinate (x0) at (3,0);
    \coordinate (x1) at (-5,0);
    \coordinate (y) at (0,0);
    
    \filldraw[fill=black!40] (-1,0) circle (4);  
    
    \begin{scope}
     \clip (-1,0) circle (4);
     \filldraw[fill=black!40] (-1,0) circle (4);
     \shade[inner color=black!20, outer color=black!40] (-1,2.5) circle (4);
    \end{scope}
    
    \draw[black, thick] (-1,0) circle (4);  
    
    \draw[dotted, thick, draw=black] (x0) to[out=90, in=90] (x1);
    \draw[thick, draw=black] (x1) to[out=-90, in=-90] (x0) ;

   \coordinate (v5) at (120:1.5);
   \coordinate (v6) at (240:1.5);
   
   \coordinate (b) at (-3.5,0);

    \draw[decoration={markings, mark=at position 0.5 with {\arrow{>}}},
        postaction={decorate},color=black, thick] (240:1.5) arc (240:480:1.5);

   \draw[decoration={markings, mark=at position 0.5 with {\arrow{>}}},
        postaction={decorate},color=black, thick] (120:1.5) arc (120:240:1.5);

    \draw[decoration={markings, mark=at position 0.6 with {\arrow{>}}},
        postaction={decorate},color=black, thick]  (v5) to[out=120, in=40] (b);   
    
    \draw[decoration={markings, mark=at position 0.6 with {\arrow{>}}},
        postaction={decorate},color=black, thick]  (v6) to[out=240, in=-40] (b);   
        
    \draw[color=black, fill=black, thick] ($(b)-(.2,.2)$) rectangle ($(b)+(.2,.2)$) node[left=2]{$\mu$};  
        
   \end{tikzpicture}
  \quad\simeq\quad
  \begin{tikzpicture}[scale=.22, baseline=(z2.base)] 

    \coordinate (z0) at (1,4);
    \coordinate (z1) at (1,-4);
    \coordinate (z2) at (-3,0);
    
    \filldraw[fill=black!15, thick, draw=black] (z0) to[out=180, in=90] (z2) to[out=-90, in=180] (z1) node[above]{IV} to[out=0, in=0] node[left]{$\mathbb{B}$} (z0);
    
    \draw[dotted, thick, draw=black] (z0) to[out=180, in=180] (z1) to[out=0, in=0] (z0) ;
    
  \end{tikzpicture}
  \begin{tikzpicture}[scale=.3, baseline=(x.base)]
  
  \coordinate (a0) at (-1,-3);
  \coordinate (a1) at (-1,3);
  \coordinate (a2) at (6,3);
  \coordinate (x3) at (6,-3);
  \coordinate (x) at (0,0);  
  
  \filldraw[fill=black!40](a0) to[out=0, in=0] (a1)to[out=0, in=180](a2) node[below]{$\mathbb{B}$} to[in=0, out=0](x3)to[out=180, in=0]node[above]{I} (a0) ;
  \filldraw[fill=black!5, draw=black!5](a0) to[out=180, in=180]node[right]{$\mathbb{A}$} (a1)to[in=0, out=0](a0);

  \draw[color=black, thick] (a0) to[out=0, in=0] (a1);
  \draw[color=black, thick] (a1) to[out=180, in=180] (a0);
        
  \draw[color=black, thick] (x3) to[out=0, in=0] (a2);
  \draw[color=black, thick, dotted] (a2) to[out=180, in=180] (x3);

  \draw[color=black, thick](a1) to[out=0, in=180] (a2) ;
  \draw[color=black, thick](a0) to[out=0, in=180] (x3);

  
   \coordinate (B) at (2.5,0);
   \draw[decoration={markings, mark=at position 0.6 with {\arrow{<}}},
        postaction={decorate},color=black, thick]  (B)node[left]{} to (7.7,1);   
   \draw[decoration={markings, mark=at position 0.6 with {\arrow{<}}},
        postaction={decorate},color=black, thick]  (B) to (7.7,-1);   
   \draw[color=black, fill=black, thick] ($(B)-(.25,.25)$) rectangle ($(B)+(.25,.25)$) node[left=2.5]{$\mu$};  
   \draw[thick, color=black, fill=black!30] (7.7,1) circle (7pt);  
   \draw[thick, color=black, fill=black!30] (7.7,-1) circle (7pt);  
   
  \end{tikzpicture}
  \begin{tikzpicture}[scale=.3, baseline=(x.base)]

  \coordinate (x0) at (-6,-3);
  \coordinate (x1) at (-6,3);
  \coordinate (x2) at (4,3);
  \coordinate (x3) at (4,-3);
  
  \filldraw[fill=black!15](x0) to[out=0, in=0] (x1)to[out=0, in=180] (x2)to[in=0, out=0]node[left]{$\mathbb{A}$}(x3)to[out=180, in=0](x0) ;
  \filldraw[fill=black!5, draw=black!15](x0) to[out=180, in=180] (x1) node[below]{$\mathbb{A}$} to[in=0, out=0](x0);
  \node[above left] at (2.8,-3) {III};
      
  \draw[color=black, thick] (x0) to[out=0, in=0] (x1);
  \draw[color=black, thick] (x1) to[out=180, in=180]node[left](x){} (x0);
        
  \draw[color=black, thick] (x3) to[out=0, in=0] (x2);
  \draw[color=black, thick, dotted] (x2) to[out=180, in=180] (x3);

  \draw[color=black, thick](x1) to[out=0, in=180] (x2) ;
  \draw[color=black, thick](x0) to[out=0, in=180] (x3);

  
   \coordinate (v5) at (120:1.5);
   \coordinate (v6) at (270:1.5);
      
   \coordinate (b6) at (-4.3,1);
   \coordinate (b7) at (-4.3,-1);

    \draw[decoration={markings, mark=at position 0.5 with {\arrow{>}}},
        postaction={decorate},color=black, thick] (270:1.5) arc (270:480:1.5);

   \draw[decoration={markings, mark=at position 0.5 with {\arrow{>}}},
        postaction={decorate},color=black, thick] (120:1.5) arc (120:270:1.5);

    \draw[decoration={markings, mark=at position 0.5 with {\arrow{>}}},
        postaction={decorate},color=black, thick]  (v5) to[out=120, in=0] (b6);   
    \draw[color=black, fill=black] (b6) circle (7pt);     
    
    \draw[decoration={markings, mark=at position 0.5 with {\arrow{>}}},
        postaction={decorate},color=black, thick]  (v6) to[out=270, in=0] (b7);   
    \draw[color=black, fill=black] (b7) circle (7pt);     
  \end{tikzpicture}   
  \begin{tikzpicture}[scale=.22, baseline=(y.base)] 

    \coordinate (x0) at (-1,4);
    \coordinate (x1) at (-1,-4);
    \coordinate (y) at (3,0);
    
    \filldraw[fill=black!15, thick, draw=black] (x0) to[out=0, in=90] (y) to[out=-90, in=0] (x1) to[out=180, in=180] node[right]{} (x0);
    
    \filldraw[fill=black!5, thick, draw=black] (x0) to[out=180, in=180] node[right]{$\mathbb{B}$} (x1) node[above]{IV} to[out=0, in=0] (x0) ;
   \end{tikzpicture}
  \]
\caption{A two-loops Feynman diagram for 2D Yang-Mills on the sphere, computed by suitably cutting the surface. The darker regions are the ones where the $2$-form $\mu$ is allowed to be nonzero.}
\label{feynman_salami}
\end{figure}

\begin{remark}
 Axial gauge, which we often use, corresponds to a singular propagator which can be obtained as a limit of metric propagators (given by smooth forms on the configuration space of two points) corresponding to the degeneration of geometry of the cylinder where the ratio of the circumference to the length tends either to zero or to infinity (corresponding to two versions of the axial gauge), see~\cite{wernli:thesis} for a detailed discussion.
 Thus, our computable answers obtained in a convenient gauge arise as a limit (corresponding to a limiting point on a certain curve in the space of metrics on the surface~$\Sigma$) of perturbative answers computed with smooth propagators.
\end{remark}

\subsection{Surfaces of non-negative Euler characteristic}

Surfaces  of non-negative Euler characteristic (possibly with boundary, but with no corners) can be decomposed into the following building blocks: 
\begin{enumerate}[(I)]
\item Cylinder with polarizations~$\mathbb{A},\mathbb{B}$ fixed on the two boundary circles,%
 \footnote{
  That is, with boundary conditions prescribing the pullback of~$\mathsf{A}$ to one boundary circle, and the pullback of~$\mathsf{B}$ to the other circle.
 }
 with a nonzero area form~$\mu$ allowed. 
\item Cylinder with polarization~$\mathbb{B}$ fixed on both boundaries, with the background $2$-form $\mu=0$\,.
\item Cylinder with polarization~$\mathbb{A}$ fixed on both boundaries, with ~$\mu=0$\,.
\item Disk with polarization~$\mathbb{B}$ on the boundary and with~$\mu=0$\,.
\end{enumerate}
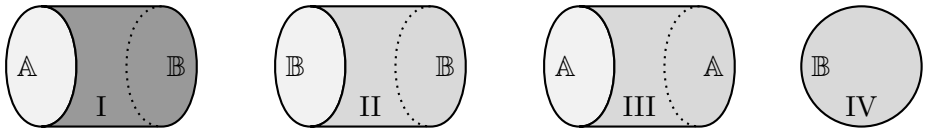
\[
   \begin{tikzpicture}[scale=.2, baseline=(x.base)] 
    \coordinate (x) at (0,0);
    \node (z0) at (0,5.2) {};
    \node (z1) at (0,-5.2) {};
   
    \coordinate (x0) at (-4,4);
    \coordinate (x1) at (-4,-4);
    \coordinate (y0) at (4,4);
    \coordinate (y1) at (4,-4);

    \filldraw[fill=black!5, thick, draw=black] (x0) to[out=180, in=180] node[right]{$\mathbb{A}$} (x1) to[out=0, in=0] (x0);
    
    \filldraw[fill=black!40, thick, draw=black] (x0) to[out=0, in=0] (x1) to[out=0, in=180] node[above]{I} (y1) to[out=0, in=0]node[left]{$\mathbb{B}$} (y0) 
      to[out=180, in=0] (x0);
    
    \draw[dotted, thick, draw=black] (y0) to[out=180, in=180] (y1);
   \end{tikzpicture}
   \hspace{.7cm}
   \begin{tikzpicture}[scale=.2, baseline=(x.base)] 
    \coordinate (x) at (0,0);
    \node (z0) at (0,5.2) {};
    \node (z1) at (0,-5.2) {};
   
    \coordinate (x0) at (-4,4);
    \coordinate (x1) at (-4,-4);
    \coordinate (y0) at (4,4);
    \coordinate (y1) at (4,-4);

    \filldraw[fill=black!5, thick, draw=black] (x0) to[out=180, in=180] node[right]{$\mathbb{B}$} (x1) to[out=0, in=0] (x0);
    
    \filldraw[fill=black!15, thick, draw=black] (x0) to[out=0, in=0] (x1) to[out=0, in=180] node[above]{II} (y1) to[out=0, in=0]node[left]{$\mathbb{B}$} (y0) 
      to[out=180, in=0] (x0);
    
    \draw[dotted, thick, draw=black] (y0) to[out=180, in=180] (y1);
   \end{tikzpicture} 
   \hspace{.7cm}
      \begin{tikzpicture}[scale=.2, baseline=(x.base)] 
    \coordinate (x) at (0,0);
    \node (z0) at (0,5.2) {};
    \node (z1) at (0,-5.2) {};
   
    \coordinate (x0) at (-4,4);
    \coordinate (x1) at (-4,-4);
    \coordinate (y0) at (4,4);
    \coordinate (y1) at (4,-4);

    \filldraw[fill=black!5, thick, draw=black] (x0) to[out=180, in=180] node[right]{$\mathbb{A}$} (x1) to[out=0, in=0] (x0);
    
    \filldraw[fill=black!15, thick, draw=black] (x0) to[out=0, in=0] (x1) to[out=0, in=180] node[above]{III} (y1) to[out=0, in=0]node[left]{$\mathbb{A}$} (y0) 
      to[out=180, in=0] (x0);
    
    \draw[dotted, thick, draw=black] (y0) to[out=180, in=180] (y1);
   \end{tikzpicture} 
   \hspace{.7cm}
   \begin{tikzpicture}[scale=.2, baseline=(y.base)] 

    \coordinate (x0) at (1,4);
    \coordinate (x1) at (1,-4);
    \coordinate (y) at (-3,0);
    
    \draw[thick, color=black, fill=black!15] (-1,0) circle (4);  
    \node[above] at (-1,-4) {IV};
    \node[right] at (-5,0) {$\mathbb{B}$};
   \end{tikzpicture}
\]
This is the premise of the BV-BFV formalism as in~\cite{CMR:pert_quantum_BV}, where the connected components of the boundary are decorated with either $\mathbb{A}$-~or $\mathbb{B}$-polarization (boundary condition), and one is allowed to glue an $\mathbb{A}$-boundary circle of one surface to a $\mathbb{B}$-boundary circle of another surface.

We manage to compute building blocks%
\footnote{
 By computing a building block, we mean computing the partition function on it.
}
(I-III) explicitly using the axial gauge, whereas the building block~(IV) can be computed in any gauge due to the vanishing of almost all Feynman diagrams by a degree counting argument.
We use the fact that the partition function can only depend on the total area of a surface to concentrate the area form~$\mu$ on building blocks~(I).%
\footnote{
 More precisely, changing the area form~$\mu$ by an exact 2-form $\mathrm{d}\gamma$ amounts in the BV language to a canonical transformation of the action (pulling back of the action by a symplectomorphism of the space of fields) and the associated change of the partition function by a BV-exact term. 
 Therefore, working modulo BV~exact terms, one can concentrate the area term in an arbitrarily small region.
 We thank Alberto S.~Cattaneo for this remark.
}
Block~(III) is the most complicated in this list. 
We only compute it modulo BV-exact terms (see Section~\ref{Quantum BV-BFV} for the definition): the latter ultimately become irrelevant once we pass from cochain-level answers to the reduced space of states and reduced partition functions (i.e. once we integrate out the bulk residual fields and pass to the cohomology of the boundary BFV differential~$\Omega$).%
\footnote{
 In particular, we exploit the gauge-invariance of the answer, known a priori from the quantum master equation, to reduce to the case of constant connections on the boundary.
}

The cohomology in degree~$0$ of the BFV differential~$\Omega$ on (the BFV model for) the space of states on a circle $\mathcal{H}_{S^1}^{\mathrm{BFV},\mathbb{A}}$ in the $\mathbb{A}$-polarization yields the standard (reduced) space of 
states of 2D Yang-Mills theory -- the space of class functions on the group~$G$\,, $\mathcal{H}_{S^1}^\mathrm{red}=L^2(G)^G$ (see Section~\ref{Omega_A_cohomology}). 
The following is the central result of this paper. 

\begin{theoremc}\label{thm: nonneg chi}
The BV-BFV partition function of 2D Yang-Mills theory for $\Sigma$ any surface with (possibly empty) boundary, with boundary circles decorated with $\mathbb{A}$-polarization, induces,  after integrating out the bulk residual fields and passing to the cohomology of the boundary BFV operator~$\Omega$\,, the Migdal-Witten non-perturbative partition function of 2D Yang-Mills:
\begin{equation}\label{[Z]=Z_non_pert}
\left[\int_\mathrm{residual\; fields} Z^\mathrm{BV-BFV}(\Sigma)\right]=
\underbrace{\sum_{R} (\dim R)^{\chi(\Sigma)}\; \mathrm{e}^{-\frac{\mathrm{i}\hbar a}{2}\cdot C_2(R)
 }\; |R\rangle^{\otimes n}}_{Z^\mathrm{non-pert}(\Sigma)} \qquad  \in (\mathcal{H}_{S^1}^\mathrm{red})^{\otimes n} ~.
\end{equation}
\end{theoremc}
Here on the left, $[\dots]$~stands for passing to the class in zeroth $\Omega$-cohomology.%
\footnote{
 In particular, the left side is independent of the details of gauge-fixing and independent of the choice of residual fields.
}
On the right side, non-perturbative partition function is given as the sum over irreducible representations~$R$ of the structure group~$G$\,, $\dim R$~is the dimension of the representation and~$C_2(R)$ is the value of the quadratic Casimir; $|R\rangle$~is the class function on~$G$ corresponding to the character of the representation~$R$\,, mapping $g\mapsto \mathrm{tr}_R g$\,; $\chi(\Sigma)$ is the Euler characteristic; $n$~is the number of boundary circles in~$\Sigma$\,; $a=\int_\Sigma \mu$ is the total area of the surface.

We first prove the comparison~\eqref{[Z]=Z_non_pert} for the case of~$\Sigma$ a disk in Section~\ref{section_gluing}, by presenting the disk as a gluing of building blocks~(I),~(III) and~(IV) above. 
We prove the Theorem~\ref{thm: nonneg chi} for a general surface in Section~\ref{corners_and_2d_YM}.

The gluing property of the r.h.s. of~\eqref{[Z]=Z_non_pert} is
\[
 Z(\Sigma_1\cup_{S^1}\Sigma_2)=\langle Z(\Sigma_1), Z(\Sigma_2) \rangle_{\mathcal{H}^{\mathrm{red}}_{S^1}} ~.
\]
Here on the right side one has the pairing in the space of states for the circle over which the surfaces are being glued.
In the BV-BFV framework it corresponds to gluing two $\mathbb{A}$-boundary circles via an ``infinitesimally short'' $\mathbb{B}-\mathbb{B}$ cylinder -- our building block~(II).

\subsection{General surfaces, surfaces with corners
}

To extend the result~\eqref{[Z]=Z_non_pert} to general surfaces 
we have to consider gluing and cutting with corners. 
In this setting we continue to decorate the codimension~$1$ strata --~circles and intervals~-- with a choice of polarization, $\mathbb{A}$ or $\mathbb{B}$, and we also decorate the codimension $2$ corners with a choice of polarization, $\alpha$ or $\beta$ (corresponding to fixing the value of either field $A$ or field $B$ in the corner).%
\footnote{
 We think of a corner carrying a polarization as the result of a \textit{collapse} of an interval carrying the same polarization, see the discussion of the picture~\ref{pict_1} and picture~\ref{pict_2} for corners in section~\ref{corners_and_2d_YM}.
\[
  \begin{tikzpicture}[scale=.35, baseline=(y.base)]
   \coordinate (y) at (0,.8);
   
   \filldraw[fill=black!40, draw=black!40] (-3,0) rectangle (3,1.2);
   \draw[dashed, black, thick] (-3,0) to (-2,0);
   \draw[black, thick] (-2,0) to (-1,0);
   \draw[black, thick] (-1,0) to (0,0);
   \draw[black, thick] (0,0) to (1,0);
   \draw[black, thick] (1,0) to (2,0);
   \draw[dashed, black, thick] (2,0) to (3,0);
   \node[above] at (0,0) {$\mathbb{A}$};
   
   \draw[black, thick] ($(2,0)+(0,.15)$) to ($(2,0)-(0,.15)$);
   \draw[black, thick] ($(-2,0)+(0,.15)$) to ($(-2,0)-(0,.15)$);
   
   \draw[decoration={markings, mark=at position 1 with {\arrow{>}}},
        postaction={decorate},color=black, thick] (0,-.3) to (0,-1.6) node[above right]{\footnotesize collapse};
   
  \begin{scope}[shift={(0,-3)}]
   \filldraw[fill=black!40, draw=black!40] (-3,0) rectangle (3,1.2);
   \draw[dashed, black, thick] (-3,0) to (3,0);
   
   \draw[black, thick] ($(0,0)+(0,.15)$) to ($(0,0)-(0,.15)$) node[below]{$\alpha$};
  \end{scope}
   
  \end{tikzpicture}
 \]
 }
For gluing, we require that if several domains are meeting at a corner, the respective corner polarizations are the same (unlike the situation with gluing over codimension~1 strata -- those should have the opposite polarization coming from the two sides of the stratum):%
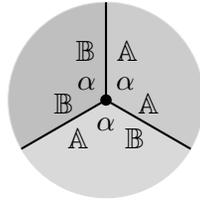
\footnote{
 Actually the BV-BFV formalism does not prescribe, in principle, a particular compatibility between polarizations for the gluing.
 What we describe here is a choice that simplifies the computations.
}
\[
 \begin{tikzpicture}[scale=.65]
  \coordinate (x0) at (0,0);
  \coordinate (x1) at (0,2);
  \coordinate (x2) at (210:2);
  \coordinate (x3) at (330:2);
  
  \filldraw[fill=black!15, draw=black!15] (x3) to (x0)node[below=3pt]{$\alpha$} to (x2) arc (210:330:2);
  \filldraw[fill=black!20, draw=black!15] (x3) to (x0)node[above right]{$\alpha$} to (x1) arc (90:-30:2);
  \filldraw[fill=black!25, draw=black!15] (x2) to (x0)node[above left]{$\alpha$} to (x1) arc (90:210:2);
  
  \draw[color=black, thick] (x0) to node[left]{$\mathbb{B}$} node[right]{$\mathbb{A}$} (x1);
  \draw[color=black, thick] (x0) to node[above]{$\mathbb{B}$} node[below right = -2.5pt]{$\mathbb{A}$}(x2);
  \draw[color=black, thick] (x0) to node[above]{$\mathbb{A}$} node[below left = -2.5pt]{$\mathbb{B}$}(x3);
  
  \filldraw (x0) circle (3pt);
 \end{tikzpicture}
\]
In this setup, one can perform the following moves on codimension~$1$ strata:
\begin{enumerate}[(a)]
 \item\label{split} One can \textit{split} an $\mathbb{A}$-interval (or an $\mathbb{A}$-circle) on the boundary of a surface into $k\geq 2$ $\mathbb{A}$-intervals separated by $\alpha$-corners. Then the partition function for a new surface is obtained by evaluating the partition function for the old surface evaluated on the concatenation of the fields $\mathbb{A}$ on the $k$ intervals.  
 
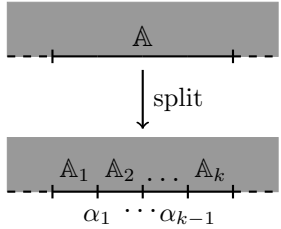
\[
  \begin{tikzpicture}[scale=.6]
   
   \filldraw[fill=black!40, draw=black!40] (-3,0) rectangle (3,1.2);
   \draw[dashed, black, thick] (-3,0) to (-2,0);
   \draw[black, thick] (-2,0) to (-1,0);
   \draw[black, thick] (-1,0) to (0,0);
   \draw[black, thick] (0,0) to (1,0);
   \draw[black, thick] (1,0) to (2,0);
   \draw[dashed, black, thick] (2,0) to (3,0);
   \node[above] at (0,0) {$\mathbb{A}$};
   
   \draw[black, thick] ($(2,0)+(0,.15)$) to ($(2,0)-(0,.15)$);
   \draw[black, thick] ($(-2,0)+(0,.15)$) to ($(-2,0)-(0,.15)$);
   
   \draw[decoration={markings, mark=at position 1 with {\arrow{>}}},
        postaction={decorate},color=black, thick] (0,-.3) to node[right]{\small split} (0,-1.6);
   
  \begin{scope}[shift={(0,-3)}]
   \filldraw[fill=black!40, draw=black!40] (-3,0) rectangle (3,1.2);
   \draw[dashed, black, thick] (-3,0) to (-2,0);
   \draw[black, thick] (-2,0) to node[above]{\small$\mathbb{A}_1$} (-1,0);
   \draw[black, thick] (-1,0) to node[above]{\small$\mathbb{A}_2$} (0,0);
   \draw[black, thick] (0,0) to node[above]{$\dots$} (1,0);
   \draw[black, thick] (1,0) to node[above]{\small$\mathbb{A}_k$} (2,0);
   \draw[dashed, black, thick] (2,0) to (3,0);
   
   \draw[black, thick] ($(2,0)+(0,.15)$) to ($(2,0)-(0,.15)$);
   \draw[black, thick] ($(-2,0)+(0,.15)$) to ($(-2,0)-(0,.15)$);
   \draw[black, thick] ($(-1,0)+(0,.15)$) to ($(-1,0)-(0,.15)$) node[below]{\small$\alpha_1$};
   \draw[black, thick] ($(0,0)+(0,.15)$) to ($(0,0)-(0,.15)$) node[below]{$\dots$};
   \draw[black, thick] ($(1,0)+(0,.15)$) to ($(1,0)-(0,.15)$) node[below]{\small$\alpha_{k-1}$};
  \end{scope}
   
  \end{tikzpicture}
 \]
 Similarly, one can split a $\mathbb{B}$-interval (or circle) into $k\geq 2$ $\mathbb{B}$-intervals separated by $\beta$-corners. 

 \item One has the inverse of the move~\eqref{split}: one can \textit{merge} $k$ $\mathbb{A}$-intervals separated by $\alpha$-corners into a single $\mathbb{A}$-interval -- this corresponds to evaluating the partition function on the field $\mathbb{A}$ restricted to the smaller sub-intervals and to the points separating them. 
 One can do the same for the $\mathbb{B}/\beta$ polarizations.

 \item One can \textit{switch} between the polarizations of the corner separating an $\mathbb{A}$-interval and a $\mathbb{B}$-interval.

 \item One can \textit{integrate out} (the field corresponding to) the $\beta$-corner separating two $\mathbb{A}$-intervals, merging them together. Likewise, one can integrate out an $\alpha$-corner separating two $\mathbb{B}$-intervals.
\end{enumerate}
The minimal set of building blocks, sufficient to construct all closed surfaces is the following:
\begin{enumerate}[(i)]
\item A disk with boundary subdivided into $k$ intervals, all in $\mathbb{A}$-polarization, and all corners taken in $\alpha$-polarization, with a possibly nonzero $2$-form $\mu$\,.  
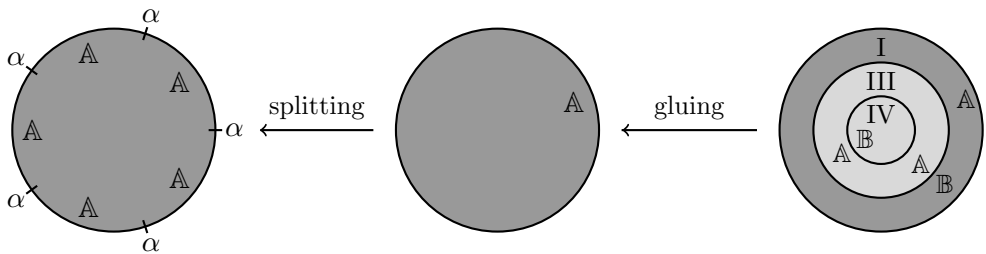
\[
 \begin{tikzpicture}[scale=.3]
  \filldraw[fill=black!40, draw=black, thick] (0,0) circle (4.5);
  
  \node at ($(36:4.5)-(36:.9)$) {$\mathbb{A}$};
  \node at ($(108:4.5)-(108:.9)$) {$\mathbb{A}$};
  \node at ($(180:4.5)-(180:.9)$) {$\mathbb{A}$};
  \node at ($(-108:4.5)-(-108:.9)$) {$\mathbb{A}$};
  \node at ($(-36:4.5)-(-36:.9)$) {$\mathbb{A}$};
  
  \draw[black, thick] (0:4.2) to (0:4.8);
  \node at (0:5.35) {$\alpha$};
  
  \draw[black, thick] (72:4.2) to (72:4.8);
  \node at (72:5.35) {$\alpha$};
  
  \draw[black, thick] (144:4.2) to (144:4.8);
  \node at (144:5.35) {$\alpha$};
  
  \draw[black, thick] (216:4.2) to (216:4.8);
  \node at (216:5.35) {$\alpha$};
  
  \draw[black, thick] (-72:4.2) to (-72:4.8);
  \node at (-72:5.35) {$\alpha$};
  
  \draw[decoration={markings, mark=at position 0.055 with {\arrow{<}}},
        postaction={decorate},color=black, thick] (6.5,0) to node[above]{\small splitting} (11.5,0);
  
  \begin{scope}[shift={(17,0)}]
   \filldraw[fill=black!40, draw=black, thick] (0,0) circle (4.5);
   \node at ($(20:4.5)-(20:.9)$) {$\mathbb{A}$};
  \end{scope}
  
  \draw[decoration={markings, mark=at position 0.055 with {\arrow{<}}},
        postaction={decorate},color=black, thick] (22.5,0) to node[above]{\small gluing} (28.5,0);
  
  \begin{scope}[shift={(34,0)}]
   \filldraw[fill=black!40, draw=black, thick] (0,0) circle (4.5);
   \filldraw[fill=black!15, draw=black, thick] (0,0) circle (3);
   \filldraw[fill=black!15, draw=black, thick] (0,0) circle (1.5);
   
   \node at (0,3.7) {\small I};
   \node at (0,2.2) {\small III};
   \node at (0,.7) {\small IV};
   \node at ($(-150:1.5)-(-150:.7)$) {\small$\mathbb{B}$};
   \node at ($(-150:1.5)+(-150:.5)$) {\small$\mathbb{A}$};
   \node at ($(-40:3)-(-40:.7)$) {\small$\mathbb{A}$};
   \node at ($(-40:3)+(-40:.7)$) {\small$\mathbb{B}$};  
   \node at ($(20:4.5)-(20:.5)$) {\small$\mathbb{A}$};
  \end{scope}  
 \end{tikzpicture}
\]
This building block is computed via the $\mathbb{A}$-disk which is expressed in terms of the building blocks (I,III,IV) above; then one applies the splitting move to the boundary. (Recall that in our convention  the shaded regions  are those that are allowed to carry a nonvanishing $2$-form $\mu$\,.)
\item A ``bean'' -- a disk with boundary subdivided into two $\mathbb{B}$-intervals, with the two corners in $\alpha$-polarization. 
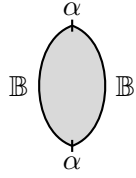
\[
 \begin{tikzpicture}[scale=.4]
  \coordinate (x1) at (0,2);
  \coordinate (x2) at (0,-2);
  
  \filldraw[fill=black!15, draw=black, thick] (x1)node[above]{$\alpha$} to[out=-20, in=20]node[right]{$\mathbb{B}$} (x2)node[below]{$\alpha$} to[out=160, in=-160]node[left]{$\mathbb{B}$} (x1);
  
  \draw[black, thick] ($(x1)+(0,.2)$) to ($(x1)-(0,.2)$);
  \draw[black, thick] ($(x2)+(0,.2)$) to ($(x2)-(0,.2)$);
 \end{tikzpicture}
\]
This block is computed from considering an axial gauge on a square and 
collapsing two opposite sides to two points.
\end{enumerate}
Then one can e.g. triangulate any surface, assign the building block~(i) with $k=3$ to each triangle and glue them using building blocks~(ii), thickening the edges of the triangulation into ``beans''.

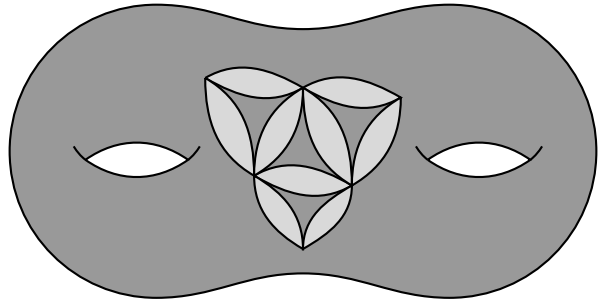
\[
 \begin{tikzpicture}[scale=.65]
  \coordinate (x0) at (-6,0);
  \coordinate (x1) at (-3,3);
  \coordinate (x2) at (0,2.5);
  \coordinate (x3) at (3,3);
  \coordinate (x4) at (6,0);
  \coordinate (x5) at (3,-3);
  \coordinate (x6) at (0,-2.5);
  \coordinate (x7) at (-3,-3);
  
  \coordinate (b1) at (-4,-.5);
  \coordinate (b2) at (3,-.5);
  
  \coordinate (y0) at (-2,1.5);
  \coordinate (y1) at (0,1.3);
  \coordinate (y2) at (2,1.1);
  \coordinate (y3) at (-1,-.5);
  \coordinate (y4) at (1,-.7);
  \coordinate (y5) at (0,-2);
  
  \filldraw[fill=black!40, thick, draw=black] (x0) to[in=180,out=90] (x1) to[in=180,out=0] (x2) to[in=180,out=0] (x3) to[in=90,out=0] (x4) to[in=0,out=-90] (x5) to[in=0,out=180] (x6) to[in=0,out=180] (x7) to[in=-90,out=180] (x0);
  
  \filldraw[fill=white, draw=black, thick]($(b1)+(-.45,.325)$) to[out=35, in=145] ($(b1)+(1.65,.325)$)to[in=-35, out=215]($(b1)+(-.45,.325)$);
  \draw[color=black, thick]($(b1)+(-.45,.325)$) to[out=141, in=145+180-22] ($(b1)+(-.69,.6)$);
  \draw[color=black, thick]($(b1)+(1.65,.325)$) to[out=39, in=35+180+22] ($(b1)+(1.89,.6)$);
  
  \filldraw[fill=white, draw=black, thick]($(b2)+(-.45,.325)$) to[out=35, in=145] ($(b2)+(1.65,.325)$)to[in=-35, out=215]($(b2)+(-.45,.325)$);
  \draw[color=black, thick]($(b2)+(-.45,.325)$) to[out=141, in=145+180-22] ($(b2)+(-.69,.6)$);
  \draw[color=black, thick]($(b2)+(1.65,.325)$) to[out=39, in=35+180+22] ($(b2)+(1.89,.6)$);
  
  \filldraw[fill=black!15, draw=black, thick] (y0) to[out=30, in=150] (y1) to[out=-150, in=-30] (y0);
  \filldraw[fill=black!15, draw=black, thick] (y1) to[out=30, in=150] (y2) to[out=-150, in=-30] (y1);
  \filldraw[fill=black!15, draw=black, thick] (y3) to[out=30, in=150] (y4) to[out=-150, in=-30] (y3);
  
  \filldraw[fill=black!15, draw=black, thick] (y0) to[out=-30, in=90] (y3) to[out=150, in=-90] (y0);
  \filldraw[fill=black!15, draw=black, thick] (y1) to[out=-150, in=90] (y3) to[out=30, in=-90] (y1);
  \filldraw[fill=black!15, draw=black, thick] (y1) to[out=-30, in=90] (y4) to[out=150, in=-90] (y1);
  \filldraw[fill=black!15, draw=black, thick] (y2) to[out=-150, in=90] (y4) to[out=30, in=-90] (y2);
  \filldraw[fill=black!15, draw=black, thick] (y3) to[out=-90, in=150] (y5) to[out=90, in=-30] (y3);
  \filldraw[fill=black!15, draw=black, thick] (y4) to[out=-150, in=90] (y5) to[out=30, in=-90] (y4);
   
 \end{tikzpicture}
\]

This way one can construct the partition function for any surface with boundary and corners, as long as corners are all in $\alpha$-polarization (boundary intervals and circles can be in any polarization). 
To produce all decorations of both boundary components and corners, one needs the following additional building block:
\begin{enumerate}[(i)]\setcounter{enumi}{2}
 \item Disk in $\mathbb{A}$-polarization, with a single $\beta$-corner.
\end{enumerate}
Then one can use the building block~(iii), together with the moves on the boundary, to create any combination of polarizations of arcs and corners on the boundary of a surface.%
\footnote{
 One starts with a surface with the corners only in $\alpha$~polarization and creates the desired $\beta$-corners surrounded by two $\mathbb{A}$-arcs by gluing in the block~(iii) --~see Figure~\ref{changing corner} in Section~\ref{corners_and_building_blocks} and formula~\eqref{b_corner_bubble}. 
 One creates the $\beta$-corners surrounded by two $\mathbb{B}$-arcs by the splitting move and $\beta$-corners surrounded by an $\mathbb{A}$-arc and a $\mathbb{B}$-arc by the switch move.
}

Using the building blocks (i), (ii), we immediately obtain the proof of Theorem~\ref{thm: nonneg chi} for a general surface $\Sigma$ (see Sections~\ref{section_gluing_arcs},~\ref{2d YM Partition Function on Surfaces with Boundaries}).

\begin{remark}\label{rem: non-or}
Theorem~\ref{thm: nonneg chi} in fact applies also to non-orientable surfaces, assuming the theory in the non-orientable case is defined as in footnote~\ref{footnote: non-or}.
\end{remark}

\begin{remark}
In this paper we are constructing a ``pragmatic'' extension of the BV-BFV framework to codimension~$2$ corners in the case of 2D Yang-Mills theory, motivated by the problem of computing explicit partition functions on surfaces (e.g., closed ones) of arbitrary genus. 
The general theory of quantization with corners in the BV-BFV formalism is work in progress and will be expanded on in a separate publication.
\end{remark}

\subsection{Main results}
\begin{itemize}
 \item Construction, in terms of explicitly computed building blocks and the gluing rule, of the partition function of 2D~Yang-Mills in BV-BFV formalism on any oriented surface with boundary and corners, with any combination of polarizations $\in \{\mathbb{A,B}\}$ assigned to the codimension~$1$ strata and polarizations $\in \{\alpha,\beta\}$ assigned to codimension~$2$ strata.
 \item Theorem~\ref{thm: nonneg chi} above, providing the comparison between the perturative BV-BFV result in the case of a surface with $\mathbb{A}$-polarized boundary and the known non-perturbative result.
 \item In Section~\ref{Corners_space_of_states} we prove that:
 \begin{itemize}
  \item The BV-BFV partition function on a surface~$\Sigma$ with corners satisfies the modified quantum master equation $(\hbar^2 \Delta + \Omega_{\partial \Sigma})Z=0$, with $\Delta$ the BV Laplacian on bulk residual fields (in the minimal realization, they are modelled on de~Rham cohomology of the surface), and with $\Omega_{\partial \Sigma}$ the boundary BFV operator. 
  We construct the operator $\Omega_{\partial \Sigma}$ explicitly. 
  In particular, apart from the edge contributions it contains quite nontrivial corner contributions, expressed in terms of the generating function for Bernoulli numbers.
  \item We prove that $\Omega$ squares to zero, and thus the space of states $\mathcal{H}$ for a stratified boundary is a cochain complex.
  \item We show that the space of states for a stratified circle can be disassembled into contributions of edges and corners, as the tensor product of certain differential graded (dg)~bimodules --~spaces of states assigned to the intervals (depending on the polarization of the interval and of its endponts)~-- over certain dg~algebras --~the spaces of states for the corners. 
  In particular, an  $\alpha$-corner gets assigned the supercommutative dg~algebra $\wedge^\bullet\mathfrak{g}^*$ with Chevalley-Eilenberg differential. 
  A $\beta$-corner gets assigned the algebra $S^\bullet \mathfrak{g}$ endowed with zero differential and a non-commutative star-product, written in terms of Baker-Campbell-Hausdorff formula.

  This picture in particular establishes a link with Baez-Dolan-Lurie setup~\cite{lurie:TFT_classification, dolan:higher_dim_alg} of extended topological quantum field theory where (in one of the models) one maps strata of the spacetime manifold of codimension $2,1,0$, respectively, to algebras, bimodules and bimodule morphisms.%
  \footnote{
   A version of extension of Atiyah's axioms accommodating the non-perturbative answers for 2D~Yang-Mills with corners was previously suggested in~\cite{oekl:2d_YM}.
   It can be obtained from our picture by fixing polarizations on all strata to $\mathbb{A},\alpha$ and passing to the zeroth cohomology of the BFV~differential~$\Omega$.
  }
 \end{itemize}
\end{itemize}

\subsection{Plan of the paper}

The paper is organized as follows.
In Section~\ref{background} we will review, for the reader's convenience, the basics of the BV-BFV formalism introduced in~\cite{CMR:pert.Segal, CMR:classical_BV,CMR:classical_quantum_BV,CMR:pert_quantum_BV} emphasizing the constructions that we will use for the analysis in the rest of the paper.

Sections~\ref{2d_YM_for_chi>0},\ref{corners_and_2d_YM} contain the main original results of this work.
We will first, in Section~\ref{2d_YM_for_chi>0}, compute the perturbative partition function for 2D Yang-Mills on disks and cylinders. 
Then, in the first part of Section~\ref{corners_and_2d_YM}, we will discuss the extension of the BV-BFV formalism to manifolds with corners for 2D YM.
This extension will be used in the second part of Section~\ref{corners_and_2d_YM} to compute the perturbative 2D YM partition function on surfaces of arbitrary genus, which induces in $\Omega$-cohomology the known non-perturbative answer~\cite{Migdal:1975zg, witten:2d_quantum_gauge}.

The reader who is well-acquainted with BV formalism and would like to take the shortest route to the proof of Theorem~\ref{thm: nonneg chi}, might want to read the sections in the following order. 
Sections~\ref{Sect:AB_cyl},~\ref{AA polarization on the cylinder},~\ref{Sec:B_disk} for the building blocks (I), (III), (IV), which are then assembled into the Yang-Mills on a disk in Sections~\ref{Sect:A_disk},~\ref{YM_Disk_in_A_polarization}. 
Then in Section~\ref{corners_and_building_blocks} the logic of extension to corners is explained and in Section~\ref{Sect:B_disk_corners} the ``bean'' is computed. Finally, in Sections~\ref{section_gluing_arcs},~\ref{2d YM Partition Function on Surfaces with Boundaries} polygons (obtained from the disk) are glued via beans into an arbitrary surface and thus the comparison theorem is proven.

In Appendix~\ref{Wilson_loop_observables} we will discuss how to compute, in this setting, Wilson loop observables for both non-intersecting and intersecting loops, recovering in $\Omega$-cohomology the known non-perturbative result.

\subsection{Notations}

Here is a table of recurring notational conventions used throughout the paper.
\[
 \begin{tabularx}{\textwidth}{l|X|l}
  Notation		&	Meaning								& 	Introduced in\\\hline
  $W^\bullet[k]$	& 	$(W^\bullet[k])^p := W^{k+p}$ for a graded vector space $W^\bullet$
												&\\	
  $\mathrm{d}$		&	De Rham differential on the source manifold			&	\\
  $\delta$		&	De Rham differential on the space of fields			&	\\
  $\mathfrak{F}_\Sigma$	&	Space of bulk fields associated to the surface $\Sigma$		&	Section~\ref{Classical_BV-BFV}\\  
  $\mathfrak{F}_\partial,~\mathfrak{F}_{\partial\Sigma},~\mathfrak{F}_{\gamma}$		
			&	Space of boundary fields, can be associated to the boundary of a surface $\partial \Sigma$ or to a curve $\gamma$ 		
												&	Section~\ref{Classical_BV-BFV}\\  
  $\mathcal{P}$		&	Polarization of a symplectic manifold				&	Section~\ref{BV-BFV_quantization}\\
  $\mathcal{B}_\gamma^{\mathcal{P}}$
			&	Quotient space $\mathfrak{F}_\gamma/\mathcal{P}$		&	Section~\ref{BV-BFV_quantization}\\
  $\mathcal{V}^\bullet_\Sigma$&	Space of residual fields or zero modes associated to a surface~$\Sigma$
												&	Section~\ref{Quantum BV-BFV}\\
  $\mathcal{Y}_\Sigma$	&	Space of bulk fields with boundary conditions			&	Section~\ref{BV-BFV_quantization}\\
  $\mathcal{Y}'_\Sigma$	&	Space of bulk fluctuations					&	Section~\ref{BV-BFV_quantization}\\ 
 \end{tabularx}
\]
\[
 \begin{tabularx}{\textwidth}{l|X|l} 
  Notation		&	Meaning								& 	Introduced in\\\hline
  $\Omega_{\partial\Sigma}$ &	BV differential, quantization of the boundary action		&	Section~\ref{Quantum BV-BFV}\\
  $\Delta_\mathcal{V}$	&	BV Laplacian on the space of residual fields			&	Section~\ref{Quantum BV-BFV}\\
  $\mathcal{H}_{\partial\Sigma},~\mathcal{H}_\gamma$		
			&	Space of quantum states						&	Section~\ref{Quantum BV-BFV}\\
  $Z_\Sigma$		&	Partition function or state, an element of $\mathcal{H}$	&	Section~\ref{Quantum BV-BFV}\\
  $\mathcal{L}$		&	Lagrangian submanifold, used as gauge-fixing			&	Section~\ref{Quantum BV-BFV}\\
  $\eta(x;x')$		&	Propagator from the point $x$ to the point $x'$			&	Section~\ref{BF Propagators and Axial Gauge}\\
  $\mathsf{A},~\mathsf{B}$ &	Bulk fields							&	Section~\ref{BV-BFV_quantization}\\
  $\mathbb{A},~\mathbb{B}$ &	Boundary fields							&	Section~\ref{BV-BFV_quantization}\\
  $\mathsf{a},~\mathsf{b}$ &	Residual fields, zero modes					&	Section~\ref{BV-BFV_quantization}\\
  $\alpha,~\beta$	&	Corner fields 
												&	Section~\ref{Corners_space_of_states}\\
 \end{tabularx}
\]

\subsection{Acknowledgements}

We would like to thank Alberto S. Cattaneo, Andrey S. Losev, Stephan Stolz, Konstantin Wernli for inspiring discussions. 
P. M. acknowledges partial support of RFBR Grant No. 17-01-00283a.

\section{Background: BV-BFV formalism}\label{background}

We will start this section reviewing the basic constructions of the BV-BFV formalism and fixing the notation.
We will then apply this construction to obtain the BV-BFV formulation of the non-abelian BF theory and Yang-Mills theory reviewing some of the known results.
For a complete and detailed discussion of this topic we refer to~\cite{CMR:pert.Segal, CMR:classical_BV,CMR:classical_quantum_BV,CMR:pert_quantum_BV}, where this formalism was first introduced.

\subsection{Classical BV-BFV}\label{Classical_BV-BFV}

\begin{definition}
 A BFV manifold is given by the triple $(\mathfrak{F}_\partial, \alpha^\partial, \mathcal{Q}^\partial)$\,, where: the space of \emph{boundary fields} $\mathfrak{F}_\partial$ is an exact graded symplectic manifold with 0-symplectic form $\omega^\partial=\delta\alpha^\partial$ and $\mathcal{Q}^\partial$ is a homological symplectic vector field of degree 1.%
 \footnote{
  For simplicity we will consider, here and in the following, all the gradings to be $\mathbb{Z}$ gradings.
  The parity, determining the commuting/anticommuting properties of coordinates, is given by the degree$\mod{2}$\,.
 }
\end{definition}
In particular the condition $L_{\mathcal{Q}^\partial}\delta\alpha^\partial=0$ for the vector field $\mathcal{Q}^\partial$\,, since $|\mathcal{Q}^\partial|+|\omega^\partial|\neq 0$\,, implies that it is also Hamiltonian: $\imath_{\mathcal{Q}^\partial}\omega^\partial=\delta\mathcal{S}^\partial$\,.
This defines the degree~1 Hamiltonian~$\mathcal{S}^\partial$\,, which we will call the \emph{boundary BFV action}.

\begin{definition}\label{classical_BV-BFV}
 A BV-BFV manifold, over a BFV manifold $(\mathfrak{F}_\partial, \alpha^\partial, \mathcal{Q}^\partial)$\,, is a quintuple $(\mathfrak{F},\omega,\mathcal{S},\mathcal{Q},\pi)$\,, where the space of \emph{bulk fields} $(\mathfrak{F},\omega)$ is a $(-1)$-symplectic manifold, the \emph{bulk action} $\mathcal{S}$ is a function of the fields, the bulk \emph{BRST operator} $\mathcal{Q}$ is a homological vector field of degree 1 and $\pi\colon\mathfrak{F}\rightarrow\mathfrak{F}_\partial$ is a surjective submersion, satisfying the following two compatibility conditions:
 \begin{enumerate}[i)]
  \item the bulk homological vector field projects on the boundary homological vector field: $\delta\pi\,\mathcal{Q}=\mathcal{Q}^\partial$\,;
  \item the \emph{modified Classical Master Equation} (mCME) holds: $\imath_\mathcal{Q}\omega=\delta\mathcal{S} + \pi^*\alpha^\partial$\,.
 \end{enumerate}
\end{definition}

A \emph{classical BV-BFV theory} is constructed for manifolds with boundaries of some fixed dimension $n$\,.
It consists of the association to each manifold with boundary~$\Sigma$ of a BV-BFV manifold~$\mathfrak{F}_\Sigma$ over the BFV manifold~$\mathfrak{F}_\partial = \mathfrak{F}_{\partial \Sigma}$ associated to the boundary~$\partial \Sigma$\,.
This association has to be compatible with disjoint union and ``gluing'' in the following sense:
\begin{enumerate}[i)]
 \item a disjoint union maps to the direct product: $\mathfrak{F}_{\Sigma_1\sqcup \Sigma_2} = \mathfrak{F}_{\Sigma_1}\times \mathfrak{F}_{\Sigma_2}$\,;
 \item a gluing of two manifolds maps to the fiber product over the space of fields associated to the gluing interface~$\gamma$\,: $\mathfrak{F}_{\Sigma_1\cup_{\gamma}\Sigma_2} = \mathfrak{F}_{\Sigma_1} \times_{\mathfrak{F}_\gamma} \mathfrak{F}_{\Sigma_2} $\,.
\end{enumerate}

\begin{remark}
 This can be interpreted as a covariant monoidal functor from the \emph{spacetime category}, with $(n-1)$ closed manifolds as objects and $n$-manifolds with boundary as morphisms with composition given by gluing,%
 \footnote{
  Depending on the specific theory, the spacetime category could have additional structures: for examples manifolds could be oriented, Riemannian, etc.
  Also, depending on the specific theory, there may be subtleties to defining the spacetime category as an actual category. 
  This discussion is beyond the scope (and is not relevant to) this paper.
 }
 to the \emph{BFV category}, where objects are BFV manifolds, morphisms are BV-BFV manifolds over (products of in- and out-) BFV manifolds, and composition is given by fiber products.
 The monoidal structure on the spacetime category is given by the disjoint union, while on the BFV category side it is given by the direct product.
\end{remark}

\begin{remark}
 On closed manifolds this construction reduces to a classical BV theory, which gives a homological resolution of the space of classical states for Lagrangian gauge field theories and is the classical starting point for the BV quantization of such theories~\cite{BV:1,teitelboim:quantization}.
\end{remark}

\subsection{Quantum BV-BFV}\label{Quantum BV-BFV}\pagestyle{section}

A \emph{quantum BV-BFV theory} associates to an $(n-1)$~manifold~$\gamma$ a graded cochain complex~$\mathcal{H}_\gamma$\,, the \emph{space of states}, with a coboundary operator~$\Omega_\gamma$ called the \emph{quantum BFV charge}.
To $n$-manifolds with boundary~$\Sigma$ the quantum theory assigns a (finite-dimensional) $(-1)$-symplectic manifold $(\mathcal{V}_\Sigma,\omega_{\mathcal{V}_\Sigma})$\,, the space of \emph{residual fields}, and the \emph{partition function}, which is an element $Z_\Sigma\in\mathcal{H}_{\partial \Sigma}\otimes\mathrm{Dens}^{\frac{1}{2}}(\mathcal{V}_\Sigma)$ in the boundary space of states tensored with the half-densities on the residual fields.
The space of residual fields is not uniquely determined, but comes in a poset of different realizations.
The partition function for a smaller realization can be reached with a BV-pushforward (see Subsection~\ref{renormalization_and_globalization} for further discussion).%
\footnote{
 BV pushforward along a symplectic fibration is the natural version of the fiber integral in BV formalism. 
 It maps solutions of the master equation on the total space to solutions of the master equation on the base, see~\cite{CMR:pert_quantum_BV} for details.
}
The partition function has to satisfy the \emph{modified Quantum Master Equation} (mQME):
\beq\label{mQME}
 (\Omega_{\partial \Sigma} + \hbar^2 \Delta_{\mathcal{V}_\Sigma}) Z_\Sigma=0	~,
\eeq
where $\Delta_{\mathcal{V}_\Sigma}$ is the canonical BV Laplacian on the half-densities of the residual fields.
The partition function is understood to be defined modulo $(\Omega_{\partial \Sigma} + \hbar^2 \Delta_{\mathcal{V}_\Sigma})$-exact terms.
Also, the quantum theory satisfies compatibility conditions with respect to the disjoint union and the gluing of spacetime manifolds:
\begin{enumerate}[i)]
 \item To disjoint unions, reflecting the quantum nature of the theory, the BV-BFV theory associates the tensor product of the spaces of states, $\mathcal{H}_{\gamma_1\sqcup \gamma_2}=\mathcal{H}_{\gamma_1}\otimes\mathcal{H}_{\gamma_2}$\,, the direct product of residual fields $\mathcal{V}_{\Sigma_1\sqcup \Sigma_2}=\mathcal{V}_{\Sigma_1}\times\mathcal{V}_{\Sigma_2}$ and the tensor product of partition functions, $Z_{\Sigma_1\sqcup \Sigma_2}=Z_{\Sigma_1}\otimes Z_{\Sigma_2}$.
 \item To the gluing of two manifolds the theory associates the partition function obtained as the pairing, in the space of states of the gluing interface, of the partition functions of the constituent manifolds: $Z_{\Sigma_1\cup_\gamma \Sigma_2} = \langle Z_{\Sigma_1} , Z_{\Sigma_2}\rangle_{\gamma}$\,.%
 \footnote{
  For oriented spacetime manifolds, this is the dual pairing between~$\mathcal{H}_\gamma$ and~$\mathcal{H}^*_\gamma$\,: since the two boundaries that are glued together must have opposite orientations, the associated vector spaces are dual to each other.
 }
\end{enumerate}

\emph{Quantum observables} are defined to be cohomology classes of the coboundary operator $Z^{-1}_\Sigma(\Omega_{\partial \Sigma} + \hbar^2 \Delta_{\mathcal{V}_\Sigma})(Z_\Sigma\cdot\dots)$ --~the conjugation of the coboundary operator appearing in the mQME~-- with expectation value computed by a \emph{BV pushforward} of a representative~$\mathcal{O}$ times the partition function, i.e. integrating their product over a Lagrangian~$\mathcal{L}\subset\mathcal{V}_\Sigma$\,:
\beq
 \langle\mathcal{O}\rangle_\Sigma := \int_{\mathcal{L}} \mathcal{O}\, Z_\Sigma 	~.
\eeq
The Lagrangian submanifold $\mathcal{L}$ has here the meaning of \emph{gauge-fixing} for the integration over residual fields and the closedness of~$\mathcal{O} Z_\Sigma$ with respect to $\Omega_{\partial \Sigma} + \hbar^2 \Delta_{\mathcal{V}_\Sigma}$ ensures that the $\Omega_{\partial \Sigma}$-cohomology class resulting from the integration does not depend on the particular choice of gauge fixing thanks to the following theorem~\cite{CMR:pert_quantum_BV, mnev:discrete_BF}.

\begin{theorem}\label{BV_pushforward}
 Let $(\mathcal{M}_1,\omega_2)$ and $(\mathcal{M}_2,\omega_2)$ be two (finite-dimensional) graded manifolds with odd symplectic forms~$\omega_i$ and canonical Laplacians $\Delta_i$\,.
 Consider $\mathcal{M} = \mathcal{M}_1\times\mathcal{M}_2$ with product symplectic form and canonical Laplacian $\Delta$ and let $\mathcal{L},\mathcal{L}'\subset\mathcal{M}_2$ be any two Lagrangian submanifolds which can be deformed into each other.
 For any half-density $f\in\mathrm{Dens}^{\frac{1}{2}}(\mathcal{M})$:
 \begin{enumerate}[i)]
  \item $\int_{\mathcal{L}} \Delta f=\Delta_1 \int_{\mathcal{L}} f$
  \item $\int_{\mathcal{L}} f-\int_{\mathcal{L}'} f= \Delta_1\xi\qquad$ for some $\xi\in\mathrm{Dens}^{\frac{1}{2}}(\mathcal{M}_1)$\,, if $\Delta f=0$\,.
 \end{enumerate}
 In particular, when $\mathcal{M}_1$ is just a point, the r.h.s. of the two equations above vanishes.
\end{theorem}

\subsection{Quantization}\label{BV-BFV_quantization}

The quantization procedure is a way to get a quantum BV-BFV theory from the data of a classical BV-BFV theory.
The first object to construct is the space of states~$\mathcal{H}_{\gamma}$\,, which is obtained from the space of boundary fields~$\mathfrak{F}_{\gamma}$ by choosing a Lagrangian foliation, or more generally a \emph{polarization}~$\mathcal{P}$\,.
We will assume in the following that the 1-form $\alpha_\gamma$ vanishes along the fibers of~$\mathcal{P}$\,; if this is not the case, we can use the transformation:
\beq\label{gauge_transf.}
 \alpha_\gamma \mapsto \alpha_\gamma - \delta f_\gamma	~, \qquad
 \mathcal{S}_\Sigma \mapsto \mathcal{S}_\Sigma +	\pi^*f_\gamma~,
\eeq
which uses an arbitrary function $f_\gamma$ of the boundary fields to shift $\alpha_\gamma$ and the bulk action in such a way that the mCME is preserved (cf. def.~\ref{classical_BV-BFV}).
The space of states of the quantum theory (a.k.a. the space of quantum states) is defined as the space of complex-valued functions%
\footnote{
 Another possible model for states uses half-densities instead of functions. 
 These two models are isomorphic, with the isomorphism given by multiplication by a fixed reference half-density.
}
on the leaf space~$\mathcal{B}_{\gamma}^\mathcal{P}= \mathfrak{F}_{\gamma}/\mathcal{P}$ (or more generally the space of polarized sections of the trivial ``prequantum'' $U(1)$-bundle  over~$\mathfrak{F}_\gamma$).
\beq
 \mathcal{H}_{\gamma} := 
  \mathrm{Fun}_{\mathbb{C}}(\mathcal{B}_{\gamma}^\mathcal{P})	~.
\eeq
In other words, the space of quantum states is obtained as the \emph{geometric quantization} of the space of boundary fields~\cite{CMR:pert_quantum_BV}.

The space of quantum states forms a cochain complex.
The coboundary operator~$\Omega_\gamma$ is constructed as the geometric quantization of the boundary action~$\mathcal{S}_\gamma$\,.
Suppose we have Darboux coordinates~$(q,p)$ on~$\mathfrak{F}_\gamma$\,, where~$q$ are also coordinates of~$\mathcal{B}_\gamma^\mathcal{P}$\,.
The operator $\Omega_\gamma$ is the standard-ordering quantization of the action:
\beq\label{Omega}
 \Omega_\gamma := \mathcal{S}_\gamma\bigg(q, -\mathrm{i}\hbar\frac{\partial}{\partial q}\bigg)~,
\eeq
where all the derivatives are positioned to the right.
For the theories we will consider in this paper, with this definition $\Omega_\gamma$ squares to zero; in general it could be needed to add quantum corrections to~\eqref{Omega} for $\Omega_\gamma$ to actually be a coboundary operator.%
\footnote{
 In general there might be cohomological obstructions to do that.
 Moreover, the partition function might be not compatible with the so constructed $\Omega$, causing the mQME to fail.
}

Let us consider now the data associated to the bulk $n$-manifolds.
The space of bulk fields has a fibration over $\mathcal{B}^\mathcal{P}_{\partial \Sigma}$ defined by composing the projection to the boundary fields with the projection given by the polarization: $\mathfrak{F}_\Sigma \overset{\pi}{\longrightarrow} \mathfrak{F}_{\partial \Sigma} \longrightarrow\mathcal{B}^\mathcal{P}_{\partial \Sigma}$\,.
Suppose for simplicity that this is a trivial bundle: $\mathfrak{F}_\Sigma = \widetilde{\mathcal{B}}^\mathcal{P}_{\partial \Sigma} \times \mathcal{Y}_\Sigma$\,, where $\widetilde{\mathcal{B}}^\mathcal{P}_{\partial \Sigma}$ is some bulk extension of ${\mathcal{B}}^\mathcal{P}_{\partial \Sigma}$ and $\mathcal{Y}_\Sigma$ is some $(-1)$-symplectic manifold.
This assumption will hold in all the theories considered in the following.

The space of residual fields can be taken to be any (finite-dimensional) symplectic subspace $\mathcal{V}_\Sigma$ of the space of fields,%
\footnote{
 In the framework of perturbation theory, the requirement that the integral~\eqref{part.funct.} below is perturbatively well-defined, imposes restrictions on the possible choices of~$\mathcal{V}_\Sigma$. 
 E.g.~$\mathcal{V}_\Sigma$ for perturbed BF~theories has to be modelled on a deformation retract of the de~Rham complex of the bulk manifold.
} 
separating it as $\mathcal{Y}_\Sigma=\mathcal{V}_\Sigma\times \mathcal{Y}'_\Sigma$\,, where~$\mathcal{Y}'_\Sigma$ is the space of \emph{fluctuations}.
The partition function is now defined as a BV pushforward of the exponentiated bulk action:
\beq\label{part.funct.}
 Z_\Sigma(\mathcal{P};\mathcal{V}_\Sigma) := \int_{\mathcal{L}} \mathrm{e}^{\frac{\mathrm{i}}{\hbar} \mathcal{S}_\Sigma}	~,
\eeq
where $\mathcal{L}\subset \mathcal{Y}'_\Sigma$ is a Lagrangian submanifold. 
If $\Delta_{\mathcal{Y}_\Sigma}\mathcal{S}_\Sigma=0$, theorem~\ref{BV_pushforward} implies that the partition function is a solution of the mQME~\eqref{mQME}.%
\footnote{
 In the theories considered in the following this condition is verified.
 However, this is not always the case and there can be theories where the bulk action needs quantum corrections in order for the mQME to hold.
 This is connected in particular to the presence of quantum anomalies in the theory.
}
Moreover, we have that $Z_\Sigma$ does not depend on (deformations of) the gauge-fixing Lagrangian~$\mathcal{L}$ used in the BV pushforward, up to $(\Omega_\gamma + \hbar^2\Delta_{\mathcal{V}_\Sigma})$-exact terms.

The discussion, until now, assumes a finite-dimensional situation.
This is usually not the setting of quantum field theories; for infinite-dimensional spaces one needs a more delicate analysis to prove the mQME and to prove that the dependence of the partition on the gauge-fixing is BV-exact.
A way to make sense of infinite-dimensional integrals is through perturbation theory, as discussed in the following section.

\subsubsection{Perturbative expansion}

The space of fields $\mathfrak{F}_\Sigma$ is typically infinite-dimensional, for example it can contain the de Rham complex of differential forms over $\Sigma$.
As a consequence, the integral~\eqref{part.funct.} defining the partition function is (almost always) ill-defined as a measure-theoretic integral and has to be understood as a \emph{perturbative series} written in terms of the Feynman diagrams coming from the interactions in the bulk action expanded around a point $\mathsf{x}_0\in\mathfrak{M}$ in the \emph{Euler-Lagrange moduli space} -- the space of solutions of classical equations of motion of~$\mathcal{S}_\Sigma$ (modulo gauge symmetries).

In order for the perturbative expansion to be well-defined, the \emph{gauge-fixed action} -- the restriction of~$\mathcal{S}_\Sigma$ to the gauge-fixing Lagrangian~$\mathcal{L}$ -- needs to have isolated critical points.
It is important to remark that this condition does not, in general, hold for every Lagrangian.

The existence of such a ``good gauge-fixing'' depends on the choice of residual fields.
In particular the quadratic part of the bulk action can have \emph{zero-modes}~$\mathcal{V}^0_\Sigma$\,, i.e. bulk fields configurations that are annihilated by the kinetic operator.%
\footnote{
 See section~\ref{BF Propagators and Axial Gauge} for the definition of zero modes in 2D YM.
}
Zero-modes correspond to the tangent directions to the Euler-Lagrange moduli space (cf.~\cite{CMR:pert_quantum_BV}, appendix F) and therefore their presence in the space of fluctuations indicates non-isolated critical points of the action and obstructs the perturbative expansion.
Thus, the space of residual fields has to at least contain the space of zero-modes for a good gauge-fixing Lagrangian to exist: $\mathcal{V}_\Sigma^0\subseteq\mathcal{V}_\Sigma$\,.
When the residual fields coincide with zero modes we say that the perturbative partition function is in its \emph{minimal realization}.

Another consequence of the infinite dimensions of~$\mathfrak{F}_\Sigma$ is that also the BV~Laplacian is ill-defined.
The equations containing it are thus only formal (or require a regularization).%
\footnote{
 The BV Laplacian becomes non-singular within the framework of renormalization theory on the level of residual fields. 
 See also~\cite{costello:BV_renorm.} for the discussion of how the RG~flow regularizes the BV~Laplacian.  
}
In particular, theorem~\ref{BV_pushforward} is proved in a finite-dimensional setting.
An important point is thus that even if the action is formally annihilated by the Laplacian, the mQME is only expected to hold and needs to be verified for each particular theory.
For perturbed BF~theories, including 2D Yang-Mills, the mQME has been proved in the infinite-dimensional perturbative setting in~\cite{CMR:pert_quantum_BV} and relies on the Stokes' theorem for integrals over compactified configuration spaces of points.

\subsubsection{Renormalization and globalization}\label{renormalization_and_globalization}

A non-minimal realization of a theory is obtained when the zero-modes are a proper subset of the space of residual fields.
Of course there are different, inequivalent, non-minimal realizations of any theory.
Given a non-minimal realization, one can obtain a smaller one by a BV pushforward.
If $\mathcal{V}_\Sigma'=\mathcal{V}_\Sigma''\times \mathcal{Y}'$\,, with $\mathcal{V}_\Sigma^0\subset\mathcal{V}''_\Sigma$\,, then:
\beq
 Z_\Sigma(\mathcal{P};\mathcal{V}'') = \int_\mathcal{L} Z_\Sigma(\mathcal{P};\mathcal{V}')
\eeq
for a Lagrangian submanifold $\mathcal{L}\subset\mathcal{Y}'$\,.
The set of all possible realizations forms therefore a partially ordered set, with the final object given by the minimal realization.
Passing from bigger to smaller realizations can be interpreted as following the \emph{renormalization group flow}.

\begin{remark}
 According to the gluing prescription (cf. section~\ref{Quantum BV-BFV}), the residual fields of the glued manifold are the direct product of the residual fields of the two manifolds being glued.
 In particular this means that, generally, if we glue together two partition functions in the minimal realization the result of the gluing will not be in the minimal realization.
 Let $\Sigma= \Sigma_1 \cup_\gamma \Sigma_2$\,; it typically happens that $\mathcal{V}_\Sigma^0\subset \mathcal{V}^0_{\Sigma_1}\times\mathcal{V}^0_{\Sigma_2}$\,.
 The minimal realization for the glued manifold has then to be obtained via a BV pushforward:
 \beq
  Z_\Sigma(\mathcal{V}_\Sigma^0) = \int_\mathcal{L} \langle Z_{\Sigma_1}(\mathcal{V}_{\Sigma_1}^0),Z_{\Sigma_2}(\mathcal{V}_{\Sigma_2}^0)\rangle~.
 \eeq
\end{remark}

Because of its perturbative definition, the partition function depends on the point~$\mathsf{x}_0\in\mathfrak{M}$ around which we are expanding and carries only local information on field configurations infinitesimally close to $\mathsf{x}_0$ (it is defined on a formal neighbourhood of~$\mathsf{x}_0$).
Nevertheless, at least in the theories considered in this paper, its minimal realization is the Taylor expansion of a \emph{global} half-density on the tangent bundle of the Euler-Lagrange moduli space (cf.~\cite{CMR:pert_quantum_BV}, appendix~F).
Thus, under some assumptions, it can be integrated on the zero section of~$T\mathfrak{M}$\,.
This corresponds to setting to zero all the zero-modes $\nu\in\mathcal{V}^0_\Sigma$ and integrating the partition function on the Euler-Lagrange moduli space
\beq\label{glob.part.funct.1}
 Z_\Sigma(\mathcal{P})= \int_\mathfrak{M} Z_\Sigma(\mathcal{P},\mathcal{V}^0_\Sigma;\mathsf{x}_0)\big|_{\nu=0} \in \mathrm{Dens}^{\frac{1}{2}}(\mathcal{B}^\mathcal{P}_{\partial \Sigma})~,
\eeq
obtaining a \emph{globalized partition function} depending only on the boundary fields in~$\mathcal{B}^\mathcal{P}_{\partial \Sigma}$\,.

Another way to obtain a partition function which does not depend on~$\mathcal{V}^0_\Sigma$ is to integrate its minimal realization over all the zero-modes, again using a BV pushforward.
Notice that this cannot be done in a perturbative way -- the propagator cannot be defined for zero modes -- but since $\mathcal{V}_\Sigma^0$ is a finite-dimensional space, the BV pushforward is well-defined as an ordinary integral on a supermanifold:
\beq
 Z_\Sigma(\mathcal{P}) = \int_{\mathcal{L}\subset\mathcal{V}^0_\Sigma} Z_\Sigma (\mathcal{P},\mathcal{V}^0_\Sigma)~.
\eeq
This can be viewed as an alternative definition of a globalized partition function and in fact, when both this integral and the one in~\eqref{glob.part.funct.1} can be computed explicitly, they coincide (cf.~section~\ref{2d YM Partition Function on Surfaces with Boundaries}).%
\footnote{
 One caveat is that one needs to take care to avoid possible overcounting when integrating over zero-modes, cf.~Remark~\ref{gribov_region} below.
}
However, the precise relation between the two globalization procedures is to be understood better.

\subsection{BV-BFV formulation of 2D YM}

We will review in this section the BV-BFV construction for 2D Yang-Mills and non-abelian BF theories; for a deeper discussion and for some of the proofs we refer to~\cite{CMR:classical_quantum_BV, CMR:pert_quantum_BV}.

Let~$G$ be a Lie group with Lie algebra~$\mathfrak{g}$ and let~$A$ be a connection 1-form on a principal $G$-bundle over a 2-dimensional surface~$\Sigma$.
In the first order formalism the classical YM action can be written in the following form:
\beq
 S_\Sigma(A,B) = \int_\Sigma \langle B , F_A \rangle + \frac{1}{2} \int_\Sigma (B,B) \, \mu ~,
\eeq
where the auxiliary field~$B$ is a zero-form valued in~$\mathfrak{g}^*$\,, $F_A$ is the curvature 2-form of $A$\,, $\langle\cdot,\cdot\rangle$~is the dual pairing between~$\mathfrak{g}$ and~$\mathfrak{g^*}$\,, $(\cdot,\cdot)$~is an invariant non-degenerate pairing on~$\mathfrak{g}^*$ and $\mu$ is the volume 2-form associated to a metric on~$\Sigma$\,.
We see that 2D YM can be treated as a perturbation of 2D non-abelian BF theory, which can be obtained in the zero-area limit~$\mu\rightarrow 0$\,. 
In the following sections we will generally find it useful to consider BF theory first, introducing the area term only afterwards.

On closed surfaces, the classical BV construction enhances the space of fields by adding differential forms of every degree, usually called ghosts and antifields for positive or negative internal degree respectively.
The BV space of fields over~$\Sigma$ is then 
\beq
 \mathfrak{F}_\Sigma = \Omega(\Sigma;\mathfrak{g})[1]\oplus\Omega(\Sigma;\mathfrak{g}^*)\ni (\mathsf{A},\mathsf{B})~\,,
\eeq
where~$\mathsf{A}$ and~$\mathsf{B}$ are the superfields associated to~$A=\mathsf{A}_{(1)}$ and~$B=\mathsf{B}_{(0)}$ which are their degree-zero components.%
\footnote{
 We will denote by $\mathsf{A}_{(n)}$ the $n$-form component of a superfield~$\mathsf{A}$\,.
}
The BV space of fields is a symplectic graded space, with $(-1)$-symplectic form given by:
\beq\label{BF_sympl.}
 \omega_\Sigma = \int_\Sigma \langle\delta \mathsf{B}, \delta \mathsf{A}\rangle~.
\eeq
The BV action on a closed manifold is
\beq\label{YM_BV_action}
 \mathcal{S}_\Sigma=\int_\Sigma \langle \mathsf{B} , \mathrm{d}\mathsf{A} +\frac{1}{2}[\mathsf{A},\mathsf{A}]  \rangle 
  + \frac{1}{2} \int_\Sigma (\mathsf{B},\mathsf{B}) \, \mu
\eeq
and the corresponding Hamiltonian vector field, the homological vector field~$\mathcal{Q}_\Sigma$\,, is:
\beq\label{YM_BV_charge}
 \mathcal{Q}_\Sigma= \int_\Sigma \Big\langle \mathrm{d}\mathsf{A}+\frac{1}{2} [\mathsf{A}, \mathsf{A}],\frac{\delta}{\delta\mathsf{A}}\Bigl\rangle 
  + \int_\Sigma \Big\langle\mathrm{d}\mathsf{B}+ \mathrm{ad}^*_{\mathsf{A}}\mathsf{B},\frac{\delta}{\delta\mathsf{B}}\Big\rangle + \int_\Sigma \Big(\mathsf{B},\frac{\delta}{\delta \mathsf{B}}\Big) \mu ~,
\eeq

In the BV-BFV construction the bulk fields, symplectic structure, action and homological vector field are again the ones described above.
Notice that now, when~$\Sigma$ has a non-empty boundary~$\partial \Sigma$\,, the homological vector field~\eqref{YM_BV_charge} is not the Hamiltonian vector field of the action~\eqref{YM_BV_action}.
Indeed, it is not even symplectic:
\beq\label{BF_mCME}
 \imath_{\mathcal{Q}_\Sigma}\omega_\Sigma = \delta \mathcal{S}_\Sigma + \int_{\partial \Sigma} \langle \mathsf{B},\delta\mathsf{A} \rangle~.
\eeq
The boundary fields~$\mathfrak{F}_{\partial \Sigma}$ are similar to the bulk (see~\cite{CMR:pert_quantum_BV}):
\beq
 \mathfrak{F}_{\partial \Sigma} = \Omega (\partial \Sigma;\mathfrak{g})[1]\oplus\Omega (\partial \Sigma;\mathfrak{g}^*) \ni (\mathbb{A},\mathbb{B}) ~.
\eeq
We can thus define the projection $\pi\colon\mathfrak{F}_{\Sigma}\longrightarrow\mathfrak{F}_{\partial \Sigma}$ to be just the restriction (pullback) to~$\partial \Sigma$ of the bulk fields.
This, taking into account the compatibility conditions of def.~\ref{classical_BV-BFV}, fixes the remaining boundary data.
From~\eqref{BF_mCME} we get
\beq\label{bound.sympl.sruct.}
 \alpha_{\partial \Sigma} = \int_{\partial \Sigma} \langle \mathbb{B},\delta\mathbb{A} \rangle~, \qquad 
   \omega_{\partial \Sigma} = \delta \alpha_{\partial \Sigma}= -\int_{\partial \Sigma} \langle \delta\mathbb{B},\delta\mathbb{A} \rangle~,
\eeq
the boundary homological vector field is the projection of the bulk homological vector field
\beq\label{YM_BFV_charge}
 \mathcal{Q}_{\partial \Sigma} = \mathrm{d}\pi\,\mathcal{Q}_\Sigma = \int_{\partial \Sigma} \Big(\Big\langle \mathrm{d}\mathbb{A}
  +\frac{1}{2}[\mathbb{A},\mathbb{A}],\frac{\delta}{\delta\mathbb{A}}\Big\rangle 
   + \Big\langle \mathrm{d}\mathbb{B}+ \mathrm{ad}^*_{\mathbb{A}}\mathbb{B},\frac{\delta}{\delta\mathbb{B}}\Big\rangle\Big)
\eeq
and thus the boundary action is obtained as the Hamiltonian of~$\mathcal{Q}_{\partial \Sigma}$\,:
\beq\label{YM_bound_action}
 \mathcal{S}_{\partial \Sigma}=\int_{\partial \Sigma} \langle \mathbb{B} , \mathrm{d}\mathbb{A} +\frac{1}{2}[\mathbb{A},\mathbb{A}]  \rangle~.
\eeq
Notice that, for degree reasons, the area form~$\mu$ does not appear in the boundary data.
The boundary BFV manifold for 2D YM is thus exactly the same as in BF theory; actually, the only difference between the two theories is the area term in the bulk action and, consequently, the state (partition function) defined by the two quantum theories.

To quantize the theory, we need to choose a polarization of the space of boundary fields.
From~\eqref{bound.sympl.sruct.} we see that there are two simple choices of polarization: the $\mathbb{A}$-polarization~$\mathcal{P}_\mathbb{A}$\,,%
\footnote{
 In the terminology of~\cite{CMR:pert_quantum_BV}, this is ``$\mathbb{A}$-representation'', or ``$\frac{\delta}{\delta \mathbb{B}}$-polarization'' (as those are the vector fields spanning the tangential Lagrangian distribution on the phase space).
} 
for which the leaf space is decribed by the $\mathbb{A}$~fields 
\beq
 \mathcal{B}_{\partial \Sigma}^{\mathcal{P}_\mathbb{A}}=\Omega(\partial \Sigma;\mathfrak{g})[1]
\eeq
and the $\mathbb{B}$-polarization~$\mathcal{P}_\mathbb{B}$\,, for which the leaf space is described by the $\mathbb{B}$~fields 
\beq
 \mathcal{B}_{\partial \Sigma}^{\mathcal{P}_\mathbb{B}}=\Omega(\partial \Sigma;\mathfrak{g}^*)~.
\eeq
We will, in the rest of this paper, always use these two transversal polarizations, arbitrarily splitting the boundary of a manifold into the disjoint%
\footnote{
 However, when we start considering corners in section~\ref{corners_and_2d_YM}, the disjointness assumption will fail along codimension~2 strata.
}
union of two components, $\partial \Sigma := \partial_{\mathbb{B}} \Sigma \sqcup \partial_{\mathbb{A}} \Sigma$\,, and choosing the product polarization which assigns the $\mathbb{B}$-polarization to the first and the $\mathbb{A}$-polarization to the latter boundary component: 
\beq
 \mathcal{B}_{\partial \Sigma}^{\mathcal{P}} = \Omega(\partial_{\mathbb{A}} \Sigma;\mathfrak{g})[1]\oplus \Omega(\partial_{\mathbb{B}} \Sigma;\mathfrak{g}^*)~.
\eeq

The boundary one-form~$\alpha_{\partial \Sigma}$ does not vanish on the fibers of this polarization (cf.~\eqref{bound.sympl.sruct.}) but it can be adapted to this choice using the transformation~\eqref{gauge_transf.}:
\beq
 \alpha^{\mathcal{P}}_{\partial \Sigma} &= \alpha_{\partial \Sigma} + \delta \int_{\partial_{\mathbb{B}}\Sigma} \langle \mathbb{B},\mathbb{A}\rangle =
  \int_{\partial_{\mathbb{A}}\Sigma} \langle \mathbb{B},\delta \mathbb{A}\rangle 
   - \int_{\partial_{\mathbb{B}}\Sigma} \langle \delta \mathbb{B},\mathbb{A}\rangle	~,\\
 \mathcal{S}^\mathcal{P}_{\Sigma} &= \mathcal{S}_{\Sigma} - \int_{\partial_{\mathbb{B}}\Sigma} \langle \mathsf{B},\mathsf{A}\rangle ~.
\eeq

We can now quantize, with the above polarization, the boundary action to obtain the coboundary operator~$\Omega_{\partial \Sigma}^\mathcal{P}$\,:
\beq\label{BF_Omega}
 \Omega_{\partial \Sigma}^\mathcal{P} = &\int_{\partial_{\mathbb{A}}\Sigma}  \mathrm{i}\hbar \bigg( \mathrm{d}\mathbb{A}^a 
  + \frac{1}{2}f^a_{bc}\mathbb{A}^b\mathbb{A}^c\bigg)\frac{\delta}{\delta \mathbb{A}^a} 
   + \int_{\partial_{\mathbb{B}}\Sigma} \bigg(\mathrm{i}\hbar\mathrm{d}\mathbb{B}_a\frac{\delta}{\delta \mathbb{B}_a} 
    -\frac{\hbar^2}{2}f^a_{bc}  \mathbb{B}_a\frac{\delta}{\delta \mathbb{B}_b}\frac{\delta}{\delta \mathbb{B}_c} \bigg)~.
\eeq
Here $f^a_{bc}$ are the structure constants of the Lie algebra $\mathfrak{g}$.%
\footnote{
 This is the quantization of the BFV action with a particular choice of ordering, putting the derivatives on the right and multiplication operators on the left. 
 Notice that, in fact, switching the ordering to the opposite one does not change the resulting operator due to unimodularity of $\mathfrak{g}$ (which follows from the existence of a non-degenerate invariant pairing).
}

To write the partition function we lift~$\mathcal{B}_{\partial \Sigma}^{\mathcal{P}}$ to~$\mathfrak{F}_\Sigma=\widetilde{\mathcal{B}}_{\partial \Sigma}^{\mathcal{P}}\times\mathcal{Y}_\Sigma$ by taking (discontinuous) bulk extensions~$(\widetilde{\mathbb{A}},\widetilde{\mathbb{B}})$ of the boundary fields (cf. the discussion in~\cite{CMR:pert_quantum_BV},~Section~3.4):
\beq
 \mathcal{Y}_\Sigma = \Omega(\Sigma,\partial_\mathbb{A} \Sigma ; \mathfrak{g})[1] \oplus \Omega(\Sigma,\partial_{\mathbb{B}} \Sigma;\mathfrak{g}^*)~.
\eeq
We then split the bulk fields~$\mathcal{Y}_\Sigma$ into residual fields~$(\mathsf{a},\mathsf{b})\in\mathcal{V}_\Sigma$ and fluctuations~$(\alpha,\beta)\in\mathcal{Y}'_\Sigma$\,:
\beq
 \mathsf{A} = \widetilde{\mathbb{A}} + \mathsf{a} + \alpha~,\qquad
 \mathsf{B} = \widetilde{\mathbb{B}} + \mathsf{b} + \beta~.
\eeq
The minimal realization of $\mathcal{V}_\Sigma$ is the space of zero-modes of the theory, which is written in terms of relative cohomology (c.f.~Section~\ref{BF Propagators and Axial Gauge}):
\beq
 \mathcal{V}_\Sigma = H(\Sigma,\partial_\mathbb{A} \Sigma ; \mathfrak{g})[1] \oplus H(\Sigma,\partial_{\mathbb{B}} \Sigma;\mathfrak{g}^*)~.
\eeq
We can finally define the partition function as the (perturbative) path integral:
\beq
 Z_\Sigma[\mathbb{A},\mathbb{B}; \mathsf{a}, \mathsf{b}] = \int_{\mathcal{L}\subset\mathcal{Y}'_\Sigma} \mathfrak{D}[\alpha,\beta] ~ 
  \mathrm{e}^{\frac{\mathrm{i}}{\hbar} \mathcal{S}^\mathcal{P}_{\Sigma}(\widetilde{\mathbb{A}} + \mathsf{a} + \alpha,\widetilde{\mathbb{B}} + \mathsf{b} + \beta)}~.
\eeq
This construction uses a choice of a Lagrangian subspace $\mathcal{L}\subset \mathcal{Y}'_\Sigma$ -- the choice of gauge-fixing.
\begin{remark}
 To avoid the appearance of ill-defined derivatives of the discontinuous fields~$(\widetilde{\mathbb{A}},\widetilde{\mathbb{B}})$ in the bulk action~$\mathcal{S}^\mathcal{P}_{\Sigma}$\,, we integrate by parts rewriting it as:
 \beq\label{YM_pert.action}
  &\mathcal{S}^\mathcal{P}_{\Sigma}(\widetilde{\mathbb{A}} + \mathsf{a} + \alpha,\widetilde{\mathbb{B}} + \mathsf{b} + \beta) \\
  &= \mathcal{S}^\mathcal{P}_{\Sigma}(\mathsf{a} + \alpha, \mathsf{b} + \beta)
    + \frac{1}{2} \int_\Sigma (\mathsf{b}+\beta,\mathsf{b}+\beta) \, \mu 
     + \int_{\partial_{\mathbb{A}}\Sigma} \langle \mathsf{b}+\beta , \widetilde{\mathbb{A}} \rangle
      - \int_{\partial_{\mathbb{B}}\Sigma} \langle \widetilde{\mathbb{B}},\mathsf{a}+\alpha\rangle	\\
  &= \int_\Sigma \Big(\langle \beta,\mathrm{d}\alpha \rangle +\langle \beta, [\alpha,\mathsf{a}] \rangle 
    +\frac12 \langle \beta , [\alpha,\alpha] \rangle  + \frac12\langle \beta , [\mathsf{a},\mathsf{a}] \rangle 
     + \langle \mathsf{b}, [\alpha,\mathsf{a}] \rangle +\frac12 \langle \mathsf{b} , [\alpha,\alpha] \rangle 
      + \frac12\langle \mathsf{b} , [\mathsf{a},\mathsf{a}] \rangle \\
  &+ \frac12\langle \beta ,\beta \rangle\mu + \langle\beta ,\mathsf{b} \rangle \mu
   + \frac12\langle \mathsf{b},\mathsf{b} \rangle\mu \Big) 
    + \int_{\partial_\mathbb{A}\Sigma} \Big(\langle \beta, \mathbb{A}\rangle + \langle \mathsf{b},\mathbb{A} \rangle\Big) 
     - \int_{\partial_\mathbb{B}\Sigma}\Big(\langle \mathbb{B}, \alpha\rangle + \langle \mathbb{B}, \mathsf{a}\rangle\Big) ~.
 \eeq
 The boundary fields thus act as currents in the perturbative expansion of the partition function.
\end{remark}

We are now in the position of writing down the diagrammatic elements of the Feynman diagrams expansion of the theory:%
\footnote{Zero-modes are here understood as ``loose'' half-hedges.}
\beq\label{vertices}
   \begin{tikzpicture}[scale=.23, baseline=(x.base)]
   \coordinate (x) at (0,0);
    \draw[decoration={markings, mark=at position 0.5 with {\arrow{>}}},
        postaction={decorate},color=black, thick] (-2,0)  to node[above=15pt]{propagator}node[below=3pt]{$\eta$} (2,0);     
   \end{tikzpicture}
   \hspace{.6cm}
   \begin{tikzpicture}[scale=.23, baseline=(x.base)]
    \coordinate (x) at (0,0);
    \draw[color=black, thick] (-3,0) to node[above=15pt]{$\mathbb{B}$ boundary source} node[below=3pt]{$\mathbb{B}$}  (3,0);  
    \draw[thick, dotted, color=black] (-4.35,0) to (4.35,0);
    \draw[decoration={markings, mark=at position 0.5 with {\arrow{<}}},
        postaction={decorate},color=black, thick] (0,1.5) to (0,0);
    \draw[color=black, thick, fill=black!10] (0,0) circle (10pt); 
   \end{tikzpicture}
   \hspace{.6cm}
   \begin{tikzpicture}[scale=.23, baseline=(x.base)]
   \coordinate (x) at (0,0);
    \draw[color=black, thick] (-3,0) to node[above=15pt]{$\mathbb{A}$ boundary source} node[below=3pt]{$\mathbb{A}$} (3,0);     
    \draw[thick, dotted, color=black] (-4.35,0) to (4.35,0);    
    \draw[decoration={markings, mark=at position 0.5 with {\arrow{>}}},
        postaction={decorate},color=black, thick] (0,1.5) to (0,0);
    \draw[fill=black, color=black, thick] (0,0) circle (10pt);     
   \end{tikzpicture}\hspace{1cm}\\
   \\
   \begin{tikzpicture}[scale=.23, baseline=(x.base)]
    \node[above=18pt] (x) at (0,0) {BF interaction};
    \draw[decoration={markings, mark=at position 0.5 with {\arrow{>}}},
        postaction={decorate},color=black, thick] (-1.5,0) to (0,0);     
    \draw[decoration={markings, mark=at position 0.5 with {\arrow{<}}},
        postaction={decorate},color=black, thick] (60:1.5) to (0,0); 
    \draw[decoration={markings, mark=at position 0.5 with {\arrow{<}}},
        postaction={decorate},color=black, thick] (-60:1.5) to (0,0);     
        
   \end{tikzpicture}
   \hspace{.6cm}
   \begin{tikzpicture}[scale=.23, baseline=(x.base)]
    \node[above=18pt] (x) at (0,0) {YM interaction};
    
    \draw[decoration={markings, mark=at position 0.5 with {\arrow{>}}},
        postaction={decorate},color=black, thick] (60:1) to (-1,0); 
    \draw[decoration={markings, mark=at position 0.5 with {\arrow{>}}},
        postaction={decorate},color=black, thick] (-60:1) to (-1,0);     
        
    \draw[fill=black, color=black, thick] (-1.3,-.3) rectangle (-.7,.3);     
    \node[left=2pt] at (-1,0) {$\mu$};
   \end{tikzpicture}
   \hspace{.6cm}
   \begin{tikzpicture}[scale=.23, baseline=(x.base)]
    \node[above=18pt] (x) at (0,0) {$\mathsf{b}$ zero-modes};
    \draw[decoration={markings, mark=at position 0.8 with {\arrow{>}}},
        postaction={decorate}, color=black, thick] (-2,0) to (0,0);

    \draw[color=black, thick, fill=black!10] (-2,0) circle (10pt);         
    \node[left=2pt] at (-2,0) {$\mathsf{b}$};
   \end{tikzpicture}
   \hspace{.6cm}
   \begin{tikzpicture}[scale=.23, baseline=(x.base)]
    \node[above=18pt] (x) at (0,0) {$\mathsf{a}$ zero-modes};
    
    \draw[decoration={markings, mark=at position 0.4 with {\arrow{>}}},
        postaction={decorate}, color=black, thick] (-2,0) to (0,0);            
    \draw[color=black, thick, fill=black!] (0,0) circle (10pt);         
    \node[right=2pt] at (0,0) {$\mathsf{a}$};    
   \end{tikzpicture}
\eeq

With these vertices we can compose a large set of non-trivial Feynman diagrams (e.g.~figure~\ref{diagrams_ex}).
The general strategy will be to cut the surface, and hence Feynman diagrams, in such a way that there is a simple choice of propagators on each component which allows us to compute the partition function for that surface.
Then, using the gluing properties of BV-BFV theories, we can glue back all the pieces to recover the partition function on the original surface we started with.
This procedure can be viewed as a method to construct a complicated propagator on the starting surface which, though, allows explicit computations.

\begin{figure}[h]
 \centering
\begin{tikzpicture}[scale=.44]
    
    \draw[decoration={markings, mark=at position 0.5 with {\arrow{>}}},
        postaction={decorate},color=black, thick] (-20:2) arc (-20:30:2);
    \draw[decoration={markings, mark=at position 0.5 with {\arrow{>}}},
        postaction={decorate},color=black, thick] (30:2) arc (30:80:2);        
    \draw[decoration={markings, mark=at position 0.5 with {\arrow{>}}},
        postaction={decorate},color=black, thick] (80:2) arc (80:120:2);
    \draw[decoration={markings, mark=at position 0.5 with {\arrow{>}}},
        postaction={decorate},color=black, thick] (120:2) arc (120:160:2);
    \draw[decoration={markings, mark=at position 0.5 with {\arrow{>}}},
        postaction={decorate},color=black, thick] (160:2) arc (160:240:2);
    \draw[decoration={markings, mark=at position 0.5 with {\arrow{>}}},
        postaction={decorate},color=black, thick] (240:2) arc (240:300:2);
    \draw[decoration={markings, mark=at position 0.5 with {\arrow{>}}},
        postaction={decorate},color=black, thick] (300:2) arc (300:340:2);

    \draw[thick] (-4,-4)node[left]{$\mathbb{A}$} to (-3,-4)node{$\bullet$};
    \draw[thick] (-3,-4) to (-1,-4)node{$\bullet$};
    \draw[thick] (-1,-4) to (1,-4)node{$\bullet$};
    \draw[thick] (1,-4) to (3,-4)node{$\bullet$};
    \draw[thick] (3,-4) to (4,-4);
    
    \draw[decoration={markings, mark=at position 0.5 with {\arrow{>}}},
        postaction={decorate},color=black, thick] (-20:2) to[out=-20, in=90] (1,-4);
    \draw[color=black, thick] (80:2) to (80:3);
    \draw[decoration={markings, mark=at position 0.6 with {\arrow{>}}},
        postaction={decorate},color=black, thick] (30:2)   to[out=30 , in=90] (3,-4);
    \draw[decoration={markings, mark=at position 0.5 with {\arrow{>}}},
        postaction={decorate},color=black, thick] (160:2) to[out=160, in=90] (-3,-4);
    \draw[decoration={markings, mark=at position 0.5 with {\arrow{>}}},
        postaction={decorate},color=black, thick] (240:2) to[out=240, in=90] (-1,-4);        
        
    \draw[decoration={markings, mark=at position 0.5 with {\arrow{>}}},
        postaction={decorate},color=black, thick] (120:2) to[out=-60, in=120] (0,0);
    \draw[decoration={markings, mark=at position 0.5 with {\arrow{>}}},
        postaction={decorate},color=black, thick] (300:2) to[out=120, in=300] (0,0);

    \draw[fill=black] (-.2,-.2) rectangle (.2,.2);
    \draw[fill=black] (80:3) circle (.2);
        
   \end{tikzpicture}
   \hspace{2.3cm}
   \begin{tikzpicture}[scale=.44]
    \node (a) at (0,-4) {};
    \node (b) at ($(a)+(90:1.5)$) {};
    \node (c1) at ($(b)+(30:1.5)$) {};
    \node (c2) at ($(b)+(150:1.5)$) {};
    \node (d1) at ($(c1)+(90:1.5)$) {};
    \node (d2) at ($(c1)+(-30:1.5)$) {$\bullet$};
    \node (e1) at ($(d1)+(30:1.5)$) {};
    \node (e2) at ($(d1)+(150:1.5)$) {};  
    \node (f1) at ($(c2)+(90:1.5)$) {};
    \node (f2) at ($(c2)+(210:1.5)$) {$\bullet$};   
    \node (g1) at ($(e1)+(90:1.5)$) {$\bullet$};
    \node (g2) at ($(e1)+(-30:1.5)$) {$\bullet$};     
    \coordinate (z) at ($(f1)+(30:.75)$);

    \draw[decoration={markings, mark=at position 0.6 with {\arrow{>}}},
        postaction={decorate},color=black, thick] (0,-4) to ($(a)+(90:1.5)$);
    \draw[decoration={markings, mark=at position 0.6 with {\arrow{>}}},
        postaction={decorate},color=black, thick] ($(a)+(90:1.5)$) to ($(b)+(30:1.5)$);    
    \draw[decoration={markings, mark=at position 0.6 with {\arrow{>}}},
        postaction={decorate},color=black, thick] ($(a)+(90:1.5)$) to ($(b)+(150:1.5)$);   
        
    \draw[decoration={markings, mark=at position 0.6 with {\arrow{>}}},
        postaction={decorate},color=black, thick] ($(b)+(30:1.5)$) to ($(c1)+(90:1.5)$);    
    \draw[color=black, thick] ($(b)+(30:1.5)$) to ($(c1)+(-30:1.5)$);
        
    \draw[decoration={markings, mark=at position 0.6 with {\arrow{>}}},
        postaction={decorate},color=black, thick] ($(c1)+(90:1.5)$) to ($(d1)+(30:1.5)$);    
    \draw[decoration={markings, mark=at position 0.6 with {\arrow{>}}},
        postaction={decorate}, color=black, thick] ($(c1)+(90:1.5)$) to (z);

    \draw[decoration={markings, mark=at position 0.6 with {\arrow{>}}},
        postaction={decorate}, color=black, thick] ($(b)+(150:1.5)$) to (z);
 
    \draw[color=black, thick] ($(b)+(150:1.5)$) to ($(c2)+(210:1.5)$);

    \draw[color=black, thick] ($(d1)+(30:1.5)$) to ($(e1)+(90:1.5)$);    
    \draw[color=black, thick] ($(d1)+(30:1.5)$) to ($(e1)+(-30:1.5)$);             
  
    \draw[thick] (0,-4) to (4,-4)node[right]{$\mathbb{B}$};
    \draw[thick] (-3,-4) to (0,-4);
    
    \draw[color=black, thick, fill=black!10] (a) circle (6pt); 

    \draw[fill=black] ($(z)-(.2,.2)$) rectangle ($(z)+(.2,.2)$);
   \end{tikzpicture}
 \caption{Two examples of the many possible Feynman diagrams for 2D YM on a surface with boundary.}
 \label{diagrams_ex}
\end{figure}
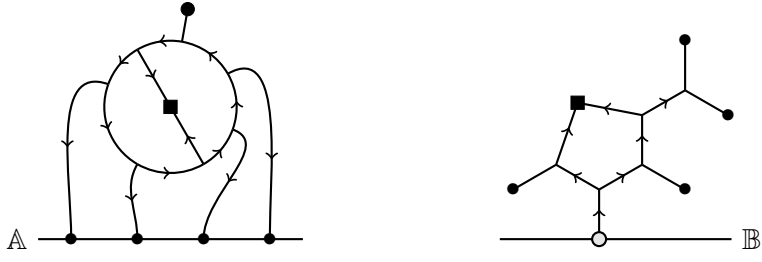

\subsubsection{$\Omega$-cohomology in $\mathbb{A}$-polarization on a circle}\label{Omega_A_cohomology}

 In~\cite{CMR:pert_quantum_BV} it was proven that the partition function for 2D YM (and other perturbations of abelian BF) solves the mQME.
 In the following sections, to simplify some computations, we will exploit this fact by choosing a suitable representative for the cohomology class of the partition function.
 In particular it will be useful to know the cohomology of $\Omega$\,, in ghost degree zero, on the space of boundary fields in $\mathbb{A}$ polarization.
 
 From equation~\eqref{BF_Omega} we see that $\Omega$ acts as a gauge transformation, thus $\Omega$-closed functionals of ghost degree zero are just gauge-invariant functionals of the connection, which are in one-to-one correspondence with class functions on the simply connected Lie group~$G$ integrating the Lie algebra~$\mathfrak{g}$\,.
 
 An alternative approach is to split~$\Omega$ into the ``abelian'' part, i.e. the de Rham differential~$\mathrm{d}$\,, plus a perturbation~$\bm\delta$ (not to be confused with the de~Rham operator $\delta$ on fields) containing the structure constants of the Lie algebra:%
 \beq\label{Omega_A}
  \Omega= \underset{\mathrm{d}}{\underbrace{\int_{S^1} \mathrm{d}\mathbb{A}^a_{(0)} \frac{\delta}{\delta \mathbb{A}^a_{(1)}} }} 
   + \underset{\bm{\delta}}{\underbrace{\int_{S^1} \bigg( \frac{1}{2}f_{bc}^a \mathbb{A}^b_{(0)}\mathbb{A}^c_{(0)} \frac{\delta}{\delta \mathbb{A}^a_{(0)}} 
    + f_{bc}^a \mathbb{A}^b_{(0)}\mathbb{A}^c_{(1)} \frac{\delta}{\delta \mathbb{A}^a_{(1)}} \bigg)}} ~.
 \eeq
 We can then compute the cohomology via the homological perturbation lemma~\cite{Gugenheim:hpt}.
 The cohomology of $\mathrm{d}$ is given by functions on the de Rham cohomology $H_{\mathrm{dR}}^\bullet(S^1;\mathfrak{g})[1]$\,; choosing the coordinate $t$ on the circle, these can be represented as functions of the ``constant fields'' $\underline{\mathbb{A}}_{(0)}$ and $\underline{\mathbb{A}}_{(1)}\mathrm{d}t$\,, where~$\underline{\mathbb{A}}_{(0)}\in\mathfrak{g}[1]$ and~$\underline{\mathbb{A}}_{(1)}\in\mathfrak{g}$\,.
 Now if we compute the cohomology of the induced differential
 \beq\label{delta_A}
  \underline{\bm{\delta}}=\frac{1}{2}f_{bc}^a \underline{\mathbb{A}}^b_{(0)}\underline{\mathbb{A}}^c_{(0)} \frac{\delta}{\delta \underline{\mathbb{A}}^a_{(0)}} + f_{bc}^a \underline{\mathbb{A}}^b_{(0)}\underline{\mathbb{A}}^c_{(1)} \frac{\delta}{\delta \underline{\mathbb{A}}^a_{(1)}}
 \eeq
 on $H_{\mathrm{dR}}^\bullet(S^1;\mathfrak{g})[1]$\,, we get that in ghost degree zero it is given by $G$-invariant functions on the Lie algebra~$\mathfrak{g}$\,.
 Comparing with the previous answer, we see that the correct $\Omega$-cohomology corresponds to the subspace of $G$-invariant functions on~$\mathfrak{g}$ coming as the pullback by the exponential map $\mathrm{exp}\colon \mathfrak{g} \rightarrow G$ of class functions on $G$\,. 
 Such functions on~$\mathfrak{g}$ are determined by their values on the fundamental domain~$B_0$ of the exponential map (e.g. for $G=\mathrm{SU}(2)$\,, $B_0$~is a ball in $\mathfrak{g}$ centered at the origin).%
 \footnote{
  Generally, $B_0$~is the connected component of the origin in $\mathfrak{g} - \phi^{-1}(0)$ where the function $\phi\colon \mathfrak{g} \rightarrow \mathbb{R}$\,, $\phi(x):= \det \frac{\sinh(\mathrm{ad}_x /2)}{\mathrm{ad}_x /2}$\,, is the Jacobian of the exponential map. 
  In other words, $B_0$~is the set of elements~$x\in \mathfrak{g}$ such that all eigenvalues of~$\mathrm{ad}_x$ 
  are contained in the interval $(-2\pi \mathrm{i}, 2\pi \mathrm{i}) \subset \mathrm{i} \mathbb{R}$\,.
 }
 The discrepancy between the correct cohomology of~$\Omega$ and the cohomology of~$\underline{\bm\delta}$ is due to a convergence issue arising in homological perturbation theory.%
 \footnote{
  This problem is a version of the Gribov ambiguity (Gribov copies) problem in 4d Yang-Mills theory -- the problem of gauge-fixing ``section'' intersecting the gauge orbits more than once.
  For that reason, we will refer to~$B_0$ as the ``Gribov region''.
 }
 
 The useful remark coming from this discussion is that, modulo $\Omega$-exact terms, the partition function and the physical observables can be represented as a ($G$-invariant) function of \emph{constant fields} valued in a neighbourhood of zero in~$\mathfrak{g}$\,.
 Moreover, in ghost degree zero, for $\Omega$-closed objects depending only on ${\mathbb{A}}_{(1)}$\,, any ``reduced wavefunction''~$\Psi(\underline{\mathbb{A}}_{(1)})$ can be lifted to an $\Omega$-closed function in the non-reduced space of states by evaluating~$\Psi$ on the logarithm of the holonomy of~$\mathbb{A}_{(1)}$\,.

\subsubsection{Hodge propagators and axial gauge}\label{BF Propagators and Axial Gauge}

The kinetic term in the YM action~\eqref{YM_BV_action} is of the kind~$\int_{\Sigma} \langle \mathsf{B},{D}\mathsf{A}\rangle$\,, where~$D$ is a differential on~$\mathcal{Y}_\Sigma^\bullet$ (in our case $D=\mathrm{d}$, but the construction presented in this section can be applied to more general cases).
Since the propagator is the integral kernel of the inverse of~$D$\,, we want to find where the differential can actually be inverted.

Let~$(K,i,p)$ be a \emph{retraction} of~$(\mathcal{Y}^\bullet_\Sigma, D)$ on its cohomology $(\mathcal{V}^\bullet_\Sigma,0)$\,, i.e. a triple where $K\colon \mathcal{Y}_\Sigma^\bullet\longrightarrow \mathcal{Y}_\Sigma^{\bullet-1}$ is a chain homotopy, $i\colon \mathcal{V}^\bullet_\Sigma \hookrightarrow \mathcal{Y}_\Sigma^\bullet $ a chain inclusion and $p\colon \mathcal{Y}_\Sigma^\bullet\twoheadrightarrow \mathcal{V}_\Sigma^\bullet$ a chain projection satisfying:
\beq\label{chain_contraction}
 K^2= p\circ K = K\circ i = 0 ~, \qquad i\circ p = \mathrm{id}~,\qquad 
  D K + K D = \mathrm{id} - i\circ p~.
\eeq
Then the complex $\mathcal{Y}^\bullet_\Sigma$ has a \emph{weak Hodge decomposition}:
\beq
 \mathcal{Y}^\bullet_\Sigma = \underset{\simeq \mathcal{V}^\bullet_\Sigma}{\underbrace{\Pi \mathcal{Y}^\bullet_\Sigma }} \oplus 
  \underset{= \mathcal{Y}'^\bullet_\Sigma}{\underbrace{K\mathcal{Y}^{\bullet+1}_\Sigma \oplus D\mathcal{Y}^{\bullet-1}_\Sigma}}~,
\eeq
where we have defined~$\Pi := i \circ p$\,.
From eq.~\eqref{chain_contraction}, the differential~$D$ is invertible as an operator from the image of the chain homotopy~$K$ to $D$-exact cochains~$D\colon K\mathcal{Y}^{\bullet+1}_\Sigma \longrightarrow D\mathcal{Y}^{\bullet-1}_\Sigma$ and its inverse is precisely the chain homotopy itself: $K = D^{-1}$\,.

The gauge can thus be fixed on the Lagrangian~$\mathcal{L} = K\mathcal{Y}_\Sigma$\,; the \emph{propagator}~$\eta(x';x)$\,, with this gauge-fixing, is defined as the integral kernel of the chain homotopy:
\beq
 K \omega(x) = \int_{\Sigma\ni x'} \eta(x;x') \wedge \omega(x') ~,\qquad \omega\in\mathcal{Y}_\Sigma ~.
\eeq

When spacetime is a product manifold, $\Sigma=\Sigma_1\times \Sigma_2$\,, there is a particular class of propagators which can be induced on~$\Sigma$ from lower-dimensional propagators on the two factors~\cite{bonechi:PSM_on_closed.}.
Since the differential forms on a product manifold are the (closure) of the sum of products of the differential forms on the two factors, we have~$\mathcal{Y}_\Sigma=\mathcal{Y}_{\Sigma_1} \otimes \mathcal{Y}_{\Sigma_2}$\,.
For each pair of contractions~$(K_\ell,i_\ell,p_\ell)$ on the factors~$\mathcal{Y}_{\Sigma_\ell}$ we have an induced weak Hodge decomposition on~$\mathcal{Y}_\Sigma$\,:
\beq
 \mathcal{Y}_\Sigma= \overbrace{\big(\Pi_1 \mathcal{Y}_{\Sigma_1} \otimes \Pi_2 \mathcal{Y}_{\Sigma_2}\big)}^{=\Pi\mathcal{Y}_\Sigma\simeq \mathcal{V}_\Sigma} 
  \oplus\overbrace{\big(\Pi_1 \mathcal{Y}_{\Sigma_1}\otimes K_2 \mathcal{Y}_{\Sigma_2}\big)
   \oplus \big(K_1 \mathcal{Y}_{\Sigma_1}\otimes \mathcal{Y}_{\Sigma_2}\big)}^{K\mathcal{Y}_{\Sigma}}\\
 \oplus{\big(\Pi_1 \mathcal{Y}_{\Sigma_1}\otimes D_2\mathcal{Y}_{\Sigma_2}\big)\oplus \big(D_1 \mathcal{Y}_{\Sigma_1}
  \otimes \mathcal{Y}_{\Sigma_2}\big)}	~.
\eeq
The zero modes are the product of the zero modes on the two factors and the induced chain homotopy is $K=\Pi_1\otimes K_2 \oplus K_1\otimes \mathrm{id}_{\mathcal{Y}_{\Sigma_2}}$\,.
The associated gauge is called \emph{axial gauge}.
If we call~$\pi_\ell$ the integral kernel of~$\Pi_\ell$\,, the \emph{axial gauge propagator} is:
\beq\label{axial_gauge_prop.}
 \eta(x_1,x_2;x'_1,x'_2) = \pi_1(x_1;x_1') \wedge \eta_2(x_2;x'_2) + \eta_1(x_1;x_1')\wedge \delta(x_2;x_2') ~.
\eeq

\section{2D YM for surfaces of non-negative Euler characteristic}\label{2d_YM_for_chi>0}

In this section we will consider 2D YM on manifolds with codimension 1 boundaries.
With a good choice of propagators and exploiting the gluing properties of BV-BFV theories, we will be able in this setting to explicitly compute all Feynman diagrams and sum the perturbative series to find the complete partition function of this theory on disks and cylinders. 
The globalized realization of the partition function on a disk in the $\mathbb{A}$~polarization will coincide with the well-known non-perturbative solution of 2D YM~\cite{Migdal:1975zg, witten:2d_quantum_gauge}.%
\footnote{
 Although we can present, e.g., the sphere and the torus as assembled from building blocks considered in this section, globalization integrals for them are perturbatively obstructed, see Section~\ref{closed_surfaces}. 
 We obtain a non-singular globalized answer in these cases as a part of the general result of Section~\ref{corners_and_2d_YM}.
}

We consider a set of generators, under gluing, for orientable surfaces of non-negative Euler characteristic: the disk and the cylinder.
At the level of the field theory constructed on such surfaces, we have to also consider the data of the polarization associated to the boundaries.
The building blocks for 2D YM can be thus chosen to be the disk in the $\mathbb{B}$ polarization, the cylinder in $\mathbb{A}-\mathbb{A}$ polarization and the cylinder in the $\mathbb{B}-\mathbb{B}$ polarization.
Moreover, using the invariance of the theory under area-preserving diffeomorphisms, as a convenient choice we can concentrate the support of the volume form~$\mu$ near the boundaries; this allows to use as generators the above surfaces in the limit of zero area, i.e. for BF theory, at the cost of introducing as fourth generator a YM cylinder in $\mathbb{A}-\mathbb{B}$ polarization with finite volume~(figure~\ref{building_blocks_x>o}).

\begin{figure}[h]
 \centering
   \begin{tikzpicture}[scale=.2, baseline=(y.base)] 

    \coordinate (x0) at (1,4);
    \coordinate (x1) at (1,-4);
    \coordinate (y) at (-3,0);
    
    \filldraw[fill=black!15, thick, draw=black] (x0) to[out=180, in=90] (y) to[out=-90, in=180] (x1) to[out=0, in=0] node[right]{$\mathbb{B}$} (x0);
    
    \draw[dotted, thick, draw=black] (x0) to[out=180, in=180] (x1) to[out=0, in=0] (x0) ;
    
    \node (k1) at (1,0) {BF};
   \end{tikzpicture}
   \hspace{.7cm}
   \begin{tikzpicture}[scale=.15, baseline=(y.base)] 
    \coordinate (x) at (6,0);
    \coordinate (z) at (-1,0);
   
    \coordinate (x0) at (-6,12);
    \coordinate (x1) at (-6,4);
    \coordinate (y0) at (-6,-4);
    \coordinate (y1) at (-6,-12);

    \filldraw[fill=black!5, thick, draw=black] (x0) to[out=180, in=180] node[left]{$\mathbb{A}$} (x1) to[out=0, in=0] (x0);
    \filldraw[fill=black!5, thick, draw=black] (y0) to[out=180, in=180] node[left]{$\mathbb{A}$} (y1) to[out=0, in=0] (y0);

    \filldraw[fill=black!15, thick, draw=black] (x0) to[out=0, in=90] (x) to[out=-90, in=0] (y1) to[out=0, in=0] (y0) 
      to[out=0, in=-90] (z) to[out=90, in=0] (x1) to[out=0, in=0] (x0);
      
    \node (k1) at (2.2,0) {BF};      
   \end{tikzpicture}
   \hspace{1cm}
   \begin{tikzpicture}[scale=.15, baseline=(y.base)] 
    \coordinate (x) at (-6,0);
    \coordinate (z) at (1,0);
   
    \coordinate (x0) at (6,12);
    \coordinate (x1) at (6,4);
    \coordinate (y0) at (6,-4);
    \coordinate (y1) at (6,-12);
      
    \filldraw[fill=black!15, thick, draw=black] (x0) to[out=180, in=90] (x) to[out=-90, in=180] (y1) to[out=0, in=0]node[right]{$\mathbb{B}$} (y0) 
      to[out=180, in=-90] (z) to[out=90, in=180] (x1) to[out=0, in=0]node[right]{$\mathbb{B}$} (x0);

    \draw[dotted, thick, draw=black] (x0) to[out=180, in=180] (x1);
    \draw[dotted, thick, draw=black] (y0) to[out=180, in=180] (y1);
    
    \node (k1) at (-2.2,0) {BF};    
   \end{tikzpicture}   
   \hspace{.7cm}
   \begin{tikzpicture}[scale=.2, baseline=(x.base)] 
    \coordinate (x) at (0,0);
    \node (z0) at (0,5.2) {};
    \node (z1) at (0,-5.2) {};
   
    \coordinate (x0) at (-6,4);
    \coordinate (x1) at (-6,-4);
    \coordinate (y0) at (6,4);
    \coordinate (y1) at (6,-4);

    \filldraw[fill=black!5, thick, draw=black] (x0) to[out=180, in=180] node[left]{$\mathbb{A}$} (x1) to[out=0, in=0] (x0);
    
    \filldraw[fill=black!40, thick, draw=black] (x0) to[out=0, in=0] (x1) to[out=0, in=180] (y1) to[out=0, in=0]node[right]{$\mathbb{B}$} (y0) 
      to[out=180, in=0] (x0);
    
    \draw[dotted, thick, draw=black] (y0) to[out=180, in=180] (y1);
    \node (z) at (0,0) {YM};
   \end{tikzpicture} 

 \caption{Building blocks for 2D YM on surfaces with non-negative Euler characteristic.}
 \label{building_blocks_x>o}
\end{figure}
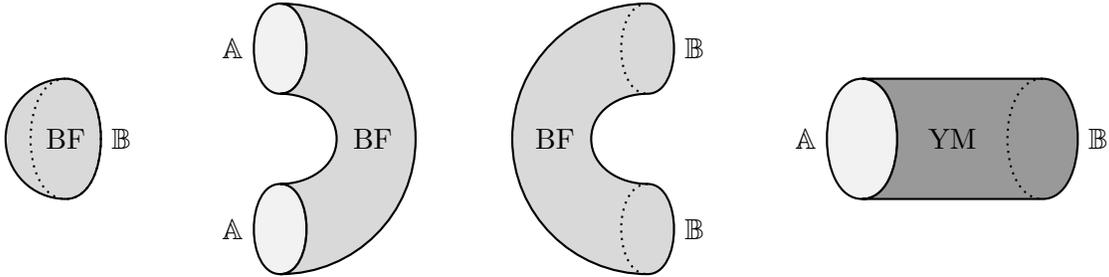

\subsection{$\mathbb{A}$-$\mathbb{B}$ polarization on the cylinder} \label{Sect:AB_cyl}

Let us start studying the BF theory on the cylinder, $\Sigma=S^1\times I\ni(\tau,t)$\,, $I=[0,1]$\,.
We will firstly choose $\mathbb{B}$~polarization on $S^1\times\{0\}=\partial_{\mathbb{B}} \Sigma$ and $\mathbb{A}$~polarization on $S^1\times\{1\}=\partial_{\mathbb{A}} \Sigma$\,.
The space of bulk fields, with this polarization, is $\mathcal{Y}=\Omega(\Sigma,\partial_{\mathbb{B}} \Sigma;\mathfrak{g})[1]\oplus\Omega(\Sigma,\partial_{\mathbb{A}} \Sigma;\mathfrak{g})$\,.
Since the relative cohomology $H(\Sigma,\partial_i \Sigma)$ is trivial with the above choice of boundaries, we have no zero-modes.
Thus the connected diagrams contributing to the effective action of the theory are trees with one root on $\partial_{\mathbb{B}} \Sigma$ and leafs on $\partial_{\mathbb{A}} \Sigma$ or 1-loop diagrams with trees rooted on a point of the loop and leafs on $\partial_{\mathbb{A}} \Sigma$ (figure~\ref{AB_cyl}).
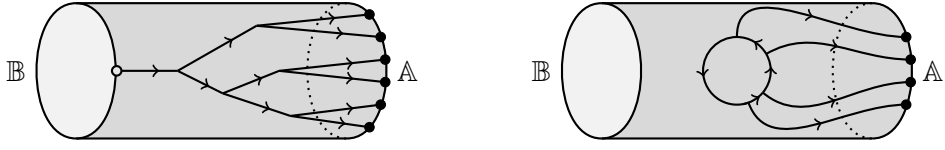
\begin{figure}[h]
 \centering
 \begin{tikzpicture}[scale=.3, baseline=(x.base)]
  
  \coordinate (x0) at (-6,-3);
  \coordinate (x1) at (-6,3);
  \coordinate (x2) at (6,3);
  \coordinate (x3) at (6,-3);
  
  \filldraw[fill=black!15](x0) to[out=0, in=0] (x1)to[out=0, in=180] (x2)to[in=0, out=0](x3)to[out=180, in=0](x0) ;
  \filldraw[fill=black!5, draw=black!15](x0) to[out=180, in=180] (x1)to[in=0, out=0](x0);

  \draw[color=black, thick] (x0) to[out=0, in=0] (x1);
  \draw[color=black, thick] (x1) to[out=180, in=180]node[left](x){$\mathbb{B}$} (x0);
        
  \draw[color=black, thick] (x3) to[out=0, in=0]node[right](x){$\mathbb{A}$} (x2);
  \draw[color=black, thick, dotted] (x2) to[out=180, in=180] (x3);

  \draw[color=black, thick](x1) to[out=0, in=180] (x2) ;
  \draw[color=black, thick](x0) to[out=0, in=180] (x3);

  
   \coordinate (B) at (-4.2,0);
   
   \coordinate (v1) at (-1.5,0);
   \coordinate (v2) at (2,2);
   \coordinate (v3) at (.5,-1);
   \coordinate (v4) at (3,0);
   \coordinate (v5) at (3.5,-2);

   \coordinate (b1) at (7,2.5);
   \coordinate (b2) at (7.5,1.5);
   \coordinate (b3) at (7.7,.5);
   \coordinate (b4) at (7.7,-.5);
   \coordinate (b6) at (7.5,-1.5);
   \coordinate (b5) at (7,-2.5);

   \draw[decoration={markings, mark=at position 0.7 with {\arrow{>}}},
        postaction={decorate},color=black, thick] (B) to (v1);   
    
   \draw[decoration={markings, mark=at position 0.7 with {\arrow{>}}},
        postaction={decorate},color=black, thick] (v1) to (v2);     
   
   \draw[decoration={markings, mark=at position 0.7 with {\arrow{>}}},
        postaction={decorate},color=black, thick] (v1) to (v3);    
   
   \draw[decoration={markings, mark=at position 0.7 with {\arrow{>}}},
        postaction={decorate},color=black, thick] (v3) to (v4);    
   
   \draw[decoration={markings, mark=at position 0.7 with {\arrow{>}}},
        postaction={decorate},color=black, thick] (v3) to (v5);

    \draw[decoration={markings, mark=at position 0.7 with {\arrow{>}}},
        postaction={decorate},color=black, thick]  (v2) to (b1);   
    \draw[color=black, fill=black] (b1) circle (6pt);      
      
    \draw[decoration={markings, mark=at position 0.7 with {\arrow{>}}},
        postaction={decorate},color=black, thick]  (v2) to (b2);   
    \draw[color=black, fill=black] (b2) circle (6pt);     
    
    \draw[decoration={markings, mark=at position 0.7 with {\arrow{>}}},
        postaction={decorate},color=black, thick]  (v4) to (b3);   
    \draw[color=black, fill=black] (b3) circle (6pt);     
    
    \draw[decoration={markings, mark=at position 0.7 with {\arrow{>}}},
        postaction={decorate},color=black, thick]  (v4) to (b4);   
    \draw[color=black, fill=black] (b4) circle (6pt);     
    
    \draw[decoration={markings, mark=at position 0.7 with {\arrow{>}}},
        postaction={decorate},color=black, thick]  (v5) to (b5);   
    \draw[color=black, fill=black] (b5) circle (6pt);     
    
    \draw[decoration={markings, mark=at position 0.7 with {\arrow{>}}},
        postaction={decorate},color=black, thick]  (v5) to (b6);   
    \draw[color=black, fill=black] (b6) circle (6pt);     
    
    \draw[color=black, thick, fill=black!10] (B) circle (6pt); 
   
 \end{tikzpicture}
 \hspace{1cm}  
 \begin{tikzpicture}[scale=.3, baseline=(x.base)]
  
  \coordinate (x0) at (-6,-3);
  \coordinate (x1) at (-6,3);
  \coordinate (x2) at (6,3);
  \coordinate (x3) at (6,-3);
  
  \filldraw[fill=black!15](x0) to[out=0, in=0] (x1)to[out=0, in=180] (x2)to[in=0, out=0](x3)to[out=180, in=0](x0) ;
  \filldraw[fill=black!5, draw=black!15](x0) to[out=180, in=180] (x1)to[in=0, out=0](x0);

  \draw[color=black, thick] (x0) to[out=0, in=0] (x1);
  \draw[color=black, thick] (x1) to[out=180, in=180]node[left](x){$\mathbb{B}$} (x0);
        
  \draw[color=black, thick] (x3) to[out=0, in=0]node[right](x){$\mathbb{A}$} (x2);
  \draw[color=black, thick, dotted] (x2) to[out=180, in=180] (x3);

  \draw[color=black, thick](x1) to[out=0, in=180] (x2) ;
  \draw[color=black, thick](x0) to[out=0, in=180] (x3);


   \coordinate (v1) at (90:1.5);
   \coordinate (v2) at (30:1.5);
   \coordinate (v3) at (-40:1.5);
   \coordinate (v4) at (-70:1.5);
   
   \coordinate (b2) at (7.5,1.5);
   \coordinate (b3) at (7.7,.5);
   \coordinate (b4) at (7.7,-.5);
   \coordinate (b5) at (7.5,-1.5);
   
    \draw[decoration={markings, mark=at position 0.5 with {\arrow{>}}},
        postaction={decorate},color=black, thick] (90:1.5) arc (90:290:1.5);
   
    \draw[decoration={markings, mark=at position 0.5 with {\arrow{<}}},
        postaction={decorate},color=black, thick] (320:1.5) arc (320:280:1.5);
        
    \draw[decoration={markings, mark=at position 0.5 with {\arrow{<}}},
        postaction={decorate},color=black, thick] (30:1.5) arc (30:-40:1.5);
        
    \draw[decoration={markings, mark=at position 0.5 with {\arrow{>}}},
        postaction={decorate},color=black, thick] (30:1.5) arc (30:90:1.5);

    \draw[decoration={markings, mark=at position 0.5 with {\arrow{>}}},
        postaction={decorate},color=black, thick]  (v1) to[out=90, in=180] (b2);   
    \draw[color=black, fill=black] (b2) circle (6pt);     
    
    \draw[decoration={markings, mark=at position 0.5 with {\arrow{>}}},
        postaction={decorate},color=black, thick]  (v2) to[out=30, in=180] (b3);   
    \draw[color=black, fill=black] (b3) circle (6pt);     
    
    \draw[decoration={markings, mark=at position 0.5 with {\arrow{>}}},
        postaction={decorate},color=black, thick]  (v3) to[out=-40, in=180] (b4);   
    \draw[color=black, fill=black] (b4) circle (6pt);     
    
    \draw[decoration={markings, mark=at position 0.5 with {\arrow{>}}},
        postaction={decorate},color=black, thick]  (v4) to[out=-70, in=180] (b5);   
    \draw[color=black, fill=black] (b5) circle (6pt);     

 \end{tikzpicture}  
 \caption{Connected diagrams for non-abelian BF on the cylinder in $\mathbb{A}$-$\mathbb{B}$ polarization.}
 \label{AB_cyl}
\end{figure}
\\
To compute these diagrams we can use the axial-gauge, with propagator (cf.~\ref{Axial_Gauge_Cylinder}):
\beq
 \begin{tikzpicture}[scale=.5, baseline=(x.base)]
  \draw[decoration={markings, mark=at position 0.55 with {\arrow{>}}},
        postaction={decorate},color=black, thick]  (-1,0) node[left](x){$(t,\tau)$} to (1,0)node[right]{$(t',\tau')$};
 \end{tikzpicture}
 =\eta(t,\tau;t',\tau')= -\Theta(t'-t)\delta(\tau-\tau')(\mathrm{d}\tau'-\mathrm{d}\tau)~.
\eeq
Looking at this propagator we immediately notice that it is a zero-form on the interval~$I$\,.
Since each bulk vertex carries an integration over $S^1\times I$\,, the differential form associated to a diagram has the right form components (that is, s.t. its integral over the configuration space doesn't vanish) only if it doesn't contain any bulk vertex.
Thus there is only one non-vanishing diagram contributing to the effective action:
\beq
 \mathcal{S}_{\mathrm{BF}}^{\mathrm{eff}}[\mathbb{B},\mathbb{A}]=
  \begin{tikzpicture}[scale=.15, baseline=(x.base)]
  
  \coordinate (x0) at (-6,-3);
  \coordinate (x1) at (-6,3);
  \coordinate (x2) at (6,3);
  \coordinate (x3) at (6,-3);
  
  \filldraw[fill=black!15](x0) to[out=0, in=0] (x1)to[out=0, in=180] (x2)to[in=0, out=0](x3)to[out=180, in=0](x0) ;
  \filldraw[fill=black!5, draw=black!15](x0) to[out=180, in=180] (x1)to[in=0, out=0](x0);

  \draw[color=black, thick] (x0) to[out=0, in=0] (x1);
  \draw[color=black, thick] (x1) to[out=180, in=180]node[right](x){$\phantom{x}$} (x0);
        
  \draw[color=black, thick] (x3) to[out=0, in=0] (x2);
  \draw[color=black, thick, dotted] (x2) to[out=180, in=180] (x3);

  \draw[color=black, thick](x1) to[out=0, in=180] (x2) ;
  \draw[color=black, thick](x0) to[out=0, in=180] (x3);

  
   \coordinate (B) at (-4.2,0);
   \draw[decoration={markings, mark=at position 0.5 with {\arrow{>}}},
        postaction={decorate},color=black, thick]  (B) to (7.8,0);   
   \draw[color=black, fill=black] (7.8,0) circle (10pt);  
   \draw[color=black, thick, fill=black!10] (-4.2,0) circle (10pt);  

 \end{tikzpicture}
   =\int_{\partial_\mathbb{A}\Sigma} \langle p^* \mathbb{B},\mathbb{A}\rangle~,
\eeq
where $p\colon\Sigma\longrightarrow \partial_\mathbb{B} \Sigma$ is a projection to the $\mathbb{B}$-boundary.

The Yang-Mills action can be rewritten as a perturbation of BF:
 \beq
  \mathcal{S}_{\mathrm{YM}} = \mathcal{S}_{\mathrm{BF}} + \frac{1}{2}\int_\Sigma \mu\;\mathrm{tr}(\mathsf{B}^2)~.
 \eeq
The additional bivalent interaction vertex is proportional to the volume form $\mu$\,.
For degree counting reasons analogous to the one described above, the only additional non-vanishing Feynman diagram is the one containing a single YM vertex:
\beq
  \mathcal{S}_{\mathrm{YM}}^{\mathrm{eff}} =
    \begin{tikzpicture}[scale=.15, baseline=(x.base)]
  
  \coordinate (x0) at (-6,-3);
  \coordinate (x1) at (-6,3);
  \coordinate (x2) at (6,3);
  \coordinate (x3) at (6,-3);
  
  \filldraw[fill=black!15](x0) to[out=0, in=0] (x1)to[out=0, in=180] (x2)to[in=0, out=0](x3)to[out=180, in=0](x0) ;
  \filldraw[fill=black!5, draw=black!5](x0) to[out=180, in=180] (x1)to[in=0, out=0](x0);

  \draw[color=black, thick] (x0) to[out=0, in=0] (x1);
  \draw[color=black, thick] (x1) to[out=180, in=180]node[right](x){$\phantom{x}$} (x0);
        
  \draw[color=black, thick] (x3) to[out=0, in=0] (x2);
  \draw[color=black, thick, dotted] (x2) to[out=180, in=180] (x3);

  \draw[color=black, thick](x1) to[out=0, in=180] (x2) ;
  \draw[color=black, thick](x0) to[out=0, in=180] (x3);

  
   \coordinate (B) at (-4.2,0);
   \draw[decoration={markings, mark=at position 0.5 with {\arrow{>}}},
        postaction={decorate},color=black, thick]  (B) to (7.8,0);   
   \draw[color=black, fill=black] (7.8,0) circle (10pt);  
   \draw[color=black, thick, fill=black!10] (-4.2,0) circle (10pt);  

 \end{tikzpicture}
   + 
 \begin{tikzpicture}[scale=.15, baseline=(x.base)]
  
  \coordinate (x0) at (-6,-3);
  \coordinate (x1) at (-6,3);
  \coordinate (x2) at (6,3);
  \coordinate (x3) at (6,-3);
  
  \filldraw[fill=black!40](x0) to[out=0, in=0] (x1)to[out=0, in=180] (x2)to[in=0, out=0](x3)to[out=180, in=0](x0) ;
  \filldraw[fill=black!5, draw=black!5](x0) to[out=180, in=180] (x1)to[in=0, out=0](x0);

  \draw[color=black, thick] (x0) to[out=0, in=0] (x1);
  \draw[color=black, thick] (x1) to[out=180, in=180] (x0);
        
  \draw[color=black, thick] (x3) to[out=0, in=0] (x2);
  \draw[color=black, thick, dotted] (x2) to[out=180, in=180]node[right](x){$\phantom{x}$}  (x3);

  \draw[color=black, thick](x1) to[out=0, in=180] (x2) ;
  \draw[color=black, thick](x0) to[out=0, in=180] (x3);

  
   \coordinate (B) at (1.2,0);
   \draw[decoration={markings, mark=at position 0.6 with {\arrow{<}}},
        postaction={decorate},color=black, thick]  (B)node[right]{$\mu$} to (-4.3,1);   
   \draw[decoration={markings, mark=at position 0.6 with {\arrow{<}}},
        postaction={decorate},color=black, thick]  (B) to (-4.3,-1);   
   \draw[color=black, fill=black, thick] ($(B)-(.3,.3)$) rectangle ($(B)+(.3,.3)$);  
   \draw[thick, color=black, fill=black!10] (-4.3,1) circle (10pt);  
   \draw[thick, color=black, fill=black!10] (-4.3,-1) circle (10pt);  
   
 \end{tikzpicture}
 = \int_{\partial_\mathbb{A}\Sigma} \langle p^*\mathbb{B},\mathbb{A}\rangle +\frac{1}{2} \int_{\partial_\mathbb{B}\Sigma}p_*\mu \; \mathrm{tr} \mathbb{B}^2 ~,
\eeq
where the last integral is the integral of a density, with $p_*\mu$ the pushforward of $\mu$, viewed as a density on the cylinder, to the $\mathbb{B}$-circle.
Thus we proved the following:

\begin{proposition}[YM on $\mathbb{A}$-$\mathbb{B}$ cylinder]
 The partition function for a YM cylinder in the $\mathbb{A}$-$\mathbb{B}$ polarization is:
 \beq
  Z[\mathbb{A},\mathbb{B}] = \exp \frac{\mathrm{i}}{\hbar} \Big( \int_{\partial_\mathbb{A}\Sigma} \langle p^*\mathbb{B},\mathbb{A}\rangle
   +\frac{1}{2} \int_{\partial_\mathbb{B}\Sigma}p_*\mu \; \mathrm{tr} \mathbb{B}^2 \Big)~.
 \eeq
\end{proposition}

A YM $\mathbb{A}$-$\mathbb{B}$ cylinder can be glued to other YM surfaces with boundary to modify their volume.
In particular in this way one can convert BF ($\mu=0$) to YM.

\subsection{$\mathbb{B}$-$\mathbb{B}$ polarization on the cylinder}

Another possible choice is to take the $\mathbb{B}$ polarization on both the boundary components of the cylinder.
This time the bulk fields are $\mathcal{Y}=\Omega(\Sigma;\mathfrak{g})[1]\oplus\Omega(\Sigma,\partial \Sigma;\mathfrak{g})$\,, with zero-modes $\mathcal{V}=H(\Sigma;\mathfrak{g})[1]\oplus H(\Sigma,\partial \Sigma;\mathfrak{g})\simeq H(S^1;\mathfrak{g})[1]\oplus H(S^1;\mathfrak{g})[-1]$\,.%
\footnote{
 Here $\simeq$ stands for the isomorphic model for cohomology.  
 The first term on the right is the cohomology of the circle (isomorphic to the cohomology of the cylinder) shifted \emph{down} by $1$ -- i.e. a complex concentrated in degrees $-1, 0$. 
 The second term is the cohomology of the cylinder relative to the boundary, which is isomorphic to the cohomology of the circle shifted \emph{up} by $1$ -- thus, it is a complex concentrated in degrees $1,2$. 
 Our general convention for homological degree shifts is $(W^\bullet[k])^p=W^{k+p}$ for a graded vector space $W^\bullet$. 
}
More explicitly, the zero-modes can be described expanding with respect to a basis $[\chi_i]$ of $H(S^1)$ and its dual $[\chi^i]$\,:
\beq
 \mathsf{a}= \mathsf{a}_i\chi^i \in H(S^1;\mathfrak{g})[1]~,\qquad \mathsf{b}= \mathsf{b}^i\chi_i\wedge \mathrm{d}t \in H(S^1;\mathfrak{g})[-1]~.
\eeq
We can again fix the gauge using the axial-gauge, obtaining the propagator~\eqref{axial_gauge_table}:
\beq\label{BB_ax.gauge_prop}
  \eta(t,\tau;t',\tau')=\big(t'-\Theta(t'-t)\big)\delta(\tau'-\tau)(\mathrm{d}\tau'-\mathrm{d}\tau) 
    +\mathrm{d}t' \big(\Theta(\tau-\tau')-\tau+\tau'-\frac{1}{2}\big)~.
\eeq

Now the effective action contains trees with root on one of the boundaries and leafs in the bulk or 1-loop diagrams with trees rooted on the loop and leafs in the bulk.
Luckily, a lot of these diagrams vanish as it is shown by the following
\begin{lemma}\label{no_trees_lemma}
 For BF theory on a cylinder with $\mathbb{B}$-$\mathbb{B}$ polarization in the axial gauge, all the diagrams containing a bulk vertex with attached two $\mathsf{a}$ zero-modes vanish:
 \beq\label{no_trees}
  \begin{tikzpicture}[scale=.2, baseline=(x.base)]
      \draw[fill=black, color=black!10] (0:2) to[out=70, in=-30] (60:2.3) to[out=150, in=40] (130:2) to[out=220, in=150] node[color=black, above left]{$\Gamma$} (220:2) to[out=330, in=250] (0:2);
      \draw[color=black, thick, dashed] (0:2) to[out=70, in=-30] (60:2.3) to[out=150, in=40] (130:2) to[out=220, in=150] (220:2) to[out=330, in=250] (0:2);
      
      \draw[decoration={markings, mark=at position 0.5 with {\arrow{>}}},
        postaction={decorate},color=black, thick] (.5,0) to (3,0);
      
      \draw[color=black, thick] (3,0) to ($(3,0)+(40:1)$);
      \draw[color=black, fill=black] ($(3,0)+(40:1)$) circle (6pt);
      
      \draw[color=black, thick] (3,0) to ($(3,0)+(-40:1)$);
      \draw[color=black, fill=black] ($(3,0)+(-40:1)$) circle (6pt);
    
    \node (x) at (5,0) [right]{$= 0~.$};

  \end{tikzpicture} 
 \eeq
\end{lemma}
\begin{proof}
 Consider any diagram of the kind depicted in formula~\eqref{no_trees}.
 The associated differential form on the configuration space of the diagram will be of the kind
 \[
  \Gamma_c(t,\tau)\;\eta(t,\tau;t',\tau')\;f_{ab}^c\;\mathsf{a}^a(\tau')\;\mathsf{a}^b(\tau')~.
 \]
 Here $\Gamma_c(t,\tau)$ is the expression associated to the subgraph $\Gamma$.
 Since $\mathsf{a}$ has no form component along $\mathrm{d}t$\,, using the axial-gauge propagator~\eqref{BB_ax.gauge_prop} we have for the corresponding amplitude:
 \beq
 \begin{aligned}
  &\int \Gamma_c(t,\tau)\;f_{ab}^c\;\mathsf{a}^a(\tau')\;\mathsf{a}^b(\tau')\;\eta(t,\tau;t',\tau')\\
  &=f_{ab}^c\;\mathsf{a}^{ia}\;\mathsf{a}^{jb} \int \Gamma_c(t,\tau)\int_{S^1} \chi^i(\tau')\;\chi_j(\tau')\, \Big(\Theta(\tau-\tau')-\tau+\tau'-\frac{1}{2}\Big) = 0~,
 \end{aligned}
 \eeq
 were the inner integral over $\tau'$ vanishes by a direct computation.
\end{proof}

In particular this means that contributions to the effective action only come from either one single zero-mode attached to one of the two boundaries or from 1-loop diagrams with $n\geqslant 2$ vertices, each attached to a single $\mathsf{a}$ zero-mode (figure~\ref{BB_cyl}).
These diagrams can be explicitly evaluated and the perturbative series can be summed to obtain the effective action.

\begin{figure}[h!t]
 \centering{
  \begin{tikzpicture}[scale=.3, baseline=(x.base)]
  
  \coordinate (x0) at (-6,-3);
  \coordinate (x1) at (-6,3);
  \coordinate (x2) at (1.5,3);
  \coordinate (x3) at (1.5,-3);
  
  \filldraw[fill=black!15](x0) to[out=0, in=0] (x1)to[out=0, in=180] (x2)to[in=0, out=0](x3)to[out=180, in=0](x0) ;
  \filldraw[fill=black!5, draw=black!15](x0) to[out=180, in=180] (x1)to[in=0, out=0](x0);

  \draw[color=black, thick] (x0) to[out=0, in=0] (x1);
  \draw[color=black, thick] (x1) to[out=180, in=180]node[left](x){$\mathbb{B}$} (x0);
        
  \draw[color=black, thick] (x3) to[out=0, in=0]node[above right]{$\widetilde{\mathbb{B}}$} (x2);
  \draw[color=black, thick, dotted] (x2) to[out=180, in=180] (x3);

  \draw[color=black, thick](x1) to[out=0, in=180] (x2) ;
  \draw[color=black, thick](x0) to[out=0, in=180] (x3);

   \coordinate (B) at (-3,0);
   \coordinate (v0) at (-1.5,0);
   \coordinate (v1) at ($(-1.5,0)+(60:1.5)$);
   \coordinate (v2) at ($(-1.5,0)+(-60:1.5)$);

   \draw[color=black, thick] (B) to (v0);
   \draw[color=black, thick] (v0) to (v1);   
   \draw[color=black, thick] (v0) to (v2);   
   
   \draw[color=black, fill=black] (v1) circle (6pt);     
   \draw[color=black, fill=black] (v2) circle (6pt);     
   \draw[color=black, thick, fill=black!10] (B) circle (6pt);    
 \end{tikzpicture}
 \hspace{.2cm}  
  \begin{tikzpicture}[scale=.3, baseline=(x.base)]
  
  \coordinate (x0) at (-6,-3);
  \coordinate (x1) at (-6,3);
  \coordinate (x2) at (1,3);
  \coordinate (x3) at (1,-3);
  
  \filldraw[fill=black!15](x0) to[out=0, in=0] (x1)to[out=0, in=180] (x2)to[in=0, out=0](x3)to[out=180, in=0](x0) ;
  \filldraw[fill=black!5, draw=black!15](x0) to[out=180, in=180] (x1)to[in=0, out=0](x0);

  \draw[color=black, thick] (x0) to[out=0, in=0] (x1);
  \draw[color=black, thick] (x1) to[out=180, in=180]node[left](x){$\mathbb{B}$} (x0);
        
  \draw[color=black, thick] (x3) to[out=0, in=0]node[above right]{$\widetilde{\mathbb{B}}$} (x2);
  \draw[color=black, thick, dotted] (x2) to[out=180, in=180] (x3);

  \draw[color=black, thick](x1) to[out=0, in=180] (x2) ;
  \draw[color=black, thick](x0) to[out=0, in=180] (x3);

   \coordinate (B) at (-4.2,0);
   \coordinate (v1) at (-1.5,0);

   \draw[color=black, thick] (B) to (v1);   
   \draw[color=black, fill=black] (v1) circle (6pt);     
    
   \draw[color=black, thick, fill=black!10] (B) circle (6pt);    
 \end{tikzpicture}
 \hspace{.2cm}  
  \begin{tikzpicture}[scale=.3, baseline=(x.base)]
  
  \coordinate (x0) at (-6,-3);
  \coordinate (x1) at (-6,3);
  \coordinate (x2) at (1,3);
  \coordinate (x3) at (1,-3);
  
  \filldraw[fill=black!15](x0) to[out=0, in=0] (x1)to[out=0, in=180] (x2)to[in=0, out=0](x3)to[out=180, in=0](x0) ;
  \filldraw[fill=black!5, draw=black!15](x0) to[out=180, in=180] (x1)to[in=0, out=0](x0);

  \draw[color=black, thick] (x0) to[out=0, in=0] (x1);
  \draw[color=black, thick] (x1) to[out=180, in=180]node[left](x){$\mathbb{B}$} (x0);
        
  \draw[color=black, thick] (x3) to[out=0, in=0]node[above right]{$\widetilde{\mathbb{B}}$} (x2);
  \draw[color=black, thick, dotted] (x2) to[out=180, in=180] (x3);

  \draw[color=black, thick](x1) to[out=0, in=180] (x2) ;
  \draw[color=black, thick](x0) to[out=0, in=180] (x3);

   \coordinate (B) at (2.7,0);
   \coordinate (v1) at (0,0);

   \draw[color=black, thick] (B) to (v1);   
   \draw[color=black, fill=black] (v1) circle (6pt);     
    
   \draw[color=black, thick, fill=black!10] (B) circle (6pt);    
 \end{tikzpicture}

 \vspace{.5cm}
 \begin{tikzpicture}[scale=.225, baseline=(x.base)]
   
   \coordinate (x0) at (-7,-4);
   \coordinate (x1) at (-7,4);
   \coordinate (x2) at (7,4);
   \coordinate (x3) at (7,-4);
   
   \filldraw[fill=black!15](x0) to[out=0, in=0] (x1)to[out=0, in=180] (x2)to[in=0, out=0](x3)to[out=180, in=0](x0) ;
   \filldraw[fill=black!5, draw=black!15](x0) to[out=180, in=180] (x1)to[in=0, out=0](x0);

   \draw[color=black, thick] (x0) to[out=0, in=0] (x1);
   \draw[color=black, thick] (x1) to[out=180, in=180]node[left](x){$\mathbb{B}$} (x0);
         
   \draw[color=black, thick] (x3) to[out=0, in=0]node[above right]{$\widetilde{\mathbb{B}}$} (x2);
   \draw[color=black, thick, dotted] (x2) to[out=180, in=180] (x3);

   \draw[color=black, thick](x1) to[out=0, in=180] (x2) ;
   \draw[color=black, thick](x0) to[out=0, in=180] (x3);

    \draw[decoration={markings, mark=at position 0.5 with {\arrow{>}}},
        postaction={decorate},color=black, thick] (-40:2) arc (-40:80:2);
        
    \draw[decoration={markings, mark=at position 0.5 with {\arrow{>}}},
        postaction={decorate},color=black, thick] (60:2) arc (60:160:2);
        
    \draw[decoration={markings, mark=at position 0.5 with {\arrow{>}}},
        postaction={decorate},color=black, thick] (160:2) arc (160:240:2);
        
    \draw[decoration={markings, mark=at position 0.5 with {\arrow{>}}},
        postaction={decorate},color=black, thick] (240:2) arc (240:340:2);
     
    \draw[color=black, thick] (-40:2) to (-40:3);
    \draw[color=black, fill=black] (-40:3) circle (8pt);
    
    \draw[color=black, thick] (60:2) to (60:3);
    \draw[color=black, fill=black] (60:3) circle (8pt);
    
    \draw[color=black, thick] (160:2) to (160:3);
    \draw[color=black, fill=black] (160:3) circle (8pt);
    
    \draw[color=black, thick] (240:2) to (240:3);
    \draw[color=black, fill=black] (240:3) circle (8pt);
    
    \draw[color=black, thick, dotted] (-20:3) arc   (-20:40:3)node[right]{$n$};
    
   \end{tikzpicture}}
 \caption{Connected diagrams for non-abelian BF on the cylinder in $\mathbb{B}$-$\mathbb{B}$ polarization.}
 \label{BB_cyl}
\end{figure}
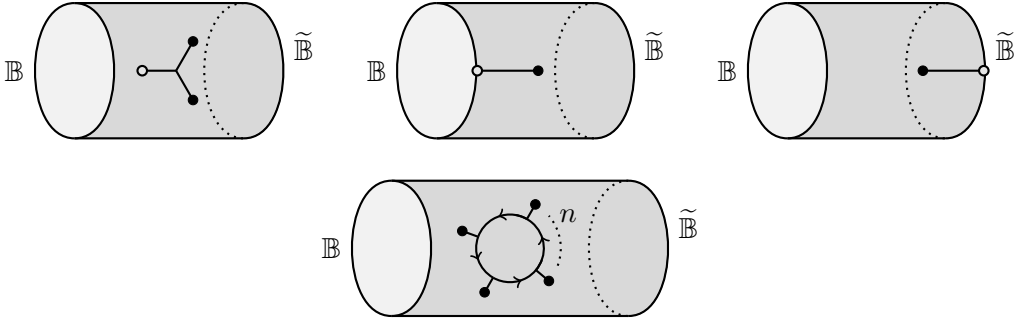

\begin{proposition}\label{BB_eff_act}
 The partition function for BF theory on the cylinder in $\mathbb{B}$-$\mathbb{B}$ polarization~is:
 \beq\label{BB_cyl_part_funct}
 \begin{aligned}
  Z[\mathbb{B},\widetilde{\mathbb{B}},{\mathsf{a}},{\mathsf{b}}] &= 
   \exp \bigg(\frac{\mathrm{i}}{2\hbar}\int_{I\times S^1} \langle \mathsf{b}, [\mathsf{a},\mathsf{a}] \rangle 
   + \frac{\mathrm{i}}{\hbar}\int_{S^1} \langle \mathbb{B}-\widetilde{\mathbb{B}},{\mathsf{a}} \rangle 
    + \sum_{n\geqslant 2} \frac{1}{n} \mathrm{tr} (\mathrm{ad}_{\mathsf{a}_1})^n \;\frac{B_n}{n!}\bigg) \cdot \rho_{\mathcal{V}} \\ 
  &= \mathrm{e}^{\frac{\mathrm{i}}{2\hbar}\int_{I\times S^1} \langle \mathsf{b}, [\mathsf{a},\mathsf{a}] \rangle 
   + \frac{\mathrm{i}}{\hbar}\int_{S^1} \langle \mathbb{B}-\widetilde{\mathbb{B}},{\mathsf{a}} \rangle } \,
    \mathrm{det}  \bigg(\frac{\sinh\big(\mathrm{ad}_{\mathsf{a}_1} /2\big)}{\mathrm{ad}_{\mathsf{a}_1} /2}\bigg) \cdot \rho_{\mathcal{V}}~.   
 \end{aligned}
 \eeq
 Here $\rho_{\mathcal{V}} = (-\mathrm{i}\hbar)^{\dim \mathfrak{g}} \, D^{\frac12}\mathsf{a} ~ D^{\frac12}\mathsf{b}$  is the reference half-density on the space of zero-modes.
\end{proposition}
\begin{proof}
We refer for the proof to Appendix~\ref{diagram_computations}.
\end{proof}

\begin{remark}[Reference half-densities on residual fields]
Generally, we choose the following reference half-density on the space of residual fields $\mathcal{V}$:%
\footnote{
 We are suppressing the index $\Sigma$ to lighten the notations.
}
\begin{equation}
 \rho_{\mathcal{V}}=\prod_{k=0}^2 (\xi_k)^{d_k}\cdot  \mathrm{D}^{\frac12}\mathsf{a}~ \mathrm{D}^{\frac12} \mathsf{b} ~.
\end{equation}
Here $\mathrm{D}^{\frac12}\mathsf{a}~ \mathrm{D}^{\frac12} \mathsf{b}$ is the standard half-density on $\mathcal{V}$, inducing the standard Berezin-Lebesgue densities $\mathrm{d}\mathsf{a}$, $\mathrm{d}\mathsf{b}$ on the Lagrangians $\mathsf{b}=0$ and $\mathsf{a}=0$, respectively. 
Also, $d_k$ is the dimension of the subspace of $\mathcal{V}$ corresponding to $\mathsf{a}$-fields of de~Rham degree $k\in \{0,1,2\}$; in particular, for $\mathcal{V}$ the zero-modes, $d_k=\dim H^k(\Sigma,\partial_\mathbb{A}\Sigma;\mathfrak{g})$ are the Betti numbers of relative de~Rham cohomology. 
Factors $\xi_k$ are as follows:
\begin{equation}
 \xi_0=-\mathrm{i}\hbar,\quad  \xi_1 = 1,\quad \xi_2 = \frac{1}{2\pi \hbar} ~.
\end{equation}
The logic behind this normalization is that for $\mathcal{V} = W[1]\oplus W^*[-2]$ with $W$ a complex (concentrated in degrees $0,1,2$) and $\mathcal{V}' = W'[1] \oplus W^{'*}[-2]$ with $W'$ a deformation retract of $W$, we would like the BV~pushforward of the half-density $\rho_{\mathcal{V}}\, \mathrm{e}^{\frac{\mathrm{i}}{\hbar} \langle \mathsf{b}, \mathrm{d}\mathsf{a} \rangle}$ on $\mathcal{V}$ (corresponding to abstract abelian BF~theory associated to $W$) to yield $\rho_{\mathcal{V'}}\, \mathrm{e}^{\frac{\mathrm{i}}{\hbar}\langle \mathsf{b}', \mathrm{d}\mathsf{a}' \rangle}$ on $\mathcal{V}'$. 
Thus, we recover the normalization of reference half-densities from the automorphicity with respect to BV~pushforwards. 
The most general normalization satisfying this condition is:
\begin{equation}
 \xi_0 = -\mathrm{i} \hbar\phi,\quad  \xi_1 = \phi^{-1},\quad \xi_2 = \frac{\phi}{2\pi \hbar} ~,
\end{equation}
with $\phi\neq 0$ an arbitrary constant. 
Our choice is to set $\phi=1$ which will ultimately lead to the number-valued partition function of 2D~Yang-Mills with standard normalization. 
Choosing any other $\phi$ would induce a rescaling of partition functions by 
\begin{equation}
 Z_\Sigma\mapsto \phi^{\chi(\Sigma)\cdot\dim\mathfrak{g}}~ Z_\Sigma~,
\end{equation} 
which reflects an inherent ambiguity of the normalization of path integral measure. 
We refer the reader to~\cite{CMR:cellular} for details on the normalization of half-densities compatible with BV~pushforwards.
\end{remark}

\subsection{$\mathbb{A}$-$\mathbb{A}$ polarization on the cylinder}\label{AA polarization on the cylinder}

The last polarization choice we will consider consists in taking $\mathbb{A}$ polarization for both the boundaries of the cylinder.
With this polarization the bulk fields are $\mathcal{Y}=\Omega(\Sigma,\partial \Sigma;\mathfrak{g})[1]\oplus\Omega(\Sigma;\mathfrak{g})$\,.
The zero-modes $\mathcal{V}=H(\Sigma,\partial \Sigma;\mathfrak{g})[1]\oplus H(\Sigma;\mathfrak{g})\simeq H(S^1;\mathfrak{g})\oplus H(S^1;\mathfrak{g})$ can be expanded as:
\beq
 \mathsf{a}= \mathsf{a}_i\chi^i\wedge \mathrm{d}t ~,\qquad \mathsf{b}= \mathsf{b}^i\chi_i ~.
\eeq
The axial-gauge propagator is now~\eqref{axial_gauge_table}:
\beq\label{AA_ax.gauge_prop}
   \eta(t,\tau;t',\tau')=\big(\Theta(t-t')-t\big)\delta(\tau'-\tau)(\mathrm{d}\tau'-\mathrm{d}\tau) 
    -\mathrm{d}t \big(\Theta(\tau-\tau')-\tau+\tau'-\frac{1}{2}\big) ~.
\eeq

The effective action contains trees with the root in the bulk and leafs on one of the boundaries or 1-loop diagrams with trees rooted on the loop and leafs either on the boundary or decorated with the zero mode~$a^1$\,.
Also with this polarization an analogue of lemma~\ref{no_trees_lemma} holds:

\begin{lemma}
 For BF theory on a cylinder with $\mathbb{A}$-$\mathbb{A}$ polarization in the axial gauge, all the diagrams containing a bulk vertex with attached two $\mathsf{a}$ zero-modes vanish:
 \beq\label{no_trees_AA}
  \begin{tikzpicture}[scale=.2, baseline=(x.base)]
      \draw[fill=black, color=black!10] (0:2) to[out=70, in=-30] (60:2.3) to[out=150, in=40] (130:2) to[out=220, in=150] node[color=black, above left]{$\Gamma$} (220:2) to[out=330, in=250] (0:2);
      \draw[color=black, thick, dashed] (0:2) to[out=70, in=-30] (60:2.3) to[out=150, in=40] (130:2) to[out=220, in=150] (220:2) to[out=330, in=250] (0:2);
      
      \draw[decoration={markings, mark=at position 0.5 with {\arrow{>}}},
        postaction={decorate},color=black, thick] (.5,0) to (3,0);
      
      \draw[color=black, thick] (3,0) to ($(3,0)+(40:1)$);
      \draw[color=black, fill=black] ($(3,0)+(40:1)$) circle (6pt);
      
      \draw[color=black, thick] (3,0) to ($(3,0)+(-40:1)$);
      \draw[color=black, fill=black] ($(3,0)+(-40:1)$) circle (6pt);
    
    \node (x) at (5,0) [right]{$= 0~.$};

  \end{tikzpicture} 
 \eeq
\end{lemma}
\begin{proof}
 The proof follows trivially from degree counting, since $\mathsf{a}$ zero-modes always have a component along $\mathrm{d}t$\,.
\end{proof}

The diagrams contributing to the effective action can be restricted further by reducing to the case of \textit{constant} fields~$\mathbb{A},\widetilde{\mathbb{A}}$\,.
Indeed, the gauge-invariance of the partition function (expressed by the mQME) implies that it is sufficient to evaluate it on constant 1-form fields $\mathbb{A}=\mathrm{d}t\,  \underline{\mathbb{A}}_{(1)}$\,, $\widetilde{\mathbb{A}}= \mathrm{d}t\, \widetilde{\underline{\mathbb{A}}}_{(1)}$\,, with 
$\underline{\mathbb{A}}_{(1)}, \widetilde{\underline{\mathbb{A}}}_{(1)} \in \mathfrak{g}$ two constants.
Then the value of the partition function for generic fields $\mathbb{A}, \widetilde{\mathbb{A}}$ is recovered (modulo a BV-exact term, cf.~\eqref{i p Z= Z+ exact term} below) by evaluating the constant-field result on the logs of holonomies $\underline{\mathbb{A}}_{(1)}= \log U(\mathbb{A})$\,,
$\widetilde{\underline{\mathbb{A}}}_{(1)}= \log U(\widetilde{\mathbb{A}})$\,.%
\footnote{
 For simplicity of notations we are omitting the subscript of the 1-form $\mathbb{A}_{(1)}$ when it appears in the holonomy~$U(\mathbb{A})$.
}
Here~$U(\cdots)$ stands for the holonomy of a connection $1$-form around a circle.
In other words, using the language of homological perturbation theory, we have a quasi-isomorphism between the two models for the space of states for an $\mathbb{A}$-circle: 
\begin{enumerate}[(i)]
 \item The full BFV model $\mathcal{H}^\mathbb{A}=\mathrm{Fun}_\mathbb{C}(\Omega^\bullet(S^1,\mathfrak{g})[1])$ given by functions of a general differential form~$\mathbb{A}$ on the circle, with differential~$\Omega$ defined by~\eqref{Omega_A}.
 \item The constant-field model $\mathcal{H}^{\mathbb{A},\mathrm{const}}=\mathrm{Fun}_\mathbb{C}(H^\bullet(S^1,\mathfrak{g})[1])$ -- functions of a constant form $\underline{\mathbb{A}}_{(0)}+\mathrm{d}t\, \underline{\mathbb{A}}_{(1)}$\,, with differential~$\underline{\bm\delta}$ defined by~\eqref{delta_A}.%
 \footnote{
  A subtlety related to Gribov ambiguities (see the discussion in Section~\ref{Omega_A_cohomology}) is that, to have a quasi-isomorphism, one actually needs to restrict to functions of $\underline{\mathbb{A}}_{(1)}$ in the Gribov region $B_0\subset H^1(S^1,\mathfrak{g})[1]$.
 }
\end{enumerate}
We have two chain maps: first, the projection $p_\mathcal{H}\colon \mathcal{H}^\mathbb{A} \rightarrow \mathcal{H}^{\mathbb{A},\mathrm{const}}$  -- evaluation of a wavefunction on constant forms or equivalently the pullback $p_\mathcal{H}=\iota^*_{\mathcal{H}}$ by the inclusion of the cohomology as constant forms $H^\bullet(S^1,\mathfrak{g})\mapsto \Omega^\bullet(S^1,\mathfrak{g})$\,. 
Second, the inclusion $i_\mathcal{H}\colon \mathcal{H}^{\mathbb{A},\mathrm{const}}\rightarrow  \mathcal{H}^\mathbb{A} $ sending $\underline\Psi\mapsto \left(\Psi\colon \mathbb{A} \mapsto \underline\Psi \big(\left.\mathbb{A}_{(0)}\right|_{p} , \log U(\mathbb{A})\big) \right)$ where~${p}$ is the base point on the circle used to define the holonomy. 
Denoting~$K_\mathcal{H}$ the chain homotopy for the retraction of chain complexes $(\mathcal{H}^\mathbb{A},\Omega) \rightsquigarrow (\mathcal{H}^{\mathbb{A},\mathrm{const}},\underline{\bm{\delta}})$\,, we have the following (cf. the discussion of the reduced partition function in~\cite{CMR:cellular}, section~7.4): 
\begin{equation}\label{i p Z= Z+ exact term}
 i_\mathcal{H}\circ p_\mathcal{H}\,  Z = (\mathrm{id}-K_\mathcal{H} \Omega-\Omega K_\mathcal{H})Z= 
 Z + (\Omega+ \hbar^2 \Delta) (\cdots)~,
\end{equation}
where $\cdots= -K_\mathcal{H}Z$ and $Z$ stands for the partition function on non-constant fields. 
The left hand side in~\eqref{i p Z= Z+ exact term} is exactly the partition function evaluated on constant $1$-form fields having the same holonomy as the original non-constant ones.
Thus, the full partition function differs from the reduced partition function by a BV-exact term.

\begin{lemma}\label{no_trees_AA_2}
 For BF theory on a cylinder with $\mathbb{A}$-$\mathbb{A}$ polarization in the axial gauge, all the diagrams containing a bulk vertex with attached two boundary fields vanish, assuming that $\mathbb{A}$ and $\widetilde{\mathbb{A}}$ are constant $1$-forms.
 \beq
   \begin{tikzpicture}[scale=.2, baseline=(x.base)]
      \draw[color=black, thick] (5,3)node[right]{$\mathbb{A}\,,~\widetilde{\mathbb{A}}$} to[out=-80,in=80] (5,-3);

      \draw[fill=black, color=black!10] (0:2) to[out=70, in=-30] (60:2.3) to[out=150, in=40] (130:2) to[out=220, in=150] node[color=black, above left]{$\Gamma$} (220:2) to[out=330, in=250] (0:2);
      \draw[color=black, thick, dashed] (0:2) to[out=70, in=-30] (60:2.3) to[out=150, in=40] (130:2) to[out=220, in=150] (220:2) to[out=330, in=250] (0:2);
      
      \draw[decoration={markings, mark=at position 0.5 with {\arrow{>}}},
        postaction={decorate},color=black, thick] (.5,0) to (3,0);
      
      \draw[decoration={markings, mark=at position 0.7 with {\arrow{>}}},
        postaction={decorate},color=black, thick] (3,0) to ($(3,0)+(30:2.6)$);
      \draw[color=black, fill=black] ($(3,0)+(30:2.6)$) circle (6pt);
      
      \draw[decoration={markings, mark=at position 0.7 with {\arrow{>}}},
        postaction={decorate},color=black, thick] (3,0) to ($(3,0)+(-30:2.6)$);
      \draw[color=black, fill=black] ($(3,0)+(-30:2.6)$) circle (6pt);
      
    \node (x) at (6,0) [right]{$\phantom{=}$};

   \end{tikzpicture}
 = 
   \begin{tikzpicture}[scale=.2, baseline=(x.base)]
      \draw[color=black, thick] (5,3)node[right]{$\widetilde{\mathbb{A}}$} to[out=-80,in=80] (5,-3);
      \draw[color=black, thick] (-5,3)node[left]{$\mathbb{A}$} to[out=-80,in=80] (-5,-3);

      \draw[fill=black, color=black!10] ($(0,1)+(0:2)$) to[out=70, in=-30] ($(0,1)+(60:2.3)$) to[out=150, in=40] ($(0,1)+(130:2)$) to[out=220, in=150] node[color=black, above left]{$\Gamma$} ($(0,1)+(220:2)$) to[out=330, in=250] ($(0,1)+(0:2)$);
      \draw[color=black, thick, dashed] ($(0,1)+(0:2)$) to[out=70, in=-30] ($(0,1)+(60:2.3)$) to[out=150, in=40] ($(0,1)+(130:2)$) to[out=220, in=150] node[color=black, above left]{$\Gamma$} ($(0,1)+(220:2)$) to[out=330, in=250] ($(0,1)+(0:2)$);
      
      \draw[decoration={markings, mark=at position 0.5 with {\arrow{>}}},
        postaction={decorate},color=black, thick] (0,1) to (0,-1.2);
      
      \draw[decoration={markings, mark=at position 0.7 with {\arrow{>}}},
        postaction={decorate},color=black, thick] (0,-1.2) to ($(-7.1,-.7)+(-30:2.6)$);
      \draw[color=black, fill=black] ($(-7.1,-.7)+(-30:2.6)$) circle (6pt);
      
      \draw[decoration={markings, mark=at position 0.7 with {\arrow{>}}},
        postaction={decorate},color=black, thick] (0,-1.2) to ($(2.9,-.7)+(-30:2.6)$);
      \draw[color=black, fill=black] ($(2.9,-.7)+(-30:2.6)$) circle (6pt);
      
    \node (x) at (6,0) [right]{$\phantom{=}$};

   \end{tikzpicture}   
  =0 ~.
 \eeq
\end{lemma}
\begin{proof}
 Using the assumption of constancy of boundary fields and the axial-gauge propagator~\eqref{AA_ax.gauge_prop}, when we have two boundary fields connected to the same bulk vertex we find the amplitude:
 \beq
 \begin{aligned}
  &\int \Gamma_c(\tilde{t},\tilde{\tau})\;f_{ab}^c\;\underline{\mathbb{A}}^a\;\underline{\mathbb{A}}^b\;\eta(\tilde{t},\tilde{\tau};t,\tau)\;\eta(t,\tau;0,\tau')\;\eta(t,\tau;0,\tau'')\\
  &=\frac{1}{2}f_{ab}^c\;\underline{\mathbb{A}}^a\;\underline{\mathbb{A}}^b \int \Gamma_c(\tilde{t},\tilde{\tau})\eta(\tilde{t},\tilde{\tau};t,\tau)	\int_{S^1} \mathrm{d}\tau' \Big(\Theta(\tau-\tau')-\tau+\tau'-\frac{1}{2}\Big) = 0~.
 \end{aligned}
 \eeq
 Similar amplitudes are found also in the case one or both the boundary fields live on the boundary at $t=1$\,.
\end{proof}

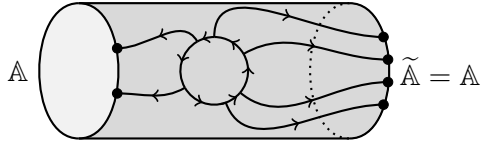
\begin{figure}[ht]
 \centering
 \begin{tikzpicture}[scale=.3, baseline=(x.base)]
  
  \coordinate (x0) at (-6,-3);
  \coordinate (x1) at (-6,3);
  \coordinate (x2) at (6,3);
  \coordinate (x3) at (6,-3);
  
  \filldraw[fill=black!15](x0) to[out=0, in=0] (x1)to[out=0, in=180] (x2)to[in=0, out=0](x3)to[out=180, in=0](x0) ;
  \filldraw[fill=black!5, draw=black!15](x0) to[out=180, in=180] (x1)to[in=0, out=0](x0);

  \draw[color=black, thick] (x0) to[out=0, in=0] (x1);
  \draw[color=black, thick] (x1) to[out=180, in=180]node[left](x){$\mathbb{A}$} (x0);
        
  \draw[color=black, thick] (x3) to[out=0, in=0]node[right](x){$\widetilde{\mathbb{A}}=\mathbb{A}$} (x2);
  \draw[color=black, thick, dotted] (x2) to[out=180, in=180] (x3);

  \draw[color=black, thick](x1) to[out=0, in=180] (x2) ;
  \draw[color=black, thick](x0) to[out=0, in=180] (x3);

  
   \coordinate (v1) at (90:1.5);
   \coordinate (v2) at (30:1.5);
   \coordinate (v3) at (-40:1.5);
   \coordinate (v4) at (-70:1.5);
   
   \coordinate (v5) at (120:1.5);
   \coordinate (v6) at (210:1.5);

   \coordinate (b2) at (7.5,1.5);
   \coordinate (b3) at (7.7,.5);
   \coordinate (b4) at (7.7,-.5);
   \coordinate (b5) at (7.5,-1.5);
   
   \coordinate (b6) at (-4.3,1);
   \coordinate (b7) at (-4.3,-1);

    \draw[decoration={markings, mark=at position 0.7 with {\arrow{>}}},
        postaction={decorate},color=black, thick] (90:1.5) arc (90:120:1.5);

   \draw[decoration={markings, mark=at position 0.5 with {\arrow{>}}},
        postaction={decorate},color=black, thick] (120:1.5) arc (120:210:1.5);
        
   \draw[decoration={markings, mark=at position 0.5 with {\arrow{>}}},
        postaction={decorate},color=black, thick] (210:1.5) arc (210:290:1.5);

    \draw[decoration={markings, mark=at position 0.5 with {\arrow{<}}},
        postaction={decorate},color=black, thick] (320:1.5) arc (320:280:1.5);
        
    \draw[decoration={markings, mark=at position 0.5 with {\arrow{<}}},
        postaction={decorate},color=black, thick] (30:1.5) arc (30:-40:1.5);
        
    \draw[decoration={markings, mark=at position 0.5 with {\arrow{>}}},
        postaction={decorate},color=black, thick] (30:1.5) arc (30:90:1.5);

    \draw[decoration={markings, mark=at position 0.5 with {\arrow{>}}},
        postaction={decorate},color=black, thick]  (v1) to[out=90, in=180] (b2);   
    \draw[color=black, fill=black] (b2) circle (6pt);     
    
    \draw[decoration={markings, mark=at position 0.5 with {\arrow{>}}},
        postaction={decorate},color=black, thick]  (v2) to[out=30, in=180] (b3);   
    \draw[color=black, fill=black] (b3) circle (6pt);     
    
    \draw[decoration={markings, mark=at position 0.5 with {\arrow{>}}},
        postaction={decorate},color=black, thick]  (v3) to[out=-40, in=180] (b4);   
    \draw[color=black, fill=black] (b4) circle (6pt);     
    
    \draw[decoration={markings, mark=at position 0.5 with {\arrow{>}}},
        postaction={decorate},color=black, thick]  (v4) to[out=-70, in=180] (b5);   
    \draw[color=black, fill=black] (b5) circle (6pt);     

    \draw[decoration={markings, mark=at position 0.5 with {\arrow{>}}},
        postaction={decorate},color=black, thick]  (v5) to[out=120, in=0] (b6);   
    \draw[color=black, fill=black] (b6) circle (6pt);     
    \draw[decoration={markings, mark=at position 0.5 with {\arrow{>}}},
        postaction={decorate},color=black, thick]  (v6) to[out=210, in=0] (b7);   
    \draw[color=black, fill=black] (b7) circle (6pt);     
 \end{tikzpicture}
 \caption{Relevant loop diagrams for the globalized effective action of BF theory on a cylinder with $\mathbb{A}$-$\mathbb{A}$ polarization}
 \label{AA_rel_diag}
\end{figure}

Like in the case of non-abelian BF theory on closed surfaces, the computation of the effective action greatly simplifies if we look at the \emph{globalized} answer.
In this case we can define the globalized partition function by integrating the perturbative partition function over a Lagrangian submanifold of the space of residual fields.
If $\mathcal{S}_{\mathrm{eff}}[\mathbb{A},\mathbb{B},\mathsf{a},\mathsf{b}]$ is the perturbative effective action on the space of boundary fields and zero-modes and $\mathcal{L}\subset\mathcal{V}$ is a Lagrangian submanifold, the globalized partition function can be defined as:
\beq
 Z=  \int_\mathcal{L} \mathrm{e}^{\frac{\mathrm{i}}{\hbar}\mathcal{S}_{\mathrm{eff}}[\mathbb{A},\mathbb{B},\mathsf{a},\mathsf{b}]} ~.
\eeq
A possible choice for the Lagrangian is $\mathcal{L}:=\{\mathsf{a}=0\}$\,, which in particular implies that all diagrams containing $\mathsf{a}$ zero-modes will not contribute to the globalized effective action.
The effective action of BF theory is always linear in the $\mathsf{b}$ zero-modes.
Moreover, for the $\mathbb{A}$-$\mathbb{A}$ polarization on the cylinder in the axial gauge, lemma~\ref{no_trees_AA_2} implies that there are no tree diagrams with this $\mathcal{L}$\,.
Thus:
\beq
 Z=\int \frac{\mathrm{d}\mathsf{b}}{(2\pi\hbar)^{\dim\mathfrak{g}}}~ \mathrm{e}^{\frac{\mathrm{i}}{\hbar}\int_{S^1}
  \langle \mathsf{b},\mathbb{A}-\widetilde{\mathbb{A}}\rangle 
   + \frac{\mathrm{i}}{\hbar} \mathcal{S}_{\mathrm{eff}}(\mathbb{A},\widetilde{\mathbb{A}}, \mathsf{a}=0,\mathsf{b}=0)}
    = \Big(\frac{\mathrm{i}}{\hbar}\Big)^{\dim\mathfrak{g}} \delta(\mathbb{A},\widetilde{\mathbb{A}})\mathrm{e}^{\frac{\mathrm{i}}{\hbar}\mathcal{S}_{\mathrm{eff}}
     (\mathbb{A},\widetilde{\mathbb{A}}=\mathbb{A}, \mathsf{a}=0,\mathsf{b}=0)}   ~.
\eeq	
Here the $\delta$-function of the pair $\mathbb{A},\widetilde{\mathbb{A}}$ is the delta of the difference. 
Our partition functions are allowed to be distributions; they are obtained, using singular gauge-fixing, as a limit of a family of smooth partition functions laying in the same BV-cohomology class. 

The loop diagrams contributing to the globalized effective action are now only those where each loop vertex is connected to a boundary field with a single propagator and the fields on the two boundary components coincide (figure~\ref{AA_rel_diag}).
The amplitude of such a diagram with $n$ boundary fields is:
\beq
 -\frac{1}{n} \mathrm{tr} (\mathrm{ad}_{\underline{\mathbb{A}}_{(1)}}^n) \int_{(S^1)^n} \mathrm{d}\tau_1\cdots\mathrm{d}\tau_n\; \eta_{S^1}(\tau_1;\tau_n)\eta_{S^1}(\tau_{n};\tau_{n-1})\cdots 
  \eta_{S^1}(\tau_2;\tau_1)   ~.
\eeq
The integrals involved are exactly the same as the ones of the case of $\mathbb{B}$-$\mathbb{B}$ polarization~\eqref{BB_loop_amp}.
Thus we have:

\begin{proposition}
 The globalized partition function for BF theory on the cylinder in the $\mathbb{A}$-$\mathbb{A}$~polarization is
 \beq
 \begin{aligned}\label{AA_cylind.part.funct.}
  Z[{\mathbb{A}},\widetilde{{\mathbb{A}}}]&= 
   \Big(\frac{\mathrm{i}}{\hbar}\Big)^{\dim\mathfrak{g}} \delta(\log U(\mathbb{A}), \log U(\mathbb{\widetilde{A}})) \cdot 
    \delta(\mathbb{A}_p, \mathbb{\widetilde{A}}_{\tilde{p}})~     
     \mathrm{det}\bigg(\frac{\sinh\big(\mathrm{ad}_{\log U(\mathbb{A})} /2\big)}{\mathrm{ad}_{\log U(\mathbb{A})} /2}\bigg)^{-1} 
 \end{aligned}
 \eeq
 where~$U(\cdots)$ is the holonomy of the connection around a circle and $\mathbb{A}_p, \mathbb{\widetilde{A}}_{\widetilde{p}}$ are the zero-form components of boundary fields~$\mathbb{A}, \widetilde{\mathbb{A}}$ evaluated at the base points~$p,\widetilde{p}$ on the two boundary circles.
\end{proposition}

\begin{remark}
 Since $\mathrm{det}\bigg(\frac{\sinh\big(\mathrm{ad}_{x} /2\big)}{\mathrm{ad}_{x} /2}\bigg)$ is the determinant of the Jacobian of the exponential map $\exp\colon \mathfrak{g} \longrightarrow G$, we can rewrite~\eqref{AA_cylind.part.funct.} in terms of the delta function on the Lie group:
 \beq
  Z[{\mathbb{A}},\widetilde{{\mathbb{A}}}] = 
   \Big(\frac{\mathrm{i}}{\hbar}\Big)^{\dim\mathfrak{g}} \delta(\mathbb{A}_p, \mathbb{\widetilde{A}}_{\tilde{p}}) \cdot 
    \delta_G( U(\mathbb{A}), U(\mathbb{\widetilde{A}})) ~.
 \eeq
\end{remark}

\subsection{$\mathbb{B}$ polarization on the disk}\label{Sec:B_disk}

Let us consider now non-abelian BF theory on the disk $D$\,.
Using the $\mathbb{B}$ polarization on the boundary, the bulk fields are $\mathcal{Y}= \Omega(D;\mathfrak{g})[1] \oplus \Omega(D,S^1;\mathfrak{g})$\,.
The zero-modes are $\mathcal{V} = \mathfrak{g}[1]\oplus\mathfrak{g}^*[-2]$ with generators the constant zero-form $[1\cdot t_a]$ and an area 2-form $[\mu\cdot t^a]$\,, where $t_a$ and $t^a$ are dual basis of $\mathfrak{g}$ and $\mathfrak{g}^*$\,.

The Feynman graphs appearing in the effective action are trees, with root either in a boundary $\mathbb{B}$-field or in a $\mathsf{b}$ zero-mode in the bulk, or 1-loop diagrams (figure~\ref{B_disk}).

\begin{figure}[h]
\centering
   \begin{tikzpicture}[scale=.25]
   
   \node[left](x) at (-5.5,0) {$\mathbb{B}$};

   \draw[fill=black!15, thick] (0,0) circle (5);   
   
   \coordinate (B) at (-5,0);
   
   \coordinate (v1) at (-3.5,0);
   \coordinate (v2) at ($(v1)+(40:1.5)$);
   \coordinate (v3) at ($(v1)+(-40:1.5)$);
   \coordinate (v4) at ($(v3)+(0:1.5)$);
   \coordinate (v5) at ($(v3)+(-80:1.5)$);
   \coordinate (v6) at ($(v4)+(-40:1.5)$);

   \coordinate (b1) at ($(v2)+(800:1)$);
   \coordinate (b2) at ($(v2)+(0:1)$);
   \coordinate (b3) at ($(v4)+(40:1)$);
   \coordinate (b4) at ($(v5)+(-40:1)$);
   \coordinate (b5) at ($(v5)+(-120:1)$);
   \coordinate (b6) at ($(v6)+(0:1)$);
   \coordinate (b7) at ($(v6)+(-80:1)$);

   \draw[decoration={markings, mark=at position 0.7 with {\arrow{>}}},
        postaction={decorate},color=black, thick] (B) to (v1);   
    
   \draw[decoration={markings, mark=at position 0.7 with {\arrow{>}}},
        postaction={decorate},color=black, thick] (v1) to (v2);     
   
   \draw[decoration={markings, mark=at position 0.7 with {\arrow{>}}},
        postaction={decorate},color=black, thick] (v1) to (v3);    
   
   \draw[decoration={markings, mark=at position 0.7 with {\arrow{>}}},
        postaction={decorate},color=black, thick] (v3) to (v4);    
   
   \draw[decoration={markings, mark=at position 0.7 with {\arrow{>}}},
        postaction={decorate},color=black, thick] (v3) to (v5);      
        
   \draw[decoration={markings, mark=at position 0.7 with {\arrow{>}}},
        postaction={decorate},color=black, thick] (v4) to (v6);

   \draw[color=black, thick, fill=black!10] (B) circle (8pt);

    \draw[color=black, thick]  (v2) to (b1);   
    \draw[color=black, fill=black] (b1) circle (6pt);      
      
    \draw[color=black, thick]  (v2) to (b2);   
    \draw[color=black, fill=black] (b2) circle (6pt);     
    
    \draw[color=black, thick]  (v4) to (b3);   
    \draw[color=black, fill=black] (b3) circle (6pt);     
    
    \draw[color=black, thick]  (v5) to (b4);   
    \draw[color=black, fill=black] (b4) circle (6pt);     
    
    \draw[color=black, thick]  (v5) to (b5);   
    \draw[color=black, fill=black] (b5) circle (6pt);     
    
    \draw[color=black, thick]  (v6) to (b6);   
    \draw[color=black, fill=black] (b6) circle (6pt);     
    
    \draw[color=black, thick]  (v6) to (b7);   
    \draw[color=black, fill=black] (b7) circle (6pt); 
    
   \end{tikzpicture}
   \hspace{1.5cm}
   \begin{tikzpicture}[scale=.25]

   \node[left](x) at (-5.5,0) {$\mathbb{B}$};

   \draw[fill=black!15, thick] (0,0) circle (5);   
   
   \coordinate (B) at (-3,0);
   
   \coordinate (v1) at (-1.5,0);
   \coordinate (v2) at ($(v1)+(40:1.5)$);
   \coordinate (v3) at ($(v1)+(-40:1.5)$);
   \coordinate (v4) at ($(v3)+(0:1.5)$);
   \coordinate (v5) at ($(v3)+(-80:1.5)$);
   \coordinate (v6) at ($(v4)+(-40:1.5)$);

   \coordinate (b1) at ($(v2)+(800:1)$);
   \coordinate (b2) at ($(v2)+(0:1)$);
   \coordinate (b3) at ($(v4)+(40:1)$);
   \coordinate (b4) at ($(v5)+(-40:1)$);
   \coordinate (b5) at ($(v5)+(-120:1)$);
   \coordinate (b6) at ($(v6)+(0:1)$);
   \coordinate (b7) at ($(v6)+(-80:1)$);

   \draw[color=black, thick] (B) to (v1);   
    
   \draw[decoration={markings, mark=at position 0.7 with {\arrow{>}}},
        postaction={decorate},color=black, thick] (v1) to (v2);     
   
   \draw[decoration={markings, mark=at position 0.7 with {\arrow{>}}},
        postaction={decorate},color=black, thick] (v1) to (v3);    
   
   \draw[decoration={markings, mark=at position 0.7 with {\arrow{>}}},
        postaction={decorate},color=black, thick] (v3) to (v4);    
   
   \draw[decoration={markings, mark=at position 0.7 with {\arrow{>}}},
        postaction={decorate},color=black, thick] (v3) to (v5);      
        
   \draw[decoration={markings, mark=at position 0.7 with {\arrow{>}}},
        postaction={decorate},color=black, thick] (v4) to (v6);

   \draw[color=black, thick, fill=black!10] (B) circle (8pt);

    \draw[color=black, thick]  (v2) to (b1);   
    \draw[color=black, fill=black] (b1) circle (6pt);      
      
    \draw[color=black, thick]  (v2) to (b2);   
    \draw[color=black, fill=black] (b2) circle (6pt);     
    
    \draw[color=black, thick]  (v4) to (b3);   
    \draw[color=black, fill=black] (b3) circle (6pt);     
    
    \draw[color=black, thick]  (v5) to (b4);   
    \draw[color=black, fill=black] (b4) circle (6pt);     
    
    \draw[color=black, thick]  (v5) to (b5);   
    \draw[color=black, fill=black] (b5) circle (6pt);     
    
    \draw[color=black, thick]  (v6) to (b6);   
    \draw[color=black, fill=black] (b6) circle (6pt);     
    
    \draw[color=black, thick]  (v6) to (b7);   
    \draw[color=black, fill=black] (b7) circle (6pt); 
    
   \end{tikzpicture}
   \hspace{1.5cm}
   \begin{tikzpicture}[scale=.25]
   
    \node[left](x) at (-5.5,0) {$\mathbb{B}$};
    \draw[fill=black!15, thick] (0,0) circle (5);

    \draw[decoration={markings, mark=at position 0.5 with {\arrow{>}}},
        postaction={decorate},color=black, thick] (-40:1.5) arc (-40:80:1.5);
        
    \draw[decoration={markings, mark=at position 0.5 with {\arrow{>}}},
        postaction={decorate},color=black, thick] (60:1.5) arc (60:160:1.5);
        
    \draw[decoration={markings, mark=at position 0.5 with {\arrow{>}}},
        postaction={decorate},color=black, thick] (160:1.5) arc (160:240:1.5);
        
    \draw[decoration={markings, mark=at position 0.5 with {\arrow{>}}},
        postaction={decorate},color=black, thick] (240:1.5) arc (240:340:1.5);

    \draw[color=black, thick] (160:1.5) to (160:2.5);
    \draw[color=black, fill=black] (160:2.5) circle (6pt);

    \draw[color=black, thick, dotted] (-20:2.5) arc  (-20:40:2.5);

   \draw[decoration={markings, mark=at position 0.7 with {\arrow{>}}},
        postaction={decorate},color=black, thick] (60:1.5) to (60:2.7);   
    
   \draw[decoration={markings, mark=at position 0.7 with {\arrow{>}}},
        postaction={decorate},color=black, thick] (60:2.7) to ($(60:2.7)+(110:1.2)$);     
   
   \draw[decoration={markings, mark=at position 0.7 with {\arrow{>}}},
        postaction={decorate},color=black, thick] (60:2.7) to ($(60:2.7)+(10:1.2)$);     
    
   \draw[color=black, thick] ($(60:2.7)+(110:1.2)$) to ($(60:2.7)+(110:1.2)+(150:1)$);
   \draw[color=black, fill=black] ($(60:2.7)+(110:1.2)+(150:1)$) circle (6pt);  
    
   \draw[color=black, thick] ($(60:2.7)+(110:1.2)$) to ($(60:2.7)+(110:1.2)+(70:1)$);
   \draw[color=black, fill=black] ($(60:2.7)+(110:1.2)+(70:1)$) circle (6pt);  
   
   \draw[color=black, thick] ($(60:2.7)+(10:1.2)$) to ($(60:2.7)+(10:1.2)+(50:1)$);
   \draw[color=black, fill=black] ($(60:2.7)+(10:1.2)+(50:1)$) circle (6pt);  
   
   \draw[color=black, thick] ($(60:2.7)+(10:1.2)$) to ($(60:2.7)+(10:1.2)+(-30:1)$);
   \draw[color=black, fill=black] ($(60:2.7)+(10:1.2)+(-30:1)$) circle (6pt);

   \draw[decoration={markings, mark=at position 0.7 with {\arrow{>}}},
        postaction={decorate},color=black, thick] (-40:1.5) to (-40:2.7);   
    
   \draw[color=black, thick] (-40:2.7) to ($(-40:2.7)+(10:1.2)$);     
   \draw[color=black, fill=black] ($(-40:2.7)+(10:1.2)$) circle (6pt);      
   
   \draw[decoration={markings, mark=at position 0.7 with {\arrow{>}}},
        postaction={decorate},color=black, thick] (-40:2.7) to ($(-40:2.7)+(-90:1.2)$);     
    
   \draw[color=black, thick] ($(-40:2.7)+(-90:1.2)$) to ($(-40:2.7)+(-90:1.2)+(-50:1)$);
   \draw[color=black, fill=black] ($(-40:2.7)+(-90:1.2)+(-50:1)$) circle (6pt);  
    
   \draw[color=black, thick] ($(-40:2.7)+(-90:1.2)$) to ($(-40:2.7)+(-90:1.2)+(-130:1)$);
   \draw[color=black, fill=black] ($(-40:2.7)+(-90:1.2)+(-130:1)$) circle (6pt);

   \draw[decoration={markings, mark=at position 0.7 with {\arrow{>}}},
        postaction={decorate},color=black, thick] (240:1.5) to (240:2.7);   
    
   \draw[color=black, thick] (240:2.7) to ($(240:2.7)+(280:1)$);     
   \draw[color=black, fill=black] ($(240:2.7)+(280:1)$) circle (6pt);      
   \draw[color=black, thick] (240:2.7) to ($(240:2.7)+(200:1)$);     
   \draw[color=black, fill=black] ($(240:2.7)+(200:1)$) circle (6pt);

   \end{tikzpicture}  
   \caption{Connected diagrams for non-abelian BF on the disk in $\mathbb{B}$ polarization.}
   \label{B_disk}
\end{figure}
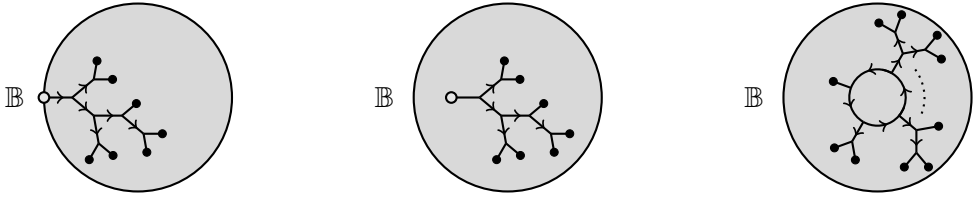

\begin{proposition}
 In the effective action for BF theory on the disk in $\mathbb{B}$-polarization, all the diagrams containing at least one propagator are vanishing.
 In particular the partition function reads:
 \beq\label{B_disk_eff_act}
  Z[\mathbb{B},\mathsf{a},\mathsf{b}] = \exp\frac{\mathrm{i}}{\hbar}\bigg(-\int_{S^1} \langle \mathbb{B},\mathsf{a}\rangle 
   + \int_D \frac{1}{2}\langle\mathsf{b},[\mathsf{a},\mathsf{a}]\rangle\bigg) \cdot \rho_{\mathcal{V}} ~,
 \eeq
 with $\rho_{\mathcal{V}}=(-\mathrm{i}\hbar)^{\dim \mathfrak{g}}\, D^{\frac12}\mathsf{a} ~ D^{\frac12} \mathsf{b}$ the reference half-density on residual fields.
\end{proposition}
\begin{proof}
 The result follows from degree counting.
 Let us consider first a tree diagram rooted in the bulk.
 If $n$ is the number of bulk vertices, then we have $n-1$ propagators, $n+1$ $\mathsf{a}$ zero-modes and one $\mathsf{b}$ zero-mode.
 Since propagators are 1-forms, $\mathsf{a}$ only has the zero-form component and $\mathsf{b}$ is a 2-form, then the differential form associated to the diagram is a $(n+1)$-form.
 This has to be integrated on the configuration space of the diagram, which is of dimension $2n$\,.
 Thus the only possibly non-vanishing diagram is for $n=1$\,.
 Its contribution is:
 \beq
  \frac{1}{2}\int_D \langle\mathsf{b},[\mathsf{a},\mathsf{a}]\rangle~.
 \eeq
 Consider now a tree diagram rooted on the boundary with $n$ bulk vertices. 
 We have $n$ propagators, $n+1$ $\mathsf{a}$ zero-modes and one boundary field $\mathbb{B}$\,.
 Thus the differential form associated to the diagram is a $n$-form or a $(n+1)$-form, depending on the form degree of~$\mathbb{B}$\,, and has to be integrated again on a $(2n+1)$-dimensional configuration space.
 In this case we only have a contribution with $n=0$\,:
 \beq
  -\int_{S^1} \langle \mathbb{B},\mathsf{a}\rangle~.
 \eeq
 Last, for a 1-loop diagram with $n\geqslant 1$ vertices in the loop and $l$ vertices in the trees rooted on the $n$ loop vertices, we have $n+l$ propagators and $n+l$ $\mathsf{a}$ zero-modes.
 Thus we have to integrate a differential form of degree $n+l$ on a $2(n+l)$-dimensional configuration space and, since $n\geqslant 1$\,, we have no non-vanishing contributions.
\end{proof}

\subsection{Gluing}\label{section_gluing}

We computed the YM partition function on the $\mathbb{A}$-$\mathbb{B}$ cylinder and the BF partition function on the $\mathbb{B}$-disk and the $\mathbb{A}$-$\mathbb{A}$ cylinder.
As we will show in this section, using the gluing property of BV-BFV theories, this is sufficient to prove a gluing formula between $\mathbb{A}$-polarized boundaries and to find the YM state on any surface with non-negative Euler characteristic.

\subsubsection{BF disk in $\mathbb{A}$ polarization}\label{Sect:A_disk}

The BF disk in $\mathbb{A}$ polarization can be obtained changing polarization to the disk in $\mathbb{B}$ polarization by gluing to it an $\mathbb{A}$-$\mathbb{A}$ BF cylinder (figure~\ref{A_disk}).

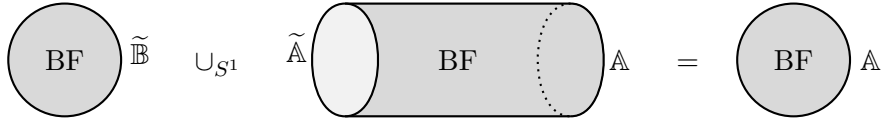
\begin{figure}[h]
 \[
 \begin{tikzpicture}[scale=.25, baseline=(x.base)]
  
    \draw[fill=black!15, thick] (0,0) circle (3);   

    \node (x) at (0,0) {BF};
    
    \node[right] (b) at (10:3) {$\widetilde{\mathbb{B}}$};

   \end{tikzpicture}
   \quad\cup_{S^1}\quad
   \begin{tikzpicture}[scale=.25, baseline=(x.base)]
  
    \coordinate (x0) at (-6,-3);
    \coordinate (x1) at (-6,3);
    \coordinate (x2) at (6,3);
    \coordinate (x3) at (6,-3);
    
    \filldraw[fill=black!15](x0) to[out=0, in=0] (x1)to[out=0, in=180] (x2)to[in=0, out=0](x3)to[out=180, in=0](x0) ;
    \filldraw[fill=black!5, draw=black!5](x0) to[out=180, in=180] (x1)to[in=0, out=0](x0);

    \draw[color=black, thick] (x0) to[out=0, in=0] (x1);
    \draw[color=black, thick] (x1) to[out=180, in=180] (x0);
        
    \draw[color=black, thick] (x3) to[out=0, in=0] (x2);
    \draw[color=black, thick, dotted] (x2) to[out=180, in=180] (x3);

    \draw[color=black, thick](x1) to[out=0, in=180] (x2) ;
    \draw[color=black, thick](x0) to[out=0, in=180] (x3);
    
    \node (x) at (0,0) {BF};
    
    \node[left] (b1) at (175:7.5) {$\widetilde{\mathbb{A}}$};
    \node[right] (b2) at (0:7.5) {$\mathbb{A}$};

   \end{tikzpicture}
   \quad=\quad
   \begin{tikzpicture}[scale=.25, baseline=(x.base)]
  
    \draw[fill=black!15, thick] (0,0) circle (3);   

    \node (x) at (0,0) {BF};
    
    \node[right] (b) at (0:3) {$\mathbb{A}$};

   \end{tikzpicture}
  \]
 \caption{$\mathbb{A}$ disk as the gluing of a $\mathbb{B}$ disk with an $\mathbb{A}$-$\mathbb{A}$ cylinder.}
 \label{A_disk} 
\end{figure}

For the $\mathbb{A}$-$\mathbb{A}$ cylinder we only know the projection in $\Omega$ cohomology of the globalized answer; since both globalization and projection to cohomology commute with gluing, we are still able to compute the partition function for the disk.
The glued partition function~is:
\beq
\begin{aligned}\label{A_disk_part.funct.}
 Z[\mathbb{A}] &= \int \mathrm{d}\mathsf{a}\, \mathrm{d}\widetilde{\mathbb{B}} \, \mathrm{d}\widetilde{\mathbb{A}} ~
  \mathrm{e}^{-\frac{\mathrm{i}}{\hbar}\int_{S^1} \langle \widetilde{\mathbb{B}} , \mathsf{a}-\widetilde{\mathbb{A}} \rangle }~
   \delta(\mathbb{A}_{p},\widetilde{\mathbb{A}}_{\widetilde{p}})\\
 &\qquad \qquad \qquad \cdot \det\bigg(\frac{\sinh\big(\mathrm{ad}_{\log U(\mathbb{{A}})} /2\big)}{\mathrm{ad}_{\log U(\mathbb{{A}})} /2}\bigg)^{-1}
    \delta(\log U(\mathbb{\widetilde{A}}),\log U(\mathbb{{A}})) &\\
 &= \det\bigg(\frac{\sinh\big(\mathrm{ad}_{\log U(\mathbb{{A}})} /2\big)}{\mathrm{ad}_{\log U(\mathbb{{A}})} /2}\bigg)^{-1}
    \delta(\log U(\mathbb{{A}}),0) = \delta_G(\mathrm{e}^{\log U(\mathbb{{A}})},\mathbb{I}) \\
 &= \delta_G(U(\mathbb{A}),\mathbb{I}) ~.
\end{aligned}    
\eeq

\begin{remark}
To have consistency with gluing, we assume that the integration measure over the boundary fields is normalized in such a way that 
\beq
 \int \mathrm{d}\widetilde{\mathbb{B}}\, \mathrm{d}\widetilde{\mathbb{A}}~ \mathrm{e}^{\frac{\mathrm{i}}{\hbar}\int_{S^1} 
  \langle \widetilde{\mathbb{B}}, \widetilde{\mathbb{A}}  \rangle }=1 ~.
\eeq
As a matter of convenience, we moreover distribute the normalization between $\mathrm{d}\widetilde{\mathbb{A}}$ and $\mathrm{d}\widetilde{\mathbb{B}}$ in such a way that 
\beq
 \int \mathrm{d}\widetilde{\mathbb{B}}~ \mathrm{e}^{\frac{\mathrm{i}}{\hbar}\int_{S^1} 
  \langle \widetilde{\mathbb{B}}, \widetilde{\mathbb{A}}  \rangle }=\delta(\widetilde{\mathbb{A}}) ~.
\eeq
\end{remark}

\subsubsection{YM disk in $\mathbb{A}$ polarization}\label{YM_Disk_in_A_polarization}

We can obtain the partition function for the YM disk in $\mathbb{A}$ polarization gluing to the BF disk a YM cylinder in $\mathbb{A}$-$\mathbb{B}$ polarization (figure~\ref{YM_A_disk}).

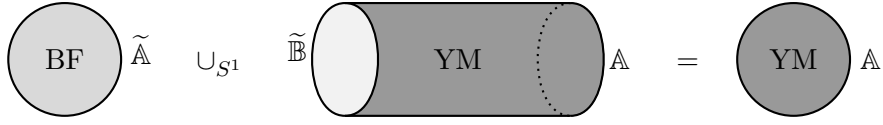
\begin{figure}[h]
 \[
 \begin{tikzpicture}[scale=.25, baseline=(x.base)]
  
    \draw[fill=black!15, thick] (0,0) circle (3);   

    \node (x) at (0,0) {BF};
    
    \node[right] (b) at (10:3) {$\widetilde{\mathbb{A}}$};

   \end{tikzpicture}
   \quad\cup_{S^1}\quad
   \begin{tikzpicture}[scale=.25, baseline=(x.base)]
  
    \coordinate (x0) at (-6,-3);
    \coordinate (x1) at (-6,3);
    \coordinate (x2) at (6,3);
    \coordinate (x3) at (6,-3);
    
    \filldraw[fill=black!40](x0) to[out=0, in=0] (x1)to[out=0, in=180] (x2)to[in=0, out=0](x3)to[out=180, in=0](x0) ;
    \filldraw[fill=black!5, draw=black!5](x0) to[out=180, in=180] (x1)to[in=0, out=0](x0);

    \draw[color=black, thick] (x0) to[out=0, in=0] (x1);
    \draw[color=black, thick] (x1) to[out=180, in=180] (x0);
        
    \draw[color=black, thick] (x3) to[out=0, in=0] (x2);
    \draw[color=black, thick, dotted] (x2) to[out=180, in=180] (x3);

    \draw[color=black, thick](x1) to[out=0, in=180] (x2) ;
    \draw[color=black, thick](x0) to[out=0, in=180] (x3);
    
    \node (x) at (0,0) {YM};
    
    \node[left] (b1) at (175:7.5) {$\widetilde{\mathbb{B}}$};
    \node[right] (b2) at (0:7.5) {$\mathbb{A}$};

   \end{tikzpicture}
   \quad=\quad
   \begin{tikzpicture}[scale=.25, baseline=(x.base)]
  
    \draw[fill=black!40, thick] (0,0) circle (3);   

    \node (x) at (0,0) {YM};
    
    \node[right] (b) at (0:3) {$\mathbb{A}$};

   \end{tikzpicture}
  \]
 \caption{YM $\mathbb{A}$ disk as the gluing of a BF $\mathbb{A}$ disk with a YM $\mathbb{A}$-$\mathbb{B}$ cylinder.}
 \label{YM_A_disk} 
\end{figure}

As the partition function for the BF disk, also the YM partition function coincides with the non-perturbative answer~\cite{witten:2d_quantum_gauge, Migdal:1975zg, cordes:lectures_2d_YM}.

\begin{proposition}
 The globalized partition function for 2D~YM on the disk in $\mathbb{A}$-polarization~is:
 \beq\label{A_disk_YM_part_funct}
 \begin{aligned}
  Z_{\mathrm{YM}}[\mathbb{A}] = \sum_R (\mathrm{dim}~R) \; \chi_R(U(\mathbb{A})) \;
   \mathrm{e}^{-\frac{\mathrm{i}\hbar a}{2}C_2(R)}	~,
 \end{aligned}
 \eeq
 where $a=\int \mu$ is the area of the disk, $\chi_R$ the character and $C_2(R)$ the quadratic Casimir of the representation~$R$\,.
\end{proposition}
\begin{proof}
 Gluing a BF disk to a YM cylinder in $\mathbb{A}$-$\mathbb{B}$~polarization we get:
 \beq
 Z_{\mathrm{YM}}[\mathbb{A}] &= \int \mathrm{d}\widetilde{\mathbb{B}} \, \mathrm{d}\widetilde{\mathbb{A}} ~
   \mathrm{e}^{-\frac{\mathrm{i}}{\hbar} \int_{S^1} \langle \widetilde{\mathbb{B}}, (\widetilde{\mathbb{A}}-\mathbb{A}) \rangle -\frac{\mathrm{i}}{2\hbar}\int_{S^1\times I} \mu\;\mathrm{tr}(\widetilde{\mathbb{B}}^2)}	\delta(U(\widetilde{\mathbb{A}}),\mathbb{I})\\
  &= \mathrm{e}^{\frac{\mathrm{i}\hbar a}{2} \big(\frac{\partial}{\partial\mathbb{A}}
  ,\frac{\partial}{\partial\mathbb{A}}\big)}~
   \delta(U({\mathbb{A}}),\mathbb{I})= \langle\mathbb{I}| \; \mathrm{e}^{-\frac{\mathrm{i}}{\hbar}H_{\mathrm{YM}}} \; |U(\mathbb{A})\rangle \\
  &= \sum_R (\mathrm{dim}~R) \; \chi_R(U(\mathbb{A})) \;
   \mathrm{e}^{-\frac{\mathrm{i}\hbar a}{2}C_2(R)}~.
\eeq
\end{proof}

\subsubsection{Gluing circles in $\mathbb{A}$ polarization} 

Two boundaries in $\mathbb{A}$ polarization can be glued together using a BF cylinder in $\mathbb{B}$-$\mathbb{B}$ polarization.
If $Z_{\Sigma_i}[\mathbb{A}_i]$ is the globalized partition function on a surface $\Sigma_i$\,, $i=1,2$\,, and $\Sigma$ is the gluing $\Sigma_1\cup_{S^1}\Sigma_2$ along a common boundary in $\mathbb{A}$ polarization, we get:
\beq\label{S1_gluing}
\begin{aligned}
 Z_\Sigma &= \int \mathrm{d}\widetilde{\mathbb{B}} \, \mathrm{d}\widetilde{\mathbb{A}} 
  \, \mathrm{d}\mathbb{B} \, \mathrm{d}\mathbb{A} \, \mathrm{d}\mathsf{a}_1 ~
   \mathrm{e}^{\frac{\mathrm{i}}{\hbar} \int_{S^1} \langle \widetilde{\mathbb{B}}, (\widetilde{\mathbb{A}}-\mathsf{a}) \rangle
    -\frac{\mathrm{i}}{\hbar} \int_{S^1} \langle {\mathbb{B}}, ({\mathbb{A}}-\mathsf{a}) \rangle 
     + \frac{\mathrm{i}}{2\hbar} \langle \mathsf{b}^{(2)},[\mathsf{a}^{(0)},\mathsf{a}^{(0)}] \rangle}\\
 &\qquad\cdot\det\bigg(\frac{\sinh\big(\mathrm{ad}_{\mathsf{a}_1} /2\big)}{\mathrm{ad}_{\mathsf{a}_1} /2}\bigg)
      Z_{\Sigma_1}[\mathbb{A}]\;Z_{\Sigma_2}[\widetilde{\mathbb{A}}] \cdot \rho_{\mathcal{V}} \\
 &= \rho_{\mathcal{V}} \cdot \mathrm{e}^{\frac{\mathrm{i}}{2\hbar} \langle \mathsf{b}^{(2)},[\mathsf{a}^{(0)},\mathsf{a}^{(0)}] \rangle}\int_G \mathrm{d}U ~ Z_{\Sigma_1}[U]\;Z_{\Sigma_2}[U]    ~,
\end{aligned}
\eeq
which coincides with the gluing formula for YM known in literature~\cite{witten:2d_quantum_gauge, Migdal:1975zg} up to a zero-mode dependent factor.
Here, instead of integrating out the zero-modes on the $\mathbb{B}$-$\mathbb{B}$~cylinder completely (which would yield an ill-defined integral), we performed a partial integration (BV pushforward), retaining the zero-modes~$\mathsf{a}^{(0)},\mathsf{b}^{(2)}$. 
Here the index in brackets stands for the form degree of a zero-mode and $\rho_{\mathcal{V}}=(-\mathrm{i}\hbar)^{\dim \mathfrak{g}}\, D^{\frac12} \mathsf{a}^{(0)}~ D^{\frac12} \mathsf{b}^{(2)}$ is the reference half-density on the remaining zero-modes.

\begin{remark}\label{gribov_region}
 In~\eqref{S1_gluing}, the domain of integration over~$\mathsf{a}_1$ is the ``Gribov region'' $B_0\subset\mathfrak{g}$ (cf.~subsection~\ref{Omega_A_cohomology}) --~the preimage of an open dense subset of the group~$G$ under the exponential map $\exp\colon\mathfrak{g} \rightarrow G$\,.
 On one hand, this is the domain corresponding to values of~$\mathsf{a}$ for which the sum of Feynman diagrams converges. 
 On the other hand, this corresponds to avoiding overcounting when performing the globalization via integrating over zero-modes as opposed to integrating over the moduli space of solutions of Euler-Lagrange equations. 
\end{remark}

\subsubsection{Other surfaces of non-negative Euler characteristic}\label{closed_surfaces}

To obtain the YM cylinder in $\mathbb{A}$-$\mathbb{A}$ polarization we can simply change polarization to the YM $\mathbb{A}$-$\mathbb{B}$ cylinder by gluing a BF cylinder (figure~\ref{YM_AA_cyl}):
\beq\label{YM_AA_cyl_pert_funct}
\begin{aligned}
 Z_{\mathrm{YM}}[\mathbb{A},\mathbb{A}'] &= \Big(\frac{\mathrm{i}}{\hbar}\Big)^{\dim \mathfrak{g}} \int \mathrm{d}\widetilde{\mathbb{B}} \,
  \mathrm{d}\widetilde{\mathbb{A}} ~
   \mathrm{e}^{-\frac{\mathrm{i}}{\hbar} \int_{S^1} \langle \widetilde{\mathbb{B}}, (\widetilde{\mathbb{A}}-\mathbb{A}) \rangle 
    -\frac{\mathrm{i}}{2\hbar}\int \mu\;\mathrm{tr}(\widetilde{\mathbb{B}}^2)}	\delta_G(U(\widetilde{\mathbb{A}}),U(\mathbb{A}')) ~ 
     \delta(\mathbb{A}'_{p'},\widetilde{\mathbb{A}}_{\widetilde{p}})\\
 &=\Big(\frac{\mathrm{i}}{\hbar}\Big)^{\dim \mathfrak{g}} \delta(\mathbb{A}'_{p'},{\mathbb{A}}_{{p}})~ \sum_R (\mathrm{dim}~R)~\chi_R\big(U^{-1}(\mathbb{A}')\cdot U(\mathbb{A})\big) \; \mathrm{e}^{-\frac{\mathrm{i}\hbar a}{2}C_2(R)}	~.
\end{aligned}
\eeq

\begin{figure}[h!]
   \[
    \begin{tikzpicture}[scale=.23, baseline=(x.base)]
  
    \coordinate (x0) at (-5,-3);
    \coordinate (x1) at (-5,3);
    \coordinate (x2) at (5,3);
    \coordinate (x3) at (5,-3);
    
    \filldraw[fill=black!15](x0) to[out=0, in=0] (x1)to[out=0, in=180] (x2)to[in=0, out=0](x3)to[out=180, in=0](x0) ;
    \filldraw[fill=black!5, draw=black!5](x0) to[out=180, in=180] (x1)to[in=0, out=0](x0);

    \draw[color=black, thick] (x0) to[out=0, in=0] (x1);
    \draw[color=black, thick] (x1) to[out=180, in=180] (x0);
        
    \draw[color=black, thick] (x3) to[out=0, in=0] (x2);
    \draw[color=black, thick, dotted] (x2) to[out=180, in=180] (x3);

    \draw[color=black, thick](x1) to[out=0, in=180] (x2) ;
    \draw[color=black, thick](x0) to[out=0, in=180] (x3);
    
    \node (x) at (0,0) {BF};
    \node[left] (b1) at (180:7.5) {$\mathbb{A}'$};
    \node[right] (b2) at (0:7.5) {$\widetilde{\mathbb{A}}$};

   \end{tikzpicture}
   \cup_{S^1}
   \begin{tikzpicture}[scale=.23, baseline=(x.base)]
  
    \coordinate (x0) at (-5,-3);
    \coordinate (x1) at (-5,3);
    \coordinate (x2) at (5,3);
    \coordinate (x3) at (5,-3);
    
    \filldraw[fill=black!40](x0) to[out=0, in=0] (x1)to[out=0, in=180] (x2)to[in=0, out=0](x3)to[out=180, in=0](x0) ;
    \filldraw[fill=black!5, draw=black!5](x0) to[out=180, in=180] (x1)to[in=0, out=0](x0);

    \draw[color=black, thick] (x0) to[out=0, in=0] (x1);
    \draw[color=black, thick] (x1) to[out=180, in=180] (x0);
        
    \draw[color=black, thick] (x3) to[out=0, in=0] (x2);
    \draw[color=black, thick, dotted] (x2) to[out=180, in=180] (x3);

    \draw[color=black, thick](x1) to[out=0, in=180] (x2) ;
    \draw[color=black, thick](x0) to[out=0, in=180] (x3);
    
    \node (x) at (0,0) {YM};
    \node[left] (b1) at (180:7.5) {$\widetilde{\mathbb{B}}$};
    \node[right] (b2) at (0:7.5) {${\mathbb{A}}$};

   \end{tikzpicture}
   =
      \begin{tikzpicture}[scale=.23, baseline=(x.base)]
  
    \coordinate (x0) at (-5,-3);
    \coordinate (x1) at (-5,3);
    \coordinate (x2) at (5,3);
    \coordinate (x3) at (5,-3);
    
    \filldraw[fill=black!40](x0) to[out=0, in=0] (x1)to[out=0, in=180] (x2)to[in=0, out=0](x3)to[out=180, in=0](x0) ;
    \filldraw[fill=black!5, draw=black!5](x0) to[out=180, in=180] (x1)to[in=0, out=0](x0);

    \draw[color=black, thick] (x0) to[out=0, in=0] (x1);
    \draw[color=black, thick] (x1) to[out=180, in=180] (x0);
        
    \draw[color=black, thick] (x3) to[out=0, in=0] (x2);
    \draw[color=black, thick, dotted] (x2) to[out=180, in=180] (x3);

    \draw[color=black, thick](x1) to[out=0, in=180] (x2) ;
    \draw[color=black, thick](x0) to[out=0, in=180] (x3);
    
    \node (x) at (0,0) {YM};
    \node[left] (b1) at (178:7.5) {$\mathbb{A}'$};
    \node[right] (b2) at (0:7.5) {$\mathbb{A}$};

   \end{tikzpicture}
   \]
 \caption{YM $\mathbb{A}$-$\mathbb{A}$ cylinder as the gluing of a BF $\mathbb{A}$-$\mathbb{A}$ cylinder with a YM $\mathbb{A}$-$\mathbb{B}$ cylinder.}
 \label{YM_AA_cyl} 
\end{figure}
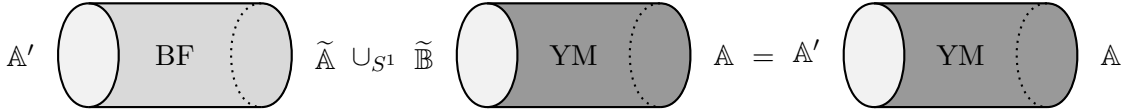

\begin{remark}
 The answer~\eqref{YM_AA_cyl_pert_funct} \emph{does not} coincide with the non-perturbative answer, which will be recovered perturbatively in Section~\ref{corners_and_2d_YM} using manifolds with corners.
 This discrepancy is due to the presence of inequivalent gauge-fixings in the globalization process.
\end{remark}

Let us now compute the YM partition function for a sphere $S^2$ obtained by the gluing of two disks: one with area~$a$ and in the $\mathbb{A}$~polarization, the other with zero area and in $\mathbb{B}$~polarization.

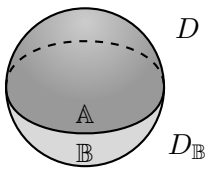
\begin{figure}[h!]
\centering
 \begin{tikzpicture}[scale=.35]
 
  \filldraw[black!15] (0,0) circle (3);
  \node at (30:4.5) {$D$};
  \node at (-30:4.5) {$D_{\mathbb{B}}$};
  
  \begin{scope}
   \clip (-3,0) to[out=-90, in=-90] (3,0) to(3,4) to (-3,4);
   \clip (0,0) circle (3);
   \filldraw[fill=black!40] (0,0) circle (3);
   \shade[inner color=black!20, outer color=black!40] (0,2.5) circle (3);
  \end{scope}
  
  \draw[black, thick] (0,0) circle (3);
  
  \draw[black, thick] (-3,0) to[out=-90, in=-90] node[below]{\small$\mathbb{B}$} node[above]{\small$\mathbb{A}$} (3,0);
  \draw[black, thick, dashed] (-3,0) to[out=90, in=90] (3,0);
  
 \end{tikzpicture}
\end{figure}

Using the \emph{globalized} partition function~\eqref{A_disk_YM_part_funct} for the $\mathbb{A}$-disk and the \emph{non-globalized} answer~\eqref{B_disk_eff_act} for the $\mathbb{B}$-disk we get a \emph{partially globalized} partition function for the sphere:
\beq\label{sphere_part_funct}
 Z_{\mathrm{YM}}^{S^2}[\mathsf{a},\mathsf{b}] &= \rho_{\mathcal{V}} \int \mathrm{d}\mathbb{A}\,\mathrm{d}\mathbb{B}~
  \mathrm{e}^{-\frac{\mathrm{i}}{\hbar} \int_{S^1} \langle \mathbb{B}, \mathbb{A} - \mathsf{a} \rangle 
   + \frac{\mathrm{i}}{2\hbar} \int_{D_{\mathbb{B}}} \langle \mathsf{b}, [\mathsf{a},\mathsf{a}] \rangle} ~Z_{\mathrm{YM}}^{D}[\mathbb{A}]
   = \rho_{\mathcal{V}} \cdot \mathrm{e}^{\frac{\mathrm{i}}{2\hbar} \int_{D_{\mathbb{B}}} \langle \mathsf{b}, [\mathsf{a},\mathsf{a}] \rangle} ~ Z_{\mathrm{YM}}^{D}[U(\mathbb{A})=\mathbb{I}]\\
 &= \rho_{\mathcal{V}} \cdot \mathrm{e}^{\frac{\mathrm{i}}{2\hbar} \int \langle \mathsf{b}, [\mathsf{a},\mathsf{a}] \rangle} ~
  \sum_R (\mathrm{dim}~R)^2 \; \mathrm{e}^{-\frac{\mathrm{i}\hbar a}{2}C_2(R)}~.
\eeq
Here $\rho_{\mathcal{V}}=(-\mathrm{i}\hbar)^{\dim\mathfrak{g}} \, D^{\frac12}\mathsf{a} ~ D^{\frac12}\mathsf{b}$ is the reference half-density on $\mathbb{B}$-disk zero-modes.
We immediately notice that this \emph{partially-globalized} answer consists of the product of a function of the zero-modes times the non-perturbative Migdal-Witten partition function for the sphere.
Moreover, the partition function~\eqref{sphere_part_funct} \emph{does not} produce well-defined global answers by integrating out the zero modes.

Similarly, trying to calculate the globalized partition function for the torus by gluing a YM cylinder in $\mathbb{A}$-$\mathbb{A}$~polarization~\eqref{YM_AA_cyl_pert_funct} with a cylinder in $\mathbb{B}$-$\mathbb{B}$~polarization~\eqref{BB_cyl_part_funct}, one obtains an ill-defined answer.

\begin{remark}
 We remark that the form of the perturbative answer here --~as the non-perturbative (number-valued) answer times the exponential of a cubic term in zero-modes~-- is similar to the form of the perturbative result for Chern-Simons theory in BV formalism on a rational homology 3-sphere~\cite{cm:remarks_on_CS}:
 \[
  Z_\mathrm{CS} = \rho_{\mathcal{V}}\cdot \mathrm{e}^{\frac{\mathrm{i}}{2\hbar} \langle \mathsf{a}^{(3)}, [\mathsf{a}^{(0)}, \mathsf{a}^{(0)}] \rangle}
   \cdot \mathrm{e}^{\frac{\mathrm{i}}{\hbar}\zeta(\hbar)}~.
 \]
 Here $\zeta(\hbar)$ is the sum of contributions of connected 3-valent graphs without leaves.
\end{remark}

\begin{remark}\label{ineq.g.f.}
 The gluing construction of Section~\ref{corners_and_2d_YM} (gluing along edges rather than circles) produces a well-defined globalized answer for all surfaces -- including the cylinder, the sphere and the torus -- coinciding with the non-perturbative answer in case of surfaces with boundary in $\mathbb{A}$-polarization~\eqref{[Z]=Z_non_pert}. 
 In particular, the gluing construction of Section~\ref{corners_and_2d_YM} produces the answer for the sphere as in~\eqref{sphere_part_funct} but without the zero-mode factor. 
 This discrepancy is due to inequivalence of gauge-fixings used in the two approaches.
\end{remark}

\subsection{Results summary}

We collect here the results for the partition functions obtained in this section.
As we discussed above, some of these partition function (which we mark with the symbol \warning) do not coincide with the non-perturbative results because of inequivalent gauge-fixing.
The main result of this section is the $\mathbb{A}$-disk partition function with finite area.

\begin{tabularx}{\textwidth}{c|X|c}
  \begin{tikzpicture}[scale=.14, baseline=-4pt] 
    \coordinate (x) at (0,0);
    \node (z0) at (0,5.2) {};
    \node (z1) at (0,-5.2) {};
   
    \coordinate (x0) at (-6,4);
    \coordinate (x1) at (-6,-4);
    \coordinate (y0) at (6,4);
    \coordinate (y1) at (6,-4);

    \filldraw[fill=black!5, thick, draw=black] (x0) to[out=180, in=180] node[right]{\footnotesize$\mathbb{A}$} (x1) to[out=0, in=0] (x0);
    
    \filldraw[fill=black!40, thick, draw=black] (x0) to[out=0, in=0] (x1) to[out=0, in=180] (y1) to[out=0, in=0]node[left]{\footnotesize$\mathbb{B}$} (y0) 
      to[out=180, in=0] (x0);
    
    \draw[dotted, thick, draw=black] (y0) to[out=180, in=180] (y1);
   \end{tikzpicture} 
  &	$Z[\mathbb{A},\mathbb{B}] = \exp \frac{\mathrm{i}}{\hbar} \Big( \int_{\partial_\mathbb{A}\Sigma} \langle p^*\mathbb{B},\mathbb{A}\rangle
   +\frac{1}{2} \int_{\partial_\mathbb{B}\Sigma}p_*\mu \; \mathrm{tr} \mathbb{B}^2 \Big)$		
      & \Large\checkmark\\
  \begin{tikzpicture}[scale=.14, baseline=-4pt] 
    \coordinate (x) at (0,0);
    \node (z0) at (0,5.2) {};
    \node (z1) at (0,-5.2) {};
   
    \coordinate (x0) at (-6,4);
    \coordinate (x1) at (-6,-4);
    \coordinate (y0) at (6,4);
    \coordinate (y1) at (6,-4);

    \filldraw[fill=black!5, thick, draw=black] (x0) to[out=180, in=180] node[right]{\footnotesize$\mathbb{B}$} (x1) to[out=0, in=0] (x0);
    
    \filldraw[fill=black!15, thick, draw=black] (x0) to[out=0, in=0] (x1) to[out=0, in=180] (y1) to[out=0, in=0]node[left]{\footnotesize$\widetilde{\mathbb{B}}$} (y0) 
      to[out=180, in=0] (x0);
    
    \draw[dotted, thick, draw=black] (y0) to[out=180, in=180] (y1);
   \end{tikzpicture}
  &	$Z[\mathbb{B},\widetilde{\mathbb{B}},{\mathsf{a}},{\mathsf{b}}] = \mathrm{e}^{\frac{\mathrm{i}}{2\hbar}\int_{I\times S^1} \langle \mathsf{b}, [\mathsf{a},\mathsf{a}] \rangle  + \frac{\mathrm{i}}{\hbar}\int_{S^1} \langle \mathbb{B}-\widetilde{\mathbb{B}},{\mathsf{a}} \rangle } \,
    \mathrm{det}  \bigg(\frac{\sinh\big(\mathrm{ad}_{\mathsf{a}_1} /2\big)}{\mathrm{ad}_{\mathsf{a}_1} /2}\bigg) \cdot \rho_{\mathcal{V}}$	
      & \Large\checkmark\\
  \begin{tikzpicture}[scale=.14, baseline=-4pt] 
    \coordinate (x) at (0,0);
    \node (z0) at (0,5.2) {};
    \node (z1) at (0,-5.2) {};
   
    \coordinate (x0) at (-6,4);
    \coordinate (x1) at (-6,-4);
    \coordinate (y0) at (6,4);
    \coordinate (y1) at (6,-4);

    \filldraw[fill=black!5, thick, draw=black] (x0) to[out=180, in=180] node[right]{\footnotesize$\mathbb{A}$} (x1) to[out=0, in=0] (x0);
    
    \filldraw[fill=black!15, thick, draw=black] (x0) to[out=0, in=0] (x1) to[out=0, in=180] (y1) to[out=0, in=0]node[left]{\footnotesize$\widetilde{\mathbb{A}}$} (y0) 
      to[out=180, in=0] (x0);
    
    \draw[dotted, thick, draw=black] (y0) to[out=180, in=180] (y1);
   \end{tikzpicture}
  &	$Z[{\mathbb{A}},\widetilde{{\mathbb{A}}}] = 
   \Big(\frac{\mathrm{i}}{\hbar}\Big)^{\dim\mathfrak{g}} \delta(\mathbb{A}_p, \mathbb{\widetilde{A}}_{\tilde{p}}) \cdot 
    \delta_G( U(\mathbb{A}), U(\mathbb{\widetilde{A}}))$	
      & \Large\warning\\[15pt]
   \begin{tikzpicture}[scale=.20, baseline=-4pt]
  
    \draw[fill=black!15, thick] (0,0) circle (3);   
     
    \node[left] (b) at (0:3) {\footnotesize$\mathbb{B}$};

   \end{tikzpicture}
  &	$Z[\mathbb{B},\mathsf{a},\mathsf{b}] = \exp\frac{\mathrm{i}}{\hbar}\bigg(-\int_{S^1} \langle \mathbb{B},\mathsf{a}\rangle 
   + \int_D \frac{1}{2}\langle\mathsf{b},[\mathsf{a},\mathsf{a}]\rangle\bigg) \cdot \rho_{\mathcal{V}}$	
      & \Large\checkmark\\[15pt]
  \begin{tikzpicture}[scale=.20, baseline=-4pt]
  
    \draw[fill=black!15, thick] (0,0) circle (3);   
     
    \node[left] (b) at (0:3) {\footnotesize$\mathbb{A}$};

   \end{tikzpicture}
  &	$Z[\mathbb{A}]=\delta_G(U(\mathbb{A}),\mathbb{I})$	
      & \Large\checkmark\\[15pt]
   \begin{tikzpicture}[scale=.20, baseline=-4pt]
  
    \draw[fill=black!40, thick] (0,0) circle (3);   
     
    \node[left] (b) at (0:3) {\footnotesize$\mathbb{A}$};

   \end{tikzpicture}
  &	$Z[\mathbb{A}]= \sum_R (\mathrm{dim}~R) \; \chi_R(U(\mathbb{A})) \;
   \mathrm{e}^{-\frac{\mathrm{i}\hbar a}{2}C_2(R)}$		
      &\Large\checkmark\\
  \begin{tikzpicture}[scale=.14, baseline=-4pt] 
    \coordinate (x) at (0,0);
    \node (z0) at (0,5.2) {};
    \node (z1) at (0,-5.2) {};
   
    \coordinate (x0) at (-6,4);
    \coordinate (x1) at (-6,-4);
    \coordinate (y0) at (6,4);
    \coordinate (y1) at (6,-4);

    \filldraw[fill=black!5, thick, draw=black] (x0) to[out=180, in=180] node[right]{\footnotesize$\mathbb{A}$} (x1) to[out=0, in=0] (x0);
    
    \filldraw[fill=black!40, thick, draw=black] (x0) to[out=0, in=0] (x1) to[out=0, in=180] (y1) to[out=0, in=0]node[left]{\footnotesize$\mathbb{A}'$} (y0) 
      to[out=180, in=0] (x0);
    
    \draw[dotted, thick, draw=black] (y0) to[out=180, in=180] (y1);
   \end{tikzpicture}
  &	$\!\begin{aligned}Z[\mathbb{A},\mathbb{A}'] =\Big(\frac{\mathrm{i}}{\hbar}\Big)^{\dim \mathfrak{g}} &\delta(\mathbb{A}'_{p'},{\mathbb{A}}_{{p}})~ \\
	  &\cdot\sum_R (\mathrm{dim}~R)~\chi_R\big(U^{-1}(\mathbb{A}')\cdot U(\mathbb{A})\big) \; \mathrm{e}^{-\frac{\mathrm{i}\hbar a}{2}C_2(R)}\end{aligned}$
      & \Large\warning\\
  \begin{tikzpicture}[scale=.25, baseline=-4pt]
 
  \filldraw[black!15] (0,0) circle (3);
  
  \begin{scope}
   \clip (-3,0) to[out=-90, in=-90] (3,0) to(3,4) to (-3,4);
   \clip (0,0) circle (3);
   \filldraw[fill=black!40] (0,0) circle (3);
   \shade[inner color=black!20, outer color=black!40] (0,2.5) circle (3);
  \end{scope}
  
  \draw[black, thick] (0,0) circle (3);
  
  \draw[black, thick] (-3,0) to[out=-90, in=-90]  (3,0);
  \draw[black, thick, dashed] (-3,0) to[out=90, in=90] (3,0);
  
 \end{tikzpicture}
 &	$Z[\mathsf{a},\mathsf{b}] = \rho_{\mathcal{V}} \cdot \mathrm{e}^{\frac{\mathrm{i}}{2\hbar} \int \langle \mathsf{b}, [\mathsf{a},\mathsf{a}] \rangle} ~\sum_R (\mathrm{dim}~R)^2 \; \mathrm{e}^{-\frac{\mathrm{i}\hbar a}{2}C_2(R)}$
      & \Large\warning\\[20pt]
  \begin{tikzpicture}[scale=.4, baseline=(x.base)]
	\draw[thick, fill=black!40] (0,0) ellipse  (100pt and 50pt);	
	\draw[thick, color=black] (-1.2,.1) to[bend right] (1.2,.1);
	\draw[thick, fill=white] (-1,0) to[bend left] (1,0)  to[bend left] (-1,0);

	\draw[color=black, thick] (-.95,-.855-47pt) to[out=180, in=180] (-.95,-.855-1pt);
	\draw[color=black, thick, dotted] (-.95,-.855-47pt) to[out=0, in=0] (-.95,-.855-1pt);
	
	\begin{scope}[shift={(1.9,1.75)}]
	 \draw[color=black, thick] (-.95,-.855-47pt) to[out=180, in=180] (-.95,-.855-1pt);
	 \draw[color=black, thick, dotted] (-.95,-.855-47pt) to[out=0, in=0] (-.95,-.855-1pt);
	\end{scope}

  \end{tikzpicture}
 &	Ill-defined global partition function
      & \Large\warning
\end{tabularx}

\section{2D Yang-Mills for general surfaces with boundaries and corners}\label{corners_and_2d_YM}

To be able to compute the partition function of 2D YM for general surfaces, we need to also consider corners, i.e. codimension 2 strata -- marked points on the boundary.
In topology surfaces can be described as collections of polygons modulo an equivalence relation which identifies pairs of edges.
The idea is to transport this description to the level of field theory: if we can compute the partition function on polygons with arbitrary combinations of polarizations associated to the edges, then we can recover the partition function on surfaces with boundary by gluing pairs of edges with transversal polarizations.

In this section we will formulate a set of rules for corners, dictated by the logic of the path integral and find a set of building blocks that generates, under gluing, 2D YM on all manifolds with boundaries and corners.
We will then discuss the mQME in presence of corners and compute the partition function of the various building blocks.
Finally, we will use the results of this analysis to prove a gluing formula in presence of corners and compute the 2D YM partition function on a generic surface with boundary, recovering the well known non-perturbative solution.

\subsection{Corners and building blocks for 2D YM}\label{corners_and_building_blocks}

The partition function is an element of the space of boundary states, which are defined by the data of a choice of polarization on the boundary; this choice reflects on the (fluctuations of the) bulk fields by imposing boundary conditions.
In the presence of corners dividing two arcs with different polarizations, we have to consider mixed boundary conditions for the bulk fields.
More generally we can associate a polarization also to corners, inducing boundary conditions for all adjacent bulk or boundary fields.
In this case, corners can be considered as collapsed arcs, with associated polarization the same as the corner they represent, but carrying only some of the boundary fields, namely the ones pulled back from the corner (i.e. constant zero-forms).

Notice that the presence of a corner with the same polarization as one of the adjacent edges has no effect on the partition function (but could require modifications of~$\Omega$\,: cf.~section~\ref{Corners_space_of_states}).
For example taking a corner with the same polarization of both adjacent edges simply means that we are formally splitting the boundary field into two concatenating fields, but this doesn't change the boundary conditions for the bulk fields: 
\beq\label{corner_rules_1}
  Z\bigg(
   \begin{tikzpicture}[scale=.21, baseline=(x.base)]
    \coordinate (x) at (0,-.5);
    \coordinate (x1) at (60:3);
    \coordinate (x2) at (0:3);
    \coordinate (x3) at (-60:3);
        
    \filldraw[fill=black!15, color=black!15] (0,0) circle (3);
    
    \draw[black, thick] (x1) arc (60:-60:3);    
        
    \draw[black, thick] (0:2.7) to (0:3.3);
    
    \node at (30:4) {$\mathbb{A}$};
    \node at (0:4) {$\alpha$};
    \node at (-30:4) {$\mathbb{B}$};
    
   \end{tikzpicture}
  \bigg)
  &\simeq 
  Z\bigg(
   \begin{tikzpicture}[scale=.2, baseline=(x.base)]
    \coordinate (x) at (0,-.5);
    \coordinate (x1) at (60:3);
    \coordinate (x2) at (0:3);
    \coordinate (x3) at (-60:3);
    
    \filldraw[fill=black!15, color=black!15] (0,0) circle (3);
    
    \draw[black, thick] (x1) arc (60:-60:3);    
    
    \draw[black, thick] (0:2.7) to (0:3.3);
    
    \draw[black, thick] (0:2.7) to (0:3.3);
    
    \node at (30:4) {$\mathbb{A}$};
    \node at (0:4) {$\beta$};
    \node at (-30:4) {$\mathbb{B}$};
    
   \end{tikzpicture}
  \bigg)~;
  \\
  Z\bigg(
  \begin{tikzpicture}[scale=.2, baseline=(x.base)]
    \coordinate (x) at (0,-.5);
    \coordinate (x1) at (60:3);
    \coordinate (x2) at (0:3);
    \coordinate (x3) at (-60:3);
    
    \filldraw[fill=black!15, color=black!15] (0,0) circle (3);
    
    \draw[black, thick] (x1) arc (60:-60:3);    
    
    \draw[black, thick] (0:2.7) to (0:3.3);
    
    \node at (30:4) {$\mathbb{A}$};
    \node at (0:4) {$\alpha$};
    \node at (-30:4) {$\mathbb{A}$};
        
   \end{tikzpicture}
   \bigg)
   \simeq
   Z\bigg(
   \begin{tikzpicture}[scale=.2, baseline=(x.base)]
    \coordinate (x) at (0,-.5);
    \coordinate (x1) at (60:3);
    \coordinate (x2) at (0:3);
    \coordinate (x3) at (-60:3);
    
    \filldraw[fill=black!15, color=black!15] (0,0) circle (3);
    
    \draw[black, thick] (x1) arc (60:-60:3);    
    
    \node at (0:3.9) {$\mathbb{A}$};

   \end{tikzpicture}
   \bigg)
  ~;&\qquad
   Z\bigg(
   \begin{tikzpicture}[scale=.2, baseline=(x.base)]
    \coordinate (x) at (0,-.50);
    \coordinate (x1) at (60:3);
    \coordinate (x2) at (0:3);
    \coordinate (x3) at (-60:3);
    
    \filldraw[fill=black!15, color=black!15] (0,0) circle (3);
    
    \draw[black, thick] (x1) arc (60:-60:3);    
    
    \draw[black, thick] (0:2.7) to (0:3.3);
        
    \node at (30:4) {$\mathbb{B}$};
    \node at (0:4) {$\beta$};
    \node at (-30:4) {$\mathbb{B}$};
    
   \end{tikzpicture}
   \bigg)
   \simeq
   Z\bigg(
   \begin{tikzpicture}[scale=.2, baseline=(x.base)]
    \coordinate (x) at (0,-.50);
    \coordinate (x1) at (60:3);
    \coordinate (x2) at (0:3);
    \coordinate (x3) at (-60:3);
    
    \filldraw[fill=black!15, color=black!15] (0,0) circle (3);
    
    \draw[black, thick] (x1) arc (60:-60:3);    
    
    \node at (0:3.9) {$\mathbb{B}$};

   \end{tikzpicture}
   \bigg)~.
\eeq
Here $\simeq$ means equality under appropriate identification of the boundary data.

Moreover, by ``freeing'' the bulk fields from the boundary conditions imposed by a corner, i.e. upon integrating over all possible values of corner fields, the partition function of the surface without that corner is recovered:
\beq\label{corner_rules_2}
   \int \mathfrak{D}\beta ~Z\bigg(
   \begin{tikzpicture}[scale=.2, baseline=(x.base)]
    \coordinate (x) at (0,-.5);
    \coordinate (x1) at (60:3);
    \coordinate (x2) at (0:3);
    \coordinate (x3) at (-60:3);
    
    \filldraw[fill=black!15, color=black!15] (0,0) circle (3);
    
    \draw[black, thick] (x1) arc (60:-60:3);    
    
    \draw[black, thick] (0:2.7) to (0:3.3);
    
    \node at (30:4) {$\mathbb{A}$};
    \node at (0:4) {$\beta$};
    \node at (-30:4) {$\mathbb{A}$};

   \end{tikzpicture}
   \bigg) ~=~ Z\bigg(
   \begin{tikzpicture}[scale=.2, baseline=(x.base)]
    \coordinate (x) at (0,-.5);
    \coordinate (x1) at (60:3);
    \coordinate (x2) at (0:3);
    \coordinate (x3) at (-60:3);
    
    \filldraw[fill=black!15, color=black!15] (0,0) circle (3);
    
    \draw[black, thick] (x1) arc (60:-60:3);    
    
    \node at (0:3.9) {$\mathbb{A}$}; 
   \end{tikzpicture}
   \bigg)
   ~;\qquad
   \int \mathfrak{D}\alpha ~Z\bigg(
   \begin{tikzpicture}[scale=.2, baseline=(x.base)]
    \coordinate (x) at (0,-.5);
    \coordinate (x1) at (60:3);
    \coordinate (x2) at (0:3);
    \coordinate (x3) at (-60:3);
    
    \filldraw[fill=black!15, color=black!15] (0,0) circle (3);
    
    \draw[black, thick] (x1) arc (60:-60:3);    
    
    \draw[black, thick] (0:2.7) to (0:3.3);
    
    \node at (30:4) {$\mathbb{B}$};
    \node at (0:4) {$\alpha$};
    \node at (-30:4) {$\mathbb{B}$};
        
   \end{tikzpicture}
   \bigg) ~=~ Z\bigg(
   \begin{tikzpicture}[scale=.2, baseline=(x.base)]
    \coordinate (x) at (0,-.5);
    \coordinate (x1) at (60:3);
    \coordinate (x2) at (0:3);
    \coordinate (x3) at (-60:3);
    
    \filldraw[fill=black!15, color=black!15] (0,0) circle (3);
    
    \draw[black, thick] (x1) arc (60:-60:3);    
    
    \node at (0:3.9) {$\mathbb{B}$}; 
   \end{tikzpicture}
   \bigg)~.
\eeq
The gluing of two arcs is analogous to the case without corners, with the only additional condition that the fields on those common corners that
will be identified in the gluing are required to coincide:
\beq\label{gluing_formula}
   \int \mathfrak{D}(\mathbb{A},\mathbb{B})~ \mathrm{e}^{-\frac{\mathrm{i}}{\hbar}\int\langle \mathbb{B},\mathbb{A}\rangle}~Z\bigg(
   \begin{tikzpicture}[scale=.2, baseline=(x.base)]
    \coordinate (x) at (0,-.5);
    \coordinate (x1) at (50:3);
    \coordinate (x2) at (0:3);
    \coordinate (x3) at (-50:3);
    
    \filldraw[fill=black!15, color=black!15] (0,0) circle (3);
    
    \draw[black, thick] (x1) arc (50:-50:3);    
    
    \draw[black, thick] (50:2.7) to (50:3.3);
    \draw[black, thick] (-50:2.7) to (-50:3.3);
    
    \node at (0:4) {$\mathbb{A}$};
    
   \end{tikzpicture}
   \bigg)~Z\bigg(
   \begin{tikzpicture}[scale=.2, baseline=(x.base)]
    \coordinate (x) at (0,-.5);
    \coordinate (x1) at (130:3);
    \coordinate (x2) at (0:3);
    \coordinate (x3) at (-50:3);
    
    \filldraw[fill=black!15, color=black!15] (0,0) circle (3);
    
    \draw[black, thick] (x1) arc (130:230:3);    
    
    \draw[black, thick] (130:2.7) to (130:3.3);
    \draw[black, thick] (230:2.7) to (230:3.3);
    
    \node at (180:4) {$\mathbb{B}$};
    
   \end{tikzpicture}   \bigg)
   ~=~
   Z\bigg(
   \begin{tikzpicture}[scale=.2, baseline=(x.base)]
    \coordinate (x) at (0,-.5);
    \coordinate (x1) at (60:3);
    \coordinate (x2) at (0:3);
    \coordinate (x3) at (-60:3);
    
    \filldraw[fill=black!15, color=black!15] (0,0) circle (3);
    
    \draw[black, thick, dashed] (90:3) to (-90:3);    
    
    \draw[black, thick] (90:2.7) to (90:3.3);    
    \draw[black, thick] (-90:2.7) to (-90:3.3);    

   \end{tikzpicture}   
   \bigg)~.
\eeq
Statements above are our working axioms extending the BV-BFV setup to corners and will be tested with explicit computations in the following sections.

\begin{remark}
 Let the surface $\Sigma$ be the result of gluing of surfaces $\Sigma_1$ and $\Sigma_2$ along an interval~$I$, as above. 
 Assume that the partition functions for $\Sigma_1$, $\Sigma_2$ are computed perturbatively, using the propagators $\eta_1,\eta_2$. 
 Then the gluing formula~\eqref{gluing_formula} above yields the partition function for the glued surface $\Sigma$ computed using the ``glued propagator'' $\eta=\eta_1*\eta_2$ on~$\Sigma$, constructed as follows:
 \begin{itemize}
  \item For $x,y\in \Sigma_1$, $\eta(x,y)=\eta_1(x,y)$.
  \item For $x,y\in \Sigma_2$, $\eta(x,y)=\eta_2(x,y)$.
  \item For $x\in \Sigma_2$, $y\in \Sigma_1$, $\eta(x,y)=0$.
  \item For $x\in \Sigma_1$, $y\in \Sigma_2$, we have 
  \begin{equation}
   \eta(x,y)=\int_{I\ni z} \eta_1(x,z) \eta_2(z,y)
  \end{equation}
 \end{itemize}
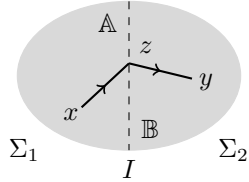
 \[
  \begin{tikzpicture}[scale=.33, baseline=(x.base)]
    
    \begin{scope}[xscale=1.5]
     \filldraw[color=black!15] (0,0) circle (3);
     \node at (105:2.2) {\small$\mathbb{A}$};
     \node at (-75:2.2) {\small$\mathbb{B}$};
     
     \node[above right] at (0,.5) {\small$z$};
     \draw[decoration={markings, mark=at position 0.5 with {\arrow{>}}}, postaction={decorate},color=black, thick]
      (225:1.8) to (0,.5);
     \node at (225:2.2) {\small$x$};
     \draw[decoration={markings, mark=at position 0.5 with {\arrow{>}}}, postaction={decorate},color=black, thick]
      (0,.5) to ($(0,.5)+(-20:1.8)$);
     \node at ($(0,.5)+(-20:2.2)$) {\small$y$};
     
     \draw[black, dashed] (90:3) to (-90:3);  
     
     \node[below] at (0,-3) {\small$I$};
     \node[below right] at (-45:3) {\small$\Sigma_2$};
     \node[below left] at (225:3) {\small$\Sigma_1$};
     
    \end{scope}
    
  \end{tikzpicture}
 \]
 This is precisely the gluing construction for propagators from~\cite{CMR:pert_quantum_BV}, which turns out to work also in the setting with corners.
\end{remark}

Assuming this set of rules for the corners, we have the following set of building blocks for 2D YM, as illustrated in figure~\ref{building_blocks}.
The disk in the $\mathbb{A}$~polarization was already computed in section~\ref{YM_Disk_in_A_polarization} and, using equation~\eqref{corner_rules_1}, it is equivalent to a polygon with an arbitrary number of edges where all the edges and the corners are in $\mathbb{A}$-polarization.
To change polarization of one of its edges, we can glue to it the BF disk with two corners in the $\alpha$~polarization and two edges in $\mathbb{B}$~polarization.
The last BF disk of figure~\ref{building_blocks}, with only one $\mathbb{A}$-edge and one corner in the opposite polarization, can be then used in combination with the other building blocks to change the polarization of one corner (figure~\ref{changing corner}).
In this way we can obtain a polygon with any number of edges and with any combination of polarizations associated to edges and corners; thus we can also obtain the partition function for any given surface with boundary (and corners).

\begin{figure}[h]
 \centering
 \begin{subfigure}[b]{0.2\textwidth}
  \begin{tikzpicture}[scale=.23, baseline=(x.base)]
    
    \filldraw[fill=black!40, thick] (0,0) circle (3);
    
    \node at (150:4) {$\mathbb{A}$};

    \node at (0,0) {YM};
    \node at (180:4.8) {};
    \node at (0:4.8) {};
    \node at (90:4.8) {};
    \node at (-90:4.8) {};
        
   \end{tikzpicture}
  \end{subfigure}
  \begin{subfigure}[b]{0.2\textwidth}
   \begin{tikzpicture}[scale=.23, baseline=(x.base)]
    
    \filldraw[fill=black!15, thick] (0,0) circle (3);
    
    \draw[black, thick] (90:2.7) to (90:3.3);    
    \draw[black, thick] (-90:2.7) to (-90:3.3);
    
    \node at (0:4) {$\mathbb{B}$};
    \node at (180:4) {$\mathbb{B}$};
    
    \node at (90:4) {$\alpha$};
    \node at (-90:4) {$\alpha$};
    
    \node at (0,0) {BF};
    
   \end{tikzpicture}
  \end{subfigure}

  \begin{subfigure}[b]{0.2\textwidth}
   \begin{tikzpicture}[scale=.23, baseline=(x.base)]
    
    \filldraw[fill=black!15, thick] (0,0) circle (3);
    
    \draw[black, thick] (0:2.7) to (0:3.3);    
    
    \node at (150:4) {$\mathbb{A}$};
    
    \node at (0:4) {$\beta$};
    
    \node at (0,0) {BF};
    
   \end{tikzpicture}
  \end{subfigure}
 \caption{Building blocks for 2D YM with corners.}
 \label{building_blocks}
\end{figure}
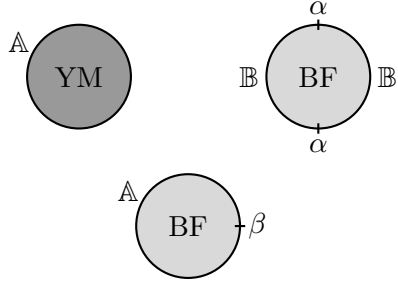

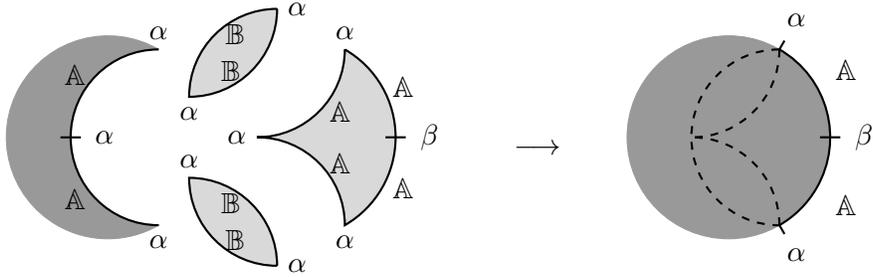
\begin{figure}[h]
 \centering
 \[
  \begin{tikzpicture}[scale=.45, baseline=(x.base)]
  
  \coordinate (x) at (0,-.5);
    \coordinate (x1) at (60:3);
    \coordinate (x2) at (0:3);
    \coordinate (x3) at (-60:3);
    
    \filldraw[fill=black!40, draw= black!40] (60:3) arc (60:300:3) arc (-90:-270:2.6); 
    \draw[black, thick] (-60:3) node[below] {$\alpha$} arc (-90:-180:2.6) node[midway, left]{$\mathbb{A}$};
    \draw[black, thick] ($(60:3)+(180:2.6)-(0,2.6)$) arc (-180:-270:2.6) node[midway, left]{$\mathbb{A}$} node[above] {$\alpha$}; 
    
    \node at ($(60:3)+(180:2.6)-(0,2.6)+(1,0)$) {$\alpha$};
    \draw[black, thick] ($(60:3)+(180:2.6)-(0,2.6)+(-.3,0)$) to ($(60:3)+(180:2.6)-(0,2.6)+(.3,0)$);

    \begin{scope}[shift={(3.5,-1.2)}]
     \filldraw[fill=black!15, draw= black, thick] (-60:3) arc (0:90:2.6) node[midway, left]{$\mathbb{B}$} node[above] {$\alpha$} arc (180:270:2.6) node[midway, right]{$\mathbb{B}$} node[right] {$\alpha$}; 
    \end{scope}

    \begin{scope}[shift={(3.5,1.2)}]
     \filldraw[fill=black!15, draw= black, thick] (60:3) arc (90:180:2.6) node[midway, right]{$\mathbb{B}$} node[below] {$\alpha$} arc (-90:0:2.6) node[midway, left]{$\mathbb{B}$} node[right] {$\alpha$};     
    \end{scope}

    \begin{scope}[shift={(5.5,0)}]
     \filldraw[fill=black!15, draw= black, thick] (60:3) node[above] {$\alpha$} arc (60:0:3) node[midway, right]{$\mathbb{A}$}  arc (0:-60:3) node[midway, right]{$\mathbb{A}$}  node[below] {$\alpha$} arc (0:90:2.6) node[midway, right]{$\mathbb{A}$}  node[left] {$\alpha$} arc (-90:0:2.6) node[midway, right]{$\mathbb{A}$};     
     \node at (0:4) {$\beta$};
     \draw[black, thick] (0:2.7) to (0:3.3);   
    \end{scope}
    
  \end{tikzpicture}
  \qquad\longrightarrow\qquad
  \begin{tikzpicture}[scale=.45, baseline=(x.base)]
    \coordinate (x) at (0,-.5);
    \coordinate (x1) at (60:3);
    \coordinate (x2) at (0:3);
    \coordinate (x3) at (-60:3);
    
    \filldraw[fill=black!40, draw=black!40] (0,0) circle (3);
    
    \draw[black, thick] (x1) arc (60:-60:3);    
    
    \draw[black, thick] (60:3) to (60:3.3);
    \draw[black, thick] (-60:3) to (-60:3.3);
    
    \draw[black, thick] (0:2.7) to (0:3.3);    

    \node at (30:4) {$\mathbb{A}$};
    \node at (-30:4) {$\mathbb{A}$};
    \node at (0:4) {$\beta$};
    \node at (60:4) {$\alpha$};
    \node at (-60:4) {$\alpha$};
    
    \draw[black, thick, dashed] (60:3) arc (90:270:2.6);    
    \draw[black, thick, dashed] (60:3) arc (0:-90:2.6);   
    \draw[black, thick, dashed] (-60:3) arc (0:90:2.6);   
    
  \end{tikzpicture}  
 \]
 \caption{The polarization on a corner can be changed by gluing. In this picture it is illustrated how to convert a corner in $\alpha$~polarization to a corner in $\beta$~polarization using the building blocks of figure~\ref{building_blocks}.}
 \label{changing corner}
\end{figure}

\subsection{Corners, spaces of states and the modified quantum master equation}\label{Corners_space_of_states}

We have two pictures for surfaces with boundary and corners.
\begin{enumerate}[I.]
 \item\label{pict_1} (\textbf{Non-polarized corners.}) Boundary circles are split into intervals by vertices (corners). 
 Each interval carries a polarization~$\mathbb{A}$ or~$\mathbb{B}$, corresponding to imposing the boundary condition on the pullback to the interval of the bulk field~$\mathsf{A}$ or~$\mathsf{B}$.  
 Corners do not carry a polarization.
 \item\label{pict_2} (\textbf{Polarized corners.}) In addition to the intervals carrying a polarization~$\mathbb{A}$ or~$\mathbb{B}$, each corner is also equipped with a polarization~$\alpha$ or~$\beta$ corresponding to prescribing the pullback of either~$\mathsf{A}$ or~$\mathsf{B}$ field to the corner.
\end{enumerate}

Picture~\ref{pict_2} is our main framework in this paper. 
One can transition from picture~\ref{pict_1} to picture~\ref{pict_2} by \emph{collapsing} every other arc on a circle (assuming that initially the number of arcs was even) into a vertex with the corresponding polarization, by the rule $\mathbb{A}\rightarrow \alpha$, $\mathbb{B}\rightarrow\beta$. 
One obtains the partition function~$Z_\mathrm{II}$ in the picture~\ref{pict_2} by evaluating the partition function $Z_\mathrm{I}$ of picture~\ref{pict_1} on constant $0$-form fields on the arcs that are being collapsed -- pullbacks of the corner fields to the arc. 
E.g., for a disk with the boundary split into $4$~arcs of alternating polarizations in picture~\ref{pict_1}, collapsing the $\mathbb{B}$-arcs into $\beta$-corners corresponds to the following:
\begin{equation}
 Z_\mathrm{II}(\mathbb{A}_1,\beta_1,\mathbb{A}_2,\beta_2;\mbox{ zero-modes})=\;\; Z_\mathrm{I}(\mathbb{A}_1,\mathbb{B}_1=\beta_1,\mathbb{A}_2,\mathbb{B}_2=\beta_2;\mbox{ zero-modes})~.
\end{equation}
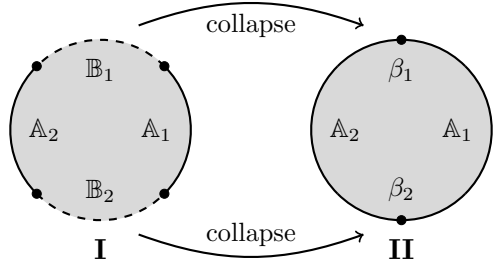
\[
 \begin{tikzpicture}[scale=.4]
  \filldraw[color=black!15] (0,0) circle (3);
  
  \draw[black, thick] (45:3) arc (45:-45:3) node[midway, left=3pt]{\small$\mathbb{A}_1$};
  \draw[black, thick, dashed] (45:3) arc (45:135:3)node[midway, below=3pt]{\small$\mathbb{B}_1$};
  \draw[black, thick] (135:3) arc (135:225:3)node[midway, right=3pt]{\small$\mathbb{A}_2$};
  \draw[black, thick, dashed] (225:3) arc (225:315:3)node[midway, above=3pt]{\small$\mathbb{B}_2$} node[midway, below=3pt]{\textbf{I}};
  
  \filldraw[black] (45:3) circle (.15);
  \filldraw[black] (135:3) circle (.15);
  \filldraw[black] (225:3) circle (.15);
  \filldraw[black] (315:3) circle (.15);
  
  \draw[decoration={markings, mark=at position 1 with {\arrow{>}}},
        postaction={decorate},color=black, thick] (70:3.7) to[out=20, in=160] node[below]{\small{collapse}} ($(10,0)+(110:3.7)$);
  \draw[decoration={markings, mark=at position 1 with {\arrow{>}}},
        postaction={decorate},color=black, thick] (-70:3.7) to[out=-20, in=-160] node[above]{\small collapse} ($(10,0)+(-110:3.7)$);        
  
  \begin{scope}[shift={(10,0)}]
   \filldraw[color=black!15] (0,0) circle (3);
  
   \draw[black, thick] (45:3) arc (45:-45:3) node[midway, left=3pt]{\small$\mathbb{A}_1$};
   \draw[black, thick] (45:3) arc (45:135:3)node[midway, below=3pt]{\small$\beta_1$};
   \draw[black, thick] (135:3) arc (135:225:3)node[midway, right=3pt]{\small$\mathbb{A}_2$};
   \draw[black, thick] (225:3) arc (225:315:3)node[midway, above=3pt]{\small$\beta_2$} node[midway, below=3pt]{\textbf{II}};
  
   \filldraw[black] (-90:3) circle (.15);
   \filldraw[black] (90:3) circle (.15);
   
  \end{scope}

 \end{tikzpicture}
\]

\subsubsection{Picture I: non-polarized corners. Modified quantum master equation}

Consider a circle  (thought of as a boundary component of a surface~$\Sigma$) split by $n$~points $p_1,p_2,\ldots,p_{2m}=p_0$ (``corners'')  into intervals $I_1,I_2,\ldots,I_{2m}$ with $I_k=[p_{k-1},p_k]$. 
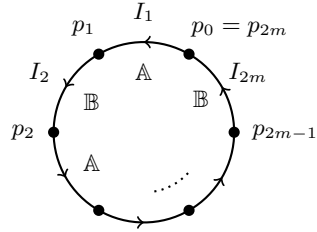
\[
 \begin{tikzpicture}[scale=.4]
  \draw[decoration={markings, mark=at position 0.5 with {\arrow{>}}},
        postaction={decorate},color=black, thick] (0:3) arc (0:60:3);
  \draw[decoration={markings, mark=at position 0.5 with {\arrow{>}}},
        postaction={decorate},color=black, thick] (60:3) arc (60:120:3);
  \draw[decoration={markings, mark=at position 0.5 with {\arrow{>}}},
        postaction={decorate},color=black, thick] (120:3) arc (120:180:3);
  \draw[decoration={markings, mark=at position 0.5 with {\arrow{>}}},
        postaction={decorate},color=black, thick] (180:3) arc (180:240:3);
  \draw[decoration={markings, mark=at position 0.5 with {\arrow{>}}},
        postaction={decorate},color=black, thick] (240:3) arc (240:300:3);
  \draw[decoration={markings, mark=at position 0.5 with {\arrow{>}}},
        postaction={decorate},color=black, thick] (300:3) arc (300:360:3);      
        
  \draw[color=black, thick, dotted] (280:2) arc (280:320:2);  
        
  \node[right] at (70:3.65) {\footnotesize$p_0=p_{2m}$};
  \node at (120:4) {\footnotesize$p_1$};
  \node at (180:4) {\footnotesize$p_2$};
  \node at (240:4) {};
  \node at (360:4.7) {\footnotesize$p_{2m-1}$};
    
  \node at (30:4) {\footnotesize$I_{2m}$};
  \node at (90:4) {\footnotesize$I_1$};
  \node at (150:4) {\footnotesize$I_2$};
  
  \node at (30:2.2) {\footnotesize$\mathbb{B}$};
  \node at (90:2) {\footnotesize$\mathbb{A}$};
  \node at (150:2) {\footnotesize$\mathbb{B}$};
  \node at (210:2) {\footnotesize$\mathbb{A}$};
  
  \filldraw[black, thick] (60:3) circle (.15);   
  \filldraw[black, thick] (120:3) circle (.15); 
  \filldraw[black, thick] (180:3) circle (.15); 
  \filldraw[black, thick] (240:3) circle (.15);  
  \filldraw[black, thick] (300:3) circle (.15);
  \filldraw[black, thick] (360:3) circle (.15); 

 \end{tikzpicture}
\]
Assume that we fix the $\mathbb{A}$-polarization on the intervals~$I_k$ with $k$~odd and the $\mathbb{B}$-polarization for $k$~even. 
We understand that we can, by a tautological transformation, further subdivide each $\mathbb{A}$- or $\mathbb{B}$-interval into several intervals carrying the same polarization. 
No polarization data is assigned to the corners~$p_k$ (this is our ``picture I'' for corners).

The BFV space of states~$\mathcal{H}$, associated to the circle with such a stratification and a choice of polarizations, is the space of complex-valued functions of the fields on the intervals:
\begin{equation}\label{H for stratified circle}
 \mathcal{H} = \Big\{ \mbox{functions }\Psi(\mathbb{A}|_{I_1},\mathbb{B}|_{I_2},\ldots, \mathbb{B}|_{I_{2m}}) \Big\} ~.
\end{equation}
The space of states is equipped with the BFV operator (which with an appropriate refinement becomes a differential, see Remark~\ref{rem: composite fields} below)
\begin{equation}\label{Omega_stratified_I}
 \Omega=\underbrace{\sum_{k~\mathrm{odd}} \Omega_{I_k}^{\mathbb{A}}+ \sum_{k~\mathrm{even}} \Omega_{I_k}^{\mathbb{B}}}_{\mathrm{edge~contributions}} 
  +  \underbrace{\sum_{k~\mathrm{odd}}  \Omega_{p_k}^{\mathbb{A}\mathbb{B}}+\sum_{k~\mathrm{even}}  
   \Omega_{p_k}^{\mathbb{B}\mathbb{A}}}_{\mathrm{corner~contributions}}~.
\end{equation}
Here the edge contributions from the intervals, depending on the polarization, are:
\begin{eqnarray}\label{Omega A}
 \Omega_I^{\mathbb{A}} & =&  \mathrm{i}\hbar\int_I \Big\langle \mathrm{d} \mathbb{A}+\frac12 [\mathbb{A},\mathbb{A}],
  \frac{\delta}{\delta \mathbb{A}} \Big\rangle~, \\ \label{Omega B}
 \Omega_I^{\mathbb{B}} & =& \int_I \mathrm{i}\hbar\, \Big\langle \mathrm{d}\mathbb{B}, \frac{\delta}{\delta \mathbb{B}} \Big\rangle 
  + (\mathrm{i}\hbar)^2  \Big\langle \mathbb{B} , \frac12\, \Big[ \frac{\delta}{\delta \mathbb{B}} , \frac{\delta}{\delta \mathbb{B}} \Big] \Big\rangle ~.
\end{eqnarray}
The corner contributions from the vertices~$p_k$ are the multiplication operators by the product of the limiting values of the $\mathbb{A}$-field and the $\mathbb{B}$-field coming from the incident arcs, with a sign depending on the order of the arcs relative to the orientation:
\begin{equation}\label{Omega corner picture I}
 \Omega_p^{\mathbb{A}\mathbb{B}} =-\langle\mathbb{B}_p,\mathbb{A}_p\rangle ~, \qquad 
  \Omega_p^{\mathbb{B}\mathbb{A}} = \langle \mathbb{B}_p, \mathbb{A}_p \rangle ~.
\end{equation}
These corner contributions to the boundary BFV operator~$\Omega$ and their necessity for the modified quantum master equation were observed by Alberto S. Cattaneo~\cite{cattaneo:private}.

The following is a refinement of Lemma~4.11 in~\cite{CMR:pert_quantum_BV} for a surface with boundary, with non-polarized corners allowed, in the case of 2D~Yang-Mills theory.

\begin{proposition}[mQME in picture I]\label{prop: mQME I}
 The BV-BFV partition function~$Z$ of 2D~Yang-Mills theory on a surface with boundary consisting of stratified circles decorated with a choice of $\mathbb{A},\mathbb{B}$ polarizations on the codimension~$1$ strata (and no polarization data on codimension~$2$ strata) satisfies the mQME
 \begin{equation}\label{mQME I}
  (\hbar^2\Delta+\Omega)Z=0~,
 \end{equation}
 where $\Omega$ is the sum of expressions~\eqref{Omega_stratified_I} for the stratified boundary circles.
\end{proposition}

\begin{proof}[Sketch of proof]
 The proof follows the proof of Lemma~4.11 in~\cite{CMR:pert_quantum_BV} where we need to take care of collapses of point near a corner. 
 Let~$\Gamma$ be a Feynman graph for the partition function; its contribution to~$Z$ is $\int_{C_\Gamma}\omega_\Gamma$\,: the integral over the configuration space~$C_\Gamma$ -- where vertices of~$\Gamma$ are restricted to the respective strata of~$\Sigma$ (bulk, boundary arcs or corners) -- of~$\omega_\Gamma$, the differential form on~$C_\Gamma$, which is the product of propagators, boundary fields and zero-modes, as prescribed by the combinatorics of~$\Gamma$.%
 \footnote{
  A tacit assumption in this proof is that the propagator is a \textit{smooth} $1$-form on the configuration space of two points. 
  E.g, the ``metric propagator'' arising from Hodge theory satisfies this property. Singular propagators considered in this paper arise as limits of such smooth propagators.}   
 One considers the Stokes' theorem for configuration space integrals:
 \begin{equation}\label{mQME Stokes}
  \mathrm{i}\hbar \sum_\Gamma \int_{C_\Gamma} \mathrm{d}\omega_\Gamma = \mathrm{i}\hbar \sum_\Gamma \int_{\partial C_\Gamma} \omega_\Gamma~.
 \end{equation}
 On the left hand side, the terms with $\mathrm{d}$~acting on the propagators assemble into $\hbar^2\Delta Z$ and the terms with $\mathrm{d}$ acting on $\mathbb{A},\mathbb{B}$ fields assemble into $\Omega_0 Z$ where $\Omega_0 = \mathrm{i}\hbar\int_{\partial_\mathbb{A} \Sigma} \big\langle \mathrm{d}\mathbb{A}\,,\, \frac{\delta}{\delta\mathbb{A}} \big\rangle + \mathrm{i}\hbar\int_{\partial_\mathbb{B} \Sigma} \big\langle \mathrm{d}\mathbb{B}\, ,\, \frac{\delta}{\delta\mathbb{B}} \big\rangle$. 
 Here $\partial_{\mathbb{A}}\Sigma$ and $\partial_\mathbb{B}\Sigma$ are the parts of the boundary equipped with polarizations $\mathbb{A}$ and $\mathbb{B}$, respectively. 
 Thus, the l.h.s. of~\eqref{mQME Stokes} is~$(\hbar^2\Delta+\Omega_0)Z$. 
 The r.h.s. contains several types of terms, corresponding to types of boundary strata of~$C_\Gamma$:
 \begin{enumerate}[(i)]
  \item Collapses of $2$ points in the bulk -- cancel out when summed over graphs, due to the classical master equation satisfied by the BV action.
  \item Collapses of $\geq 3$ points in the bulk -- vanish by the standard vanishing arguments for hidden strata of the configuration spaces~\cite{konstevich:def_quant}.
  \item Collapses of one or more points at a point on a boundary arc. These contributions assemble into~$-\Omega_{1}Z$, where contributions to the differential operator $\Omega_{1}$ are given by the collapsed subgraphs.
  \item Collapses of several points at a corner -- they assemble into~$-\Omega_{2}Z$.
 \end{enumerate}
 Thus, one obtains the modified quantum master equation~\eqref{mQME I} with $\Omega=\Omega_0+\Omega_1+\Omega_2$. 
 Analyzing the possible contributing collapses at an arc yields two graphs contributing to~$\Omega_1$:
 \beq\label{Omega_I_edge_graphs}
  \begin{tikzpicture}[scale=.4,  baseline=(x.base)]
   \coordinate (x) at (0,.7);
   \begin{scope}
    \clip (-3, 0) rectangle (3, 3) ;
    \shade[inner color=black!15, outer color=black!2] (0,0) circle (3);
   \end{scope}
   \draw [thick] (-3,0) -- (3,0) node[midway, below=3pt]{\small$\mathbb{A}$};
   \filldraw[black] (-1.5,0) circle (.15);
   \filldraw[black] (1.5,0) circle (.15);
   \draw[decoration={markings, mark=at position 0.5 with {\arrow{<}}},
        postaction={decorate},color=black, thick] (-1.5,0) to (0,1.5);
   \draw[decoration={markings, mark=at position 0.5 with {\arrow{<}}},
        postaction={decorate},color=black, thick] (1.5,0) to (0,1.5);
   \draw[decoration={markings, mark=at position 0.45 with {\arrow{>}}},
        postaction={decorate},color=black, thick, dashed] (0,3.5) to (0,1.5);
   \draw[black, thick, dotted] (2.5,0) arc (0:180:2.5);
  \end{tikzpicture}
  \rightarrow \mathrm{i}\hbar \int_{\partial_\mathbb{A}\Sigma} \Big\langle\frac12  \big[\mathbb{A},\mathbb{A}\big]\, , \, \frac{\delta}{\delta \mathbb{A}} \Big\rangle~, \quad
  \begin{tikzpicture}[scale=.4,  baseline=(x.base)]
   \coordinate (x) at (0,.7);
   \begin{scope}
    \clip (-3, 0) rectangle (3, 3) ;
    \shade[inner color=black!15, outer color=black!2] (0,0) circle (3);
   \end{scope}
   
   \draw [thick] (-3,0) -- (3,0) node[midway, below=3pt]{\small$\mathbb{B}$};

   \draw[decoration={markings, mark=at position 0.4 with {\arrow{<}}},
        postaction={decorate},color=black, thick, dashed] (-1.5,3.5) to (0,1.5);
   \draw[decoration={markings, mark=at position 0.4 with {\arrow{<}}},
        postaction={decorate},color=black, thick, dashed] (1.5,3.5) to (0,1.5);
   \draw[decoration={markings, mark=at position 0.5 with {\arrow{>}}},
        postaction={decorate},color=black, thick] (0,0) to (0,1.5);
   \draw[black, thick, dotted] (2.5,0) arc (0:180:2.5);
   \filldraw[draw=black, fill=black!10, thick] (0,0) circle (.17);
  \end{tikzpicture}
  \rightarrow  (\mathrm{i}\hbar)^2 \int_{\partial_\mathbb{B}\Sigma} \Big\langle  \mathbb{B} \, , \,
   \frac{1}{2}\Big[\frac{\delta}{\delta\mathbb{B}}, \frac{\delta}{\delta\mathbb{B}} \Big] \Big\rangle ~.
 \eeq
 For $\Omega_2$, the only contributing graphs are 
 \begin{equation}\label{Omega_I_AB_and_BA_collapse}
  \begin{tikzpicture}[scale=.4,  baseline=(x.base)]
   \coordinate (x) at (0,.7);
   \begin{scope}
    \clip (-3, 0) rectangle (3, 3) ;
    \shade[inner color=black!15, outer color=black!2] (0,0) circle (3);
   \end{scope}
   \draw [thick] (-3,0) -- (3,0);
      
   \draw[decoration={markings, mark=at position 0.5 with {\arrow{>}}},
        postaction={decorate},color=black, thick] (1.5,0) arc (0:180:1.5);
        
   \draw[black, thick, dotted] (2.5,0) arc (0:180:2.5);
   
   \filldraw[black] (-1.5,0) circle (.15) node[below=3pt]{\small$\mathbb{A}$};
   \filldraw[draw=black, fill=black!10, thick] (1.5,0) circle (.17) node[below=3pt]{\small$\mathbb{B}$};   
   
   \draw[black, thick] (0,.25) to (0,-.25) node[below]{\small$p$};
  \end{tikzpicture}
  \rightarrow -\langle \mathbb{B}_p ,\mathbb{A}_p\rangle ~,\qquad\quad
  \begin{tikzpicture}[scale=.4,  baseline=(x.base)]
   \coordinate (x) at (0,.7);
   \begin{scope}
    \clip (-3, 0) rectangle (3, 3) ;
    \shade[inner color=black!15, outer color=black!2] (0,0) circle (3);
   \end{scope}
   \draw [thick] (-3,0) -- (3,0);
      
   \draw[decoration={markings, mark=at position 0.5 with {\arrow{<}}},
        postaction={decorate},color=black, thick] (1.5,0) arc (0:180:1.5);
        
   \draw[black, thick, dotted] (2.5,0) arc (0:180:2.5);
   
   \filldraw[black] (1.5,0) circle (.15) node[below=3pt]{\small$\mathbb{A}$};
   \filldraw[draw=black, fill=black!10, thick] (-1.5,0) circle (.17) node[below=3pt]{\small$\mathbb{B}$};   
   
   \draw[black, thick] (0,.25) to (0,-.25) node[below]{\small$p$};
  \end{tikzpicture}
  \rightarrow \langle\mathbb{B}_p, \mathbb{A}_p\rangle~.
 \end{equation}
 
\end{proof}

\begin{remark}\label{rem: composite fields}
 It was found out in~\cite{CMR:pert_quantum_BV} that, in order to have the property $\Omega^2=0$ for the BFV operator, generally one should consider a certain refinement of the space of states, allowing the states to depend on the so-called ``composite fields'' on the boundary, which correspond in Feynman diagrams to boundary vertices of valency $\geq 2$.%
 \footnote{
  In fact, the operator $\Omega$ constructed above (\ref{Omega_stratified_I}) with edge contributions~(\ref{Omega A},\ref{Omega B}) and corner contributions~\eqref{Omega corner picture I} does not satisfy $\Omega^2=0$ right away, whenever corners are present. 
  In the setting of~\cite{CMR:pert_quantum_BV} this is remedied by adding corrections to $\Omega$, depending on composite boundary fields. 
  Then in addition to the diagrams~\eqref{Omega_I_AB_and_BA_collapse} at a corner one should consider other diagrams, involving boundary  vertices of valency~$\geq 2$.
 }
 We are not considering composite fields in this paper: below, in Section~\ref{sec: Picture II for corners}, we manage to construct~$\Omega$ for the setting of polarized corners, which squares to zero on the nose, without having to introduce composite fields. 
\end{remark}

\subsubsection{Picture II: polarized corners}\label{sec: Picture II for corners}

Now consider a circle split by $n$~points $p_1,\ldots,p_n=p_0$ (``corners'')  into intervals $I_1,\ldots,I_n$.
Assume that for each $k$ we fix on the interval $I_k$ the polarization $\mathbb{P}_k \in \{\mathbb{A},\mathbb{B}\}$ -- i.e. we  prescribe  either the the pullback of $\mathbb{A}_k$ field $A$ or the pullback $\mathbb{B}_k$ of the field $B$ on $I_k$ (by an abuse of notations, we denote the differential form $\mathbb{A}_k$ or $\mathbb{B}_k$ also by $\mathbb{P}_k$). 
Likewise, we fix a polarization $\xi_k\in \{\alpha,\beta\}$ on the corners $p_k$.
\beq\label{stratified circle II}
 \begin{tikzpicture}[scale=.4]
  \draw[decoration={markings, mark=at position 0.5 with {\arrow{>}}},
        postaction={decorate},color=black, thick] (0:3) arc (0:72:3);
  \draw[decoration={markings, mark=at position 0.5 with {\arrow{>}}},
        postaction={decorate},color=black, thick] (72:3) arc (72:144:3);
  \draw[decoration={markings, mark=at position 0.5 with {\arrow{>}}},
        postaction={decorate},color=black, thick] (144:3) arc (144:216:3);
  \draw[decoration={markings, mark=at position 0.5 with {\arrow{>}}},
        postaction={decorate},color=black, thick] (216:3) arc (216:288:3);
  \draw[decoration={markings, mark=at position 0.5 with {\arrow{>}}},
        postaction={decorate},color=black, thick] (288:3) arc (288:360:3);
        
  \draw[color=black, thick, dotted] (268:2) arc (268:308:2);  
  \draw[color=black, thick, dotted] (278:4) arc (278:298:4);

  \node at (62:4.2) {\footnotesize$p_0=p_n$};
  \node at (144:4) {\footnotesize$p_1$};
  \node at (216:4) {\footnotesize$p_2$};
  \node at (288:4) {};
  \node at (360:4.7) {\footnotesize$p_{n-1}$};
  
  \node at (72:2) {\footnotesize$\xi_n$};
  \node at (144:2) {\footnotesize$\xi_1$};
  \node at (216:2) {\footnotesize$\xi_2$};
  \node at (288:2) {};
  \node at (360:1.6) {\footnotesize$\xi_{n-1}$};

  \node at (36:4) {\footnotesize$I_n$};
  \node at (108:4) {\footnotesize$I_1$};
  \node at (180:4) {\footnotesize$I_2$};
  \node at (252:4) {};
  \node at (324:4) {};
  
  \node at (36:2.2) {\footnotesize$\mathbb{P}_n$};
  \node at (108:2) {\footnotesize$\mathbb{P}_1$};
  \node at (180:2) {\footnotesize$\mathbb{P}_2$};
  
  \filldraw[black, thick] (72:3) circle (.15);   
  \filldraw[black, thick] (144:3) circle (.15); 
  \filldraw[black, thick] (216:3) circle (.15); 
  \filldraw[black, thick] (288:3) circle (.15);  
  \filldraw[black, thick] (0:3) circle (.15); 

 \end{tikzpicture}
\eeq

The BFV space of states $\mathcal{H}$, associated to the circle with such a stratification and a choice of polarizations, is the space of complex-valued functions of the fields on the intervals and the corners, subject to the natural corner value conditions:
\begin{equation}\label{H for stratified circle}
 \mathcal{H} = \Bigg\{ \mbox{functions }\Psi(\mathbb{P}_1,\xi_1,\mathbb{P}_2,\xi_n,\ldots, \mathbb{P}_n,\xi_n) \;\; \Big{|} \;\; 
 \begin{array}{ll}
  \left.\mathbb{P}_k\right|_{p_k}=\xi_k &\mbox{if polarizations } \mathbb{P}_k \mbox{ and }\xi_k \mbox{ agree}\\
  \left.\mathbb{P}_k\right|_{p_{k-1}}=\xi_{k-1} &\mbox{if polarizations } \mathbb{P}_k \mbox{ and }\xi_{k-1} \mbox{ agree}
 \end{array}
 \Bigg\}~.
\end{equation}
Here we say that the polarization of an interval ``agrees'' with the polarization of the incident corner if this pair of polarizations is either $(\mathbb{A},\alpha)$ or $(\mathbb{B},\beta)$.
The space of states is a cochain complex with the differential
\begin{equation}\label{Omega_stratified}
 \Omega=\sum_k \underbrace{\Omega_{I_k}^{\mathbb{P}_k}}_{\mathrm{edge\;contribution\;from\;} I_k} 
  + \sum_k \underbrace{\Omega_{p_k}^{\mathbb{P}_{k}\xi_k\mathbb{P}_{k+1}}}_{\mathrm{corner\;contribution\;from\;}p_k} ~,
\end{equation}
where the edge contributions are given by~(\ref{Omega A},\ref{Omega B}).

The corner contributions to $\Omega$ depend on the polarization $\xi_k$ at the corner and polarizations of the incident edges $\mathbb{P}_k, \mathbb{P}_{k+1}$ and are assembled from the contribution of the corner itself and the contributions of the corner interacting with the incident edges:
\begin{equation}\label{Omega corner = HE+corner+HE}
 \Omega_{p_k}^{\mathbb{P}_{k}\xi_k\mathbb{P}_{k+1}}=\Omega_{p_k}^{\mathbb{P}_{k}\xi_k}+\Omega_{p_k}^{\xi_k}+\Omega_{p_k}^{\xi_k\mathbb{P}_{k+1}} ~.
\end{equation}
Here the pure corner contributions are:
\begin{equation}\label{Omega pure corner}
 \Omega^\alpha_p= \mathrm{i}\hbar \big\langle \frac12\big[\alpha,\alpha\big],\frac{\partial}{\partial \alpha} \big\rangle ~ , 
  \qquad   \Omega^\beta_p=0 ~.
\end{equation}
The corner-edge contributions $\Omega^{\mathbb{P}\xi}_p, \Omega^{\xi\mathbb{P}}_p$ vanish if the polarization $\xi$ at the corner matches the polarization $\mathbb{P}$ of the incident edge. For mismatching corner-edge polarizations, we have nontrivial contributions to $\Omega$:

\beq\label{HE_contributions_to_Omega}
 \begin{tikzpicture}[scale=.4, baseline=(x.base)]
  \coordinate (x) at (0,-0.3);
  \draw[decoration={markings, mark=at position 0.5 with {\arrow{>}}},
        postaction={decorate},black, thick] (-2,0) to node[above]{\small$\mathbb{A}$} (1,0);
  \filldraw[black] (1,0) node[above]{\footnotesize$\beta$} node[below]{\footnotesize$p$} circle (.15);
 \end{tikzpicture}
  ~ &\longrightarrow &&  \big\langle \beta , \mathsf{F}_-\Big(\mathrm{ad}_{\mathrm{i}\hbar\frac{\partial}{\partial \beta}}\Big) \mathbb{A}_p\big\rangle ~,
   \qquad&\begin{tikzpicture}[scale=.4, baseline=(x.base)]
  \coordinate (x) at (0,-0.3);
  \draw[decoration={markings, mark=at position 0.5 with {\arrow{>}}},
        postaction={decorate},black, thick] (-2,0) to node[above]{\small$\mathbb{A}$} (1,0);
  \filldraw[black] (-2,0) node[above]{\footnotesize$\beta$} node[below]{\footnotesize$p$} circle (.15);
 \end{tikzpicture}
  &&\longrightarrow &~~  \big\langle \beta , \mathsf{F}_+\Big(\mathrm{ad}_{\mathrm{i}\hbar\frac{\partial}{\partial \beta}}\Big) \mathbb{A}_p\big\rangle~,\\
 \begin{tikzpicture}[scale=.4, baseline=(x.base)]
  \coordinate (x) at (0,-0.3);
  \draw[decoration={markings, mark=at position 0.5 with {\arrow{>}}},
        postaction={decorate},black, thick] (-2,0) to node[above]{\small$\mathbb{B}$} (1,0);
  \filldraw[black] (1,0) node[above]{\footnotesize$\alpha$} node[below]{\footnotesize$p$} circle (.15);
 \end{tikzpicture}
  ~ &\longrightarrow &&  \big\langle \mathbb{B}_p , \mathsf{F}_+\Big(\mathrm{ad}_{\mathrm{i}\hbar\frac{\partial}{\partial \mathbb{B}_p}}\Big) \alpha \big\rangle ~,
   &\begin{tikzpicture}[scale=.4, baseline=(x.base)]
  \coordinate (x) at (0,-0.3);
  \draw[decoration={markings, mark=at position 0.5 with {\arrow{>}}},
        postaction={decorate},black, thick] (-2,0) to node[above]{\small$\mathbb{B}$} (1,0);
  \filldraw[black] (-2,0) node[above]{\footnotesize$\alpha$} node[below]{\footnotesize$p$} circle (.15);
 \end{tikzpicture}
   &&\longrightarrow &~~  \big\langle \mathbb{B}_p , \mathsf{F}_-\Big(\mathrm{ad}_{\mathrm{i}\hbar\frac{\partial}{\partial \mathbb{B}_p}}\Big) \alpha \big\rangle ~.
\eeq
Here we have introduced the following functions:
\beq
 \mathsf{F}_+(x) &= \frac{x}{1-\mathrm{e}^{-x}} &&=\sum_{j=0}^\infty(-1)^j\frac{ B_{j}}{j!}x^j &&= 1+\frac{x}{2}+\frac{x^2}{12}-\frac{x^4}{720}+\cdots~,\\
 \mathsf{F}_-(x) &= \frac{x}{1-\mathrm{e}^{x}} &&=-\sum_{j=0}^\infty\frac{ B_{j}}{j!}x^j &&=-1+\frac{x}{2}-\frac{x^2}{12}+\frac{x^4}{720}+\cdots~,
\eeq
where $B_j$ are the Bernoulli numbers $B_0=1,\, B_1=-\frac12,\, B_2=\frac16,\, B_3=0,\, B_4=-\frac{1}{30},\,\ldots$
In~\eqref{HE_contributions_to_Omega}, functions $\mathsf{F}_\pm$ are evaluated on $x=\mathrm{ad}_{\mathrm{i}\hbar\frac{\partial}{\partial\beta}}$, producing $\mathrm{End}(\mathfrak{g})$-valued derivations (of infinite order) of the space of functions of $\beta$. 
Further in this section we will also need the following two functions, related to the generating functions for Bernoulli polynomials:
\begin{equation}\label{G pm}
 \mathsf{G}_+(t,x)=\frac{1-\mathrm{e}^{-tx}}{1-\mathrm{e}^{-x}}~,\qquad \mathsf{G}_-(t,x)=\frac{1-\mathrm{e}^{(1-t)x}}{1-\mathrm{e}^{x}}~.
\end{equation}

Note that, when acting on the partition function, the complicated operators $\Omega_p^{\mathbb{B}\alpha}, \Omega_p^{\alpha\mathbb{B}}$ from~\eqref{HE_contributions_to_Omega} act simply as multiplication operators 
\begin{equation}
 \Omega_p^{\mathbb{B}\alpha}\sim \langle \mathbb{B}_p,\alpha\rangle ~, \qquad \Omega_p^{\alpha\mathbb{B}}\sim -\langle \mathbb{B}_p,\alpha\rangle ~,
\end{equation}
since the derivative in the corner value of the field $\mathbb{B}$ acts by zero. 
Thus, we have:
\begin{equation}
 \Omega_p^{\mathbb{B}\alpha}= \langle \mathbb{B}_p,\alpha\rangle +\cdots ~,\qquad \Omega_p^{\alpha\mathbb{B}}= -\langle \mathbb{B}_p,\alpha\rangle+\cdots ~,
\end{equation} 
where we have added the terms $\cdots$ (irrelevant for the master equation) so as to have the property $\Omega^2=0$.  
To be precise, we impose the following mild restriction on the states (or, in other words, it is a clarification of our model for the space of states).
\begin{assumption}[Admissible states] \label{assump: admissible states}
 We assume that the states do not depend explicitly on the limiting values of $1$-form components of fields $\mathbb{A,B}$ at corners. 
 I.e., for~$p$ a corner, the derivatives $\frac{\partial}{\partial \mathbb{A}_p^{(1)}}$, $\frac{\partial}{\partial \mathbb{B}_p^{(1)}}$ act by zero on on admissible states.
\end{assumption}

Then, by a direct computation, one verifies the following (we give the explicit proof in Appendix~\ref{Omega^2=0}).

\begin{proposition}\label{prop: Omega^2=0}
 For a stratified circle, with any choice of polarizations on the strata, the operator $\Omega$ as defined by~(\ref{Omega_stratified},\ref{Omega pure corner},\ref{HE_contributions_to_Omega}) satisfies $\Omega^2=0$ on admissible states in the sense of Assumption~\ref{assump: admissible states}. 
\end{proposition}

Let us introduce the following terminology. 
For a product of intervals (or circles) $I\times J$ with $I$ parameterized by coordinate $t$ and $J$ parameterized by $\tau$, we call the axial gauge propagator containing $\delta(\tau-\tau')$ \textit{parallel} to~$I$ (and \textit{perpendicular} to~$J$), since the intervals on which the $\delta$-term is supported are parallel to~$I$. 
Note that in all the computations of Section~\ref{2d_YM_for_chi>0}, the axial gauge was always chosen to be perpendicular to the boundary.

Consider a surface $\Sigma$ with stratified boundary circles in picture~\ref{pict_2}, as in~\eqref{stratified circle II}, i.e., with arcs and corners carrying polarization data. 
We can view such a surface as a limit at $s\rightarrow 0$ of a family of surfaces $\Sigma_s$, with corners of $\Sigma$ expanded into arcs of corresponding polarization (thus, surfaces $\Sigma_s$ for $s>0$ are in picture~\ref{pict_1}). 
Let $\eta_s$ be a family of propagators (corresponding to a family of gauge-fixings) on the surfaces $\Sigma_s$, converging to a propagator $\eta$ on $\Sigma$. 
We make the following assumption.

\begin{assumption}[Collapsible gauge condition]\label{assump: collapsible gauge}
 The contraction of the propagator $\eta_s$ with a $1$-form $\mathbb{B}^{(1)}$ on the  $\mathbb{B}$-interval $I_s\subset \partial\Sigma_s$ that is being collapsed into a $\beta$-corner $p$ of $\Sigma$, becomes supported at $p$ in the limit $s\rightarrow 0$.
\end{assumption}
This assumption can be realized by considering an $s$-dependent family of metric gauge-fixings associated to equipping $\Sigma_s$ with a metric $g_s$ in which the $\mathbb{B}$-arc undergoing the collapse is placed at the end of a long ``tentacle''. 
Thus, at $s\rightarrow 0$, the $\beta$-corner is placed infinitely far from the rest of the surface.
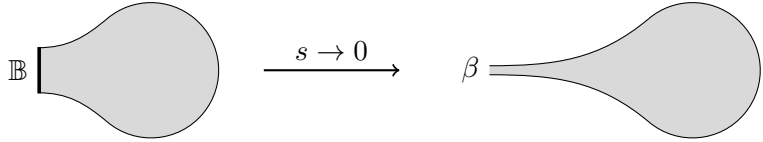
\[
 \begin{tikzpicture}[scale=.3]
  \filldraw[draw=black, fill=black!15] (-130:3) arc (-130:130:3) to[out=220, in=0] (-5,1) to node[left]{$\mathbb{B}$} (-5,-1) to[out=0, in=-220] (-130:3);
  
  \filldraw[black] (-5,-1) rectangle (-4.9,1);
  
  \draw[decoration={markings, mark=at position 1 with {\arrow{>}}},
        postaction={decorate},black, thick] (5,0) to node[above]{$s\rightarrow 0$}(11,0);
  
  \begin{scope}[shift={(24,0)}]
   \filldraw[draw=black, fill=black!15] (-9,-.2) to[out=0, in=-220] (-130:3) arc (-130:130:3) to[out=220, in=0] (-9,.2);
   \node[left] at (-9,0) {$\beta$};
  \end{scope}

 \end{tikzpicture}
\]
Put another way, if both arcs adjacent to $I_s$ are in $\mathbb{A}$-polarization, the assumption requires that $\eta_s$ asymptotically approaches the axial gauge propagator $\eta(t,\tau;t',\tau')=(\Theta(t-t')-t)$ $\cdot\delta(\tau-\tau')\, (\mathrm{d}\tau'- \mathrm{d}\tau)-\mathrm{d}t\, \Theta(\tau'-\tau)$ 
(the axial propagator \textit{parallel} to $I_s$)
near the collapsing interval $I_s$, as $s\rightarrow 0$, with $t$ the coordinate along $I_s$ and $\tau$ the coordinate along the ``tentacle''.

\begin{proposition}[mQME in picture~\ref{pict_2}] \label{prop: mQME II}
 Under the assumption above, the partition function~$Z$ for the surface~$\Sigma$ with boundary and corners equipped with polarization data, satisfies the modified quantum master equation 
 \[
  (\hbar^2 \Delta+\Omega)Z=0~,
 \]
 where $\Omega$ is given as the sum of expressions~\eqref{Omega_stratified} over the boundary circles, with edge contributions given by~(\ref{Omega A},\ref{Omega B}) and corner contributions defined by~(\ref{Omega corner = HE+corner+HE},\ref{Omega pure corner},\ref{HE_contributions_to_Omega}). 
\end{proposition}
Here $Z$ is understood as the limit $s\rightarrow 0$ of the evaluation of  partition function of picture~\ref{pict_1} on $\Sigma_s$ on the fields pulled back from edges and corners of $\Sigma$ along the collapse map $\Sigma_s\rightarrow \Sigma$. 
See Remark~\ref{rem: proof of mQME via Abeta disk} below for an explicit example of the mQME with corners in the picture~\ref{pict_2}. 
Also, in Remark~\ref{rem: collapses of AB disk} we will have a non-example showing that the mQME does indeed fail without the Assumption~\ref{assump: collapsible gauge}.

\begin{proof}[Sketch of proof]
Proposition~\ref{prop: mQME II} arises as a corollary of Proposition~\ref{prop: mQME I}, since~$Z$ is understood as a limit $s\rightarrow 0$, in the sense explained above, of partition functions $Z_s$ on surfaces $\Sigma_s$ in picture~\ref{pict_1}, which do satisfy the modified quantum master equation by Proposition~\ref{prop: mQME I}. 
Contribution $\mathrm{i}\hbar \big\langle \frac12\big[\alpha,\alpha\big],\frac{\partial}{\partial \alpha} \big\rangle $ to $\Omega$ in picture~\ref{pict_2} at a $\mathbb{B}-\alpha-\mathbb{B}$ corner arises as an $s\rightarrow 0$ limit of the $\mathbb{A}$-edge contribution~\eqref{Omega A} from the $\mathbb{A}$-edge of $\Sigma_s$ collapsing to the corner. 

Next, consider an $\mathbb{A}-\beta-\mathbb{A}$ corner where, in addition to Assumption~\ref{assump: collapsible gauge}, we assume for the moment that the $0$-form component of the $\mathbb{A}$ field is continuous through the corner. 
In this case, the corner contribution to $\Omega$ given by~(\ref{Omega corner = HE+corner+HE},\ref{HE_contributions_to_Omega}) simplifies to $\Omega^{\mathbb{A}\beta\mathbb{A}}= \mathrm{i}\hbar\big\langle [\mathbb{A}_p,\beta],\frac{\partial}{\partial \beta} \big\rangle$ and it arises from the fact that $Z_s$ depends on the $1$-form field $\mathbb{B}^{(1)}$ at the edge $I_s$ collapsing into the $\beta$-corner $p$, and this dependence is important for the mQME in picture~\ref{pict_1}. 
Using the Assumption~\ref{assump: collapsible gauge} and the continuity of $\mathbb{A}^{(0)}$ through the corner, the dependence of $Z_s$ on $\mathbb{B}^{(1)}$ for small $s$ is: $Z_s\sim \mathrm{e}^{\frac{\mathrm{i}}{\hbar}\int_{I_s}\langle\mathbb{B}^{(1)},\mathbb{A}_p\rangle}$, and one has:
\[
  \Omega_{I_s}^{\mathbb{B}}Z_s\sim \Big(\mathrm{i}\hbar\int_{I_s}\big\langle [\mathbb{A}_p,\mathbb{B}^{(0)}],\frac{\partial}{\partial \mathbb{B}^{(0)}} \big\rangle \Big) Z_s \rightarrow  \mathrm{i}\hbar\big\langle [\mathbb{A}_p,\beta],\frac{\partial}{\partial \beta} \big\rangle Z ~.
\]
Therefore, one can compensate for the loss of dependence on $\mathbb{B}^{(1)}$ during the collapse by inclusion of the term $\mathrm{i}\hbar\big\langle [\mathbb{A}_p,\beta],\frac{\partial}{\partial \beta} \big\rangle$ in~$\Omega$.

Finally, consider the $\mathbb{A}-\beta-\mathbb{A}$ corner without assuming the continuity of $\mathbb{A}^{(0)}$ through the corner. 
To analyze the dependence of $Z_s$ on $\mathbb{B}^{(1)}$, we cut a rectangle $\mathcal{R}$ out of $\Sigma_s$ at the collapsing edge:
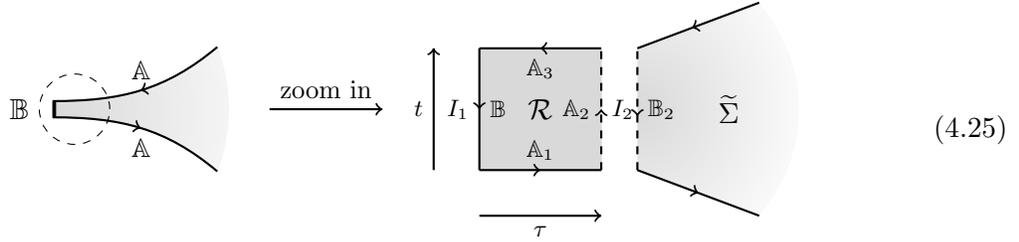
\beq
 \begin{tikzpicture}[scale=.3]
 
  \begin{scope}[scale=1.2]
    \clip (-8,-.3) to[out=0, in=-220] (-130:3) arc (-130:130:3) to[out=220, in=0] (-8,.3);
    \shade[inner color=black!15, outer color=black!4] (-8,0) circle (6.5);
  \end{scope}
  
  \begin{scope}[scale=1.2]
   \draw[decoration={markings, mark=at position 0.5 with {\arrow{>}}},
        postaction={decorate},draw=black, thick] (-8,-.3) to[out=0, in=-220]node[below]{\small$\mathbb{A}$} (-130:3);
   \draw[decoration={markings, mark=at position 0.5 with {\arrow{>}}},
        postaction={decorate},black, thick] (130:3) to[out=220, in=0] node[above]{\small$\mathbb{A}$} (-8,.3);
   \node[left] at (-8.5,0) {$\small\mathbb{B}$};
   \draw[dashed, black] (-7.2,0) circle (1.3);
   \filldraw (-8,.3) rectangle (-7.9,-.3);
  \end{scope}
  
  \draw[decoration={markings, mark=at position 1 with {\arrow{>}}},
        postaction={decorate},black, thick] (0,0) to node[above]{\small zoom in}(5,0);
        
  \begin{scope}[shift={(12,0)}, scale=1.35]
   \filldraw[color=black!15] (-2,-2) rectangle (2,2);   
   
   \draw[decoration={markings, mark=at position .5 with {\arrow{>}}},
        postaction={decorate},black, thick] (-2,-2) to node[above]{\footnotesize$\mathbb{A}_1$}(2,-2);
        
   \draw[decoration={markings, mark=at position .5 with {\arrow{>}}},
        postaction={decorate},black, thick, dashed] (2,-2) to node[left]{\footnotesize$\mathbb{A}_2$} node[right]{\footnotesize$I_2$}(2,2);
        
   \draw[decoration={markings, mark=at position .5 with {\arrow{>}}},
        postaction={decorate},black, thick] (2,2) to node[below]{\footnotesize$\mathbb{A}_3$}(-2,2);
        
   \draw[decoration={markings, mark=at position .5 with {\arrow{>}}},
        postaction={decorate},black, thick] (-2,2) to node[right]{\footnotesize$\mathbb{B}$} node[left]{\footnotesize${I}_1$}(-2,-2);  
   \node at (0,0) {$\small\mathcal{R}$};
   
   \draw[decoration={markings, mark=at position 1 with {\arrow{>}}},
        postaction={decorate},black, thick] (-2,-3.5) to node[below]{\footnotesize$\tau$}(2,-3.5);
        
   \draw[decoration={markings, mark=at position 1 with {\arrow{>}}},
        postaction={decorate},black, thick] (-3.5,-2) to node[left]{\footnotesize$t$}(-3.5,2);
  \end{scope}
  
  \begin{scope}[shift={(19,0)}, scale=1.35]
    \clip (-2,2) to (-2,-2) to (2,-3.5) to (5,-3.5) to (5,3.5) to (2,3.5);
    \shade[inner color=black!15, outer color=black!2] (-2,0) circle (5.3);
  \end{scope}
  
  \begin{scope}[shift={(19,0)}, scale=1.35]
   \draw[decoration={markings, mark=at position .5 with {\arrow{<}}},
        postaction={decorate},black, thick, dashed] (-2,-2) to node[right]{\footnotesize$\mathbb{B}_2$}(-2,2);
        
   \draw[decoration={markings, mark=at position .5 with {\arrow{<}}},
        postaction={decorate},black, thick] (-2,2) to (2,3.5);
        
   \draw[decoration={markings, mark=at position .5 with {\arrow{>}}},
        postaction={decorate},black, thick] (-2,-2) to (2,-3.5);
        
   \node at (1,0) {$\widetilde{\Sigma}$};     
  \end{scope} 
  
 \end{tikzpicture}
\eeq
Thus, we present the surface $\Sigma_s$ as $\mathcal{R}\cup_{I_2} \widetilde{\Sigma}$. 
Computing the partition function on the rectangle in the axial gauge,%
\footnote{
 We are using the axial gauge propagator parallel to $I_s$, i.e. $\eta(t,\tau \,;\, t',\tau')= {(\Theta(t-t')-t)\, \delta(\tau-\tau') \, (\mathrm{d}\tau'-\mathrm{d}\tau)-\mathrm{d}t\,\Theta(\tau'-\tau)}$ with $t,\tau$ the vertical and horizontal coordinate on the rectangle. 
 This choice is the one consistent with the Assumption~\ref{assump: collapsible gauge}.
} 
setting $\mathbb{A}_1=\mathbb{A}_{p+0}$, $\mathbb{A}_3=\mathbb{A}_{p-0}$ -- constant zero-forms, the limiting values of $\mathbb{A}^{(0)}$ to the right and left of the corner $p$ on $\Sigma$, and setting $\mathbb{A}_2= \mathrm{d}t\,\underline{\mathbb{A}}$ -- a constant $1$-form, we find the following:
\begin{equation}
 Z_\mathcal{R}=\mathrm{e}^{\frac{\mathrm{i}}{\hbar}\int_{\bar{I}_s} \mathrm{d}t\,\langle\mathbb{B}^{(0)},\underline{\mathbb{A}}\rangle
  + \langle\mathbb{B}^{(1)}, \mathsf{G}_-(t,\mathrm{ad}_{\underline{\mathbb{A}}})\,\mathbb{A}_{p+0}
   + \mathsf{G}_+(t,\mathrm{ad}_{\underline{\mathbb{A}}})\,\mathbb{A}_{p-0}\rangle ~,
}
\end{equation}
with $\mathsf{G}_\pm$ as in~\eqref{G pm}. 
Here the integral is over $t\in [0,1]$, or equivalently over $I_s$ with reversed orientation. 
This implies 
\begin{equation}
 (\Omega^{\mathbb{B}}_{I_s}+\langle \beta, \mathbb{A}_{p+0}\rangle - \langle \beta, \mathbb{A}_{p-0} \rangle ) Z_\mathcal{R}\Big|_{\mathbb{B}=\beta}
  = -\langle \beta , \mathsf{F}_-(\mathrm{ad}_{\underline{\mathbb{A}}})\,\mathbb{A}_{p+0} + \mathsf{F}_+(\mathrm{ad}_{\underline{\mathbb{A}}})\,
   \mathbb{A}_{p-0}\rangle\; \mathrm{e}^{\frac{\mathrm{i}}{\hbar}\langle \beta, \underline{\mathbb{A}} \rangle} ~.
\end{equation}
The operator acting on $Z_\mathcal{R}$ on the left hand side is the part of the $\Omega$ in picture~\ref{pict_1} corresponding to the collapsing interval $I_s$ and its two endpoints. 
Thus, combining with the gluing formula for partition functions we have:
\begin{multline}
 (\Omega^{\mathbb{B}}_{I_s}+\langle \beta, \mathbb{A}_{p+0}\rangle - \langle \beta, \mathbb{A}_{p-0} \rangle )\; 
  \underbrace{\int \mathrm{d}\underline{\mathbb{A}}\; \mathrm{d}\beta' \;Z_\mathcal{R} \cdot \mathrm{e}^{-\frac{i}{\hbar}\langle \beta', \underline{\mathbb{A}} \rangle}\cdot Z_{\widetilde \Sigma}(\beta',\cdots)}_{Z_{\Sigma_s}} =\\
=\langle \beta , \mathsf{F}_+(\mathrm{ad}_{i\hbar \frac{\partial}{\partial \beta}})\,\mathbb{A}_{p+0} + \mathsf{F}_-(\mathrm{ad}_{i\hbar \frac{\partial}{\partial \beta}})\,\mathbb{A}_{p-0}\rangle\; Z_{\widetilde\Sigma}(\beta,\cdots) ~.
\end{multline}
Thus, the action on $Z_{\Sigma_s}$ of the part of $\Omega$ in picture~\ref{pict_1} corresponding to the collapsing interval (with its endpoints) is compensated by the action on the partition function in picture~\ref{pict_2} of the operator appearing on the right hand side -- which is precisely our anticipated corner contribution in picture~\ref{pict_2}, $\Omega^{\mathbb{A}\beta \mathbb{A}}=\Omega^{\mathbb{A}\beta }+\Omega^{\beta \mathbb{A}}$, see~\eqref{HE_contributions_to_Omega}. 
\end{proof}

\begin{remark}\label{rem: proof of mQME via Abeta disk}
Another argument for the contribution to $\Omega$ from a $\beta$-corner is as follows. 
In Section~\ref{BF A-disk with b corners} we will obtain the explicit partition function for an $\mathbb{A}$-disk $D$ with a single $\beta$-corner 
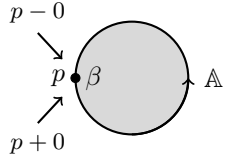
\[
 \begin{tikzpicture}[scale=.5]
   \filldraw[fill=black!15, thick, draw=black] (0,0) circle (1.5);
   
   \draw[black, thick, ->] (0,-1.5) arc (-90:0:1.5) node[right=2pt]{\small$\mathbb{A}$};
   
   \filldraw[black] (-1.5,0) circle (.12) node[right]{\small$\beta$} node[left]{\small$p$};
   
   \draw[black, thick, ->] (-2.5,1.2) node[above]{\footnotesize$p-0$} to (-1.85,.5);
   
   \draw[black, thick, ->] (-2.5,-1.2) node[below]{\footnotesize$p+0$} to (-1.85,-.5);
   
 \end{tikzpicture}
\]
in the form
\begin{equation}
 Z_D=\mathrm{e}^{\frac{\mathrm{i}}{\hbar}\langle \beta, \log U(\mathbb{A}) \rangle} ~,
\end{equation}
with $U(\mathbb{A})$ the holonomy of the $1$-form field $\mathbb{A}^{(1)}$ along the boundary circle. 
From Baker-Campbell-Hausdorff formula, one finds 
\begin{equation}
 \Omega^\mathbb{A}\log U(\mathbb{A})=-\mathrm{i}\hbar\, (\mathsf{F}_+(\mathrm{ad}_{\log U(\mathbb{A})})\, \mathbb{A}_{p-0}
  +\mathsf{F}_-(\mathrm{ad}_{\log U(\mathbb{A})})\, \mathbb{A}_{p+0})~,
\end{equation}
which implies
\begin{equation}
 \Omega^\mathbb{A} Z_D= \langle \beta,  \mathsf{F}_+(\mathrm{ad}_{\log U(\mathbb{A})})\, \mathbb{A}_{p-0}
  +\mathsf{F}_-(\mathrm{ad}_{\log U(\mathbb{A})})\, \mathbb{A}_{p+0}\rangle \cdot Z_D ~.
\end{equation}
From this one immediately sees that 
\begin{equation}
 (\Omega^\mathbb{A}+\Omega^{\beta \mathbb{A}}_p+\Omega^{\mathbb{A}\beta}_p)Z_D=0 ~,
\end{equation}
with the corner contributions as prescribed by~\eqref{HE_contributions_to_Omega}. 
Thus, the mQME works by a direct computation. 
For a general surface $\Sigma$ containing a $\beta$-corner, surrounded by $\mathbb{A}$-edges, one can cut out a disk around the corner and the mQME will follow from the one we just checked for the disk and from the one for the remaining part of the surface (thus by induction one can reduce to the case of surfaces without $\mathbb{A}-\beta-\mathbb{A}$-corners).
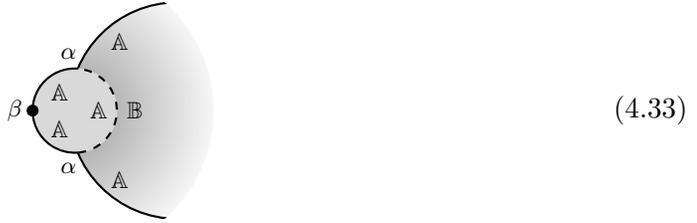
\beq\label{b_corner_bubble}
 \begin{tikzpicture}[scale=.4, baseline=(x.base)]
  
  \begin{scope}[scale=1.2, shift={(.6,0)}]
   \clip (-3,-3) rectangle (3,3);
   \clip (0,0) circle (3);
   \shade[inner color= black!40, outer color=black!3] (-3,0) circle (4);
  \end{scope}

  \begin{scope}[scale=1.2, shift={(.6,0)}]
   \clip (-3,0) circle (4);
   \draw[thick] (0,0) circle (3);
  \end{scope}

  \begin{scope}[shift={(.3,0)}]
   \filldraw[black!15] (-3,0) circle (1.4);
   \draw[thick, dashed] ($(-3,0)+ (-85:1.4)$) arc (-85:85:1.4);
   \draw[thick] ($(-3,0)+ (85:1.4)$) arc (85:275:1.4);
   
   \filldraw[black] (-4.4,0) circle (.18) node[left]{\footnotesize$\beta$};
  \end{scope}
  
  \node at (-1.9,0) {\footnotesize$\mathbb{A}$};
  \node at (-.7,0) {\footnotesize$\mathbb{B}$};
  
  \node at (-3.2,.6) {\footnotesize$\mathbb{A}$};
  \node at (-3.2,-.6) {\footnotesize$\mathbb{A}$};
  
  \node at (-2.9,1.9) {\footnotesize$\alpha$};
  \node at (-2.9,-1.9) {\footnotesize$\alpha$};
  
  \node at (-1.2,2.3) {\footnotesize$\mathbb{A}$};
  \node at (-1.2,-2.3) {\footnotesize$\mathbb{A}$};
    
 \end{tikzpicture}
\eeq

\end{remark}

Yet another approach to the proof of Proposition~\ref{prop: mQME II}, explaining the corner contributions~\eqref{HE_contributions_to_Omega}, is in the vein of the proof of Proposition~\ref{prop: mQME I}, with $\Omega$ given by Feynman subgraphs collapsing at the boundary/corners. 
Consider e.g. a collapse at a $\mathbb{A}-\beta-\mathbb{A}$ corner. 
The following subgraphs are contributing:
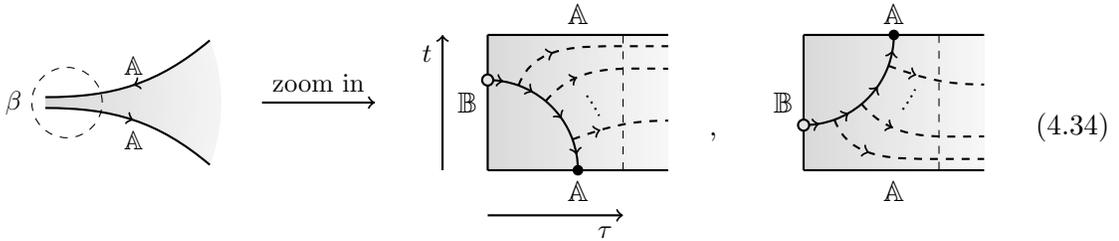
\beq\label{AbetaA corner collapse}
 \begin{tikzpicture}[scale=.3]
 
  \begin{scope}[scale=1.2]
    \clip (-8,-.2) to[out=0, in=-220] (-130:3) arc (-130:130:3) to[out=220, in=0] (-8,.2);
    \shade[inner color=black!15, outer color=black!4] (-8,0) circle (6.5);
  \end{scope}
  
  \begin{scope}[scale=1.2]
   \draw[decoration={markings, mark=at position 0.5 with {\arrow{>}}},
        postaction={decorate},draw=black, thick] (-8,-.2) to[out=0, in=-220]node[below]{\small$\mathbb{A}$} (-130:3);
   \draw[decoration={markings, mark=at position 0.5 with {\arrow{>}}},
        postaction={decorate},black, thick] (130:3) to[out=220, in=0] node[above]{\small$\mathbb{A}$} (-8,.2);
   \node[left] at (-8.5,0) {\small$\beta$};
   \draw[dashed, black] (-7.2,0) circle (1.3);
  \end{scope}
  
  \draw[decoration={markings, mark=at position 1 with {\arrow{>}}},
        postaction={decorate},black, thick] (0,0) to node[above]{\small zoom in}(5,0);

  \begin{scope}[yscale=-1, shift={(13,0)}]
  
    \shade[left color=black!15, right color=black!3] (-3,-3) rectangle (5,3);

    \draw[black, thick] (-3,-3) to node[left]{$\mathbb{B}$} (-3,3);
    \draw[black, thick] (-3,3) to node[below]{$\mathbb{A}$} (5,3);
    \draw[black, thick] (-3,-3) to node[above]{$\mathbb{A}$} (5,-3);
    
    \draw[black, dashed] (3,3) to (3,-3);
    
    \draw[->, thick] (-5,3) to (-5,-3) node[below left]{\small$t$};
    \draw[<-, thick] (3,5) node[below left]{\small$\tau$} to (-3,5);

    \draw[decoration={markings, mark=at position 0.5 with {\arrow{>}}},
        postaction={decorate},color=black, thick] (-3,-1) arc (-90:-70:4);
    \draw[decoration={markings, mark=at position 0.3 with {\arrow{>}}},
        postaction={decorate}, black, thick, dashed] ($(-3,3)+(-70:4)$) to[out=-70, in=180] (3,-2.5) to[out=0, in=180] (5,-2.5);

    \draw[decoration={markings, mark=at position 0.5 with {\arrow{>}}},
        postaction={decorate},color=black, thick] ($(-3,3)+(-70:4)$) arc (-70:-50:4);
    \draw[decoration={markings, mark=at position 0.3 with {\arrow{>}}},
        postaction={decorate}, black, thick, dashed] ($(-3,3)+(-50:4)$) to[out=-50, in=180] (5,-1.5);

    \draw[decoration={markings, mark=at position 0.5 with {\arrow{>}}},
        postaction={decorate},color=black, thick] ($(-3,3)+(-50:4)$) arc (-50:-20:4);
    \draw[decoration={markings, mark=at position 0.3 with {\arrow{>}}},
        postaction={decorate}, black, thick, dashed] ($(-3,3)+(-20:4)$) to[out=-20, in=180] (5,.8);

    \draw[decoration={markings, mark=at position 0.5 with {\arrow{>}}},
        postaction={decorate},color=black, thick] ($(-3,3)+(-20:4)$) arc (-20:0:4);
    
    \draw[color=black, thick, fill=black!10] (-3,-1) circle (7pt);
    \draw[color=black, fill=black] (1,3) circle (6pt);
    
    \draw[color=black, thick, dotted] ($(-3,3)+(-37:5.5)$) arc   (-37:-26:5.5);
    
   \end{scope}   

   \node at (20,-1.5) {,};
   
   \begin{scope}[yscale=1, shift={(27,0)}]
  
    \shade[left color=black!15, right color=black!3] (-3,-3) rectangle (5,3);

    \draw[black, thick] (-3,-3) to node[left]{$\mathbb{B}$} (-3,3);
    \draw[black, thick] (-3,3) to node[above]{$\mathbb{A}$} (5,3);
    \draw[black, thick] (-3,-3) to node[below]{$\mathbb{A}$} (5,-3);
    
    \draw[black, dashed] (3,3) to (3,-3);

    \draw[decoration={markings, mark=at position 0.5 with {\arrow{>}}},
        postaction={decorate},color=black, thick] (-3,-1) arc (-90:-70:4);
    \draw[decoration={markings, mark=at position 0.3 with {\arrow{>}}},
        postaction={decorate}, black, thick, dashed] ($(-3,3)+(-70:4)$) to[out=-70, in=180] (3,-2.5) to[out=0, in=180] (5,-2.5);

    \draw[decoration={markings, mark=at position 0.5 with {\arrow{>}}},
        postaction={decorate},color=black, thick] ($(-3,3)+(-70:4)$) arc (-70:-50:4);
    \draw[decoration={markings, mark=at position 0.3 with {\arrow{>}}},
        postaction={decorate}, black, thick, dashed] ($(-3,3)+(-50:4)$) to[out=-50, in=180] (5,-1.5);

    \draw[decoration={markings, mark=at position 0.5 with {\arrow{>}}},
        postaction={decorate},color=black, thick] ($(-3,3)+(-50:4)$) arc (-50:-20:4);
    \draw[decoration={markings, mark=at position 0.3 with {\arrow{>}}},
        postaction={decorate}, black, thick, dashed] ($(-3,3)+(-20:4)$) to[out=-20, in=180] (5,.8);

    \draw[decoration={markings, mark=at position 0.5 with {\arrow{>}}},
        postaction={decorate},color=black, thick] ($(-3,3)+(-20:4)$) arc (-20:0:4);
    
    \draw[color=black, thick, fill=black!10] (-3,-1) circle (7pt);
    \draw[color=black, fill=black] (1,3) circle (6pt);
    
    \draw[color=black, thick, dotted] ($(-3,3)+(-37:5.5)$) arc   (-37:-26:5.5);
    
   \end{scope}
    
 \end{tikzpicture}
\eeq
One computes these contributions to $\Omega$ using the propagator $\eta=(\Theta(t-t')-t)\delta(\tau-\tau')(\mathrm{d}\tau-\mathrm{d}\tau')-\mathrm{d}t\, \Theta (\tau'-\tau)$ in the rectangle that we see when zooming into the corner. 
In the zoomed-in picture we  are considering configurations of points modulo the horizontal rescalings $\tau\mapsto c\cdot \tau$. 
We fix a representative of the quotient by fixing the horizontal position of one marked vertex. 
Edges leaving the collapsing subgraph are assigned the expression $\mathrm{d}t\cdot \mathrm{i}\hbar\frac{\partial}{\partial \beta}$ (the factor $\mathrm{d}t$ comes from the propagator associated to the external edge). 
The graphs in~\eqref{AbetaA corner collapse} are easily computed and yield, when summed over the number of external edges, $\Omega^{\mathbb{A}\beta\mathbb{A}}=\Omega^{\mathbb{A}\beta} + \Omega^{\beta\mathbb{A}}$ with $\Omega^{\mathbb{A}\beta},\Omega^{\beta\mathbb{A}}$ given by the formulae~\eqref{HE_contributions_to_Omega}.

\subsubsection{Space of states for the stratified circle as assembled from spaces of states for edges and corners}

One can regard the space of states~\ref{H for stratified circle} for the stratified circle as constructed from the spaces of states for individual edges.
One assigns  to an interval (with chosen polarization $\mathbb{P}$ in the bulk and $\xi,\xi'\in \{\alpha,\beta\}$ on the endpoints) a space of states -- a cochain complex -- 
constructed as the space of functions on the $\mathbb{P}$-field at the edge and fields at the corners (understood as independent fields if the corner and edge polarizations disagree; if the polarizations agree, the corner field is the limiting value of the edge field):
\beq
 \begin{tikzpicture}[scale=.4, baseline=(x.base)]
  \coordinate (x) at (0,-0.3);
  
  \draw[decoration={markings, mark=at position 0.5 with {\arrow{>}}},
        postaction={decorate},black, thick] (-2,0) to node[above]{\small$\mathbb{A}$} node[below]{\small$I$} (2,0);

  \filldraw[black] (-2,0) node[above]{\footnotesize$\xi$} node[below]{\footnotesize$p_\mathrm{in}$} circle (.15);
  \filldraw[black] (2,0) node[above]{\footnotesize$\xi'$} node[below]{\footnotesize$p_\mathrm{out}$} circle (.15);
 \end{tikzpicture}
 ~ \longrightarrow ~\mathcal{H}_I^{\xi,\mathbb{A},\xi'}
  &= \mathrm{Fun}_\mathbb{C}\Bigg( \Bigg\{\begin{array}{l} \mathfrak{g}[1] \\ \mathfrak{g}^*\end{array} \Bigg\} \times_{\mathfrak{g}[1]}
   \Omega^\bullet(I,\mathfrak{g})[1] \times_{\mathfrak{g}[1]} \Bigg\{\begin{array}{l} \mathfrak{g}[1] \\  
    \mathfrak{g}^*\end{array} \Bigg\} \Bigg)~,\\ 
 \Omega_I^{\xi,\mathbb{A},\xi'} &=\Omega_\mathrm{in}^{\xi}+\Omega_\mathrm{in}^{\xi\mathbb{A}}
  +\Omega_I^{\mathbb{A}} +\Omega_\mathrm{out}^{\mathbb{A}\xi'} +\Omega_\mathrm{out}^{\xi'} ~,\\
 \begin{tikzpicture}[scale=.4, baseline=(x.base)]
  \coordinate (x) at (0,-0.3);
  
  \draw[decoration={markings, mark=at position 0.5 with {\arrow{>}}},
        postaction={decorate},black, thick] (-2,0) to node[above]{\small$\mathbb{B}$} node[below]{\small$I$} (2,0);

  \filldraw[black] (-2,0) node[above]{\footnotesize$\xi$} node[below]{\footnotesize$p_\mathrm{in}$} circle (.15);
  \filldraw[black] (2,0) node[above]{\footnotesize$\xi'$} node[below]{\footnotesize$p_\mathrm{out}$} circle (.15);
 \end{tikzpicture}
 ~ \longrightarrow  ~\mathcal{H}_I^{\xi,\mathbb{B},\xi'}
  &= \mathrm{Fun}_\mathbb{C}\Bigg( \Bigg\{\begin{array}{l}\mathfrak{g}[1] \\  \mathfrak{g}^* \end{array} \Bigg\}
   \times_{ \mathfrak{g}^*} \Omega^\bullet(I,\mathfrak{g}^*) \times_{ \mathfrak{g}^*} \Bigg\{\begin{array}{l}\mathfrak{g}[1] \\  \mathfrak{g}^* \end{array} \Bigg\} \Bigg)~,\\  
 \Omega_I^{\xi,\mathbb{A},\xi'} &=\Omega_\mathrm{in}^{\xi}+ \Omega_\mathrm{in}^{\xi\mathbb{B}}+ \Omega_I^{\mathbb{B}}+\Omega_\mathrm{out}^{\mathbb{B}\xi'} 
  + \Omega_\mathrm{out}^{\xi'}~.
\eeq
Here top/bottom choice for the fiber product factors on the left/right corresponds to $\alpha$ or $\beta$ polarization on the left/right endpoint.
Note that the polarizations of the endpoints affect the BFV differential, which is given by the edge term defined by~(\ref{Omega A},\ref{Omega B}) plus the two endpoint-edge terms defined by~\eqref{HE_contributions_to_Omega}, plus two pure endpoint terms defined by~\eqref{Omega pure corner}.

One can also assign a space of states to a corner $p$ in $\alpha$- or $\beta$-polarization as follows:
\beq\label{H alpha}
 \mathcal{H}^\alpha_p &= \mathrm{Fun}_\mathbb{C}(\mathfrak{g}[1]) =\mathbb{C}\otimes\wedge^\bullet \mathfrak{g}^*~, \quad
  &&\Omega_p^\alpha= \frac{\mathrm{i}\hbar}{2} \Big\langle [\alpha,\alpha],\frac{\partial}{\partial \alpha} \Big\rangle~, \\
 \mathcal{H}^\beta_p &= \mathrm{Fun}_\mathbb{C}(\mathfrak{g}^*)=\mathbb{C}\otimes S^\bullet \mathfrak{g} ~,\quad 
  &&\Omega_p^\beta= 0~.
\eeq

Note that, as a cochain complex, $\mathcal{H}_I^{\alpha,\mathbb{A},\alpha}$ is quasi-isomorphic to $\mathcal{H}^\alpha_p$ -- the Chevalley-Eilenberg complex of the Lie algebra $\mathfrak{g}$.
Geometrically, this corresponds to the collapse of an $\mathbb{A}$-interval with endpoints in $\alpha$-polarization into a single $\alpha$-point. 
Likewise, the cochain complex $\mathcal{H}_I^{\beta,\mathbb{B},\beta}$ is quasi-isomorphic to $\mathcal{H}^\beta_p$:
\beq\label{collapse quasi-iso}
 \mathcal{H}_I^{\alpha,\mathbb{A},\alpha} \rightsquigarrow \mathcal{H}^\alpha_p~,\qquad 
 \mathcal{H}_I^{\beta,\mathbb{B},\beta} \rightsquigarrow \mathcal{H}^\beta_p~.
\eeq

One can regard $\mathcal{H}_p^\alpha$ and $\mathcal{H}_p^\beta$ as differential graded algebras.
The algebra structure on $\mathcal{H}_p^\alpha$ is the standard supercommutative multiplication in the exterior algebra, while for $\mathcal{H}_p^\beta$ we need to deform the na\"ive commutative product in the symmetric algebra into a star-product~$*_\hbar$ -- the deformation quantization of the Kirillov-Kostant-Souriaux Poisson structure on~$\mathfrak{g}^*$, as we explain below.

One can regard the space of states for the interval as a bimodule over the spaces of states associated to the end-points. 
The action of the end-point algebra $\mathcal{H}_p^\xi$ on the space of states $\mathcal{H}_I^{\xi,\mathbb{P},\xi'}$ for the edge is via multiplication in the algebra, e.g. $\psi(\alpha)\otimes \Psi(\alpha,\mathbb{P},\xi') \mapsto \psi(\alpha) \Psi(\alpha,\mathbb{P},\xi')$, $\psi(\beta)\otimes \Psi(\beta,\mathbb{P},\xi') \mapsto \psi(\beta) *_\hbar\Psi(\beta,\mathbb{P},\xi')$.
The reason we need to deform the product in $\mathcal{H}^\beta$ from the commutative one is that we want the edge to give a \emph{differential graded} bimodule over the corner spaces. 
In particular, the module structure map $\mathcal{H}_p^\beta\otimes \mathcal{H}_I^{\beta,\mathbb{A},\xi'} \rightarrow  \mathcal{H}_I^{\beta,\mathbb{A},\xi'}$ should be a chain map with respect to the differential $\Omega^{\beta \mathbb{A}}_p+\Omega^{\mathbb{A}}_I + \Omega^{\mathbb{A}\xi'}_{p'}$. 
This requirement is incompatible with the commutative product on $\mathcal{H}^\beta_p$ and forces the following associative non-commutative deformation $*_\hbar: \mathcal{H}^\beta\otimes \mathcal{H}^\beta\rightarrow \mathcal{H}^\beta$.

\begin{proposition}
 The associative product strucure $*_\hbar$ on $\mathcal{H}^\beta$ is fixed uniqely by the two  properties: 
 \begin{enumerate}[(i)]
  \item The module structure map $m:\mathcal{H}_p^\beta\otimes \mathcal{H}_I^{\beta,\mathbb{A},\xi'} \rightarrow  \mathcal{H}_I^{\beta,\mathbb{A},\xi'}$, obtained by extending $*_\hbar$ by linearity in the second factor, is a chain map.
  \item \label{prop star-product (ii)} $*_\hbar$ is unital with $\psi(\beta)=1$ the unit.
 \end{enumerate}
 The product $*_\hbar$ is explicitly described as follows:
 \begin{equation}\label{BCH star product}
  \mathrm{e}^{-\frac{\mathrm{i}}{\hbar}\langle \beta, x \rangle} *_\hbar \mathrm{e}^{-\frac{\mathrm{i}}{\hbar}\langle \beta, y \rangle} 
   =  \mathrm{e}^{-\frac{\mathrm{i}}{\hbar}\langle \beta, \mathrm{BCH}(x,y) \rangle}
 \end{equation}
 Here $x,y\in \mathfrak{g}$ are arbitrary parameters in the Lie algebra and $\mathrm{BCH}(x,y)=\log(\mathrm{e}^x \mathrm{e}^y)$ is the Baker-Campbell-Hausdorff group law.
\end{proposition}

\begin{proof}
 Let us check that the star-product~\eqref{BCH star product} does indeed make the module structure map a chain map. 
 Note that, for $\psi(\beta)=\mathrm{e}^{-\frac{\mathrm{i}}{\hbar}\langle \beta, x \rangle}$, the action of $\Omega^{\beta\mathbb{A}}_p$ on $\psi$ can be written as 
 \beq
  \Omega^{\beta\mathbb{A}}_p\psi =  \langle \beta, \mathsf{F}_+(\mathrm{ad}_x)\mathbb{A}_p \rangle\,\psi 
   = \mathrm{i}\hbar\frac{\mathrm{d}}{\mathrm{d}\epsilon}\Big|_{\epsilon=0} 
    \mathrm{e}^{-\frac{\mathrm{i}}{\hbar}\langle \beta,\mathrm{BCH}(x,\epsilon\, \mathbb{A}_p) \rangle}
     =\mathrm{i}\hbar\frac{\mathrm{d}}{\mathrm{d}\epsilon}\Big|_{\epsilon=0} \psi *_\hbar 
      \mathrm{e}^{-\frac{\mathrm{i}}{\hbar}\langle \beta, \epsilon\, \mathbb{A}_p \rangle}~,
 \eeq
 with $\epsilon$ an odd, ghost degree~$-1$ parameter and $*_\hbar$ defined by~\eqref{BCH star product}. 
 Here we have used the identity $\mathrm{BCH}(x,y)=x+\mathsf{F}_+(\mathrm{ad}_x)y+\mathcal{O}(y^2)$. 
 This implies that for any $\Psi\in \mathcal{H}^{\beta,\mathbb{A},\xi'}_I$ we have:
 \beq
  \Omega^{\beta\mathbb{A}}_p\Psi = \mathrm{i}\hbar\frac{\mathrm{d}}{\mathrm{d}\epsilon}\Big|_{\epsilon=0} \Psi *_\hbar 
   \mathrm{e}^{-\frac{\mathrm{i}}{\hbar}\langle \beta, \epsilon\, \mathbb{A}_p \rangle} ~.
 \eeq
 Therefore, for any $\widetilde{\psi}\in \mathcal{H}^\beta_p$ we have:
 \beq
  m\circ (\mathrm{id}\otimes\Omega^{\beta\mathbb{A}})(\widetilde\psi\otimes \Psi) 
   = \mathrm{i}\hbar\frac{\mathrm{d}}{\mathrm{d}\epsilon}\Big|_{\epsilon=0} \widetilde{\psi}*_\hbar \big(\Psi*_\hbar 
    \mathrm{e}^{-\frac{\mathrm{i}}{\hbar}\langle \beta, \epsilon\, \mathbb{A}_p \rangle}\big) \\
  = \mathrm{i}\hbar\frac{\mathrm{d}}{\mathrm{d}\epsilon}\Big|_{\epsilon=0} \big(\widetilde{\psi}*_\hbar \Psi \big) *_\hbar 
   \mathrm{e}^{-\frac{\mathrm{i}}{\hbar}\langle \beta, \epsilon\, \mathbb{A}_p \rangle} 
    = \Omega^{\beta\mathbb{A}}\circ m (\widetilde{\psi}\otimes \Psi) ~.
 \eeq
 Here we used the associativity of the star-product~\eqref{BCH star product}. 
 Note that the other pieces of the differential, $\Omega^\mathbb{A}_I$ and $\Omega^{\xi'}_{p'}$, clearly commute with the module structure map $m$. 
 Thus we have proven that $m$, defined by~\eqref{BCH star product} and extended by $\mathrm{Fun}(\mathbb{A},\xi')$-linearity in the second factor, is indeed a chain map. 
 
 Moreover, assume that $\bullet$ is some unital associative product on $\mathcal{H}^\beta$ with $\psi(\beta)=1$ the unit. 
 Then the argument above shows that the module structure map $m$ defined using $\bullet$ is a chain map if and only if $\psi_1\bullet (\psi_2 *_\hbar \psi_3)=(\psi_1\bullet \psi_2)*_\hbar \psi_3$ for any $\psi_{1,2,3}\in \mathcal{H}^\beta$. 
 Choosing $\psi_2=1$, we obtain $\psi_1 \bullet \psi_3=\psi_1 *_\hbar \psi_3$. 
 This proves uniqueness of the star-product~\eqref{BCH star product}.
\end{proof}

Gluing two intervals over a point corresponds to taking the tensor product of the spaces of states for the intervals over the algebra associated to the point:%
\footnote{
 Note that, in dg setting, when taking the tensor product $M_1\otimes_\mathcal{A} M_2$ of a right $\mathcal{A}$-module $M_1$ and a left  $\mathcal{A}$-module $M_2$ over a dg algebra $\mathcal{A}$, the total differential is the sum of the differentials on $M_1$, $M_2$ minus the differential on $\mathcal{A}$.
}
\begin{equation}
 \begin{tikzpicture}[scale=.4, baseline=(x.base)]
  \coordinate (x) at (0,-0.3);
  
  \draw[decoration={markings, mark=at position 0.5 with {\arrow{>}}},
        postaction={decorate},black, thick] (-2,0) to node[above]{\small$\mathbb{P}_1$} node[below]{\small$I_1$} (2,0);

  \filldraw[black] (-2,0) node[above]{\footnotesize$\xi'$} circle (.15);
  \filldraw[black] (2,0) node[above]{\footnotesize$\xi$} node[below]{\footnotesize$p$} circle (.15);
  
  \draw[decoration={markings, mark=at position 0.5 with {\arrow{>}}},
        postaction={decorate},black, thick] (2,0) to node[above]{\small$\mathbb{P}_2$} node[below]{\small$I_2$} (6,0);

  \filldraw[black] (6,0) node[above]{\footnotesize$\xi''$} circle (.15);
 \end{tikzpicture}
 ~ \longrightarrow \quad
 \mathcal{H}= \mathcal{H}_{I_1}^{\xi',\mathbb{P}_1,\xi} \otimes_{\mathcal{H}_p^\xi}
 \mathcal{H}_{I_2}^{\xi,\mathbb{P}_2,\xi''} ~.
\end{equation}

The space of states~\eqref{H for stratified circle} for the stratified circle can then be written, in terms of the spaces of states for intervals and corners introduced above, as:
\begin{equation}\label{H on stratified circle as a tensor product}
 \mathcal{H}= \Big( \mathcal{H}_{I_1}^{\xi_n,\mathbb{P}_1,\xi_1} \otimes_{\mathcal{H}_{p_1}^{\xi_1}} \mathcal{H}_{I_2}^{\xi_1,\mathbb{P}_2,\xi_2}  \cdots 
  \otimes_{\mathcal{H}_{p_{n-1}}^{\xi_{n-1}}} \mathcal{H}_{I_n}^{\xi_{n-1},\mathbb{P}_n,\xi_n} \Big) 
   \otimes_{\mathcal{H}_{p_n}^{\xi_n} \otimes \big(\mathcal{H}_{p_n}^{\xi_n}\big)^{\mathrm{op}} } \mathcal{H}_{p_n}^{\xi_n} ~.
\end{equation}
Here the superscript $\mathrm{op}$ stands for the opposite algebra.

\begin{remark} 
 Let $\mathcal{I}=\cup_{k=1}^l I_k$ be the union of $l$ consecutive intervals on the stratified circle~\eqref{stratified circle II} and $\mathcal{J}=\cup_{k=l+1}^n I_k$ the union of the remaining intervals, and let $p=p_0$, $q=p_l$ be the points separating $\mathcal{I}$ and $\mathcal{J}$.
 The globalized partition function $Z$ for the disk $D$ filling the stratified circle is, by the mQME, an $\Omega$-closed element of the space of states 
 \begin{equation}
  \mathcal{H}_{S^1}=\mathcal{H}_{\mathcal{I}}\otimes_{\mathcal{H}_q}\mathcal{H}_{\mathcal{J}}\otimes_{\mathcal{H}_{p}\otimes \mathcal{H}_{p}^\mathrm{op}} \mathcal{H}_{p}\cong 
  \mathrm{Hom}_{ (\mathcal{H}_q , \mathcal{H}_p)-\mathrm{bimod} }(\mathcal{H}_{\bar{\mathcal{I}}},\mathcal{H}_{\mathcal{J}})~.
 \end{equation}
 Here on the right hand side we have the space of morphisms of dg bimodules over $\mathcal{H}_q$ on the left and $\mathcal{H}_p$ on the right; bar on $\bar{\mathcal{I}}$ stands for orientation reversal. 
 Thus, the partition function for a disk can be seen as a bimodule morphism between two bimodules associated to the two arcs constituting the boundary. 
 
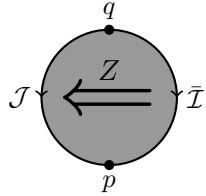
\[
  \begin{tikzpicture}[scale=.3]
   \filldraw[black!40] (0,0) circle (3);
   \draw[decoration={markings,mark=at position 1 with
    {\arrow[scale=4]{>}}},postaction={decorate}] (-2,0) to (-2.05,0);
   \draw[black, line width={1.3pt}, line cap=round] (-1.55,-.3) to (1.8,-.3);
   \draw[black, line width={1.3pt}, line cap=round] (-1.55,.3) to (1.8,.3);
   \node at (0,1.2) {$Z$};
   
   \draw[decoration={markings, mark=at position 0.5 with {\arrow{>}}}, postaction={decorate},black, thick] 
    (0,3) arc (90:270:3) node[midway, left]{\small$\mathcal{J}$};
    
   \draw[decoration={markings, mark=at position 0.5 with {\arrow{>}}}, postaction={decorate},black, thick] 
    (0,3) arc (90:-90:3) node[midway, right]{\small$\bar{\mathcal{I}}$};
    
   \filldraw (0,3) circle (.2) node[above]{\small$q$};
   \filldraw (0,-3) circle (.2) node[below]{\small$p$};
  \end{tikzpicture}
 \]
 Note that the picture for 2D~Yang-Mills we just described, mapping points to algebras, intervals to bimodules and disks to morphisms of bimodules, is in agreement with Baez-Dolan-Lurie setting of extended topological quantum field theory~\cite{lurie:TFT_classification, dolan:higher_dim_alg}, with the correction that our spaces of states depend on the choice of polarization and partition functions depend on the area of the surface (and pre-globalization partition functions additionally depend on residual fields).
\end{remark}

\subsubsection*{Towards quantization of codimension 2 corners in more general BV-BFV theories}

The algebra $\mathcal{H}^\beta,  *_\hbar$ is isomorphic to $U_{\hbar}(\mathfrak{g})$ -- the enveloping algebra of $\mathfrak{g}$ with the normalized Lie bracket~$\mathrm{i}\hbar [-,-]$.
This algebra arises as Kontsevich's deformation quantization~\cite{konstevich:def_quant, cattaneo:kontsevich_quant.} of the algebra of functions on $\mathfrak{g}^*$ equipped with the Kirillov-Kostant-Souriaux linear Poisson structure. 
This observation fits well into the following expected picture of quantization of corners of codimension~$2$. 

In a general gauge theory, a codimension~$2$ stratum $\gamma$ is classically associated a BFV ``corner phase space''~\cite{CMR:classical_BV} $\Phi_\gamma$ equipped with a degree~$+1$ symplectic form $\omega_\gamma$ and  a BFV~charge~$S_\gamma$ of degree~$+2$. 
On the level of quantization, we impose a polarization $\Phi_\gamma\simeq T^*[1]\mathcal{B}_\gamma$. 
The BFV~charge $S_\gamma$ generates a $P_\infty$ (Poisson up-to-homotopy) algebra structure on $C^\infty(\mathcal{B}_\gamma)$, coming from interpreting $S_\gamma$ as a self-commuting polyvector $\Pi$ on $\mathcal{B}_\gamma$.%
\footnote{
 Equivalently, the $P_\infty$ structure arises from $S_\gamma$ via the derived bracket construction, $\{\psi_1,\ldots,\psi_n\}_\Pi:=(\cdots (S_\gamma,\psi_1),\cdots,\psi_n)$ with $\psi_1,\ldots,\psi_n\in C^\infty(\mathcal{B}_\gamma)$. 
 The brackets $(-,-)$ on the r.h.s. are the Poisson brackets on functions on the phase space $\Phi_\gamma$ defined by $\omega_\gamma$.
} 
Then the quantum space of states $\mathcal{H}_\gamma$ is expected to be the $A_\infty$ algebra obtained as Kontsevich's deformation quantization of the $P_\infty$ algebra $C^\infty(\mathcal{B}_\gamma)$. 
In particular, the $A_\infty$ structure maps arise from Feynman diagrams on a thickening of $\gamma$ to $\gamma\times D$, with $D$ a $2$-disk, for a field theory coming from the AKSZ construction on the mapping space $\mathrm{Map}(T[1]D,\Phi_\gamma)$. 
We plan to revisit this construction in more detail in a future paper on corners on BV-BFV formalism.

Note that, in the case of 2D~Yang-Mills theory, the corner phase space is $\Phi_p=\mathfrak{g}[1]\oplus \mathfrak{g}^*$, with $\omega_p=\langle \delta\beta, \delta\alpha \rangle$, $S_p=\frac12 \langle \beta, [\alpha,\alpha] \rangle$. 
Deformation quantization of $C^\infty(\mathfrak{g}^*)$ with Poisson bivector $\Pi= \frac12 \langle \beta , [\frac{\partial}{\partial \beta}\stackrel{\wedge}{,} \frac{\partial}{\partial \beta}] \rangle$ yields the algebra~$\mathcal{H}^\beta$. 
Taking the opposite polarization, one gets the deformation quantization of $C^\infty(\mathfrak{g}[1])=\wedge \mathfrak{g}^*$ with $1$-vector $\Pi=\frac12 \langle [\alpha,\alpha], \frac{\partial}{\partial \alpha} \rangle$, which is the dg algebra~$\mathcal{H}^\alpha$.

In general, one expects all the structure maps on (and between) the spaces of states associated to various strata to come from Feynman diagrams. 

A related picture was obtained in~\cite{CF:coisotropic} in the context of Poisson sigma model on a disk with intervals on the boundary decorated with coisotropic submanifolds~$C_i$ of the Poisson target~$M$.
In this setting the quantization yields algebras assigned to intervals (deformation quantization of the rings of functions on~$C_i$) and bimodules assigned to the corners separating the intervals. 
In particular, the algebra $\mathcal{H}^\beta$ arises in this context as a quantization of the space-filling coisotropic in~$M=\mathfrak{g}^*$. 
This picture can be thought of as Poincaré dual to our picture on the boundary of a disk.

\subsubsection{Gluing regions along an interval and the Fourier transform property of BFV differentials}

Recall that the BFV differentials for an $\mathbb{A}$-circle and a $\mathbb{B}$-circle are related by Fourier transform. 
This property in particular implies that mQME is compatible with gluing: if $\Sigma=\Sigma_1\cup_{S^1}\Sigma_2$ a union of surfaces over a circle and if the partition functions $Z_{\Sigma_1}$, $Z_{\Sigma_2}$ are known to satisfy mQME, then the glued partition function $Z_\Sigma = \langle Z_{\Sigma_1},Z_{\Sigma_2} \rangle_{\mathcal{H}_{S^1}}$ automatically satisfies mQME on the glued surface. 

One has an analogous property in the setting with corners. 
Consider e.g. an $\mathbb{A}$-interval~$I$ parameterized by $t\in[0,1]$ with endpoints in polarizations $\xi,\xi'$ and consider a $\mathbb{B}$-interval $\widetilde{I}$ parameterized by $\widetilde{t}\in [0,1]$, with endpoints in $\xi',\xi$. 
Let $r\colon I\rightarrow \widetilde{I}$ be an orientation-reversing diffeomorphism $t\mapsto \widetilde{t}=1-t$. 
Gluing along $r$ corresponds to the following pairing of states on $I$ and $\widetilde{I}$:
\beq
 \langle -,- \rangle_I \colon \quad  \mathcal{H}^{\xi, \mathbb{A}, \xi'}_I &\otimes \mathcal{H}^{\xi', \mathbb{B}, \xi}_{\widetilde{I}}  
  &&\longrightarrow   &&\mathbb{C} \\ 
 \psi_1 &\otimes \psi_2  &&\mapsto  &&\int \mathcal{D} \mathbb{A}\, \mathcal{D} \mathbb{B}\; \psi_1(\xi, \mathbb{A},\xi')
  \cdot \mathrm{e}^{-\frac{\mathrm{i}}{\hbar}\int_I \langle r^*\mathbb{B},\mathbb{A} \rangle} \cdot \psi_2(\xi', \mathbb{B},\xi)
\eeq
One easily verifies the following:
\begin{equation}
 \langle (\Omega_\mathrm{in}^{\xi \mathbb{A}}+\Omega_I^{\mathbb{A}}+ \Omega_\mathrm{out}^{\mathbb{A}\xi'})\,\psi_1 , \psi_2 \rangle_I 
  = -\langle \psi_1 ,(\Omega_\mathrm{in}^{\xi' \mathbb{B}}+\Omega_{\widetilde{I}}^{\mathbb{B}}+ \Omega_\mathrm{out}^{\mathbb{B}\xi})\, \psi_2 \rangle_I~.
\end{equation}
Here we are making the Assumption~\ref{assump: admissible states} on states $\psi_1,\psi_2$. 
In other words, the operators $\Omega_\mathrm{in}^{\xi \mathbb{A}}+\Omega_I^{\mathbb{A}}+ \Omega_\mathrm{out}^{\mathbb{A}\xi'}$ and $\Omega_\mathrm{in}^{\xi' \mathbb{B}}+\Omega_{\widetilde{I}}^{\mathbb{B}}+ \Omega_\mathrm{out}^{\mathbb{B}\xi}$ are, up to sign, the Fourier transform of each other (when acting on admissible states).

This immediately implies the following. 
Assume that $\Sigma$ is a result of gluing of surfaces~$\Sigma_1$ and $\Sigma_2$ via attaching an interval $I\subset \partial \Sigma_1$ to $\widetilde{I}\subset \partial \Sigma_2$ along the diffeomorphism $r$. Then for $\Psi_1\in \mathcal{H}_{\partial \Sigma_1}$, $\Psi_2\in \mathcal{H}_{\partial \Sigma_2}$ any two states on the boundary of~$\Sigma_1$, $\Sigma_2$, we have
\begin{equation}
 \Omega_{\partial \Sigma} \langle\Psi_1,\Psi_2 \rangle_I =  \langle \Omega_{\partial \Sigma_1}\Psi_1,\Psi_2 \rangle_I 
  +  \langle \Psi_1,\Omega_{\partial \Sigma_2}\Psi_2 \rangle_I~,
\end{equation}
where $ \langle\Psi_1,\Psi_2 \rangle_I \in \mathcal{H}_{\partial \Sigma}$ is understood as the ``gluing'' of states $\Psi_1$, $\Psi_2$ along~$I$. 
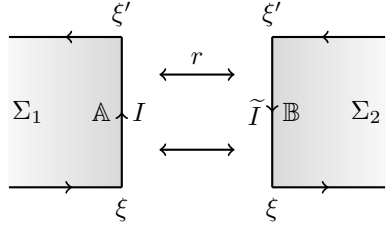
\[
 \begin{tikzpicture}[scale=.5, baseline=(x.base)]
  
  \begin{scope}
   \shade[right color=black!15, left color=black!3] (-1.5,-2) rectangle (1.5,2);
   \node at (-1,0) {\small$\Sigma_1$};
   
   \draw[decoration={markings, mark=at position 0.5 with {\arrow{>}}}, postaction={decorate},black, thick] 
    (1.5,-2) node[below]{\small$\xi$} to node[left]{\small$\mathbb{A}$} node[right]{\small$I$} (1.5,2) node[above]{\small$\xi'$};
    
   \draw[decoration={markings, mark=at position 0.5 with {\arrow{>}}}, postaction={decorate},black, thick] 
    (1.5,2) to (-1.5,2);
   \draw[decoration={markings, mark=at position 0.5 with {\arrow{<}}}, postaction={decorate},black, thick] 
    (1.5,-2) to (-1.5,-2);
  \end{scope}
  
  \begin{scope}[ shift={(7,0)}, rotate=180]
   \shade[left color=black!15, right color=black!3] (-1.5,-2) rectangle (1.5,2);
   \node at (-1,0) {\small$\Sigma_2$};
   
   \draw[decoration={markings, mark=at position 0.5 with {\arrow{>}}}, postaction={decorate},black, thick] 
    (1.5,-2) node[above]{\small$\xi'$} to node[right]{\small$\mathbb{B}$} node[left]{\small$\widetilde{I}$} (1.5,2) node[below]{\small$\xi$};
    
   \draw[decoration={markings, mark=at position 0.5 with {\arrow{>}}}, postaction={decorate},black, thick] 
    (1.5,2) to (-1.5,2);
   \draw[decoration={markings, mark=at position 0.5 with {\arrow{<}}}, postaction={decorate},black, thick] 
    (1.5,-2) to (-1.5,-2);
  \end{scope}
  
  \draw[<->, thick] (2.5,1) to node[above]{\small$r$} (4.5,1);
  \draw[<->, thick] (2.5,-1) to (4.5,-1);
 \end{tikzpicture}
\]
In particular, if the partition functions on~$\Sigma_1,\Sigma_2$ are known to satisfy the mQME, the glued partition function $Z_\Sigma=\langle Z_{\Sigma_1},Z_{\Sigma_2} \rangle_I$ automatically satisfies the mQME on~$\Sigma$.

\subsubsection{Small model for states on an $\mathbb{A}$-interval}\label{sec: small model for A-states}

In preparation for the calculations of section~\ref{BF A-disk with b corners}, we want to present a ``small model'' for the space of states on an $\mathbb{A}$-interval, corresponding to the passage to a constant $1$-form field~$\mathbb{A}^{(1)}$ on the interval.
This is an extension of the discussion of section~\ref{AA polarization on the cylinder} (and in particular, formula~\eqref{i p Z= Z+ exact term}), and of section~\ref{Omega_A_cohomology}.

Consider a single interval in $\mathbb{A}$-polarization:
\begin{equation}\label{beta_A_beta_interval}
 \begin{tikzpicture}[scale=.4, baseline=(x.base)]
  \coordinate (x) at (0,-0.3);
  
  \draw[decoration={markings, mark=at position 0.5 with {\arrow{>}}},
        postaction={decorate},black, thick] (-2,0) to node[above]{\small$\mathbb{A}$} node[below]{\small$I$} (3,0);

  \filldraw[black] (-2,0) node[above]{\footnotesize$\mathbb{A}_0$} node[below]{\footnotesize$p_0$} circle (.15);
  \filldraw[black] (3,0) node[above]{\footnotesize$\mathbb{A}_1$} node[below]{\footnotesize$p_1$} circle (.15);
  
 \end{tikzpicture}
\end{equation}
We view its endpoints as corners in picture~\ref{pict_1} (non-polarized), with $\mathbb{A}_{0}$, $\mathbb{A}_{1}$ the limiting values of the $0$-form field $\mathbb{A}$ at the endpoints $p_0$, $p_1$. 
Equivalently, we can treat the endpoint in picture~\ref{pict_2},  putting $\alpha$-polarization on them, with  corner fields $\alpha_{0,1}$ identified with $\mathbb{A}_{0,1}$.

The space of states for the interval~\eqref{beta_A_beta_interval} is a cochain complex
\begin{equation}
 \mathcal{H}=\mathrm{Fun}_\mathbb{C}\big(\Omega^\bullet(I,\mathfrak{g})[1]\big) =\big\{\Psi(\mathbb{A})\big\}
\end{equation}
with differential
\begin{equation}\label{Omega for A disk with beta corners}
 \Omega = \mathrm{i}\hbar \Big( \int \Big\langle \mathrm{d}\mathbb{A}^{(0)}
  +[\mathbb{A}^{(0)},\mathbb{A}^{(1)}]\; ,\; \frac{\delta}{\delta \mathbb{A}^{(1)}} \Big\rangle 
   + \int \Big\langle \frac12\ [\mathbb{A}^ {(0)},\mathbb{A}^{(0)}] \;,\; \frac{\delta}{\delta \mathbb{A}^{(0)}} \Big\rangle  \Big) ~.
\end{equation}

One has the following ``small'' quasi-isomorphic model for the space of states -- the cochain complex
\begin{equation}\label{H_small}
 \mathcal{H}^\mathrm{small}=\mathrm{Fun}_\mathbb{C}\Big( C^\bullet(I,\mathfrak{g})[1]\Big)
  = \Big\{ \underline\Psi\big(\underline{\mathbb{A}}_{0},\, \underline{\mathbb{A}},\,  \underline{\mathbb{A}}_{1}\big) \Big\}~.
\end{equation}
Here $C^\bullet(I,\mathfrak{g})=\mathfrak{g}\oplus \mathfrak{g}[-1]\oplus \mathfrak{g}$ is the complex of  $\mathfrak{g}$-valued cellular cochains on the interval~$I$ endowed with the standard CW complex structure, with two $0$-cells $p_{0}$, $p_1$ and a single $1$-cell~$I$. 
Variables $\underline{\mathbb{A}}_{0},\, \underline{\mathbb{A}}_{1}\in \mathfrak{g}[1]$ are the values of the cochain on the $0$-cells~$p_{0}$ and~$p_1$ (endpoints), respectively, and $\underline{\mathbb{A}}\in \mathfrak{g}$ is the value of the cochain on the $1$-cell~$I$ itself. 
The differential on~$\mathcal{H}^\mathrm{small}$ is given by:
\beq\label{Omega small}
 \Omega^{\mathrm{small}} = \mathrm{i}\hbar \Big( \Big\langle \frac12 \big[\underline{\mathbb{A}}_{0},\, \underline{\mathbb{A}}_{0}\big]\; 
  ,\; \frac{\partial}{\partial \underline{\mathbb{A}}_{0}} \Big\rangle
   + \Big\langle \frac12 \big[\underline{\mathbb{A}}_{1},\, \underline{\mathbb{A}}_{1}\big]\; 
    ,\; \frac{\partial}{\partial \underline{\mathbb{A}}_{1}} \Big\rangle +&\\
 -\Big\langle \mathsf{F}_-(\mathrm{ad}_{\underline{\mathbb{A}}})\circ \underline{\mathbb{A}}_{0} 
  + \mathsf{F}_+(\mathrm{ad}_{\underline{\mathbb{A}}})\circ \underline{\mathbb{A}}_{1}~,
   ~\frac{\partial}{\partial \underline{\mathbb{A}}} \Big\rangle \Big)~.&
\eeq
The chain projection $p_\mathcal{H} \colon \mathcal{H} \rightarrow \mathcal{H}^\mathrm{small}$ is the following map:
\beq\label{p_H on A beta circle}
 \Psi \mapsto \bigg( \underline{\Psi} \colon \{\underline{\mathbb{A}}_{0},\, \underline{\mathbb{A}},\, \underline{\mathbb{A}}_{1}\} \mapsto 
  \Psi \Big( \mathbb{A}^{(0)}= \mathsf{G}_-(t,\mathrm{ad}_{\underline{\mathbb{A}}})\circ \underline{\mathbb{A}}_{0} 
   + \mathsf{G}_+(t,\mathrm{ad}_{\underline{\mathbb{A}}})\circ \underline{\mathbb{A}}_{1} ~,
    ~ \mathbb{A}^{(1)}= \mathrm{d}t\cdot \underline{\mathbb{A}} \Big) \bigg)~.
\eeq
Here we parameterize the interval by the coordinate $t\in [0,1]$ and $\mathsf{G}_\pm$ are the generating functions for Bernoulli polynomials~\eqref{G pm}.

The chain inclusion $i_\mathcal{H}\colon \mathcal{H}^\mathrm{small}\rightarrow \mathcal{H} $ is given as follows:
\begin{equation}\label{i_H on A beta circle}
 i_\mathcal{H} \colon \underline{\Psi} \mapsto \bigg( \Psi:\;\; \mathbb{A} \mapsto \underline{\Psi}
  \Big( \ \underline{\mathbb{A}}_{0}= \mathbb{A}_0~,~\underline{\mathbb{A}}_{1}= \mathbb{A}_{1}~ , ~ \underline{\mathbb{A}} =\log U(\mathbb{A}) \Big)\bigg)
\end{equation}
Here the group element $U(\cdots)\in G$ is the holonomy of the connection $1$-form along the interval~$I$.

\begin{remark}  
The space of states $\mathcal{H},\Omega$ is the Chevalley-Eilenberg complex (or the dual of the bar complex) of the differential graded Lie algebra of $\mathfrak{g}$-valued differential forms on the interval, $\Omega^\bullet(I,\mathfrak{g}),\mathrm{d},[-,-]$. 
Likewise, $\mathcal{H}^\mathrm{small},\Omega^\mathrm{small}$ is the Chevalley-Eilenberg complex for the $L_\infty$~algebra structure on $\mathfrak{g}$-valued cellular cochains on an interval, constructed in~\cite[see also~\cite{lawrence-sullivan:deformation, cheng:transferring_homotopy}]{mnev:simplicial_BF, mnev:discrete_BF}. 
This $L_\infty$ algebra arises as the homotopy transfer of the ``big'' algebra $\Omega^\bullet(I,\mathfrak{g})$ onto the deformation retract $C^\bullet(I,\mathfrak{g})$ -- cochains, realized as Whitney forms on the interval. 
Chain map~\eqref{p_H on A beta circle} corresponds to the $L_\infty$~morphism from $C^\bullet(I,\mathfrak{g})$ to $\Omega^\bullet(I,\mathfrak{g})$ constructed explicitly in~\cite[-- Statement 14]{mnev:simplicial_BF, mnev:discrete_BF}; it is a non-abelian deformation of the inclusion of cochains as Whitney forms. 
The map~\eqref{i_H on A beta circle}, constructed via holonomies, corresponds to the $L_\infty$ morphism from forms to cochains - the non-abelian version of the integration-over-cells map, cf.~\cite{bandiera:discretize}.
We give a proof of the chain map property of\eqref{i_H on A beta circle} in Appendix~\ref{app:chain}.
\end{remark}
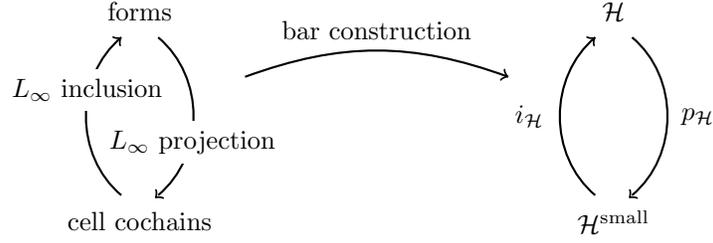
\[
  \begin{tikzpicture}[scale=.35]
  
   \begin{scope}[shift={(18,0)}]
    \draw[decoration={markings, mark=at position .9999 with {\arrow{>}}},
         postaction={decorate},color=black, thick] (-.7,-3) to[out=140, in=-140] (-.7,3);
    \draw[decoration={markings, mark=at position 0.015 with {\arrow{<}}},
         postaction={decorate},color=black, thick] (.7,-3) to[out=40, in=-40] (.7,3);
   
    \node at (0,4) {\small $\mathcal{H}$};
    \node at (0,-4) {\small $\mathcal{H}^{\mathrm{small}}$};
    \node at (-3.2,0) {\small $i_\mathcal{H}$};
    \node at (3.2,0) {\small $p_\mathcal{H}$};
   \end{scope}
            
   \draw[decoration={markings, mark=at position .9999 with {\arrow{>}}},
        postaction={decorate},color=black, thick] (-.7,-3) to[out=140, in=-140] (-.7,3);
   \draw[decoration={markings, mark=at position 0.015 with {\arrow{<}}},
        postaction={decorate},color=black, thick] (.7,-3) to[out=40, in=-40] (.7,3);
  
   \node at (0,4) {\small forms};
   \node at (0,-4) {\small cell cochains};
   \node[fill=white] at (-2,1) {\small $L_\infty$ inclusion};
   \node[fill=white] at (2,-1) {\small $L_\infty$ projection};
   
   \draw[decoration={markings, mark=at position .9999 with {\arrow{>}}},
        postaction={decorate},color=black, thick] (4,1.5) to[out=20, in=160] node[above]{\small bar construction} (14,1.5);
   
  \end{tikzpicture}  
 \]

One has similar small models for the space of states on the $\mathbb{A}$-interval with endpoints in any combination of polarizations $\xi_0,\xi_1$. 
E.g. for both endpoints in $\beta$-polarization, we have the small model~(\ref{H_small},\ref{Omega small}) and the maps~(\ref{p_H on A beta circle},\ref{i_H on A beta circle}), where we adjoin the corner variables~$\beta_0$, $\beta_1$ on which the wavefunctions~$\Psi$, $\underline{\Psi}$ are allowed to depend, and we add corner-edge terms~$\Omega^{\beta\mathbb{A}}$,~$\Omega^{\mathbb{A}\beta}$~\eqref{HE_contributions_to_Omega} to the respective differentials  $\Omega$ and $\Omega^\mathrm{small}$.

Finally, consider a surface~$\Sigma$ with stratified boundary circles decorated with an arbitrary combination of polarizations of arcs and corners.
By the discussion above, we have a small quasi-isomorphic model $\mathcal{H}'$ for the space of states $\mathcal{H}$ corresponding to replacing the states on some (or all) $\mathbb{A}$-arcs with respective small models for $\mathbb{A}$-arcs in the formula (\ref{H on stratified circle as a tensor product}), and we have chain maps  $p_\mathcal{H}: \mathcal{H}\rightarrow \mathcal{H}'$, $i_\mathcal{H}: \mathcal{H}' \rightarrow \mathcal{H}$. 
They correspond to a quasi-isomorphism of complexes and thus there exists a chain homotopy $K_\mathcal{H}:\mathcal{H}\rightarrow  \mathcal{H}$ between the identity and the projection $i_\mathcal{H}\circ p_\mathcal{H}$. Therefore, we can apply the argument~\eqref{i p Z= Z+ exact term} to the partition function~$Z$ on~$\Sigma$:
\begin{equation}\label{reduction to const fields on an A-interval}
 i_\mathcal{H} \circ p_\mathcal{H}\  Z = Z+ (\Omega+\hbar^2 \Delta)(\cdots)~.
\end{equation}
In particular, one can recover~$Z$, modulo BV exact terms, by evaluating it on constant $1$-forms $\mathrm{d}t\cdot \underline{\mathbb{A}}_k$ on the boundary arcs, provided that their holonomy coincides with the holonomy of the original $\mathbb{A}^{(1)}$ field along the respective intervals, i.e. $\underline{\mathbb{A}}_k=\log U_{I_k}(\mathbb{A})$.

\subsection{BF $\mathbb{B}$-disk with two $\alpha$ corners} \label{Sect:B_disk_corners}

Let us consider now the case of a BF disk with the boundary split into two arcs $\gamma_i\colon [0,1]\rightarrow S^1$\,, $i=1,2$\,, with $\gamma_i(0)=v_0$ and $\gamma_i(1)=v_1$\,, both in $\mathbb{B}$ polarization.
On both vertices of the arcs we fix the value $\alpha_i$ for the restriction of the bulk $\mathsf{A}$ fields.
Expanding the vertices~$v_i$ into two edges in $\mathbb{A}$ polarization we can think of this disk as a square (figure~\ref{BB_disk}).

\begin{figure}[h]
   \[
   \begin{tikzpicture}[scale=.25, baseline=(x.base)]
  
    \draw[fill=black!15, thick] (0,0) circle (3);   

    \node (x) at (0,0) {};
    
    \node (a1) at (180:5) {${\mathbb{B}}$};
    \node (a2) at (3:5) {$\widetilde{\mathbb{B}}$};
    
    \draw[thick] (90:2.5) to (90:3.5) node[above] {$\alpha$};
    \draw[thick] (270:2.5) to (270:3.5) node[below] {$\widetilde{\alpha}$};

   \end{tikzpicture}
   \qquad\longleftrightarrow\qquad
   \begin{tikzpicture}[scale=.25, baseline=(x.base)]
  
    \filldraw[fill=black!15, thick, draw=black] (-3,-3) to (-3,3) to (3,3) to (3,-3) to (-3,-3);
    
    \draw[decoration={markings, mark=at position 0.6 with {\arrow{>}}},
      postaction={decorate},color=black, thick] (-3,-3) to node[above]{$t$} (3,-3);
    \draw[decoration={markings, mark=at position 0.6 with {\arrow{>}}},
      postaction={decorate},color=black, thick] (-3,-3) to node[right]{$\tau$} (-3,3);
      
    \draw[color=black, thick, dashed] (-3,-1) to (3,-1) node[above left]{$\mathsf{b}$};
    \draw[color=black, thick, dashed] (-1,-3) to (-1,3) node[below right]{$\mathsf{a}$};

    \node (x) at (0,0) {};
    
    \node (a1) at (45:5.5) {$\alpha$};
    \node (a1) at (225:5.5) {$\widetilde{\alpha}$};
    \node (a2) at (-45:5.5) {$\widetilde{\alpha}$};
    \node (b2) at (135:5.5) {$\alpha$};
    
    \node (a1) at (270:5) {$\widetilde{\mathbb{A}}=\widetilde{\alpha}$};
    \node (a1) at (90:5) {${\mathbb{A}}=\alpha$};
    \node (a2) at (177:5) {$\mathbb{B}$};
    \node (b2) at (0:5) {$\widetilde{\mathbb{B}}$};
    
   \end{tikzpicture}
   \]
   \caption{$\mathbb{B}$ disk with the boundary split into two arcs separated by points in A-polarization.}
   \label{BB_disk}
\end{figure}
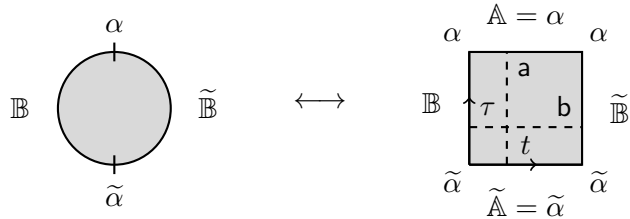

The square can be viewed as the product of two intervals, with $\mathbb{A}$ or $\mathbb{B}$ polarization on both endpoints respectively.
The zero-modes now contain 1-form components for the $\mathsf{a}$ and~$\mathsf{b}$ fields: $\mathsf{a}=\mathsf{a}_1 \mathrm{d}\tau $\,, $\mathsf{b}=\mathsf{b}^1 \mathrm{d}t$\,.
A possible choice for axial-gauge propagator is (cf.~appendix~\eqref{One-Dimensional Propagators}):
\beq\label{ABAB_square_prop}
\begin{aligned}
 &\eta(t,\tau;t'\tau') = (\Theta(\tau-\tau')-\tau) \mathrm{d}t' - (\Theta(t'-t)-t')\delta(\tau'-\tau)(\mathrm{d}\tau'-\mathrm{d}\tau)	 ~.
\end{aligned}
\eeq

The contributing Feynman diagrams to the effective action are wheels with~$n$ $\mathsf{a}$~zero-modes and trees, rooted either on the $\mathbb{B}_{(1)}$ boundary field or on the $\mathsf{b}$~zero-mode and ending on one $\mathbb{A}_{(0)}$~boundary field, with no bifurcations and the insertion of $n$~leafs decorated with $\mathsf{a}$~zero-modes~(Figure~\ref{BB_disk_diagrams}).

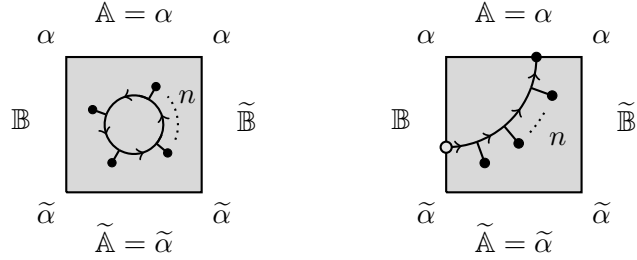
\begin{figure}[h]
   \[
   \begin{tikzpicture}[scale=.3, baseline=(x.base)]
  
    \filldraw[fill=black!15, thick, draw=black] (-3,-3) to (-3,3) to (3,3) to (3,-3) to (-3,-3);
    
    \node (x) at (0,0) {};
    
    \node (a1) at (45:5.5) {$\alpha$};
    \node (a1) at (225:5.5) {$\widetilde{\alpha}$};
    \node (a2) at (-45:5.5) {$\widetilde{\alpha}$};
    \node (b2) at (135:5.5) {$\alpha$};
    
    \node (a1) at (270:5) {$\widetilde{\mathbb{A}}=\widetilde{\alpha}$};
    \node (a1) at (90:5) {${\mathbb{A}}=\alpha$};
    \node (a2) at (177:5) {$\mathbb{B}$};
    \node (b2) at (5,.35) {$\widetilde{\mathbb{B}}$};
    
    \begin{scope}[scale=.65]
     \draw[decoration={markings, mark=at position 0.5 with {\arrow{>}}},
        postaction={decorate},color=black, thick] (-40:2) arc (-40:80:2);
        
    \draw[decoration={markings, mark=at position 0.5 with {\arrow{>}}},
        postaction={decorate},color=black, thick] (60:2) arc (60:160:2);
        
    \draw[decoration={markings, mark=at position 0.5 with {\arrow{>}}},
        postaction={decorate},color=black, thick] (160:2) arc (160:240:2);
        
    \draw[decoration={markings, mark=at position 0.5 with {\arrow{>}}},
        postaction={decorate},color=black, thick] (240:2) arc (240:340:2);
     
    \draw[color=black, thick] (-40:2) to (-40:3);
    \draw[color=black, fill=black] (-40:3) circle (8pt);
    
    \draw[color=black, thick] (60:2) to (60:3);
    \draw[color=black, fill=black] (60:3) circle (8pt);
    
    \draw[color=black, thick] (160:2) to (160:3);
    \draw[color=black, fill=black] (160:3) circle (8pt);
    
    \draw[color=black, thick] (240:2) to (240:3);
    \draw[color=black, fill=black] (240:3) circle (8pt);
    
    \draw[color=black, thick, dotted] (-20:3) arc   (-20:40:3)node[right]{$n$};
    \end{scope} 
    
   \end{tikzpicture}
   \qquad\qquad
   \begin{tikzpicture}[scale=.3, baseline=(x.base)]
  
    \filldraw[fill=black!15, thick, draw=black] (-3,-3) to (-3,3) to (3,3) to (3,-3) to (-3,-3);
    
    \node (x) at (0,0) {};
    
    \node (a1) at (45:5.5) {$\alpha$};
    \node (a1) at (225:5.5) {$\widetilde{\alpha}$};
    \node (a2) at (-45:5.5) {$\widetilde{\alpha}$};
    \node (b2) at (135:5.5) {$\alpha$};
    
    \node (a1) at (270:5) {$\widetilde{\mathbb{A}}=\widetilde{\alpha}$};
    \node (a1) at (90:5) {${\mathbb{A}}=\alpha$};
    \node (a2) at (177:5) {$\mathbb{B}$};
    \node (b2) at (5,.35) {$\widetilde{\mathbb{B}}$};

    \draw[decoration={markings, mark=at position 0.5 with {\arrow{>}}},
        postaction={decorate},color=black, thick] (-3,-1) arc (-90:-70:4);
    \draw[black, thick] ($(-3,3)+(-70:4)$) to ($(-3,3)+(-70:5)$);
    \draw[color=black, fill=black] ($(-3,3)+(-70:5)$) circle (6pt);

    \draw[decoration={markings, mark=at position 0.5 with {\arrow{>}}},
        postaction={decorate},color=black, thick] ($(-3,3)+(-70:4)$) arc (-70:-50:4);
    \draw[black, thick] ($(-3,3)+(-50:4)$) to ($(-3,3)+(-50:5)$);   
    \draw[color=black, fill=black] ($(-3,3)+(-50:5)$) circle (6pt);

    \draw[decoration={markings, mark=at position 0.5 with {\arrow{>}}},
        postaction={decorate},color=black, thick] ($(-3,3)+(-50:4)$) arc (-50:-20:4);
    \draw[black, thick] ($(-3,3)+(-20:4)$) to ($(-3,3)+(-20:5)$);   
    \draw[color=black, fill=black] ($(-3,3)+(-20:5)$) circle (6pt);

    \draw[decoration={markings, mark=at position 0.5 with {\arrow{>}}},
        postaction={decorate},color=black, thick] ($(-3,3)+(-20:4)$) arc (-20:0:4);
    
    \draw[color=black, thick, fill=black!10] (-3,-1) circle (7pt);
    \draw[color=black, fill=black] (1,3) circle (6pt);
    
    \draw[color=black, thick, dotted] ($(-3,3)+(-42:5)$) arc   (-42:-28:5)node[midway, below right]{$n$};
    
   \end{tikzpicture}
   \]
   \caption{Examples of the tree and 1-loop Feynman diagrams contributing to the effective action for the BF disk in $\mathbb{B}$ polarization with two $\alpha$ corners.}
   \label{BB_disk_diagrams}
\end{figure}

\begin{proposition}\label{BB_disk_part.funct.}
 The partition function for the BF~disk in $\mathbb{B}$~polarization with two $\alpha$~corners is:
 \beq\label{B_disk_2_corners}
  Z = \exp \Big( &\frac{\mathrm{i}}{\hbar} 
   \int_{\gamma} \langle \mathbb{B}, \mathsf{a} + \mathsf{G}_+(\tau, \mathrm{ad}_{\mathsf{a}_1}) \alpha 
    + \mathsf{G}_-(\tau, \mathrm{ad}_{\mathsf{a}_1})  \widetilde{\alpha} \rangle +\\ 
  -&\frac{\mathrm{i}}{\hbar}  \int_{\widetilde{\gamma}} \langle \widetilde{\mathbb{B}}, \mathsf{a} + \mathsf{G}_+(\tau, \mathrm{ad}_{\mathsf{a}_1}) \alpha 
   + \mathsf{G}_-(\tau, \mathrm{ad}_{\mathsf{a}_1}) \widetilde{\alpha} \rangle +\\
  +&\frac{\mathrm{i}}{\hbar}  \langle \mathsf{b}^1,  \mathsf{F}_+(\mathrm{ad}_{\mathsf{a}_1}) \alpha 
   + \mathsf{F}_-(\mathrm{ad}_{\mathsf{a}_1}) \widetilde{\alpha} \rangle \Big) ~
    \mathrm{det}\bigg(\frac{\sinh\big(\mathrm{ad}_{\mathsf{a}_1} /2\big)}{\mathrm{ad}_{\mathsf{a}_1} /2}\bigg) \cdot \rho_{\mathcal{V}} ~,
 \eeq
 where $\rho_{\mathcal{V}}= D^{\frac12} \mathsf{a} ~D^{\frac12} \mathsf{b}$ is the reference half-density on zero-modes.
\end{proposition}
\begin{proof}
 See Appendix~\ref{diagram_computations}.
\end{proof}

\subsection{BF $\mathbb{B}$-disk with one $\alpha$ corner}\label{B_disk_with_one_a_corner}

Let us consider a disk in $\mathbb{B}$-polarization with a single corner in $\alpha$-polarization.
We will denote by $\alpha$ the value of the zero-form component of the $A$-fields on the corner.
Notice that the space of zero-modes is now empty, in contrast to the $\mathbb{B}$-disk without corners or with two $\alpha$~corners.

\begin{figure}[h!]
 \centering
 \[
  \begin{tikzpicture}[scale=.25, baseline=(x.base)]
  
    \draw[fill=black!15, thick] (0,0) circle (3);   

    \node (x) at (0,0) {};
    
    \node (a1) at (180:5) {${\mathbb{B}}$};
    
    \draw[thick] (90:2.5) to (90:3.5) node[above] {$\alpha$};

   \end{tikzpicture}
   \qquad\longleftrightarrow\qquad
   \begin{tikzpicture}[scale=.25, baseline=(x.base)]
  
    \draw[fill=black!15, thick] (0,0) circle (3);   

    \node (x) at (0,0) {};
    
    \node (a1) at (45:5) {$\alpha$};
    \node (a1) at (225:5) {${\alpha}$};
    \node (a2) at (-45:5) {${\alpha}$};
    \node (b2) at (135:5) {$\alpha$};
    
    \node (a1) at (270:5) {$\widetilde{\mathbb{A}}={\alpha}$};
    \node (a1) at (90:5) {${\mathbb{A}}=\alpha$};
    \node (a2) at (177:5) {$\mathbb{B}$};
    \node[right] (b2) at (0:4) {$\mathbb{A}'=\alpha$};

    \draw[thick] (45:2.5) to (45:3.5);
    \draw[thick] (-45:2.5) to (-45:3.5);
    \draw[thick] (135:2.5) to (135:3.5);
    \draw[thick] (225:2.5) to (225:3.5);

   \end{tikzpicture}
  \]
 \caption{$\mathbb{B}$ disk with one $\alpha$ corner as the ``collapse'' of three edges in a square.}
 \label{Balpha_disk}
\end{figure}
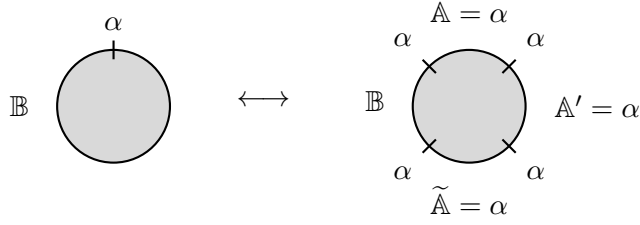

The corner can be expanded to an $\mathbb{A}$-polarized edge with $\mathbb{A}=\alpha$, which can be then split in three consecutive edges.
We then get a square, which is the product of an $\mathbb{A}$-$\mathbb{A}$ interval times an $\mathbb{A}$-$\mathbb{B}$ interval  (figure~\ref{Balpha_disk}).
We can thus choose the axial gauge propagator to compute the effective action.
If we denote with $t$ the coordinate on the $\mathbb{A}$-$\mathbb{A}$ interval and with $\tau$ the coordinate on the $\mathbb{A}$-$\mathbb{B}$ interval we have:
\beq
 \eta(t,\tau; t',\tau') = -\Theta(\tau'-\tau)\delta(t'-t)(\mathrm{d}t'-\mathrm{d}t) ~.
\eeq

Since there are no zero-modes and the boundary $\mathbb{A}$-field has only the zero-form component~$\alpha$\,, from degree counting we get that the only non-vanishing diagrams contributing to the partition function are the ones containing no interaction vertices:
\beq\label{B_disk_1_corner}
 Z[\mathbb{B},\alpha] = \mathrm{e}^{-\frac{\mathrm{i}}{\hbar}\int_I \langle\mathbb{B},\alpha \rangle} ~.
\eeq
\begin{remark}
If we compare this effective action with the one of the $\mathbb{B}$-disk without corners~\eqref{B_disk_eff_act} we notice that the corner field $\alpha$ plays here the role of the $\mathsf{a}$ zero-mode (the other term for action~\eqref{B_disk_eff_act}, containing only the zero-modes, is vanishing when restricted to the globalizing Lagrangian $\mathcal{L}=\{\mathsf{b}=0\}$).
Thus, integrating over the fields on the corner reproduces the globalized effective action for the $\mathbb{B}$-disk without corners.\\
We can also compare~\eqref{B_disk_1_corner} with the partition function of the $\mathbb{B}$-disk with two corners computed in~\eqref{B_disk_2_corners}.
We recover the partition function for the disk with one corner globalizing~\eqref{B_disk_2_corners} over $\mathcal{L}=\{\mathsf{a}=0\}$ and then integrating out one corner field $\alpha$.
\end{remark}

\subsection{BF $\mathbb{A}$-disk with one $\beta$-corner}\label{BF A-disk with b corners}

In order to calculate the one remaining building block of the theory, the partition function for an $\mathbb{A}$-disk with a single $\beta$-corner, we do the following. 
We first consider a disk~$D$ with boundary split into two intervals in $\mathbb{A}$ and $\mathbb{B}$-polarization with the two corners not decorated by polarization data (i.e. in the setting of the ``picture~\ref{pict_1}'' for corners, cf.~subsection~\ref{Corners_space_of_states}).%
\footnote{
 In fact, we can decorate the two corners with an arbitrary choice of polarizations~$\xi,\xi'$. 
 The partition function does not depend on this choice.
}
\beq
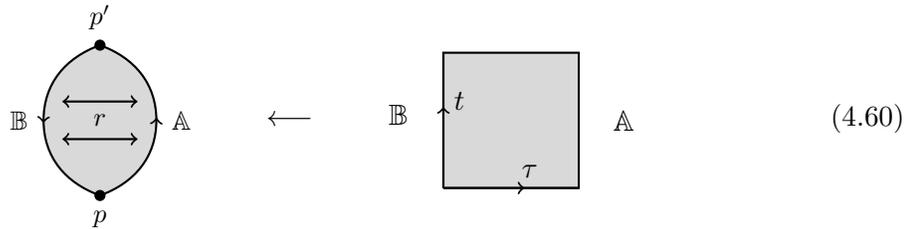
\label{AB disk}
 \begin{tikzpicture}[scale=.5, baseline=(x.base)]
  \coordinate (x1) at (0,2);
  \coordinate (x2) at (0,-2);
  
  \filldraw[fill=black!15, draw=black!15] (x1)node[above=2pt]{\small$p'$} to[out=-20, in=90] (1.5,0)  node[right=2pt]{\small$\mathbb{A}$} to[out=-90, in=20] (x2)node[below=2pt]{\small$p$} to[out=160, in=-90] (-1.5,0) node[left=2pt]{\small$\mathbb{B}$} to[out=90, in=-160] (x1);
  
  \draw[decoration={markings, mark=at position 0.52 with {\arrow{<}}},
      postaction={decorate},, color=black, thick] (x1) to[out=-20, in=90] (1.5,0) to[out=-90, in=20] (x2);
  \draw[decoration={markings, mark=at position 0.52 with {\arrow{<}}},
      postaction={decorate},, color=black, thick] (x2) to[out=160, in=-90] (-1.5,0) to[out=90, in=-160] (x1);
  
  \filldraw[black, thick] (x1) circle (0.12);
  \filldraw[black, thick] (x2) circle (0.12);
  
  \draw[black, thick, <->] (-1,0.5) to (1,0.5);
  \draw[black, thick, <->] (-1,-0.5) to (1,-0.5);
  \node (x) at (0,0) {\small$r$};
 \end{tikzpicture}
 \qquad\longleftarrow\qquad
 \begin{tikzpicture}[scale=.3, baseline=(x.base)]
  
    \filldraw[fill=black!15, thick, draw=black] (-3,-3) to (-3,3) to (3,3) to (3,-3) to (-3,-3);
    
    \draw[decoration={markings, mark=at position 0.6 with {\arrow{>}}},
      postaction={decorate},color=black, thick] (-3,-3) to node[above right]{$\tau$} (3,-3);
    \draw[decoration={markings, mark=at position 0.6 with {\arrow{>}}},
      postaction={decorate},color=black, thick] (-3,-3) to node[above right]{$t$} (-3,3);
      
    \node (x) at (0,-0.25) {};
    
    \node (a1) at (45:5.5) {};
    \node (a1) at (225:5.5) {};
    \node (a2) at (-45:5.5) {};
    \node (b2) at (135:5.5) {};
    
    \node (a1) at (270:5) {};
    \node (a1) at (90:5) {};
    \node (a2) at (177:5) {$\mathbb{B}$};
    \node (b2) at (0:5) {${\mathbb{A}}$};
    
 \end{tikzpicture}
\eeq
The partition function is easily computed by expanding the corners into two intervals (with arbitrary polarization) and putting the axial gauge on the square.%
\footnote{
 Here we use the axial gauge with the propagator  $\eta(t,\tau;t',\tau')=\delta(t'-t)\,(\mathrm{d}t'-\mathrm{d}t)\,\Theta(\tau'-\tau) $.
} 
This yields the answer
\begin{equation}\label{Z for AB-disk}
 Z(\mathbb{A},\mathbb{B})=\mathrm{e}^{\frac{\mathrm{i}}{\hbar}\int_{\partial_{\mathbb{A}}D}\langle r^*\mathbb{B},\mathbb{A} \rangle}~,
\end{equation}
where~$r$ is an orientation-reversing involution on the boundary of the disk, mapping the $\mathbb{A}$-arc diffeomorphically onto the $\mathbb{B}$-arc, having the two corners as fixed points (in terms of the square, $r$~is the involution $(t,0)\leftrightarrow (t,1)$). 
This partition function satisfies the mQME, $\Omega Z=0$, with $\Omega=\Omega_{\partial_{\mathbb{A}}}^\mathbb{A}+\Omega_{\partial_{\mathbb{B}}}^\mathbb{B}+ \langle\mathbb{B}_p,\mathbb{A}_p\rangle-\langle\mathbb{B}_{p'},\mathbb{A}_{p'}\rangle$, as per Proposition~\ref{prop: mQME I}, and as one can easily check explicitly.

\begin{remark} \label{rem: collapses of AB disk}
One can consider collapsing the $\mathbb{A}$- or $\mathbb{B}$-arc on the boundary of the disk~\eqref{AB disk}:
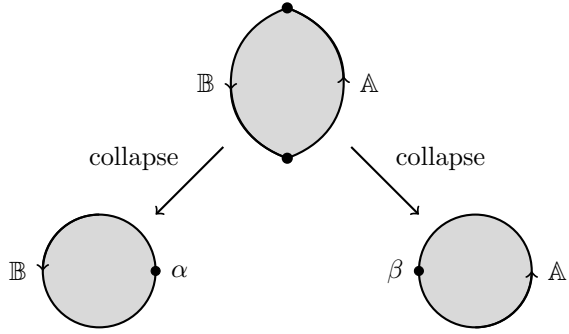
\[
 \begin{tikzpicture}[scale=.5]
  \coordinate (x1) at (0,2);
  \coordinate (x2) at (0,-2);
  
  \filldraw[fill=black!15, draw=black, thick] (x1) to[out=-20, in=90] (1.5,0)  node[right=2pt]{\small$\mathbb{A}$} to[out=-90, in=20] (x2) to[out=160, in=-90] (-1.5,0) node[left=2pt]{\small$\mathbb{B}$} to[out=90, in=-160] (x1);
  
  \draw[-<, color=black, thick] (x1) to[out=-20, in=90] (1.5,0);
  \draw[-<, color=black, thick] (x2) to[out=160, in=-90] (-1.5,0);
  
  \filldraw[black, thick] (x1) circle (0.12);
  \filldraw[black, thick] (x2) circle (0.12); 
  
  \begin{scope}[shift={(-5,-5)}, rotate=180]
   \filldraw[fill=black!15, thick, draw=black] (0,0) circle (1.5);
   \draw[black, thick, ->] (0,-1.5) arc (-90:0:1.5) node[left=2pt]{\small$\mathbb{B}$};
   \filldraw[black] (-1.5,0) circle (.12) node[right=2pt]{\small$\alpha$};
  \end{scope}
  
  \begin{scope}[shift={(5,-5)}]
   \filldraw[fill=black!15, thick, draw=black] (0,0) circle (1.5);
   \draw[black, thick, ->] (0,-1.5) arc (-90:0:1.5) node[right=2pt]{\small$\mathbb{A}$};
   \filldraw[black] (-1.5,0) circle (.12) node[left=2pt]{\small$\beta$};
  \end{scope}
  
  \draw[->, thick, black] (-1.7,-1.7) to node[above left]{\small collapse} (-3.5,-3.5);
  \draw[->, thick, black] (1.7,-1.7) to node[above right]{\small collapse} (3.5,-3.5);
  
 \end{tikzpicture}
\]
\begin{itemize}
 \item Collapsing the $\mathbb{A}$-arc into an $\alpha$-corner, we obtain a $\mathbb{B}$-disk with a single $\alpha$-corner (in the picture~\ref{pict_2}).
 Moreover, evaluating the partition function~\eqref{Z for AB-disk} on $\mathbb{B}=\beta$, we obtain the partition function $Z(\mathbb{B},\alpha)=\mathrm{e}^{-\frac{\mathrm{i}}{\hbar}\int_{\partial D} \langle \mathbb{B},\alpha\rangle}$, which agrees with our result~\eqref{B_disk_1_corner} from Section~\ref{B_disk_with_one_a_corner} and, indeed, satisfies the mQME with $\Omega=\Omega_{\partial D}^\mathbb{B}-\langle \mathbb{B}\big|^{+0}_{-0} , \alpha \rangle + \mathrm{i}\hbar \langle \frac12 [\alpha,\alpha],\frac{\partial}{\partial \alpha} \rangle$. 
 Here $\mathbb{B}\big|^{+0}_{-0}$ is the jump of the field $\mathbb{B}$ when passing through the $\alpha$-corner in positive direction.
 \item Collapsing the $\mathbb{B}$-arc of the disk~\eqref{AB disk} into a $\beta$-corner, we obtain a $\mathbb{A}$-disk with a single $\beta$-corner. 
 However, evaluating the partition function~\eqref{Z for AB-disk} on $\mathbb{B}=\beta$ yields $\mathrm{e}^{\frac{\mathrm{i}}{\hbar}\int_{\partial D} \langle \beta,\mathbb{A}\rangle}$ which \emph{does not} satisfy the mQME! 
 The reason for this is that the gauge-fixing on the disk~\eqref{AB disk} which was used to compute the partition function~\eqref{Z for AB-disk}, which in turn came from the axial gauge on a square, is not ``collapsible'', i.e. fails Assumption~\ref{assump: collapsible gauge}, and therefore Proposition~\ref{prop: mQME II} does not apply and we obtained a nonsensical answer after the collapse of the $\mathbb{B}$-arc.
\end{itemize}
\end{remark}

Using the construction of Section~\ref{sec: small model for A-states}, we can consider the projection $p_\mathcal{H}$ to the ``small model'' for the states on the $\mathbb{A}$-arc followed by respective inclusion $i_\mathcal{H}$, cf.~(\ref{p_H on A beta circle},\ref{i_H on A beta circle}). 
Thus we obtain a version of the partition function, factored through the small model for $\mathbb{A}$-states:
\beq\label{Z AB disk factored}
 \widetilde{Z}(\mathbb{A},\mathbb{B})&=i_\mathcal{H}\circ  p_\mathcal{H}\, Z \\
 &= \mathrm{e}^{\frac{\mathrm{i}}{\hbar} \int_{\partial_{\mathbb{A}}D} \big\langle r^*\mathbb{B}^{(0)}, \mathrm{d}t \log U(\mathbb{A}) \big\rangle
  + \big\langle r^*\mathbb{B}^{(1)}, \mathsf{G}_-(t,\mathrm{ad}_{\log U(\mathbb{A})}) \, \mathbb{A}_p 
   +  \mathsf{G}_+(t,\mathrm{ad}_{\log U(\mathbb{A})}) \, \mathbb{A}_{p'} \big\rangle ~.
}
\eeq
Note that, by~\eqref{reduction to const fields on an A-interval}, $\widetilde{Z}=Z+\Omega(\cdots)$ -- a modification of the answer~\eqref{Z for AB-disk} by an $\Omega$-exact term; this deformation can be interpreted as corresponding to a computation in a different gauge.%
\footnote{
 We also remark that the answer~\eqref{Z AB disk factored} can be obtained directly, by starting with an $\mathbb{A}$-disk with two $\alpha$-corners, and gluing it along one of the boundary arcs to the ``bean''~\eqref{B_disk_2_corners}.
}
Also, observe that in~\eqref{Z AB disk factored}, the field~$\mathbb{B}^{(1)}$ only interacts with the corner values of~$\mathbb{A}^{(0)}$, and thus the gauge corresponding to the answer~\eqref{Z AB disk factored} is ``collapsible'', i.e., satisfies the Assumption~\ref{assump: collapsible gauge}. 
Therefore, we can collapse the $\mathbb{B}$-arc into a $\beta$-corner, as in Section~\ref{sec: Picture II for corners}, by setting $\mathbb{B}^{(0)}=\beta$, $\mathbb{B}^{(1)}=0$ in~\eqref{Z AB disk factored}. 
Thus we finally arrive to the following result.

\begin{proposition}
The partition function for an $\mathbb{A}$-disk with a single $\beta$-corner is:
\begin{equation}\label{Z A beta}
 Z(\mathbb{A},\beta)=\mathrm{e}^{\frac{\mathrm{i}}{\hbar}\langle \beta, \log U(\mathbb{A}) \rangle}~.
\end{equation}
\end{proposition}

Note that this answer has a rigidity property: it cannot be changed by a BV-exact term~$\Omega(\cdots)$ for a degree reason -- there no boundary/corner fields of negative degree needed to construct a degree~$-1$ primitive. 
The answer~\eqref{Z A beta} does indeed satisfy the mQME, i.e. is $\Omega$-closed, as we have verified explicitly in Remark~\ref{rem: proof of mQME via Abeta disk} above.

\subsection{Gluing arcs in $\mathbb{A}$ polarization} \label{section_gluing_arcs}

We want now to recover the YM gluing law of two arcs in $\mathbb{A}$ polarization.
To compute this gluing law we can use an intermediate BF disk with the boundary split in two arcs with $\mathbb{B}$ polarization, separated by points in $\alpha$ polarization (figure~\ref{gluing_A_arcs}).
Thus, gluing together two $\mathbb{A}$-arcs with endpoints in $\alpha$-polarization, via the ``bean''~\eqref{B_disk_2_corners}, for the partition function of the glued surface we obtain the following:
\beq\label{gluing_A_complete}
\begin{aligned}
  Z_\Sigma &= \int  \mathrm{d}\mathbb{B} \,\mathrm{d}\widetilde{\mathbb{B}} \, \mathrm{d}\widetilde{\mathbb{A}} \,
   \mathrm{d}\mathbb{A} \, \mathrm{d}\mathsf{a}_1 ~
    \mathrm{e}^{\frac{\mathrm{i}}{\hbar} 
     \int_{\gamma} \langle \mathbb{B}, \mathsf{a} - \mathbb{A} + \mathsf{G}_-(\tau, \mathrm{ad}_{\mathsf{a}_1}) \alpha_0 
      + \mathsf{G}_+(\tau, \mathrm{ad}_{\mathsf{a}_1})  \alpha_1 \rangle} \\
 &\qquad\qquad\cdot\mathrm{e}^{-\frac{\mathrm{i}}{\hbar}  \int_{\widetilde{\gamma}} \langle \widetilde{\mathbb{B}}, \mathsf{a} - \widetilde{\mathbb{A}} 
  + \mathsf{G}_-(\tau, \mathrm{ad}_{\mathsf{a}_1}) \alpha_0 
   + \mathsf{G}_+(\tau, \mathrm{ad}_{\mathsf{a}_1}) \alpha_1 \rangle} ~
    \mathrm{det}\bigg(\frac{\sinh\big(\mathrm{ad}_{\mathsf{a}_1} /2\big)}{\mathrm{ad}_{\mathsf{a}_1} /2}\bigg)
     Z_{\Sigma_1}[\mathbb{A}]\;Z_{\Sigma_2}[\widetilde{\mathbb{A}}]\\
 &= \int_{\mathsf{a}_1\in B_0} \mathrm{d}\mathsf{a}_1 ~
  ~\mathrm{det}\bigg(\frac{\sinh\big(\mathrm{ad}_{\mathsf{a}_1} /2\big)}{\mathrm{ad}_{\mathsf{a}_1} /2}\bigg) 
   Z_{\Sigma_1}\big[\mathbb{A}=\mathsf{a}+\mathsf{G}_-(\tau, \mathrm{ad}_{\mathsf{a}_1}) \alpha_0 
    + \mathsf{G}_+(\tau, \mathrm{ad}_{\mathsf{a}_1}) \alpha_1\big] \\
 &\qquad\qquad\cdot  \;Z_{\Sigma_2}\big[\widetilde{\mathbb{A}}=\mathsf{a}+\mathsf{G}_-(\tau, \mathrm{ad}_{\mathsf{a}_1}) \alpha_0 
   + \mathsf{G}_+(\tau, \mathrm{ad}_{\mathsf{a}_1}) \alpha_1\big] ~.
\end{aligned}
\eeq
Here the integration domain for the zero-mode~$\mathsf{a}_1$ is the ``Gribov region''~$B_0\subset \mathfrak{g}$\,.
Notice that in this gluing formula the states on the $\mathbb{A}$-arcs factor through the ``small model'' for the space of states introduced in Section~\ref{sec: small model for A-states}.

If we assume also that the all boundary strata of $\Sigma_1$, $\Sigma_2$ are in $\mathbb{A}$-polarization and that partition functions $Z_{\Sigma_1}$, $Z_{\Sigma_2}$ are globalized, then the partition functions of $\Sigma_1$, $\Sigma_2$ does not depend on the ghost fields $\mathbb{A}_{(0)}$,$\widetilde{\mathbb{A}}_{(0)}$%
\footnote{
 Independence on the ghosts can be seen by assembling the surface with $\mathbb{A}$-boundary by gluing $\mathbb{A}$-polygons using beans as above. 
 Partition functions for polygons do not depend on the ghosts and the gluing formula~\eqref{gluing_A_complete} does not generate ghost dependence. 
 A curious point is that the  answer for $\mathbb{A}$-$\mathbb{A}$~cylinder in Section~\ref{AA polarization on the cylinder} did contain ghost dependence which seems to contradict what we are saying here. 
 In fact, there is no contradiction, rather there are inequivalent gauge-fixings: one can obtain an $\mathbb{A}$-$\mathbb{A}$~cylinder from an $\mathbb{A}$-square, gluing two opposite sides using the bean~\eqref{B_disk_2_corners}. 
 Choosing the gauge-fixing for the globalization on the bean as in~\eqref{gluing_A_complete} -- integrating over $\mathsf{a}_1$ -- we get the answer for the cylinder without the ghost delta-function. 
 If instead we use the opposite globalization on the bean -- integrating over $\mathsf{b}^1$ -- we obtain the answer of Section~\ref{AA polarization on the cylinder}, involving the ghost delta-function.
} 
and so the gluing formula reduces to
\beq\label{gluing_A}
 Z_\Sigma = \int_G \mathrm{d}U(\mathbb{A}) ~ Z_{\Sigma_1}[U(\mathbb{A})]\;Z_{\Sigma_2}[U(\mathbb{A})]    ~,
\eeq
which coincides with the gluing formula for YM known in literature~\cite{witten:2d_quantum_gauge, Migdal:1975zg}.

\begin{figure}[h!]
\[
   \begin{tikzpicture}[scale=.25, baseline=(x.base)]
    \node (x) at (0,0) {};
    \draw[fill=black!15, draw=black!15] (0,0) circle (3);   

    \node (a1) at (0:5) {$\mathbb{A}$};

    \node (a1) at (40:5) {${{\alpha_1}}$};
    \node (a2) at (-40:5) {${{\alpha_0}}$};

    \draw[color=black, thick] (40:3) arc (40:110:3);
    \draw[color=black, thick] (-40:3) arc (-40:40:3);
    \draw[color=black, thick] (-40:3) arc (-40:-110:3);
    \draw[color=black, thick, dotted] (110:3) arc (110:250:3);

    \draw[thick] (40:2.5) to (40:3.5);
    \draw[thick] (110:2.5) to (110:3.5);
    \draw[thick] (-40:2.5) to (-40:3.5);
    \draw[thick] (-110:2.5) to (-110:3.5);
   \end{tikzpicture}
   \quad
   \begin{tikzpicture}[scale=.25, baseline=(x.base)]
  
    \draw[fill=black!15, thick] (0,0) circle (3);   

    \node (x) at (0,0) {};
    \node at (0,0) {BF};
    
    \node (a1) at (180:5) {${\mathbb{B}}$};
    \node (a2) at (3:5) {$\widetilde{\mathbb{B}}$};
    
    \draw[thick] (90:2.5) to (90:3.5) node[above] {$\alpha_1$};
    \draw[thick] (270:2.5) to (270:3.5) node[below] {$\alpha_0$};

   \end{tikzpicture}
   \quad
   \begin{tikzpicture}[scale=.25, baseline=(x.base)]
    \node (x) at (0,0) {};
    \draw[fill=black!15, draw=black!15] (0,0) circle (3);   

    \node (a1) at (180:5) {$\widetilde{\mathbb{A}}$};

    \node (a1) at (-140:5) {${{\alpha_0}}$};
    \node (a2) at (140:5) {${{\alpha_1}}$};

    \draw[color=black, thick] (60:3) arc (60:300:3);
    \draw[color=black, thick, dotted] (-60:3) arc (-60:60:3);

    \draw[thick] (140:2.5) to (140:3.5);
    \draw[thick] (60:2.5) to (60:3.5);
    \draw[thick] (-60:2.5) to (-60:3.5);
    \draw[thick] (-140:2.5) to (-140:3.5);
   \end{tikzpicture}
   \quad=\quad
   \begin{tikzpicture}[scale=.25, baseline=(x.base)]
    \node (x) at (0,0) {};
    \draw[fill=black!15, draw= black!15] (0,0) circle (3);

    \draw[color=black, thick, dotted] (-30:3) arc (-30:30:3);
    \draw[color=black, thick, dotted] (150:3) arc (150:210:3);
    \draw[color=black, thick] (30:3) arc (30:150:3);
    \draw[color=black, thick] (210:3) arc (210:330:3);
    
    \draw[thick] (30:2.5) to (30:3.5);
    \draw[thick] (90:2.5) to (90:3.5);
    \draw[thick] (150:2.5) to (150:3.5);
    \draw[thick] (210:2.5) to (210:3.5);
    \draw[thick] (270:2.5) to (270:3.5);
    \draw[thick] (330:2.5) to (330:3.5);
    \draw[dashed, thick] (90:3) to (-90:3);
    
    \node (a1) at (-90:5) {${{\alpha_0}}$};
    \node (a2) at (90:5) {${{\alpha_1}}$};

   \end{tikzpicture}
   \]
  \caption{Two $\mathbb{A}$-polarized boundaries glued together using an intermediate $\mathbb{B}$-$\mathbb{B}$ BF disk.}
  \label{gluing_A_arcs}
\end{figure}
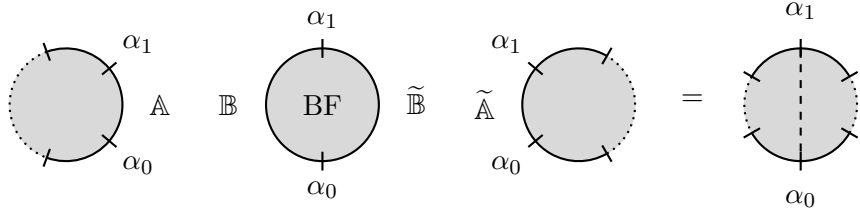

\subsection{2D YM partition function on surfaces with boundaries}\label{2d YM Partition Function on Surfaces with Boundaries}

We can now compute the partition function on a general surface with boundaries.
Indeed, any surface with boundary can be obtained by gluing edges of some polygon (or a collection of polygons -- any triangulation or a cellular decomposition gives a presentation of the surface of this kind).
Thus using the gluing properties of BV-BFV theories we can compute the YM partition function on a general surface with boundary starting from the partition function on the disk with the boundary split in several arcs $\gamma_i$\,:
\beq
 Z_{D^2}[\mathbb{A}_1,\dots,\mathbb{A}_n] = \sum_{R} (\mathrm{dim}~R)~ \chi_R\big(U_{\gamma_1}(\mathbb{A}_1)\cdots U_{\gamma_n}(\mathbb{A}_n)\big)~
  \mathrm{e}^{-\frac{\mathrm{i}\hbar a}{2}C_2(R)} ~.
\eeq
Each time we glue together two arcs, using the property~\eqref{gluing_A}, we have an integral of the kind:
\beq
 \int_G \mathrm{d} U ~ \chi_R(VUWU^\dag ) = \frac{\chi_R(V)\chi_R(W)}{\mathrm{dim}~R}	~,\\
 \int_G \mathrm{d} U ~ \chi_R(VU)\chi_R(WU^\dag) = \frac{\chi_R(VW)}{\mathrm{dim}~R}	~. 
\eeq
This way we get the following result.

\begin{theorem}
The globalized YM partition function on a surface with genus~$g$ and $b$~boundaries in the $\mathbb{A}$~polarization is:
\beq\label{2d-YM_part.funct}
 Z_{\Sigma_{g,b}}[\mathbb{A}_1,\dots,\mathbb{A}_b] = \sum_{R} (\mathrm{dim}~R)^{2-2g-b} \mathrm{e}^{-\frac{\mathrm{i}\hbar a}{2}C_2(R)} \prod_{i=1}^b 
  \chi_R\big(U_{b_i}(\mathbb{A}_i)\big)~.
\eeq
\end{theorem}

\appendix

\section{Wilson loop observables}\label{Wilson_loop_observables}

Let us consider now observables in 2D~YM.
These are operators on the Hilbert space $\mathcal{H}_\Gamma$ associated to some boundary $\Gamma$ which, in $\mathbb{A}$ polarization, is the space of functions of the holonomy $U_\Gamma(\mathbb{A})$\,.

Let us consider for example the multiplication operator for the factor $\chi_R(U_\Gamma(\mathbb{A}))$ for some representation $R$\,.
We can compute the matrix element of this operator between two states defined by the partition functions on two surfaces $\Sigma_1$ and $\Sigma_2$ with the same boundary $\Gamma = S^1$\,.
Using the gluing rule~\eqref{S1_gluing} for boundaries in $\mathbb{A}$ polarization we get:
\beq\label{non-int.loop}
 &\langle Z_{\Sigma_1} | \chi_R(U_\Gamma(\mathbb{A})) | Z_{\Sigma_2} \rangle = 
  \int_G \mathrm{d}U ~ Z_{\Sigma_1}(U^\dag) \chi_R(U) Z_{\Sigma_2}(U) \\
 &= \sum_{R_1,R_2} (\mathrm{dim}~R_1)^{1-2g_1} (\mathrm{dim}~R_2)^{1-2g_2} \mathrm{e}^{-\frac{\mathrm{i}\hbar a_1}{2}C_2(R_1) 
  -\frac{\mathrm{i}\hbar a_2}{2}C_2(R_2)} \int_G \mathrm{d}U ~ \chi_{R_1}(U^\dag) \chi_{R}(U) \chi_{R_2}(U)\\
 &= \sum_{R_1,R_2} (\mathrm{dim}~R_1)^{1-2g_1} (\mathrm{dim}~R_2)^{1-2g_2} \mathrm{e}^{-\frac{\mathrm{i}\hbar a_1}{2}C_2(R_1) 
  -\frac{\mathrm{i}\hbar a_2}{2}C_2(R_2)} N^{R_1}_{R,R_2}  ~,
\eeq
where we used the expression~\eqref{2d-YM_part.funct} for the partition functions of the surfaces~$\Sigma_i$\,, with genus~$g_i$\,, and where $N^{R_1}_{R,R_2}$ are the fusion numbers defined by the decomposition of the product of irreducible representations: $R\otimes R_2 = \oplus_{R_1} N^{R_1}_{R,R_2} R_1$\,.
This quite obviously corresponds to the computation of the expectation value of a non self-intersecting Wilson loop $W_R(\Gamma)$ on the surface $\Sigma_1\bigcup_\Gamma \Sigma_2$\,.

More generally, we can consider operators going from some space of ``inbound'' states to some ``outbound'' states: $\mathcal{O}\colon \mathcal{H}_{\mathrm{in}} \rightarrow  \mathcal{H}_{\mathrm{out}}$\,.
Such an operator can be represented by a surface (possibly with corners) with the appropriate boundary components, i.e. such that the boundary Hilbert space is $\mathcal{H}^*_{\mathrm{in}}\otimes\mathcal{H}_{\mathrm{out}}$\,, and a particular state corresponding to~$\mathcal{O}$\,.
The operator now acts on the inbound states by gluing.
For example to the (non self-intersecting) Wilson loop we computed above we can associate a cylinder in $\mathbb{A}$-$\mathbb{A}$ polarization and the state $\chi_R(U(\mathbb{A})) \delta(U(\mathbb{A}),U(\mathbb{A}'))$\,.

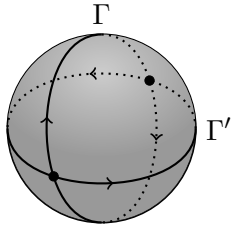
\begin{figure}[h]
\centering
\begin{tikzpicture}[scale=.25]
 \draw[fill=black!40, thick] (0,0) circle (5);   
 \begin{scope}
   \clip (0,0) circle (5);
   \filldraw[fill=black!40] (0,0) circle (5);
   \begin{scope}[yscale=.75]
   \shade[inner color=black!20, outer color=black!40] (0,4) circle (6);
   \end{scope}
 \end{scope}
 
 \draw[decoration={markings, mark=at position 0.55 with {\arrow{>}}},
      postaction={decorate},color=black, thick] (0,-5) to[in=180, out=180] (0,5) node[above]{$\Gamma$};
 \draw[decoration={markings, mark=at position 0.55 with {\arrow{>}}},
      postaction={decorate},color=black, thick, dotted] (0,5) to[in=0, out=0] (0,-5);
 
 \draw[decoration={markings, mark=at position 0.55 with {\arrow{>}}},
      postaction={decorate},color=black, thick] (-5,0) to[in=-90, out=-90] (5,0) node[right]{$\Gamma'$};
 \draw[decoration={markings, mark=at position 0.55 with {\arrow{>}}},
      postaction={decorate},color=black, thick, dotted] (5,0) to[in=90, out=90] (-5,0);
      
 \draw[color=black, fill=black] (-2.55,-2.55) circle (7pt); 
 \draw[color=black, fill=black] (2.55,2.55) circle (7pt); 

\end{tikzpicture}
\caption{Two intersecting Wilson loops $\Gamma$ and $\Gamma'$ on the 2-sphere.}
\label{inters.loops}
\end{figure}

Consider now the case of two Wilson loops $W_\Gamma (R) W_{\Gamma'}(R')$ intersecting in 2 points: $\Gamma \cap \Gamma' = v_1\cup v_2$\,.
We can view them as 4 separate arcs $\gamma_i,\gamma_i'$\,, $i=1,2$\,, joining the two intersection points.
These intersecting Wilson loops can be thought as a multiplication operator on the space of states of the 4 arcs -- multiplication by the factor $\chi_R\big(U_{\gamma_1}(\mathbb{A}_1)U_{\gamma_2}(\mathbb{A}_2)\big)\chi_{R'}\big(U_{\gamma'_1}(\mathbb{A}'_1)U_{\gamma'_2}(\mathbb{A}'_2)\big)$\,.
We can thus compute matrix elements between states defined by surfaces with opportune boundary components.
Let us consider for example four disks, each with the boundary circle separated into two arcs, glued into a sphere with two intersecting Wilson loops as in figure~\ref{inters.loops}:
\begin{equation}
\begin{aligned}
 \langle W_\Gamma (R) W_{\Gamma'}(R') \rangle_{S^2} = \sum_{\stackrel{R_1,R_2,}{_{R_3,R_4}}} 
  \mathrm{e}^{-\frac{\mathrm{i}\hbar}{2} \sum_i a_iC_2(R_i)} \prod_i (\mathrm{dim}~R_i) \int_G \mathrm{d}U_1 \mathrm{d}U_2 \mathrm{d}U_3 \mathrm{d}U_4 ~ 
   \chi_R(U_1 U_3)  \\
 \cdot \chi_{R'}(U_2 U_4) \chi_{R_1}(U_4 U_1) \chi_{R_2}(U_1^\dag U_2)\chi_{R_3}(U_2^\dag U_3^\dag) \chi_{R_4}(U_3 U_4^\dag)	~.
\end{aligned}
\end{equation}
This integral can be evaluated using the Peter-Weyl theorem (part 3) which implies:
\beq\label{edge_int}
 \int_G \mathrm{d}U~ R_1(U)_{i'}^i R_2(U)_{j'}^j R_3(U^\dag)^{k'}_k = \frac{1}{\mathrm{dim}~R_3} \sum_{\mu} C_\mu(R_1,R_2;R_3)^{ij}_k C^*_\mu(R_1,R_2;R_3)_{i'j'}^{k'}	~,
\eeq
where $C_\mu(R_1,R_2;R_3)^{ij}_k$ are Clebsch-Gordan coefficients.%
\footnote{
 If we have representations $R_1,R_2$ we can decompose their product into the sum of irreducible representations.
 Let $\{e_1^i\}$ and $\{e_2^j\}$ be two basis of the representation spaces of $R_1$ and $R_2$ respectively, and let $\{e_{\mu_a}^k\}$ be a basis of their tensor product such that the product representation is in the block-diagonal form, where $a$ denotes the irreps and $\mu_a$ labels the various copies of the representation $R_a$ appearing in the product $R_1\otimes R_2$\,.
 The Clebsh-Gordan coefficients are defined as the basis changing coefficients: \[e_1^i\otimes e_2^j = \sum_a\sum_{\mu_a} C_{\mu_a}(1,2;a)^{ij}_k e_{\mu_a}^k ~.\]
} 
We get:
\begin{equation}
\begin{aligned} 
 &\langle W_\Gamma (R) W_{\Gamma'}(R') \rangle_{S^2} =\sum_{\stackrel{R_1,R_2,}{_{R_3,R_4}}} \mathrm{e}^{-\frac{\mathrm{i}\hbar }{2}\sum_i a_iC_2(R_i)} 
  \frac{\mathrm{dim}~R_1}{\mathrm{dim}~R_3} \sum_{\overset{\mu_1,\mu_2}{_{\mu_3,\mu_4}}} C_{\mu_1}(R,R_1;R_2)^{ij}_k C^*_{\mu_1}(R,R_1;R_2)_{i'j'}^{k'}	  \\
 &\cdot  C_{\mu_2}(R',R_2;R_3)^{lk}_m  C^*_{\mu_2}(R',R_2;R_3)_{l'k'}^{m'} C_{\mu_3}(R,R_4;R_3)^{i'n'}_{m'} 
  C^*_{\mu_3}(R,R_4;R_3)_{in}^{m} C_{\mu_4}(R',R_1;R_4)^{l'j'}_{n'} \\
 &\cdot C^*_{\mu_4}(R',R_1;R_4)_{lj}^{n}= \sum_{\stackrel{R_1,R_2,}{_{R_3,R_4}}} \mathrm{e}^{-\frac{\mathrm{i}\hbar }{2} \sum_i a_iC_2(R_i)} \prod_{i=1,\dots,4} (\mathrm{dim}~R_i)~  	
  \sum_{\overset{\mu_1,\mu_2}{_{\mu_3,\mu_4}}}								\big\{\begin{smallmatrix}
                                                                                                           R&R_1&R_2\\
                                                                                                           R'&R_3&R_4
                                                                                                        \end{smallmatrix}\big\}^{\mu_1\mu_2}_{\mu_3\mu_4}~
													\big\{\begin{smallmatrix}
                                                                                                           R'&R_1&R_4\\
                                                                                                           R&R_3&R_2
                                                                                                        \end{smallmatrix}\big\}_{\mu_1\mu_2}^{\mu_3\mu_4} ~,
\end{aligned}
\end{equation}
where $\big\{\begin{smallmatrix} R&R_1&R_2\\R'&R_3&R_4\end{smallmatrix}\big\}^{\mu_1\mu_2}_{{\mu_3\mu_4}}$ are Wigner 6-j symbols.%
\footnote{
 6-j symbols are defined by:
 \[
  \bigg\{\begin{matrix} R_1&R_2&R_3\\R_4&R_5&R_6\end{matrix}\bigg\}^{\mu_1\mu_2}_{{\mu_3\mu_4}} = 
  \frac{C_{\mu_1}(R_1,R_2;R_3)^{ij}_k}{(\mathrm{dim}~R_3)^{\frac{1}{2}}}
  \frac{C_{\mu_2}(R_4,R_3;R_5)^{lk}_m}{(\mathrm{dim}~R_5)^{\frac{1}{2}}}
  \frac{C^*_{\mu_3}(R_1,R_6;R_5)_{in}^m}{(\mathrm{dim}~R_5)^{\frac{1}{2}}} 
  \frac{C^*_{\mu_4}(R_4,R_2;R_6)_{lj}^n}{(\mathrm{dim}~R_6)^{\frac{1}{2}}} ~.
 \]
}

We can generalize this to compute the value of any number of (possibly intersecting) Wilson loops over any surface with boundary.
Given a set of Wilson loops we can consider separately the various Wilson lines connecting intersection points.%
\footnote{
 If a loop has no intersections, then its contribution will be $N^{R_3}_{R_1,R_2}$ as in equation~\eqref{non-int.loop}.
}
Each line carries a group variable and and contributes with the integral~\eqref{edge_int}, where $R_1$ is the representation of the Wilson loop containing that line and $R_2,R_3$ are the representations carried by the two regions adjacent to that line.
The main observation is that this integral factorises into the product of two Clebsch-Gordan coefficents, each depending only on indices living on one of the two edges of the line.%
\footnote{
 Each oriented boundary, or Wilson loop, carries the character of the holonomy of $\mathbb{A}$\,.
 If we split the circle into various arcs, then it will carry the character of the products of the holonomies over different arcs, multiplied according to the orientation of the loop. 
 This defines inbound and outbound indices for the holonomy over each oriented arc.
}
Thus when 4 lines meet at an intersection point, the factors associated to that intersection combine to give a 6-j symbol, as in the previous example:
\beq\label{6-j}
 \begin{tikzpicture}[scale=.55, baseline=(x2.base)]
  \coordinate (x) at (0,0);
 
  \coordinate (v1) at (-2,-2);
  \coordinate (v2) at (2,2);
  
  \coordinate (v1') at (-2,2);
  \coordinate (v2') at (2,-2);
  
  \node (x1) at (0,2) {$R_1$};
  \node (x2) at (2,0) {$R_2$};
  \node (x3) at (0,-2) {$R_3$};
  \node (x4) at (-2,0) {$R_4$};
  
  \draw[decoration={markings, mark=at position 0.5 with {\arrow{>}}},
      postaction={decorate},color=black, thick] (v1) node[left, below]{$R$}node[left]{$\mu_3$} to (x);
  \draw[decoration={markings, mark=at position 0.6 with {\arrow{>}}},
      postaction={decorate},color=black, thick] (x) to (v2)node[right, above]{$R$} node[right]{$\mu_1$};      
  \draw[decoration={markings, mark=at position 0.5 with {\arrow{>}}},
      postaction={decorate},color=black, thick] (v1') node[right, above]{$R'$}node[left]{$\mu_4$} to (x);
  \draw[decoration={markings, mark=at position 0.6 with {\arrow{>}}},
      postaction={decorate},color=black, thick] (x) to (v2') node[right, below]{$R'$} node[right]{$\mu_2$};
      
  \draw[color=black, fill=black] (0,0) circle (6pt); 

 \end{tikzpicture}
 \qquad\rightsquigarrow\qquad \bigg\{\begin{matrix} R&R_1&R_2\\R'&R_3&R_4\end{matrix}\bigg\}^{\mu_1\mu_2}_{{\mu_3\mu_4}} \equiv G(R,R';\dots)	~.
\eeq
Finally, the expectation value of a set of Wilson loops $\{\Gamma_l\}$ on a surface $\Sigma$ is given by the following formula:
\beq
 \langle \prod_l W_{\Gamma_l}(R_l) \rangle_\Sigma = \sum_{R_\lambda} \prod_\lambda \mathrm{e}^{-\frac{\mathrm{i}\hbar }{2} a_\lambda C_2(R_\lambda)} 
  (\mathrm{dim}~R_\lambda)^{2-2g_\lambda - b_\lambda} \prod_{b_\lambda}\chi_{R_\lambda}(U_{b_\lambda}) \\
 \cdot \prod_v\sum_{\mu_i} G_v(R_l; R_\lambda; \mu_i) \prod_{l_0} N(R_{l_0};R_\lambda) 	~.
\eeq
where the index $\lambda$ runs over connected components $\Sigma_\lambda$ of the surface obtained by cutting $\Sigma$ along the Wilson loops, $b_\lambda$ labels the boundaries of $\Sigma$ contained in $\Sigma_\lambda$\,, $v$ labels the intersections between loops, 
$G_v$ indicates the 6-j symbol at the vertex $v$ evaluated on the surrounding representations according to~\eqref{6-j} and $N(R_{l_0};R_\lambda)$ denotes the fusion numbers for the non-intersecting Wilson loops labelled by~$l_0$\,.

\section{Propagators}

We collect in this appendix the computations of the propagators used in this paper.
We will firstly consider one-dimensional BF propagator on the circle and on the interval with the various possible polarizations on the two end-points.
Then we will use these to compute the axial-gauge propagator on 2D surfaces, in particular on the cylinder~$S^1\times I$ and on the square~$I\times I$\,.

\subsection{One-dimensional propagators}\label{One-Dimensional Propagators}

\subsubsection{Propagator on the circle}

Let us consider non-abelian BF theory on the circle~$S^1$\,.
We are looking for the propagator when we expand the action with respect to the trivial connection.
In this case the kinetic term is $\int_{S^1} \langle{B},\mathrm{d}{A}\rangle$\,.
The space of zero-modes is thus given by the de Rham cohomology: $\mathcal{V} = H_{\mathrm{dR}}(S^1;\mathfrak{g})[1]\oplus H_{\mathrm{dR}}(S^1;\mathfrak{g}^*)$\,.
If~$\tau\in[0,1]$ is the coordinate of the circle, we have the corresponding basis~$[1],[\mathrm{d}\tau]$ for the cohomology of the circle and the following coordinate expression for the zero modes:
\beq
 \mathsf{a} = \mathsf{a}_0 + \mathsf{a}_1\mathrm{d}\tau ~,\qquad \mathsf{b} = \mathsf{b}^0 + \mathsf{b}^1\mathrm{d}\tau ~,
\eeq
where~$\mathsf{a}_{(i)}\in\mathfrak{g}$ and~$\mathsf{b}_{(i)}\in\mathfrak{g}^*$\,.
A Hodge decomposition for the de Rham complex of the circle is given by the following induction data:
\beq
 \Pi\omega(\tau) & 
  = \int_{S^1}(\mathrm{d}\tau'-\mathrm{d}\tau)\omega(\tau') ~, \\
 K\omega(\tau) &
  = \int_{S^1} \Big(\Theta(\tau-\tau')-\tau + \tau'-\frac{1}{2}\Big) \omega(\tau')~.
\eeq
The extension to the space of fields --Lie-algebra valued differential forms-- is immediate and the resulting propagator is:%
\footnote{
 The Lie-algebra part of the propagator in this paper is always the identity~$\delta^a_b$ and will be often omitted.
}
\beq
 \begin{tikzpicture}[baseline=(x.base)]
  \node[above] (0,0) {$\eta_{S_1}$};
  \node[right] (x) at (1,0) {$a,\tau'$};
  \draw[decoration={markings, mark=at position 0.5 with {\arrow{>}}},
        postaction={decorate},thick] (-1,0)node[left]{$b,\tau$} to (x);
 \end{tikzpicture}
 ~=\eta_{S^1}(\tau,b;\tau',a) = \Big(\Theta(\tau-\tau')-\tau + \tau'-\frac{1}{2}\Big) \delta^a_b~.
\eeq

\subsubsection{Propagators on the interval}

\subsubsection*{Interval in $\mathbb{A}$-$\mathbb{B}$ polarization}

Let us consider now BF theory on the unit interval~$I=[0,1]$ with~$\mathbb{B}$ polarization at~$\{0\}$ and~$\mathbb{A}$ polarization at~$\{1\}$\,.
The space of bulk fields is now given by differential forms with Dirichlet boundary conditions on one of the two endpoints.
The cohomology of the differential~$\mathrm{d}$ on this space of differential forms is vanishing, thus the space of zero-modes is empty~$\mathcal{V}=0$\,.
The chain homotopy~$K$ is now
\beq
 K\omega (t) = \int_t^1 \omega(t')
\eeq
and we have the corresponding propagator
\beq
 \begin{tikzpicture}[baseline=(x.base)]
  \node[above] (0,0) {$\eta$};
  \node[right] (x) at (1,0) {$t'$};
  \draw[decoration={markings, mark=at position 0.5 with {\arrow{>}}},
        postaction={decorate},thick] (-1,0)node[left]{$t$} to (x);
 \end{tikzpicture}
 ~=\eta(t;t') = -\Theta(t'-t) ~.
\eeq
Notice that propagation can only occur if~$t<t'$\,, i.e. moving away from the~$\mathbb{B}$ endpoint and toward the~$\mathbb{A}$ endpoint of the interval.

\subsubsection*{Interval in $\mathbb{A}$-$\mathbb{A}$ polarization}

If we take the~$\mathbb{A}$ polarization on both endpoints of the interval, the~${A}$ fields will have Dirichlet boundary conditions at the endpoints while the~${B}$ fields will have free boundary conditions: $\mathcal{Y}=\Omega(I,\partial I;\mathfrak{g})[1]\oplus\Omega(I;\mathfrak{g}^*)$\,.
The cohomology is concentrated in form-degree~1 for the~$A$ fields and in form-degree~0 for the~$B$ fields
\beq
 \mathcal{V} = H(I,\partial I;\mathfrak{g})[1]\oplus H(I;\mathfrak{g}^*) \simeq \mathfrak{g}[1]\oplus\mathfrak{g}^* ~,
\eeq
so that the form-degree expansion of the zero modes is $ \mathsf{a} = \mathsf{a}_1\mathrm{d}t \,,~ \mathsf{b} = \mathsf{b}^0$\,.
The chain retraction is given by the following data:
\beq
 \eta(t,t') = \Theta(t-t') -t ~,\qquad \pi(t,t')=-\mathrm{d}t~.
\eeq

\subsubsection*{Interval in $\mathbb{B}$-$\mathbb{B}$ polarization}

The interval with~$\mathbb{B}$ polarization on both endpoints has the role of~$A$ and~$B$ fields reversed with respect to the previous case.
The space of bulk fields is $\mathcal{Y}=\Omega(I;\mathfrak{g})[1]\oplus\Omega(I,\partial I;\mathfrak{g}^*)$ and the zero-modes are: $ \mathsf{a} = \mathsf{a}_0\,,~ \mathsf{b} = \mathsf{b}^1\mathrm{d}t $\,.
The propagator and the projection to cohomology are:
\beq
 \eta(t,t') = -\Theta(t'-t) +t' ~,\qquad \pi(t,t')=\mathrm{d}t'~.
\eeq

\subsection{Axial gauge propagators on the cylinder}\label{Axial_Gauge_Cylinder}

Consider now a cylinder $S^1\times I$ and let~$t$ denote the coordinate of the interval, $\tau$~the coordinate along the circle, $\chi_i$~a basis for the cohomology of~$S^1$ and~$\chi^i$ its dual basis.
Using the 1-dimensional propagators of appendix~\ref{One-Dimensional Propagators}, from the axial-gauge formula~\eqref{axial_gauge_prop.} we get the following propagators on the cylinder.
\beq\label{axial_gauge_table}
\begin{tabular}{cc|c|l}
  &&zero-modes&\multicolumn{1}{c}{ 
  \begin{tikzpicture}[baseline=(x.base)]
  \node[above] (0,0) {$\eta(\tau,t;\tau',t')$};
  \node[right] (x) at (1.2,0) {$\tau',t'$};
  \draw[decoration={markings, mark=at position 0.5 with {\arrow{>}}},
        postaction={decorate},thick] (-1.2,0)node[left]{$\tau,t$} to (x);
 \end{tikzpicture}} 	\\ \cline{2-4}
 \multirow{5}{*}{\rotatebox[origin=c]{90}{Polarization}} & \multicolumn{1}{l|}{$\mathbb{A-B}$}	
   & $0$	
    & $-\Theta(t'-t)\delta(\tau-\tau')(\mathrm{d}\tau'-\mathrm{d}\tau)$ 	\\ \cline{2-4}
  &\multirow{2}{*}{$\mathbb{A-A}$}	& \multicolumn{1}{l|}{$\mathsf{a} = \mathsf{a}_i \; \mathrm{d}t\wedge\chi^i $}
    & \multirow{2}{*}{$(\Theta(t-t')-t)\delta(\tau'-\tau)(\mathrm{d}\tau'-\mathrm{d}\tau) - \mathrm{d}t\big( \Theta(\tau-\tau')-\tau-\tau'-\frac{1}{2} \big)$}      \\ 
  &&\multicolumn{1}{l|}{$\mathsf{b} = \mathsf{b}^i \; \chi_i $}&   \\ \cline{2-4}
  &\multirow{2}{*}{$\mathbb{B-B}$}	& \multicolumn{1}{l|}{$\mathsf{a} = \mathsf{a}_i \; \chi^i $}
    & \multirow{2}{*}{$(t'-\Theta(t'-t))\delta(\tau'-\tau)(\mathrm{d}\tau'-\mathrm{d}\tau) + \mathrm{d}t'\big( \Theta(\tau-\tau')-\tau-\tau'-\frac{1}{2} \big)$}      \\ 
  &&\multicolumn{1}{l|}{$\mathsf{b} = \mathsf{b}^i \; \mathrm{d}t\wedge\chi_i $}&   \\
\end{tabular}
\eeq
Reversing the role of the circle and the interval in formula~\eqref{axial_gauge_prop.} we would obtain different expressions for the propagator, called for the cylinder \emph{horizontal gauge}, but we don't need this choice in this paper.

\section{Computations of some Feynman diagrams}\label{diagram_computations}

We present here the proofs of Propositions~\ref{BB_eff_act},~\ref{BB_disk_part.funct.}, consisting in the evaluation of tree and loop diagrams in the axial gauge.
These computations are variations of the ones contained in~\cite{mnev:discrete_BF},~Lemma~3~and~4, obtained in the 1-dimensional setting.

\begin{proof}[Proof of Proposition~\ref{BB_eff_act}]
We have to evaluate the 1-loop diagrams of figure~\ref{BB_cyl}.
The amplitude for a diagram with $n\geqslant 2$ vertices is:
\beq\label{BB_loop_amp}
 \begin{aligned}
  &\frac{1}{n} \mathrm{tr} (\mathrm{ad}_{\mathsf{a}_1}^n) \int_{(S^1)^n}\mathrm{d}\tau_1\cdots\mathrm{d}\tau_n\; 
   \eta_{S^1}(\tau_1;\tau_2)\cdots\eta_{S^1}(\tau_{n-1};\tau_{n}) \eta_{S^1}(\tau_n;\tau_1)\\
  &= \frac{1}{n} \mathrm{tr} (\mathrm{ad}_{\mathsf{a}_1})^n\; \mathrm{tr}\big(K(\chi_1\wedge\bullet)\big)^n ~.
 \end{aligned}
\eeq
where we chose the basis $\chi_0=1,\chi_1=\mathrm{d}\tau$ for $H^\bullet (S^1)$ and $K$ is the chain homotopy with integral kernel $\eta_{S^1}(\tau;\tau')=\Theta(\tau-\tau')-\tau+\tau'-\frac{1}{2}$\,.
We will compute $\mathrm{tr}\big(K(\chi_1\wedge\bullet)\big)^n$ in the monomial basis $1,\tau,\tau^2,\dots$\,. 
Let us define the generating function:
\beq
 f_m(x,\tau) = \sum_{n=0}^\infty x^n \big(K(\chi_1\wedge\bullet)\big)^n \tau^m ~.
\eeq
Applying $xK(\chi_1\wedge\bullet)$ on both sides we get
\beq\label{gen.func.eq}
 xK(\chi_1f_m)(x,\tau)= f_m(x,\tau) - \tau^m
\eeq
and, differentiating w.r.t. $\tau$\,, we obtain the differential equation:
\beq
 \frac{\partial }{\partial \tau} f_m = xf_m + m \tau^{m-1} -x \int_0^1\mathrm{d}\tau \; f_m ~.
\eeq
Solutions to the above equation are of the form
\beq\label{fm1}
 f_m(x,\tau) = A(x)\mathrm{e}^{x\tau} + B(x) + \mathrm{e}^{x\tau}\int_0^\tau \mathrm{d}\tilde{\tau}\; m\tilde{\tau}^{m-1}\mathrm{e}^{-x\tilde{\tau}} ~,
\eeq
where $A(x)=\frac{1}{\mathrm{e}^x-1}(1-\mathrm{e}^x\int_0^1\mathrm{d}\tau\; m\tau^{m-1}\mathrm{e}^{-x\tau})$ and $B(x)$ is to be determined from the boundary conditions.
Since $K(\chi_1 f_m)(x,0)=K(\chi_1f_m)(x,1)$\,, from~\eqref{gen.func.eq} we have:
\beq
\begin{aligned}
 f_m(x,1)-1&=f_m(x,0)= xK(\chi_1f_m)(x,0)\\
 &= x\int_0^1\mathrm{d}{\tau}\; \Big( A(x)\mathrm{e}^{x{\tau}}+B(x) + \mathrm{e}^{x\tau}\int_0^\tau \mathrm{d}\tilde{\tau}\; 
  m\tilde{\tau}^{m-1}\mathrm{e}^{-x\tilde{\tau}} \Big)({\tau}-\frac{1}{2})\\
 &=A(x)g(x) + C(x) ~,
\end{aligned}
\eeq
where $g(x)=x\int_0^1\mathrm{d}\tilde{\tau}\;(\tilde{\tau}-\frac{1}{2})\mathrm{e}^{x\tilde{\tau}}$ and $C (x)= x\int_0^1\mathrm{d}{\tau}\;
 \mathrm{e}^{x\tau}({\tau}-\frac{1}{2})\int_0^\tau \mathrm{d}\tilde{\tau}\; m\tilde{\tau}^{m-1}\mathrm{e}^{-x\tilde{\tau}} $\,.
Moreover from~\eqref{fm1} we have $f_m(x,0)=A(x)+B(x)$ and thus:
\beq\label{fm2}
 f_m(x,\tau)= \frac{\mathrm{e}^{x\tau}-1}{\mathrm{e}^{x}-1}\bigg(1-\mathrm{e}^x\int_0^1\mathrm{d}\tilde{\tau}\; 
  m\tilde{\tau}^{m-1}\mathrm{e}^{-x\tilde{\tau}}\bigg) + \mathrm{e}^{x\tau}\int_0^\tau \mathrm{d}\tilde{\tau}\; 
  m\tilde{\tau}^{m-1}\mathrm{e}^{-x\tilde{\tau}} + f_m(x,0) ~.
\eeq
We can now extract the trace of powers of $\mathcal{M}:= K(\chi_1\wedge\bullet)$ from the series of the coefficients of $\tau^m$ in the expansion of $f_m$\,:
\beq
\begin{aligned}
 &f_{mm}(x) := \sum_{n=0}^\infty \langle \tau^m|  x^n \mathcal{M}^n |\tau^m\rangle~,\\
 \Rightarrow \quad &\sum_{m=1}^{\infty} (f_{mm}(x) -1) = \sum_{n=1}^{\infty} x^n \mathrm{tr}\mathcal{M}^n ~.	
\end{aligned}
\eeq
The coefficients $f_{mm}(x)$ can be read from~\eqref{fm2}:
\beq
\begin{aligned}
 f_m(x,\tau)&= \frac{\mathrm{e}^{x\tau}-1}{\mathrm{e}^{x}-1} \sum_{k=0}^{m-1} \frac{m!}{(m-k)!}x^{-k} - \sum_{k=0}^{m-1} \frac{m!}{(m-k)!}\tau^{m-k}x^{-k} + f_m(x,0),\\
 \Rightarrow\quad  f_{mm}(x)&=1-\frac{1}{\mathrm{e}^{x}-1} \sum_{k=m+1}^\infty \frac{x^k}{k!}~.
\end{aligned}
\eeq
Thus we get:
\beq
\begin{aligned}
 &\sum_{n=1}^{\infty} x^n \mathrm{tr}\mathcal{M}^n &&= -\frac{1}{\mathrm{e}^{x}-1} \sum_{m=1}^\infty \sum_{k=m+1}^\infty \frac{x^k}{k!} =
  -\frac{1}{\mathrm{e}^{x}-1} \sum_{k=2}^\infty  \frac{k-1}{k!}x^k  \\
 &&&= 1-x-\frac{x}{\mathrm{e}^{x}-1} = -\frac{1}{2}x - \sum_{n=2}^\infty \frac{B_n}{n!} x^n~,\\
 &\Rightarrow\quad \mathrm{tr}\mathcal{M}^n&&=-\frac{B_n}{n!}\qquad \mathrm{for}~n\geqslant2~,
\end{aligned}
\eeq
where $B_n$ are the Bernoulli numbers.
\end{proof}

\begin{proof}[Proof of Proposition~\ref{BB_disk_part.funct.}]
We have to evaluate the diagrams of the kind depicted in figure~\ref{BB_disk_diagrams}.
The amplitude~$I_n$ for a tree rooted on a $\mathbb{B}$~boundary field and ending on~$\alpha$ is:
\beq
 I_n=(-1)^{n+1} \int_{I^{\times (n+1)}} \langle \mathbb{B}_{(1)}(\tau_0), \eta_I(\tau_0,\tau_1)\cdots \eta_I(\tau_{n-1},\tau_n) \,
  \eta_I(\tau_n,1)\, \mathrm{ad}^n_{\mathsf{a}_1} \alpha \rangle \mathrm{d}\tau_0 \cdots \mathrm{d}\tau_n ~,
\eeq
where $\eta_I(\tau,\tau')= \Theta(\tau-\tau')-\tau$.
The result of this integral can be expressed in terms of the Bernoulli polynomials:
\beq\label{tree_ampl}
 I_n=(-1)^n \int_{I} \big\langle \mathbb{B}_{(1)}(\tau), \frac{B_{n+1}(\tau) - B_{n+1}}{(n+1)!}\, \mathrm{ad}^n_{\mathsf{a}_1} \alpha \big\rangle~.
\eeq
To prove~\eqref{tree_ampl}, let's define the operator $Kg(\tau):= \int_I \eta_I (\tau,\tau') g(\tau') \mathrm{d}\tau'$ and the generating function
\beq
 f(x;\tau) := \sum_{n=0}^\infty x^n K^n(t)~.
\eeq
This function satisfies the differential equation:
\beq
 \frac{\partial}{\partial \tau} f(x;\tau) = x f(x;\tau) + 1- \int_I f(x;\tau') \mathrm{d}\tau' = x f(x;\tau) + C(x)~,
\eeq
where $C(x)$ does not depend on $\tau$\,.
Since only the term $n=0$ contributes to~$f$ evaluated on the endpoints $\tau=0,1$, $f$ satisfies $f(x;0)=0$, $f(x;1)=1$.
Solving the differential equation with this boundary conditions we get:
\beq
 f(x;\tau) = \frac{1-\mathrm{e}^{x\tau}}{1-\mathrm{e}^{x}} = \frac1x \big( \frac{x}{1-\mathrm{e}^{x}} 
  - \frac{x\mathrm{e}^{x\tau}}{1-\mathrm{e}^{x}} \big) = \sum_{n=0}^\infty \frac{B_{n+1}(\tau) - B_{n+1}}{(n+1)!} x^n ~.
\eeq

Since $K(1)=0$, we have $K^n(\eta_I(\tau;1))= - K^n(\eta_I(\tau;0))$.
Thus, similar contributions to~\ref{tree_ampl} come from trees ending on~$\widetilde{\alpha}$ (the main difference being in the term for $n=0$) or rooted on $\widetilde{\mathbb{B}}$ or on $\mathsf{b}$.
By summing over~$n$ we get, for the tree part of the effective action:
\beq
 \mathcal{S}^{\mathrm{eff.}}_{\mathrm{tree}} = \int_0^1 \langle \mathbb{B}_{(1)}(\tau) - \widetilde{\mathbb{B}}_{(1)}(\tau), 
  \mathsf{G}_+(\tau, \mathrm{ad}_{\mathsf{a}_1}) \alpha + \mathsf{G}_-(\tau, \mathrm{ad}_{\mathsf{a}_1}) \widetilde{\alpha} \rangle \mathrm{d}\tau + \\
 + \langle \mathsf{b}^1,  \mathsf{F}_+(\mathrm{ad}_{\mathsf{a}_1}) \alpha + \mathsf{F}_-(\mathrm{ad}_{\mathsf{a}_1}) \widetilde{\alpha} \rangle ~.
\eeq

The amplitude for a wheel diagram is:
\beq
 -\frac{\mathrm{i}\hbar}{n}\mathrm{tr}(\mathrm{ad}_{\mathsf{a}_1}^n) \int_{I^n} \mathrm{d}\tau_1\cdots\mathrm{d}\tau_n\; 
   \eta_{I}(\tau_1;\tau_2)\cdots\eta_{I}(\tau_{n-1};\tau_{n}) \eta_{I}(\tau_n;\tau_1) ~.
\eeq
This integral is the same as the one appearing in~\cite{mnev:discrete_BF} and can be computed with the technique used to prove equation~\ref{BB_loop_amp}. 
The result for the loop contribution to the effective action is thus:
\beq
 \mathcal{S}^{\mathrm{eff.}}_{\mathrm{loop}} =- \mathrm{i}\hbar \sum_{n\geqslant 2} \frac{1}{n!} \mathrm{tr}(\mathrm{ad}_{\mathsf{a}_1}^n) \frac{B_n}{n} 
  =  - \mathrm{i}\hbar~ \mathrm{tr} \bigg( \log \bigg( \frac{\mathrm{sinh} (\mathrm{ad}_{\mathsf{a}_1}/2)}{\mathrm{ad}_{\mathsf{a}_1}/2} \bigg)\bigg)~.
\eeq
\end{proof}

\section{Proof of Proposition~\ref{prop: Omega^2=0}}\label{Omega^2=0}

Here we present a direct computational proof that $\Omega^2=0$ for a stratified circle, with any choice of polarizations on the edges and corners.

First note that edge contributions $\Omega^\mathbb{A}_I$, $\Omega^\mathbb{B}_I$ and pure corner contributions $\Omega^\alpha_p$, $\Omega^\beta_p$ all square to zero. 
Also, edge contributions and pure corner contributions commute. 
In particular, we have
\begin{equation}\label{Omega^2 corner contributions}
 \Omega^2=\sum_{k} \big(\Omega^{ \mathbb{P}_{k}\xi_k \mathbb{P}_{k+1}}_{p_k}\big)^2 
  + \big[\Omega^{ \mathbb{P}_{k}\xi_k \mathbb{P}_{k+1}}_{p_k}, \Omega_{I_k}^{\mathbb{P}_k}+\Omega_{I_{k+1}}^{\mathbb{P}_{k+1}}\big]~.
\end{equation}

Denote $\mathrm{BCH}(x,y)=\log(\mathrm{e}^x \mathrm{e}^y)$ for $x,y\in \mathfrak{g}$ -- the Baker-Campbell-Hausdorff group law. 
We will need the following identities
\begin{equation}\label{BCH identities}
 \mathrm{BCH}(x,y)=x+\mathsf{F}_+(\mathrm{ad}_x)y+\mathcal{O}(y^2)~ , \qquad
  \mathrm{BCH}(x,y)=y-\mathsf{F}_-(\mathrm{ad}_y)x+\mathcal{O}(x^2)
\end{equation}
which are the cases of the BCH formula when either first or second argument is infinitesimal.

Let us study e.g. a $\beta$-corner $p$ surrounded by a $\mathbb{B}$-edge $I'$ on the left and an $\mathbb{A}$-edge $I$ on the right. 
We have $\Omega^{\mathbb{B}\beta\mathbb{A}}_p=\Omega^{\beta\mathbb{A}}_p$ given by~\eqref{HE_contributions_to_Omega}. 
Applying this operator to a wavefunction of form $\psi(\beta)=\mathrm{e}^{-\frac{\mathrm{i}}{\hbar}\langle \beta,x \rangle}$, with $x\in\mathfrak{g}$ a parameter, yields
\begin{equation}
 \Omega^{\beta\mathbb{A}}_p\psi = \langle \beta, \mathsf{F}_+(\mathrm{ad}_x)\mathbb{A}_p \rangle\,\psi 
  = \mathrm{i}\hbar\frac{\mathrm{d}}{\mathrm{d}\epsilon}\, \mathrm{e}^{-\frac{\mathrm{i}}{\hbar}\langle \beta,\mathrm{BCH}(x,\epsilon\, \mathbb{A}_p) \rangle}~,
\end{equation}
with $\epsilon$ an odd parameter. 
Here we have used the first identity in~\eqref{BCH identities}. 
Note that similarly one can write $\Omega^{\mathbb{A}\beta}_p\psi = -\mathrm{i}\hbar \frac{\mathrm{d}}{\mathrm{d}\epsilon} \mathrm{e}^{-\frac{\mathrm{i}}{\hbar}\langle \beta ,\mathrm{BCH}(\epsilon\, \mathbb{A}_p,x) \rangle}$, using the second identity in~\eqref{BCH identities}.
This implies
\beq\label{Omega^betaA squared BCH computation}
 \big(\Omega^{\beta\mathbb{A}}_p\big)^2\psi
  &= (\mathrm{i}\hbar)^2\frac{\mathrm{d}}{\mathrm{d}\epsilon_2}\, \frac{\mathrm{d}}{\mathrm{d}\epsilon_1}\, \mathrm{e}^{-\frac{\mathrm{i}}{\hbar}\langle \beta,
   \mathrm{BCH}(\mathrm{BCH}(x,\epsilon_1\, \mathbb{A}_p),\epsilon_2\, \mathbb{A}_p) \rangle} \\
 &= (\mathrm{i}\hbar)^2\frac{\mathrm{d}}{\mathrm{d}\epsilon_2}\, \frac{\mathrm{d}}{\mathrm{d}\epsilon_1}\, \mathrm{e}^{-\frac{\mathrm{i}}{\hbar}\langle \beta,
  \mathrm{BCH}(x,\mathrm{BCH}(\epsilon_1\, \mathbb{A}_p,\epsilon_2\, \mathbb{A}_p)) \rangle} \\
 &= (\mathrm{i}\hbar)^2\frac{\mathrm{d}}{\mathrm{d}\epsilon_2}\, \frac{\mathrm{d}}{\mathrm{d}\epsilon_1}\, \mathrm{e}^{-\frac{\mathrm{i}}{\hbar}\langle \beta,
  \mathrm{BCH}(x,(\epsilon_1+\epsilon_2)\, \mathbb{A}_p-\frac12 \epsilon_1\epsilon_2\, [\mathbb{A}_p,\mathbb{A}_p]) \rangle}\\
 &= (\mathrm{i}\hbar)^2\frac{\mathrm{d}}{\mathrm{d}\epsilon_2}\, \frac{\mathrm{d}}{\mathrm{d}\epsilon_1}\, \mathrm{e}^{-\frac{\mathrm{i}}{\hbar}\langle \beta,
    \mathrm{BCH}(x,-\frac12 \epsilon_1\epsilon_2\, [\mathbb{A}_p,\mathbb{A}_p]) \rangle} \\
 &= -\frac{\mathrm{i}\hbar}{2}\Big\langle  [\mathbb{A}_p,\mathbb{A}_p] , \frac{\partial}{\partial \mathbb{A}_p}\Big\rangle \Omega^{\beta \mathbb{A}}_p \psi 
  = - [\Omega_I^\mathbb{A},\Omega^{\beta \mathbb{A}}_p]\psi ~.
\eeq
Note that the main trick of this computation is the use of associativity of the BCH formula. 
Operators $\big(\Omega^{\beta\mathbb{A}}_p\big)^2$, $[\Omega_I^\mathbb{A},\Omega^{\beta \mathbb{A}}_p]$ are multiplication operators in the variable $\mathbb{A}_p$, thus the computation above, for $\psi$ independent of $\mathbb{A}_p$ is sufficient to ascertain that $\big(\Omega^{\beta\mathbb{A}}_p\big)^2 + [\Omega^{\beta \mathbb{A}}_p, \Omega_I^\mathbb{A}]=0$ as operators. 
Further, note that $[\Omega^{\beta\mathbb{A}}_p, \Omega_{I'}^\mathbb{B}]$ contains derivatives in $\mathbb{B}^{(1)}_p$ and therefore vanishes on \emph{admissible} states, in the sense of Assumption~\ref{assump: admissible states}. 
This proves that the contribution of a $\mathbb{B}\beta\mathbb{A}$ corner to $\Omega^2$ (cf.~the right hand side of~\eqref{Omega^2 corner contributions}) vanishes. 
The case $\mathbb{A}\beta\mathbb{B}$ is an orientation reversal of the case we just studied; it is treated analogously and also yields a zero contribution to the r.h.s. of~\eqref{Omega^2 corner contributions}.

Case of an $\mathbb{A}\beta\mathbb{A}$ corner is treated similarly. 
Here $\Omega^{\mathbb{A}\beta\mathbb{A}}_p=\Omega^{\mathbb{A}\beta}_p+ \Omega^{\beta\mathbb{A}}_p$. 
We have $\big(\Omega^{\beta\mathbb{A}}_p\big)^2+[\Omega^{\beta \mathbb{A}}_p, \Omega_I^\mathbb{A}]=0$ as above, and similarly $\big(\Omega^{\mathbb{A}\beta}_p\big)^2+[\Omega^{\mathbb{A}\beta}_p, \Omega_{I'}^\mathbb{A}]=0$. 
We also need to understand the term $[\Omega^{\mathbb{A}\beta}_p,\Omega^{\beta \mathbb{A}}_p]$, which is done similarly to~\eqref{Omega^betaA squared BCH computation}:
\beq
 \Omega^{\mathbb{A}\beta}_p \Omega^{\beta \mathbb{A}}_p \psi 
  &= -(\mathrm{i}\hbar)^2 \frac{\mathrm{d}}{\mathrm{d}\epsilon_2} \frac{\mathrm{d}}{\mathrm{d}\epsilon_1} 
   \mathrm{e}^{-\frac{\mathrm{i}}{\hbar}\langle \beta , \mathrm{BCH}(\epsilon_2\, \mathbb{A}_{p-0} , \mathrm{BCH}(x, \epsilon_1\, \mathbb{A}_{p+0}) ) \rangle}\\
 &= -(\mathrm{i}\hbar)^2 \frac{\mathrm{d}}{\mathrm{d}\epsilon_2} \frac{\mathrm{d}}{\mathrm{d}\epsilon_1} 
  \mathrm{e}^{-\frac{\mathrm{i}}{\hbar}\langle \beta , \mathrm{BCH}( \mathrm{BCH}(\epsilon_2\, \mathbb{A}_{p-0} ,x), \epsilon_1\, \mathbb{A}_{p+0}) ) \rangle}
   = - \Omega^{\beta \mathbb{A}}_p \Omega^{\mathbb{A}\beta}_p \psi ~.
\eeq
Hence, $[\Omega^{\mathbb{A}\beta}_p,\Omega^{\beta \mathbb{A}}_p]=0$ and the contribution of an $\mathbb{A}\beta\mathbb{A}$ corner to the r.h.s. of~\eqref{Omega^2 corner contributions} vanishes.

In the case of an $\mathbb{A}\alpha\mathbb{B}$ corner, we have $\Omega^{\mathbb{A}\alpha\mathbb{B}}_p=\Omega^\alpha_p+\Omega^{\alpha\mathbb{B}}_p$. 
By a computation similar to~\eqref{Omega^betaA squared BCH computation}, one shows that $(\big(\Omega^{\alpha\mathbb{B}}_p\big)^2+ [\Omega^\alpha_p,\Omega^{\alpha\mathbb{B}}_p])\psi=0$ for $\psi=\mathrm{e}^{-\frac{\mathrm{i}}{\hbar}\langle \mathbb{B}_p,x \rangle}$. 
Together with $\big(\Omega^\alpha_p\big)^2=0$, this shows that $\big(\Omega^{\mathbb{A}\alpha\mathbb{B}}_p \big)^2=0$. 
Furthermore, $[\Omega^{\mathbb{A}\alpha\mathbb{B}}_p ,\Omega^\mathbb{A}_{I'}]=0$ and $[\Omega^{\mathbb{A}\alpha\mathbb{B}}_p ,\Omega^\mathbb{B}_{I}]=0$ on admissible states. 
Thus, the contribution an $\mathbb{A}\alpha\mathbb{B}$ to the r.h.s. of~\eqref{Omega^2 corner contributions} also vanishes. 
Orientation-reversed case $\mathbb{B}\alpha\mathbb{A}$ is similar.

Case of a $\mathbb{B}\alpha\mathbb{B}$ corner is similar to the above: we have  $\Omega^{\mathbb{B}\alpha\mathbb{B}}_p =\Omega^{\mathbb{B}\alpha}_p +\Omega^\alpha_p + \Omega^{\alpha\mathbb{B}}_p$. 
As above, we have $\big(\Omega^{\alpha\mathbb{B}}_p\big)^2+ [\Omega^\alpha_p,\Omega^{\alpha\mathbb{B}}_p]=0$ and similarly $\big(\Omega^{\mathbb{B}\alpha}_p\big)^2 + [\Omega^\alpha_p,\Omega^{\mathbb{B}\alpha}_p]=0$. 
One also trivially has $[\Omega^{\mathbb{B}\alpha}_p,\Omega^{\alpha\mathbb{B}}_p]=0$. 
Thus, $\big(\Omega^{\mathbb{B}\alpha\mathbb{B}}_p \big)^2=0$. 
Also, the corner contribution to $\Omega$ commutes with the edge terms on admissible states. 
This proves that the contribution of a $\mathbb{B}\alpha\mathbb{B}$ corner to the r.h.s. of~\eqref{Omega^2 corner contributions} vanishes, too.

Cases of $\mathbb{A}\alpha\mathbb{A}$ and  $\mathbb{B}\beta\mathbb{B}$ corners are trivial.
This finishes the proof that all terms in the sum~\eqref{Omega^2 corner contributions} over the corners vanish, and thus $\Omega^2=0$ for an arbitrarily stratified and polarized circle.

\section{A check of the chain map property of the inclusion of the small model for $\mathbb{A}$-states on an interval into the full model}\label{app:chain}

One can check directly that~\eqref{i_H on A beta circle} is indeed a chain map. 
First, it is clearly an algebra morphism (w.r.t. the standard supercommutative pointwise product on $\mathrm{Fun}(\cdots)$), so it is enough to check the chain map property on a set of generators of $\mathcal{H}^\mathrm{small}$. 
Assume for simplicity that $\mathfrak{g}\subset \mathrm{Mat}_N$ is a matrix Lie algebra and choose as generators 
\beq
 f_{k,\rho}:=\mathrm{tr}\, \rho\, \underline{\mathbb{A}}_{k} ~,\qquad g_{\rho}:=\mathrm{tr}\,\rho\, \mathrm{e}^{\underline{\mathbb{A}}} ~,
\eeq
with $\rho\in \mathrm{Mat}_N$ arbitrary parameter and $k=0,1$. 
From~\eqref{Omega small} and~\eqref{i_H on A beta circle}, we immediately obtain that $i_\mathcal{H}\circ \Omega^\mathrm{small}= \Omega\circ i_\mathcal{H} $ when applied to the generators $f_{k,\rho}$.
For $g_{\rho}$, from~(\ref{Omega small},\ref{i_H on A beta circle}) and from the rule for the deformation of holonomy under an infinitesimal gauge transformation, we obtain:
\beq
 (i_\mathcal{H}\circ \Omega^\mathrm{small}) g_{\rho} = (\Omega\circ i_\mathcal{H}) g_{\rho} 
  = -\mathrm{i}\hbar\, \mathrm{tr}\, \rho \,\big(U(\mathbb{A})\cdot  \mathbb{A}_1 - \mathbb{A}_0\cdot U(\mathbb{A}) \big)~.
\eeq
Here we used the observation that $\Omega^\mathrm{small}g_{\rho}=-\mathrm{i}\hbar \, \mathrm{tr}\, \rho ( \mathrm{e}^X c_1 - c_0 \mathrm{e}^X )$ with shorthand notation $c_0=\underline{\mathbb{A}}_{0}$, $c_1=\underline{\mathbb{A}}_{1}$, $X=\underline{\mathbb{A}}$. 
Indeed, we have 
\beq
 &\frac{\mathrm{i}}{\hbar}\Omega^\mathrm{small}g_{\rho} =\mathrm{tr}\, \rho \sum_{p,q\geq 0} \frac{1}{(p+q+1)!}X^p(\mathsf{F}_+(\mathrm{ad}_X)c_1 
  + \mathsf{F}_-(\mathrm{ad}_X)c_0)X^q \\
 &= \mathrm{tr}\, \rho \sum_{p,q,j,l\geq 0}\frac{1}{(p+q+1)!} \bigg(\frac{(-1)^l B^+_{j+l}}{j!l!} X^{j+p} c_1 X^{l+q} 
  - \frac{(-1)^l B^-_{j+l}}{j!l!} X^{j+p}c_0 X^{l+q}\bigg)~,
\eeq
where $B^+_i$ and $B^-_i$ are the Taylor coefficients of $\mathsf{F}_+(x)$ and $-\mathsf{F}_-(x)$, respectively. 
Note that, for~$x,y$ scalars, we have
\beq
 &\sum_{p,q,j,l\geq 0}\frac{1}{(p+q+1)!}\frac{(-1)^l B^+_{j+l}}{j! l!}x^{j+p}y^{l+q}
  =\Bigg(\sum_{p,q\geq 0}\frac{x^py^q}{(p+q+1)!} \Bigg) \Bigg(\sum_{j,l\geq 0}\frac{(-1)^l B^+_{j+l}}{j!l!}x^jy^l\Bigg) \\
 &=\frac{\mathrm{e}^x-\mathrm{e}^y}{x-y}\cdot \mathsf{F}_+(x-y)=\mathrm{e}^x
\eeq
and, similarly, $\sum_{p,q,j,l\geq 0}\frac{1}{(p+q+1)!}\frac{(-1)^l B^-_{j+l}}{j! l!}x^{j+p}y^{l+q}=\mathrm{e}^y$.
Thus:
\beq
 \frac{\mathrm{i}}{\hbar}\Omega^\mathrm{small}g_{\rho}=\mathrm{tr}\,\rho \big(\mathrm{e}^X c_1 - c_0 \mathrm{e}^X\big)
\eeq
as claimed.

\newpage
\bibliography{Bibliography}

\providecommand{\href}[2]{#2}\begingroup\raggedright\begin{thebibliography}{10}

\bibitem{dolan:higher_dim_alg}
J.~Baez and J.~Dolan, ``\emph{Higher-dimensional algebra and topological
  quantum field theory}'', {\emph{Journal of Mathematical Physics} {\bfseries
  36} (1998) } [\href{https://arxiv.org/abs/q-alg/9503002}{{\ttfamily
  q-alg/9503002}}].

\bibitem{bandiera:discretize}
R.~Bandiera and F.~Schaetz, ``\emph{How to discretize the differential forms on
  the interval}'',  \href{https://arxiv.org/abs/1607.03654}{{\ttfamily
  arXiv:1607.03654}}, 2016.

\bibitem{BV:1}
I.~A. Batalin and G.~A. Vilkovisky, ``\emph{Gauge algebra and quantization}'',
  \href{https://doi.org/10.1016/0370-2693(81)90205-7}{\emph{Physics Letters
  B102} (1981) 27}.

\bibitem{bonechi:PSM_on_closed.}
F.~Bonechi, A.~S. Cattaneo and P.~Mnev, ``\emph{The {Poisson} sigma model on
  closed surfaces}'', {\emph{Journal of High Energy Physics, Volume 2012}
  (2012) } [\href{https://arxiv.org/abs/1110.4850}{{\ttfamily
  arXiv:1110.4850}}].

\bibitem{cattaneo:private}
A.~S. Cattaneo{\emph{,~Private communications$\!\!$} }.

\bibitem{cattaneo:kontsevich_quant.}
A.~S. Cattaneo and G.~Felder, ``\emph{A path integral approach to the
  {Kontsevich} quantization formula}'', {\emph{Commun.Math.Phys. 212} (2000)
  591} [\href{https://arxiv.org/abs/math/9902090}{{\ttfamily math/9902090}}].

\bibitem{CF:coisotropic}
A.~S. Cattaneo and G.~Felder, ``\emph{Coisotropic submanifolds in {P}oisson
  geometry and branes in the {P}oisson sigma model}'', {\emph{Letters in
  Mathematical Physics} {\bfseries 69} (2003) }
  [\href{https://arxiv.org/abs/math/0309180}{{\ttfamily math/0309180}}].

\bibitem{cm:remarks_on_CS}
A.~S. Cattaneo and P.~Mnev, ``\emph{Remarks on {C}hern–{S}imons
  invariants}'', {\emph{Communications in Mathematical Physics} {\bfseries 293}
  (2008) } [\href{https://arxiv.org/abs/0811.2045}{{\ttfamily
  arXiv:0811.2045}}].

\bibitem{CMR:classical_quantum_BV}
A.~S. Cattaneo, P.~Mnev and N.~Reshetikhin, ``\emph{Classical and quantum
  lagrangian field theories with boundary}'', {\emph{Proceedings of Science}
  (2012) } [\href{https://arxiv.org/abs/1207.0239}{{\ttfamily
  arXiv:1207.0239}}].

\bibitem{CMR:classical_BV}
A.~S. Cattaneo, P.~Mnev and N.~Reshetikhin, ``\emph{Classical {BV} theories on
  manifolds with boundary}'', {\emph{Communications in Mathematical Physics}
  {\bfseries 332} (2014) 535}
  [\href{https://arxiv.org/abs/1201.0290}{{\ttfamily arXiv:1201.0290}}].

\bibitem{CMR:pert.Segal}
A.~S. Cattaneo, P.~Mnev and N.~Reshetikhin, ``\emph{Perturbative {BV} theories
  with {S}egal-like gluing}'',
  \href{https://arxiv.org/abs/1602.00741}{{\ttfamily arXiv:1602.00741}}, 2016.

\bibitem{CMR:cellular}
A.~S. Cattaneo, P.~Mnev and N.~Reshetikhin, ``\emph{A cellular topological
  field theory}'',  \href{https://arxiv.org/abs/1701.05874}{{\ttfamily
  arXiv:1701.05874}}, 2017.

\bibitem{CMR:pert_quantum_BV}
A.~S. Cattaneo, P.~Mnev and N.~Reshetikhin, ``\emph{Perturbative quantum gauge
  theories on manifolds with boundary}'', {\emph{Communications in Mathematical
  Physics} {\bfseries 357} (2018) 631}
  [\href{https://arxiv.org/abs/1507.01221}{{\ttfamily arXiv:1507.01221}}].

\bibitem{cordes:lectures_2d_YM}
S.~Cordes, G.~Moore and S.~Ramgoolam, ``\emph{Lectures on 2d {Yang-Mills}
  theory, equivariant cohomology and topological field theories}'',
  \href{https://arxiv.org/abs/hep-th/9411210}{{\ttfamily hep-th/9411210}},
  1994.

\bibitem{costello:BV_renorm.}
K.~J. Costello, ``\emph{Renormalisation and the {B}atalin-{V}ilkovisky
  formalism}'',  \href{https://arxiv.org/abs/0706.1533}{{\ttfamily
  arXiv:0706.1533}}, 2007.

\bibitem{Gugenheim:hpt}
V.~K. A.~M. Gugenheim and L.~A. Lambe, ``\emph{Perturbation theory in
  differential homological algebra {I}}'', {\emph{Illinois J. Math. 33} (1989)
  566}.

\bibitem{teitelboim:quantization}
M.~Henneaux and C.~Teitelboim, \emph{Quantization of Gauge Systems}. Princeton
  University Press, 1994.

\bibitem{konstevich:def_quant}
M.~Kontsevich, ``\emph{Deformation quantization of {P}oisson manifolds}'',
  {\emph{Lett. Math. Phys.} {\bfseries 66} (2003) 157}
  [\href{https://arxiv.org/abs/q-alg/9709040}{{\ttfamily q-alg/9709040}}].

\bibitem{lawrence-sullivan:deformation}
R.~Lawrence and D.~Sullivan, ``\emph{A formula for topology/deformations and
  its significance}'', {\emph{Fundamenta Mathematicae} {\bfseries 225} (2014) }
  [\href{https://arxiv.org/abs/math/0610949}{{\ttfamily math/0610949}}].

\bibitem{lurie:TFT_classification}
J.~Lurie, ``\emph{On the classification of topological field theories}'',
  {\emph{Current Developments in Mathematics} (2009) }
  [\href{https://arxiv.org/abs/0905.0465}{{\ttfamily arXiv:0905.0465}}].

\bibitem{Migdal:1975zg}
A.~A. Migdal, ``\emph{Recursion equations in gauge theories}'', {\emph{Sov.
  Phys. JETP} {\bfseries 42} (1975) 413}.

\bibitem{mnev:simplicial_BF}
P.~Mnev, ``\emph{Notes on simplicial {BF} theory}'', {\emph{Moscow Math. J.}
  {\bfseries 9} (2006) }
  [\href{https://arxiv.org/abs/hep-th/0610326}{{\ttfamily hep-th/0610326}}].

\bibitem{mnev:discrete_BF}
P.~Mnev, ``\emph{Discrete {BF} theory}'',
  \href{https://arxiv.org/abs/0809.1160}{{\ttfamily arXiv:0809.1160}}, 2008.

\bibitem{oekl:2d_YM}
R.~Oeckl, ``\emph{Two-dimensional quantum {Y}ang–{M}ills theory with
  corners}'', {\emph{Journal of Physics A: Mathematical and Theoretical}
  {\bfseries 41} (2008) 135401}
  [\href{https://arxiv.org/abs/hep-th/0608218}{{\ttfamily hep-th/0608218}}].

\bibitem{wernli:thesis}
K.~Wernli, \emph{Perturbative Quantization of Split {C}hern-{S}imons Theory on
  Handlebodies and Lens Spaces by the {BV-BFV} formalism}. Doctoral thesis,
  2018.

\bibitem{witten:2d_quantum_gauge}
E.~Witten, ``\emph{On quantum gauge theories in two dimensions}'',
  \href{https://doi.org/10.1007/BF02100009}{\emph{Commun. Math. Phys. 141}
  (1991) 153}.

\bibitem{cheng:transferring_homotopy}
X.~Zhi~Cheng and E.~Getzler, ``\emph{Transferring homotopy commutative
  algebraic structures}'', {\emph{Journal of Pure and Applied Algebra}
  {\bfseries 212} (2006) }
  [\href{https://arxiv.org/abs/math/0610912}{{\ttfamily math/0610912}}].

\end{thebibliography}\endgroup
\end{document}